\documentclass[11pt,letterpaper]{article}

\usepackage{tk}
\usepackage{hyperref}
\usepackage{scalerel}
\usepackage{tikz}
\usepackage{algpseudocode}
\usepackage{relsize}
\usepackage{arydshln}

\hypersetup{
  linkcolor  = mylinkcolor,
  citecolor  = mycitecolor,
  urlcolor   = myurlcolor,
  colorlinks = true,
}

\newcommand{\even}{\mathsf{even}} 
\newcommand{\odd}{\mathsf{odd}}

\newcommand{\M}{\mathsf{M}}

\newcommand{\tout}{\mathsf{out}}

\newcommand{\tens}{\mathsf{tens}}
\newcommand{\conn}{\mathsf{conn}}
\newcommand{\tree}{\mathsf{tree}}

\newcommand{\pow}{\mathsf{pow}}
\newcommand{\offdiag}{\mathsf{offdiag}}
\newcommand{\EvenPart}{\mathsf{EvenPart}}
\newcommand{\err}{\mathsf{err}}
\newcommand{\corr}{\mathsf{corr}}
\newcommand{\main}{\mathsf{main}}
\newcommand{\tie}{\mathsf{tie}}
\newcommand{\tied}{\mathsf{tied}}

\renewcommand{\root}{\mathsf{root}}
\newcommand{\set}{\mathsf{set}}

\newcommand{\overbar}[1]{\mkern 4mu\overline{\mkern-4mu#1}}
\newcommand{\sFbar}{\overbar{\sF}}

\newcommand{\Plans}{\mathsf{Plans}}
\newcommand{\MaxSpan}{\mathsf{MaxSpan}}
\newcommand{\pat}{\mathsf{pat}}
\newcommand{\Pat}{\mathsf{Pat}}
\newcommand{\pairs}{\mathsf{pairs}}
\newcommand{\id}{\mathsf{id}}
\newcommand{\freq}{\mathsf{freq}}
\newcommand{\coll}{\mathsf{coll}}
\newcommand{\ind}{\mathsf{ind}}
\newcommand{\isovec}{\mathsf{isovec}}

\newcommand{\SOS}{\mathsf{SOS}}

\newcommand{\tEE}{\widetilde{\EE}}

\renewcommand{\epsilon}{\varepsilon}

\title{Positivity-preserving extensions of sum-of-squares pseudomoments over the hypercube}
\author{Dmitriy Kunisky\thanks{Email: \textit{kunisky@cims.nyu.edu}. Partially supported by NSF grants DMS-1712730 and DMS-1719545.}}
\affil{Department of Mathematics, Courant Institute of Mathematical Sciences, New York University}
\date{September 15, 2020}

\begin{document}

\pagenumbering{gobble}

\maketitle

\begin{abstract}
    We introduce a new method for building higher-degree sum-of-squares (SOS) lower bounds over the hypercube $\bx \in \{\pm 1\}^N$ from a given degree 2 lower bound.
    Our method constructs pseudoexpectations that are positive semidefinite by design, lightening some of the technical challenges common to other approaches to SOS lower bounds, such as pseudocalibration.
    The construction is based on a ``surrogate'' random symmetric tensor that plays the role of $\bx^{\otimes d}$, formed by conditioning a natural gaussian tensor distribution on consequences of both the hypercube constraints and the spectral structure of the degree 2 pseudomoments.

    We give general ``incoherence'' conditions under which degree 2 pseudomoments can be extended to higher degrees.
    As an application, we extend previous lower bounds for the Sherrington-Kirkpatrick Hamiltonian from degree 4 to degree 6.
    (This is subsumed, however, in the stronger results of the parallel work \cite{GJJPR-2020-SK}.)
    This amounts to extending degree 2 pseudomoments given by a random low-rank projection matrix.
    As evidence in favor of our construction for higher degrees, we also show that random high-rank projection matrices (an easier case) can be extended to degree $\omega(1)$.
    We identify the main obstacle to achieving the same in the low-rank case, and conjecture that while our construction remains correct to leading order, it also requires a next-order adjustment.

    Our technical argument involves the interplay of two ideas of independent interest.
    First, our pseudomoment matrix factorizes in terms of \emph{multiharmonic polynomials} associated with the degree 2 pseudomoments being extended.
    This observation guides our proof of positivity.
    Second, our pseudomoment values are described graphically by sums over forests, with coefficients given by the \emph{\Mobius\ function} of a partial ordering of those forests.
    This connection with inclusion-exclusion combinatorics guides our proof that the pseudomoments satisfy the hypercube constraints.
    We trace the reason that our pseudomoments can satisfy both the hypercube and positivity constraints simultaneously to a remarkable combinatorial relationship between multiharmonic polynomials and this \Mobius\ function.
\end{abstract}

\clearpage

\setcounter{tocdepth}{2}
\tableofcontents

\clearpage

\pagenumbering{arabic}

\section{Introduction}



The problem of \emph{certifying bounds} on optimization problems or \emph{refuting feasibility} of constraint satisfaction problems, especially on random instances, has received much attention in the computer science literature.
In certification, rather than searching for a single high-quality solution to a problem, an algorithm must produce an easily verifiable \emph{proof} of a bound on the quality of \emph{all} feasible solutions.
Determining the computational cost of certification, in particular as compared to that of search, is a fundamental problem in the analysis of algorithms.

The sum-of-squares (SOS) hierarchy of semidefinite programming (SDP) convex relaxations is a powerful family of algorithms that gives a unified way to certify bounds on polynomial optimization problems \cite{Shor-1987-SumOfSquares, Lasserre-2001-GlobalOptimizationMoments,Parrilo-2003-SDPSemialgebraic,Laurent-2009-SOS}.
For many problems, both for worst-case instances and in the average case where instances are drawn at random, SOS relaxations enjoy the best known performance among certification algorithms; often, rounding techniques post-processing the output of SOS also give optimal search algorithms \cite{BBHKSZ-2012-Hypercontractivity,BS-2014-SOSQuest,BKS-2014-SOSRounding,BKS-2015-TensorDecompositionSOS,HSS-2015-TensorPCA,Hopkins-2018-Thesis}.
Moreover, a remarkable general theory has emerged recently showing that SOS relaxations are, in a suitable technical sense, optimal among all efficient SDP relaxations for various problems \cite{LRST-2014-SymmetricSDP,LRS-2015-LowerBoundsSDP}.
Conversely, in light of this apparent power, \emph{lower bounds} against the SOS hierarchy are an exceptionally strong form of evidence for the difficulty of efficiently certifying bounds on a problem \cite{MW-2015-SOSSparsePCA,BCK-2015-SOSPairwiseIndependence,KMOW-2017-SOSCSP,BHKKMP-2019-PlantedClique}.
Especially in the average case setting, where lower bounds against arbitrary efficient algorithms are far out of reach of current techniques (even with standard complexity-theoretic assumptions like $\mathsf{P} \neq \mathsf{NP}$), lower bounds against SOS have emerged as an important standard of computational complexity.

It is therefore valuable to identify general techniques for proving SOS lower bounds, which requires constructing fictitious \emph{pseudosolutions} to the underlying problem that can ``fool'' the SOS certifier.
We know of only one broadly applicable method for doing this in random problems, a technique called \emph{pseudocalibration} introduced in the landmark paper \cite{BHKKMP-2019-PlantedClique} to prove lower bounds for the largest clique problem in random graphs.\footnote{Other works with different, more problem-specific approaches to SOS include \cite{BCK-2015-SOSPairwiseIndependence,MW-2015-SOSSparsePCA,KBG-2017-CommunityDetection}.}
Pseudocalibration uses the idea that certification performs an ``implicit hypothesis test.''
Namely, whenever it is possible to certify a bound over a distribution $\QQ$ of problem instances, it is also possible to distinguish between $\QQ$ and any variant $\PP$ where an unusually high-quality solution is ``planted'' in the problem instance.
If there exists such $\PP$ that appears difficult to distinguish from $\QQ$, then it should also be difficult to certify bounds under $\QQ$.\footnote{Some other recent works, including \cite{BKW-2019-ConstrainedPCA,BBKMW-2020-SpectralPlantingColoring} in which the author participated, have used this idea of \emph{computationally-quiet planting} to give indirect evidence that certification is hard without proving lower bounds against specific certification algorithms, by performing this reduction and then using other techniques to argue that testing between $\PP$ and $\QQ$ is hard.}
Taking such $\PP$ as input, pseudocalibration builds a pseudosolution that appears to the SOS certifier to mimic the solution that is planted under $\PP$ (in a technical sense involving both averaging over the instance distributions and restricting to low-degree polynomial statistics).
Beyond largest clique, pseudocalibration has since successfully yielded SOS lower bounds for several other problems \cite{HKPRSS-2017-SOSSpectral,MRX-2019-SOS4,BCR-2020-ExtendedFormulationCSP}, and has inspired other offshoot techniques including novel spectral methods \cite{HKPRSS-2017-SOSSpectral,RSS-2018-EstimationSOS}, the low-degree likelihood ratio analysis \cite{HS-2017-BayesianEstimation,HKPRSS-2017-SOSSpectral,Hopkins-2018-Thesis,KWB-2019-NotesLowDegree}, and the local statistics SDP hierarchy \cite{BMR-2019-LocalStatistics,BBKMW-2020-SpectralPlantingColoring}.

We suggest, however, that pseudocalibration suffers from two salient drawbacks.
The first is that it requires rather notoriously challenging technical analyses, mostly pertaining to the positive semidefiniteness of certain large structured matrices that it constructs.
The mechanics of the calculations involved in these arguments is a subject unto itself \cite{AMP-2016-GraphMatrices,CP-2020-ZGraphMatrices}, and it is natural to wonder whether there might be a more conceptual explanation for this positivity, which the pseudocalibration construction \emph{a priori} gives little reason to expect.
Second, for several certification problems, it seems to have been (or, in some cases, to remain) quite challenging to advance from lower bounds against small ``natural SDPs'' (such as the Goemans-Williamson SDP for maximum cut \cite{GW-1995-MaxCutApprox,MS-2016-SDPSparseRandomGraphs,DMOSS-2019-SDPNAE,MRX-2019-SOS4}, the \Lovasz\ $\vartheta$ function for graph coloring \cite{Lovasz-1979-ShannonCapacityGraph,CO-2003-LovaszRandomGraphs,BKM-2017-LovaszThetaRandomGraphs}, or the SDP suggested by \cite{MR-2015-NonNegative} for non-negative PCA), which usually coincide with degree~2 SOS, to even the next-smallest degrees of the hierarchy.\footnote{The SOS hierarchy is graded by an even positive integer called the \emph{degree}. The degree $2d$ SDP may be solved, for sufficiently well-behaved problems, in time $N^{O(d)}$ \cite{ODonnell-2017-SOSNotAutomatizable,RW-2017-BitComplexity}.}
Outside of some specially structured deterministic problems \cite{Grigoriev-2001-SOSKnapsack,Laurent-2003-CutPolytopeSOS}, the question of whether a degree 2 lower bound can in itself suggest an extension to higher degrees without deeper reasoning about planted distributions has not been thoroughly explored.

In this paper, we introduce an alternative framework to pseudocalibration for proving higher-degree SOS lower bounds that attempts to address the two points raised above, in the context of the particular problem of optimizing quadratic forms over the hypercube.
While incorporating some intuition gleaned from a planted distribution common to such problems, our technique does not require detailed analysis of its moments, and instead proceeds by building the simplest possible higher-degree object (in a suitable technical sense) that extends a given degree 2 feasible point.
We also build this extension to be positive semidefinite by construction, reasoning from the beginning in terms of a Gram or Cholesky factorization of the matrix involved.
This gives a novel and intuitive interpretation of the positive semidefiniteness discussed above, and appears to ease some of the technicalities typical of pseudocalibration.

\paragraph{Quadratic forms over the hypercube}
We will focus for the remainder of the paper on the specific problem of optimizing a quadratic form over the $\pm 1$ hypercube:
\begin{equation}
    \M(\bW) \colonequals \max_{\bx \in \{\pm 1\}^N} \bx^{\top}\bW\bx \text{ for } \bW \in \RR^{N \times N}_{\sym}.
\end{equation}
Perhaps the most important application of this problem, at least within combinatorial optimization, is computing the maximum cut in a graph, which arises when $\bW$ is a graph Laplacian.
Accordingly, by the classical result of Karp \cite{Karp-1972-Reducibility}, it is $\mathsf{NP}$-hard to compute $\M(\bW)$ in the worst case, making average-case settings especially interesting to gain a more nuanced picture of the computational complexity of this class of problems.

Such problems also admit a simple and convenient benchmark certification algorithm, the \emph{spectral certificate} formed by ignoring the special structure of $\bx$:
\begin{equation}
    \M(\bW) \leq \max_{\|\bx\| = \sqrt{N}}\bx^{\top}\bW\bx \leq N \lambda_{\max}(\bW).
\end{equation}
Though this approach to certification seems quite naive, several SOS lower bounds for specific problems (discussed below) as well as the general heuristic of \cite{BKW-2019-ConstrainedPCA} suggest that it is often optimal.
Thus a central question about certification for $\M(\bW)$ is whether, given a particular distribution of $\bW$, polynomial-time SOS relaxations can certify bounds that are typically tighter than the spectral bound.

More specifically, the following three distributions of $\bW$ have emerged as basic challenges for proving average-case lower bounds for certification.\footnote{Another notable, though more complicated, constraint satisfaction problem that fits into this framework is not-all-equal-3SAT, which corresponds to $\bW$ the graph Laplacian of a different, non-uniform distribution of sparse regular graphs \cite{DMOSS-2019-SDPNAE}.}
\begin{enumerate}
\item \textbf{Sherrington-Kirkpatrick (SK) Hamiltonian:} $\bW$ is drawn from the \emph{gaussian orthogonal ensemble}, $\bW \sim \GOE(N)$, meaning that $W_{ij} = W_{ji} \sim \sN(0, 1 / N)$ for $i < j$ and $W_{ii} \sim \sN(2 / N)$, independently for distinct index pairs.
\item \textbf{Sparse random regular graph Laplacian:} $\bW$ is the graph Laplacian of a random $\gamma$-regular graph $G$ on $N$ vertices for $\gamma$ held constant as $N \to \infty$, normalized so that when $\bx \in \{\pm 1\}^N$, then $\bx^{\top}\bW\bx$ counts the edges of $G$ crossing the cut given by the signs of $\bx$.
\item \textbf{Sparse \Erdos-\Renyi\ random graph Laplacian:} $\bW$ is the graph Laplacian of a random \Erdos-\Renyi\ graph on $N$ vertices with edge probability $\gamma / N$ (and therefore mean vertex degree $\gamma$) for $\gamma$ held constant as $N \to \infty$, normalized as above.
\end{enumerate}
The SK Hamiltonian has a long and remarkable history in the statistical physics of spin glasses, the bold conjectures of \cite{Parisi-1979-SK} on the asymptotic value of $N^{-1} \M(\bW) \approx 1.526$ inspiring a large body of mathematical work to justify them \cite{Guerra-2003-BrokenRSB,Talagrand-2006-Parisi,Panchenko-2013-SK}.
For our purposes, it provides a convenient testbed for the difficulty of certification, since a basic result of random matrix theory shows that $\lambda_{\max}(\bW) \approx 2$, giving a precise gap between the spectral certificate and the best possible certifiable value.
In fact, the recent result of Montanari \cite{Montanari-2019-SKOptimization} also showed that search algorithms succeed (up to small additive error and assuming a technical conjecture) in finding $\what{\bx} \in \{\pm 1\}^N$ with $\what{\bx}^{\top}\bW\what{\bx} \approx 1.526$, suggesting that, if the spectral certificate is optimal, then the same gap obtains between the two algorithmic tasks of search and certification.

The sparse random graph models, which are more natural problems for combinatorial optimization, are in fact also closely related to the SK Hamiltonian.
In the limit $\gamma \to \infty$, \cite{DMS-2017-CutsSparseRandomGraphs} showed that in both graph models the asymptotics of $\M(\bW)$ reduce to those of the SK model, while \cite{MS-2016-SDPSparseRandomGraphs} showed the same for the value of the degree~2 SOS relaxation.\footnote{Generally, degree fluctuations make irregular graphs more difficult to work with in this context; one may, for example, contrast the proof techniques of \cite{MS-2016-SDPSparseRandomGraphs} for random regular and \Erdos-\Renyi\ random graphs, or those of \cite{BKM-2017-LovaszThetaRandomGraphs} and \cite{BT-2019-VectorColoringIrregular} which treat lower bounds for graph coloring.}
Our results will be inspired by the case of the SK Hamiltonian, and we will not work further with the random graph models here, since the gaussian instance distribution of the SK Hamiltonian greatly simplifies the setting.
Based on the results cited above, we do expect that lower bounds for the SK Hamiltonian should be possible to import to either random graph model for large average degree~$\gamma$, though perhaps indirectly and with substantial technical difficulties.

\begin{remark}[Constrained PCA]
    The problem $\M(\bW)$ may be generalized to the natural broader class of problems where $\bx \in \{\pm 1\}^N$ is replaced by $\bX \in \sV \subset \RR^{N \times k}$ where $k$ is a small constant and the columns of matrices in $\sV$ are constrained to lie on a sphere of fixed radius, and the objective function is replaced with $\Tr(\bX^{\top}\bW\bX)$.
    These are sometimes called \emph{constrained PCA} problems, which search for structured low-rank matrices aligned with the top of the spectrum of $\bW$.
    Certification for these problems shares many of the same phenomena as $\M(\bW)$: there is again a natural spectral certificate, and, as argued in \cite{BKW-2019-ConstrainedPCA,BBKMW-2020-SpectralPlantingColoring}, the spectral certificate is likely often optimal.
    As our construction depends in part on the hypercube constraints but perhaps more deeply on the goal of producing SOS pseudosolutions aligned with the top of the spectrum of $\bW$, our methods may be applicable in these other similarly-structured problems as well.
\end{remark}

\paragraph{Sum-of-squares relaxations}
We now give the formal definition of the sum-of-squares relaxations of $\M(\bW)$.
These are formed by writing the constraints in polynomial form as $x_i^2 - 1 = 0$ for $i = 1, \dots, N$, and applying a standard procedure to build the following feasible set and optimization problem (see, e.g., \cite{Laurent-2009-SOS} for details on this and other generalities on SOS constructions).
\begin{definition}[Hypercube pseudoexpectation]
    \label{def:pe}
    Let $\tEE: \RR[x_1, \dots, x_N]_{\leq 2d} \to \RR$ be a linear operator.
    We say $\tEE$ is a \emph{degree $2d$ pseudoexpectation over $\bx \in \{\pm 1\}^N$}, or, more precisely, \emph{with respect to the constraint polynomials $\{x_i^2 - 1\}_{i = 1}^N$}, if the following conditions hold:
    \begin{enumerate}
    \item $\tEE[1] = 1$ (normalization),
    \item  $\tEE[(x_i^2 - 1) p(\bx)] = 0$ for all $i \in [N]$, $p \in \RR[x_1, \dots, x_N]_{\leq 2d - 2}$ (ideal annihilation),
    \item $\tEE[p(\bx)^2] \geq 0$ for all $p \in \RR[x_1, \dots, x_N]_{\leq d}$ (positivity).
    \end{enumerate}
    In this paper, we abbreviate and simply call such $\tEE$ a \emph{degree $2d$ pseudoexpectation}.
\end{definition}
\noindent
Briefly, a pseudoexpectation is an object that imitates an expectation with respect to a probability distribution supported on $\{\pm 1\}^N$, but only up to the consequences that this restriction has for low-degree moments.
As the degree $2d$ increases, the constraints on pseudoexpectations become more and more stringent, eventually (at degree $2d \geq N$) forcing them to be genuine expectations over such a probability distribution \cite{Laurent-2003-CutPolytopeSOS,FSP-2016-SOSLifts}.

\begin{definition}[Hypercube SOS relaxation]
    The \emph{degree $2d$ SOS relaxation} of the optimization problem $\M(\bW)$ is the problem
    \begin{equation}
        \SOS_{2d}(\bW) \colonequals \max_{\substack{\tEE \text{ degree } 2d \\ \text{pseudoexpectation}}} \tEE[\bx^{\top}\bW \bx].
    \end{equation}
\end{definition}
\noindent
Optimization problems of this kind can be written as SDPs \cite{Laurent-2009-SOS}, converting Condition 3 from Definition~\ref{def:pe} into an associated $N^d \times N^d$ matrix being positive semidefinite (psd).
The results of~\cite{RW-2017-BitComplexity} imply that the SDP of $\SOS_{2d}$ may be solved up to fixed additive error in time $N^{O(d)}$.

How do we build $\tEE$ to show that SOS does not achieve better-than-spectral certification for $\M(\bW)$, i.e., to show that $\SOS_{2d}(\bW) \approx N \lambda_{\max}(\bW)$?
We want to have $\tEE[\bx^{\top}\bW \bx] = \langle \tEE[\bx\bx^{\top}], \bW \rangle \approx N\lambda_{\max}(\bW)$.
Since $\Tr(\tEE[\bx\bx^{\top}]) = N$ and $\tEE[\bx\bx^{\top}] \succeq \bm 0$, we see that $\tEE[\bx\bx^{\top}]$ must be closely aligned with the leading eigenvectors of $\bW$ (those having the largest eigenvalues).
Indeed, the degree 2 SOS lower bounds in the SK Hamiltonian and random regular graph Laplacian instance distributions (both treated in \cite{MS-2016-SDPSparseRandomGraphs}) build $\tEE[\bx\bx^{\top}]$ essentially as a rescaling of a projection matrix to the leading eigenvectors of $\bW$.
In the graph case this is indirectly encoded in the ``gaussian wave'' construction of near-eigenvectors of the infinite tree \cite{CGHV-2015-InvariantGaussian}.
In the case of the SK Hamiltonian, this idea is applied directly and the projection matrix involved is especially natural: since the distribution of the frame of eigenvectors of $\bW \sim \GOE(N)$ is invariant under orthogonal transformations, the span of any collection of leading eigenvectors is a uniformly random low-dimensional subspace.

We thus reach the following distilled form of the task of proving that SOS relaxations of $\M(\bW)$ achieve performance no better than the spectral certificate.
\begin{question}
    \label{ques:main}
    Can the rescaled projection matrix to a uniformly random or otherwise ``nice'' low-dimensional subspace typically arise as $\tEE[\bx\bx^{\top}]$ for a degree $\omega(1)$ pseudoexpectation $\tEE$?
\end{question}
\noindent
As mentioned above, \cite{MS-2016-SDPSparseRandomGraphs} showed that this is the case for degree~2, for both the uniformly random projection matrices arising in the SK model and the sparser approximate projection matrices in the random regular graph model.
For higher-degree SOS relaxations, the only previous known results are those of the concurrent works \cite{KB-2019-Degree4SK-Arxiv,MRX-2019-SOS4} for degree~4; both treat the SK model, while the latter also handles the random regular graph model and, more generally, gives a generic extension from degree~2 to degree~4, an insightful formulation that we follow here.
The approach of \cite{MRX-2019-SOS4} is based on pseudocalibration, while the approach of \cite{KB-2019-Degree4SK-Arxiv}, in which the author participated, uses a modified version of the degree~4 special case of the techniques we will develop here.

Finally, while this paper was being prepared, the parallel work \cite{GJJPR-2020-SK} was released, which performs a deeper technical analysis of pseudocalibration for the SK Hamiltonian and proves degree~$\omega(1)$ lower bounds.
This subsumes some of our results, but we emphasize that we are also able to give distribution-independent results for extending \emph{any} reasonably-behaved degree 2 pseudomoments, making progress towards the conjecture, discussed in their Section 8, ``that the Planted Boolean Vector problem...is still hard for SoS if the input is no  longer i.i.d.\  Gaussian or boolean entries, but is drawn from a `random enough' distribution.''

\paragraph{Local-global tension in SOS lower bounds}
We briefly remark on our technical contributions with the following perspective on what is difficult about SOS lower bounds.
In building a pseudoexpectation $\tEE$ to satisfy Definition~\ref{def:pe}, or equivalently its \emph{pseudomoment matrix} $\tEE[(\bx^{\otimes d})(\bx^{\otimes d})^{\top}]$, there is a basic tension between Properties~1 and 2 from the definition on the one hand and Property~3 on the other.
In the pseudomoment matrix, Properties~1 and 2 dictate that various entries of the matrix should equal one another, giving \emph{local} constraints that concern a few entries at a time.
Property 3, on the other hand, dictates that the matrix should be psd, a \emph{global} constraint concerning how all of the entries behave in concert.\footnote{We are using the terms ``local'' and ``global'' in the sense of \cite{RV-2018-RandomMatrixLocalGlobal}. The more typical distinction is between linear and semidefinite constraints in an SDP, which matches our distinction between local and global constraints, but we wish to emphasize the locality of the linear constraints in that they each concern only a small number of entries.}
It is hard to extend an SOS lower bound to higher degrees because it is hard to satisfy both types of constraint, which are at odds with each other since making many local changes---setting various collections of entries equal to one another---has unpredictable effects on the global spectrum, while making large global changes---adjusting the spectrum to eliminate negative eigenvalues---has unpredictable effects on the local entries.

To the best of our knowledge, SOS lower bound techniques in the literature, most notably pseudocalibration, all proceed by determining sensible entrywise values for each $\tEE[x_{i_1}\cdots x_{i_{2d}}]$, and then verifying positivity by other, often purely technical means.
As a result, there is little intuitive justification for \emph{why} these constructions should satisfy positivity.
We take a step towards rectifying this imbalance: the heuristic underlying our construction gives a plausible reason for both positivity and many of the local constraints to hold.
Some mysterious coincidences do remain in our argument, concerning the family of local constraints that we do not enforce by construction.
Still, we hope that our development of some of the combinatorics that unite the local and global constraints in this case will lead to a clearer understanding of how other SOS lower bound constructions have managed to negotiate these difficulties.

\subsection{Main Results}

\paragraph{Positivity-preserving extension}
Our main result is a generic procedure for building a feasible high-degree pseudoexpectation from a given collection of degree 2 pseudomoments.
We will focus here on describing the result of this construction, which does not in itself show why it is ``positivity-preserving'' as we have claimed---that is explained in Section~\ref{sec:informal-deriv}, where we present the underlying derivation.
We denote the given matrix of degree 2 pseudomoments by $\bM$ throughout.
Our task is then to build $\tEE$ a degree $2d$ pseudoexpectation with $\tEE[\bx\bx^{\top}] = \bM$.
This pseudoexpectation is formed as a linear combination of a particular type of polynomial in the degree 2 pseudomoments, which we describe below.
\begin{definition}[Contractive graphical scalar]
    \label{def:cgs}
    Suppose $G = (V, E)$ is a graph with two types of vertices, which we denote $\bullet$ and $\square$ visually and whose subsets we denote $V = V^{\bullet} \sqcup V^{\square}$.
    Suppose also that $V^{\bullet}$ is equipped with a labelling $\kappa: V^{\bullet} \to [ | V^{\bullet} | ]$.
    For $\bs \in [N]^{|V^{\bullet}|}$ and $\ba \in [N]^{V^{\square}}$, let $f_{\bs, \ba}: V \to [N]$ have $f_{\bs, \ba}(v) = s_{\kappa(v)}$ for $v \in V^{\bullet}$ and $f_{\bs, \ba}(v) = a_v$ for $v \in V^{\square}$.
    Then, for $\bM \in \RR^{N \times N}_{\sym}$, we define
    \begin{equation}
        Z^G(\bM; \bs) \colonequals \sum_{\ba \in [N]^{V^{\square}}} \prod_{\{v, w\} \in E} M_{f_{\bs, \ba}(v) f_{\bs, \ba}(w)}.
    \end{equation}
    We call this quantity a \emph{contractive graphical scalar (CGS)} whose \emph{diagram} is the graph $G$.
    When $S$ is a set or multiset of elements of $[N]$ with $|S| = |V^{\bullet}|$, we define $Z^G(\bM; S) \colonequals Z^G(\bM; \bs)$ where $\bs$ is the tuple of the elements of $S$ in ascending order.
\end{definition}
\noindent
As an intuitive summary, the vertices of the underlying diagram $G$ correspond to indices in $[N]$, and edges specify multiplicative factors given by entries of $\bM$.
The $\bullet$ vertices are ``pinned'' to the indices specified by $\bs$, while the $\square$ vertices are ``contracted'' over all possible index assignments.
CGSs are also a special case of existing formalisms, especially popular in the physics literature, of \emph{trace diagrams} and \emph{tensor networks} \cite{BB-2017-TensorNetworksNutshell}.

\begin{remark}
    Later, in Section~\ref{sec:cgm}, we will also study \emph{contractive graphical matrices (CGMs)}, set- or tuple-indexed matrices whose entries are CGSs with the set $S$ varying according to the indices.
    CGMs are similar to \emph{graphical matrices} as used in other work on SOS relaxations  \cite{AMP-2016-GraphMatrices,BHKKMP-2019-PlantedClique,MRX-2019-SOS4}.
    Aside from major but ultimately superficial notational differences, the main substantive difference is that graphical matrices require all indices labelling the vertices in the summation to be different from one another, while CGMs and CGSs do not.
    This restriction is natural in the combinatorial setting---if $\bM$ is an adjacency matrix then the entries of graphical matrices count occurrences of subgraphs---but perhaps artificial more generally.
    While the above works give results on the spectra of graphical matrices, and tensors formed with tensor networks have been studied at length elsewhere, the spectra of CGM-like matrix ``flattenings'' of tensor networks remain poorly understood.\footnote{One notable exception is the calculations with the trace method in the recent work \cite{MW-2019-SpectralTensorNetworks}.}
    We develop some further tools for working with such objects in Appendix~\ref{app:cgm-tools}.
\end{remark}

Next, we specify the fairly simple class of diagrams whose CGSs will actually appear in our construction.
\begin{definition}[Good forest]
    \label{def:good-forest}
    We call a forest \emph{good} if it has the following properties:
    \begin{enumerate}
    \item no vertex is isolated, and
    \item the degree of every internal (non-leaf) vertex is even and at least 4.
    \end{enumerate}
    We count the empty forest as a good forest.
    Denote by $\sF(m)$ the set of good forests on $m$ leaves, equipped with a labelling $\kappa$ of the leaves by the set $[m]$.
    We consider two labelled forests equivalent if they are isomorphic as partially labelled graphs; thus, the same underlying forest may appear in $\sF(m)$ with some but not all of the $m!$ ways that it could be labelled.
    For $F \in \sF(m)$, we interpret $F$ as a diagram by calling $V^{\bullet}$ the leaves of $F$ and calling $V^{\square}$ the internal vertices of $F$.
    Finally, we denote by $\sT(m)$ the subset of $F \in \sF(m)$ that are connected (and therefore trees).
\end{definition}
\noindent
We note that, for $m$ odd, the constraints imply that $\sF(m)$ is empty.
We give some examples of these forests and the associated CGSs in Figure~\ref{fig:cgs-examples}.

Finally, we define the coefficients that are attached to each forest diagram's CGS in our construction.
\begin{definition}[\Mobius\ function of good forests]
    \label{def:mu-F}
    For $F = (V^{\bullet} \sqcup V^{\square}, E) \in \sF(m)$, define
    \begin{equation}
        \mu(F) \colonequals \prod_{v \in V^{\square}}\big(-(\deg(v) - 2)!\big) = (-1)^{|V^{\square}|}\prod_{v \in V^{\square}}(\deg(v) - 2)!.
    \end{equation}
    For $F$ the empty forest, we set $\mu(F) = 1$ by convention.
\end{definition}
\noindent
These constants have an important interpretation in terms of the combinatorics of $\sF(m)$: as we will show in Section~\ref{sec:poset}, when $\sF(m)$ is endowed with a natural partial ordering, $\mu(F)$ is (up to sign) the \emph{\Mobius\ function} of the ``interval'' of forests lying below $F$ in this ordering.
In general, \Mobius\ functions encode the combinatorics of inclusion-exclusion calculations under a partial ordering \cite{Rota-1964-Foundations}.
In our situation, $\mu(F)$ ensures that, even if we allow repeated indices in the monomial index $S$ in the definition below, a suitable cancellation occurs such that the pseudoexpectation of $\bx^{S}$ still approximately satisfies the ideal annihilation constraint in Definition~\ref{def:pe}.

\begin{figure}
    \begin{center}
        \setlength{\tabcolsep}{15pt}
        \begin{tabular}{c:c:c}
           & & \\[-0.5em]
          \hspace{-0.5cm}\includegraphics[scale=0.65]{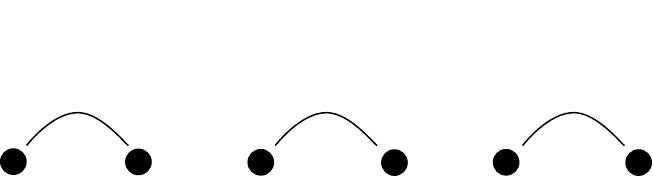} & \includegraphics[scale=0.65]{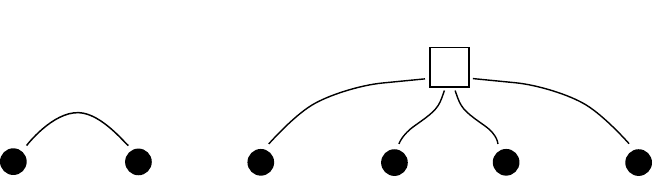} & \includegraphics[scale=0.65]{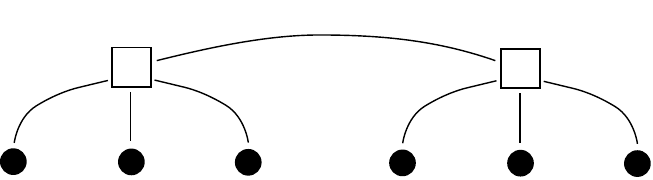} \\
          \hspace{-0.48cm}\footnotesize 1 \hspace{0.5cm} 2 \hspace{0.38cm} 3 \hspace{0.53cm} 4 \hspace{0.36cm} 5 \hspace{0.54cm} 6 & \footnotesize 1 \hspace{0.5cm} 2 \hspace{0.38cm} 3 \hspace{0.53cm} 4 \hspace{0.36cm} 5 \hspace{0.54cm} 6 & \footnotesize 1 \hspace{0.39cm} 2 \hspace{0.45cm} 3 \hspace{0.6cm} 4 \hspace{0.46cm} 5 \hspace{0.4cm} 6 \\[-0.5em] & & \\
          \hspace{-0.5cm}\small $M_{ij}M_{k\ell}M_{mn}$ & \small $-2M_{ij}\displaystyle\sum_{a = 1}^N M_{ak}M_{a\ell}M_{am}M_{an}$ & \small $4\displaystyle\sum_{a, b = 1}^NM_{ai}M_{aj}M_{ak}M_{ab}M_{b\ell}M_{bm}M_{bn}$ \\[-1em] & & 
        \end{tabular}
    \end{center}
    \caption{\textbf{Forests, polynomials, and coefficients.} We show three examples of forests $F \in \sF(6)$ with labelled leaves, together with the corresponding CGS terms $\mu(F) \cdot Z^F(\bM; (i, j, k, \ell, m, n))$ appearing in the pseudoexpectation of Definition~\ref{def:lifting}.}
    \label{fig:cgs-examples}
\end{figure}

With these ingredients defined, we are prepared to define our pseudoexpectation.
\begin{definition}[Extending pseudoexpectation]
    \label{def:lifting}
    For $\bM \in \RR^{N \times N}_{\sym}$, define $\tEE_{\bM}: \RR[x_1, \dots, x_N] \to \RR$ to be a linear operator with $\tEE_{\bM}[x_i^2p(\bx)] = \tEE_{\bM}[p(\bx)]$ for all $i \in [N]$ and $p \in \RR[x_1, \dots, x_N]$, and values on multilinear monomials given by
    \begin{equation}
        \label{eq:lifting-prediction}
        \tEE_{\bM}\left[\prod_{i \in S} x_i\right] \colonequals \sum_{F \in \sF(|S|)} \mu(F) \cdot Z^F(\bM; S) \text{ for all } S \subseteq [N].
    \end{equation}
\end{definition}
\noindent
Our main result is that, for ``nice'' $\bM$, the restriction of $\tEE_{\bM}$ to low-degree polynomials is a valid pseudoexpectation.

First, we introduce several quantities measuring favorable behavior of $\bM$ that, taken together, will describe whether $\bM$ is sufficiently well-behaved.
As a high-level summary, these quantities capture various aspects of the ``incoherence'' of $\bM$ with respect to the standard basis vectors $\be_1, \dots, \be_N$.
To formulate the subtlest of the incoherence quantities precisely, we will require the following preliminary technical definition, whose relevance will only become clear in the course of our proof in Section~\ref{sec:main-error-terms}.
There, it will describe a residual error term arising from allowing repeated indices in $S$ in Definition~\ref{def:lifting}, after certain cancellations are taken into account.

\begin{definition}[Maximal repetition-spanning forest]
    \label{def:repetitions-forest}
    For each $F \in \sF(m)$ and $\bs \in [N]^m$, let $\mathsf{MaxSpan}(F, \bs)$ be the subgraph of $F$ formed by the following procedure.
    Let $C_1, \dots, C_k$ be the connected components of $F$.
    \begin{leftbar}
        \begin{algorithmic}
            \State{Initialize with $\mathsf{MaxSpan}(F, \bs) = \emptyset$.}
    \For {$i = 1, \dots, N$}
    \For {$j = 1, \dots, k$}
    \If {\,$C_j$ has two leaves $\ell_1 \neq \ell_2$ with $s_{\kappa(\ell_1)} = s_{\kappa(\ell_2)} = i$}
    \State Let $T$ be the minimal spanning tree of all leaves $\ell$ of $C_j$ with $s_{\kappa(\ell)} = i$.
    \If {\,$T$ is vertex-disjoint from $\mathsf{MaxSpan}(F, \bs)$}
    \State Add $T$ to $\mathsf{MaxSpan}(F, \bs)$.
    \EndIf
    \EndIf
    \EndFor
\EndFor
\end{algorithmic}
\end{leftbar}
\noindent
We say that $\ba \in [N]^{V^{\square}(F)}$ is \emph{$(F, \bs)$-tight} if, for all connected components $C$ of $\MaxSpan(F, \bs)$ with $s_{\kappa(\ell)} = i$ for all leaves $\ell$ of $C$, for all $v \in V^{\square}(C)$, $a_v = i$.
    Otherwise, we say that $\ba$ is \emph{$(F, \bs)$-loose}.
\end{definition}
\noindent
With this, we define the following functions of $\bM$.
Below, $\bM^{\circ k}$ denotes the $k$th entrywise power of $\bM$, and $\set(\bs)$ for a tuple $\bs$ denotes the set of indices occurring in $\bs$.
\begin{definition}[Incoherence quantities]
    \label{def:error-quantities}
    For $\bM \in \RR^{N \times N}_{\sym}$, define the following quantities:
    \begin{align}
      \epsilon_{\offdiag}(\bM) &\colonequals \max_{1 \leq i < j \leq N} |M_{ij}|, \\
      \epsilon_{\corr}(\bM) &\colonequals \left(\max_{1 \leq i < j \leq N} \sum_{k = 1}^N M_{ik}^2 M_{jk}^2 \right)^{1/2}, \\
      \epsilon_{\pow}(\bM) &\colonequals \max_{k \geq 2} \|\bM^{\circ k} - \bm I_N\|, \\
      \epsilon_{\tree}(\bM; 2d) &\colonequals \max_{0 \leq d^{\prime} \leq d} \max_{T \in \sT(2d^{\prime})} \max_{\bs \in [N]^{2d^{\prime}}} \left|Z^T(\bM; \bs) - \One\{s_1 = \cdots = s_N\}\right|, \\
      \epsilon_{\err}(\bM; 2d) &\colonequals \max_{0 \leq d^{\prime} \leq d} \max_{T \in \sT(2d^{\prime})} \max_{\bs \in [N]^{2d^{\prime}}} N^{|\set(\bs)| / 2} \Bigg|\sum_{\substack{\ba \in [N]^{V^{\square}} \\ \ba \text{ } (T, \bs)\text{-loose}}} \prod_{(v, w) \in E(T)} M_{f_{\bs, \ba}(v)f_{\bs, \ba}(w)}\Bigg|, \\
      \epsilon(\bM; 2d) &\colonequals \epsilon_{\offdiag}(\bM) + \epsilon_{\corr}(\bM) + \epsilon_{\pow}(\bM) + \epsilon_{\tree}(\bM; d) + \epsilon_{\err}(\bM; 2d).
    \end{align}
\end{definition}

Our main result then states that $\bM$ may be extended to a high-degree pseudoexpectation so long as its smallest eigenvalue is not too small compared to the sum of the incoherence quantities.
\begin{theorem}
    \label{thm:lifting}
    Let $\bM \in \RR^{N \times N}_{\sym}$ with $M_{ii} = 1$ for all $i \in [N]$.
    Suppose that
    \begin{equation}
      \lambda_{\min}(\bM) \geq (12d)^{32}\|\bM\|^{5}\epsilon(\bM; 2d)^{1/d}.
    \end{equation}
    Then, $\tEE_{\bM}$ is a degree $2d$ pseudoexpectation with $\tEE_{\bM}[\bx\bx^{\top}] = \bM$.
\end{theorem}

In practice, Theorem~\ref{thm:lifting} will not be directly applicable to the $\bM$ we wish to extend, which, as mentioned earlier, will be rank-deficient and therefore have $\lambda_{\min}(\bM) = 0$ (or very small).
This obstacle is easily overcome by instead extending $\bM^{\prime} = (1 - \alpha)\bM + \alpha \bm I_N$ for $\alpha \in (0, 1)$ a small constant, whereby $\lambda_{\min}(\bM^{\prime}) \geq \alpha$.
Unfortunately, it seems difficult to make a general statement about how the more intricate quantities $\epsilon_{\tree}$ and $\epsilon_{\err}$ transform when $\bM$ is replaced with $\bM^{\prime}$; however, we will show in our applications that directly analyzing these quantities for $\bM^{\prime}$ is essentially no more difficult than analyzing them for $\bM$.
Indeed, we expect these to only become smaller under this replacement since $\bM^{\prime}$ equals $\bM$ with the off-diagonal entries multiplied by $(1 - \alpha)$.

\begin{remark}[Different ways of nudging]
    \label{rem:nudging}
A similar ``nudging'' operation to the one we propose above, moving $\bM$ towards the identity matrix, has been used before in \cite{KB-2019-Degree4SK-Arxiv,MRX-2019-SOS4} for degree~4 SOS and in the earlier work \cite{AU-2003-LPMaxCut} for LP relaxations.\footnote{I thank Aida Khajavirad for bringing the reference \cite{AU-2003-LPMaxCut} to my attention.}
However, the way that this adjustment propagates through our construction is quite different: while \cite{KB-2019-Degree4SK-Arxiv,MRX-2019-SOS4} consider, in essence, a convex combination of the form $(1 - \alpha) \tEE_{\bM} + \alpha \tEE_{\bm I_N}$, we instead consider $\tEE_{(1 - \alpha)\bM + \alpha \bm I_N}$.
The mapping $\bM \mapsto \tEE_{\bM}$ is highly non-linear, so this is a major difference, which indeed turns out to be crucial for the adjustment to effectively counterbalance the error terms in our analysis.
\end{remark}

We expect the following general quantitative behavior from this result.
Typically, we will have $\epsilon(\bM; 2d) = O(N^{-\gamma})$ for some $\gamma > 0$.
We will also have $\|\bM\| = O(1)$ and $\lambda_{\min}(\bM) = \widetilde{\Omega}(1)$ after the adjustment discussed above.
Therefore, Theorem~\ref{thm:lifting} will ensure that $\bM$ is extensible to degree $2d$ so long as $N^{-\gamma / d} \poly(d) = O(1)$, whereby the threshold scaling at which the condition of Theorem~\ref{thm:lifting} is no longer satisfied is slightly smaller than $d \sim \log N$; for instance, such $\bM$ will be extensible to degree $d \sim \log N / \log\log N$.
See the brief discussion after Proposition~\ref{prop:size-F} for an explanation of why this scaling of the degree is likely the best our proof techniques can achieve.

\paragraph{Application 1: Laurent's parity lower bound}
As a first application of Theorem~\ref{thm:lifting}, we show that we can recover a ``soft version'' of the following result of \cite{Laurent-2003-CutPolytopeSOS}, which says that a parity inequality that holds for $\bx \in \{\pm 1\}^N$ fails for pseudoexpectations with degree less than $N$.

\begin{theorem}[Theorem 6 of \cite{Laurent-2003-CutPolytopeSOS}]
    Define $\tEE: \RR[x_1, \dots, x_N]_{\leq N - 1 - \One\{N \text{ even}\}} \to \RR$ to be a linear operator with $\tEE[x_i^2p(\bx)] = \tEE[p(\bx)]$ for all $i \in [N]$ and $p \in \RR[x_1, \dots, x_N]$ and values on multilinear monomials given by
    \begin{equation}
        \label{eq:laurent-pseudomoments}
        \tEE\left[\prod_{i \in S} x_i\right] \colonequals \One\{|S| \text{ even}\} \cdot \frac{(-1)^{|S| / 2} (|S| - 1)!!}{\prod_{1 \leq k \leq |S| / 2}(N - 2k + 1)} \text{ for all } S \subseteq [N].
    \end{equation}
    Then, $\tEE$ is a degree $(N - 1 - \One\{N \text{ even}\})$ pseudoexpectation which satisfies
    \begin{align}
      \tEE[\bx\bx^{\top}] &= \left(1 + \frac{1}{N - 1}\right)\bm I_N - \frac{1}{N - 1}\one_N\one_N^{\top}, \\
      \tEE[(\one^{\top}\bx)^2] &= 0.
    \end{align}
\end{theorem}
\noindent
For $N$ odd and $\bx \in \{\pm 1\}^N$ we always have $(\one^{\top}\bx)^2 = (\sum_{i = 1}^N x_i)^2 \geq 1$, while the result shows that pseudoexpectations must have degree at least $N$ before they are forced by the constraints to obey this inequality.
(\cite{FSP-2016-SOSLifts} later showed that this result is tight as well.)

The version of this that follows from Theorem~\ref{thm:lifting} is as follows.
\begin{theorem}[Soft version of Laurent's theorem]
    \label{thm:appl-laurent-approx}
    Let $\alpha = \alpha(N) = (\log\log N)^{-50}$.
    Then, for all $N$ sufficiently large, there exists $\tEE$ a degree $\frac{1}{100}\log N / \log\log N$ pseudoexpectation satisfying
    \begin{align}
      \tEE\left[\prod_{i \in S} x_i\right] &= \One\{|S| \text{ even}\} \cdot \left(\frac{(-1)^{|S| / 2}(|S| - 1)!!}{(N / (1 - \alpha))^{|S| / 2}} + O_{|S|}\left(\frac{1}{N^{|S| / 2 +  1}}\right)\right), \label{eq:laurent-approx-1} \\
      \tEE[\bx\bx^{\top}] &= \left(1 + \frac{1 - \alpha}{N - 1}\right)\bm I_N - \frac{1 - \alpha}{N - 1}\one_N\one_N^{\top}. \label{eq:laurent-approx-2}
    \end{align}
\end{theorem}
\noindent
This is weaker than the original statement; most importantly, it only gives a  pseudoexpectation $\tEE$ with $\tEE[(\one^{\top}\bx)^2] \approx \alpha N$, and thus does not show that the parity inequality above fails for $\tEE$.
However, it has two important qualitative features: (1) it implies that we need only add to $\tEE[\bx\bx^{\top}]$ an adjustment with operator norm $o(1)$ to obtain an automatically-extensible degree 2 pseudomoment matrix, and (2) it gives the correct leading-order behavior of the pseudomoments.
Elaborating on the latter point, our derivation in fact shows how the combinatorial interpretation of $(|S| - 1)!!$ as the number of perfect matchings of a set of $|S|$ objects is related to the appearance of this quantity in Laurent's construction.
In the original derivation this arises from assuming the pseudomoments depend only on $|S|$ and satisfy $\tEE[(\one^{\top}\bx)\bx^S] = 0$, which determines the pseudomoments inductively starting from $\tEE[1] = 1$.
In our derivation, this coefficient simply comes from counting the diagrams of $\sF(|S|)$ making leading-order contributions, which are the diagrams of perfect matchings.

\paragraph{Application 2: random high-rank projectors}
We also consider a random variant of the setting of Laurent's theorem, where the special subspace spanned by $\one_N$ is replaced with a random low-dimensional subspace.
This is also essentially identical to the setting we would like to treat to give SOS lower bounds for the SK Hamiltonian, except for the dimensionality of the subspace.
\begin{theorem}[Random high-rank projectors]
    \label{thm:appl-high-rank}
    Suppose $n: \NN \to \NN$ is an increasing function with $\log(N) \ll n(N) \ll N / \log N$ as $N \to \infty$.
    Let $V$ be a uniformly random $(N - n)$-dimensional subspace of $\RR^N$.
    Then, with high probability as $N \to \infty$, there exists $\tEE$ a degree $\frac{1}{300}\log(N / n) / \log\log N$ pseudoexpectation satisfying
    \begin{alignat}{2}
      \frac{\tEE[\langle \bx, \bv \rangle^2]}{\|\bv\|^2} &\in \left[1, 1 + 4\frac{n}{N}\right] &&\text{ for all } \bv \in V \setminus \{ \bm 0 \}, \\
      \frac{\tEE[\langle \bx, \bv \rangle^2]}{\|\bv\|^2} &\in \left[0, \frac{1}{(\log\log N)^{32}} + 4\frac{n}{N} \right] &&\text{ for all } \bv \in V^{\perp} \setminus \{ \bm 0 \}.
    \end{alignat}
\end{theorem}
\noindent
As in the case of our version of Laurent's theorem, this result does not imply an SOS integrality gap that is in itself particularly interesting.
Indeed, results in discrepancy theory have shown that hypercube vectors can avoid random subspaces of sub-linear dimension ($V^{\perp}$, in our case) unusually effectively; see, e.g., \cite{TMR-2020-BalancingGaussianVectors} for the recent state-of-the-art.
Rather, we present this example as another qualitative demonstration of our result, showing that it is possible to treat the random case in the same way as the deterministic case above, and that we can again obtain an automatic higher-degree extension after an adjustment with operator norm $o(1)$ of $\tEE[\bx\bx^{\top}]$ from a random projection matrix.

\paragraph{Application 3: Sherrington-Kirkpatrick Hamiltonian}
Unfortunately, our approach above does not appear to extend directly to the low-rank setting.
We discuss the reasons for this in greater detail in Section~\ref{sec:future}, but, at a basic level, if $M_{ij} = \langle \bv_i, \bv_j \rangle$ for unit-norm ``Gram vectors'' $\bv_i$, then, as captured in the incoherence quantity $\epsilon_{\pow}(\bM)$, our construction relies on the $\bv_i^{\otimes k}$ for all $k \geq 2$ behaving like a nearly-orthonormal set.
Once $n = \Theta(N)$ in the setting of Theorem~\ref{thm:appl-high-rank}, this is no longer the case: for $k \geq 3$ the $\bv_i^{\otimes k}$ still behave like an orthonormal set, but the $\bv_i^{\otimes 2}$, which equivalently may be viewed as the matrices $\bv_i\bv_i^{\top}$, are too ``crowded'' in $\RR^{r \times r}_{\sym}$ and have an overly significant collective bias in the direction of the identity matrix.

However, for low degrees of  SOS, we can still make a manual correction for this and obtain a lower bound.
That is essentially what was done in \cite{KB-2019-Degree4SK-Arxiv} for degree 4, and the following result extends this to degree~6 with a more general formulation.
(As part of the proof we also give a slightly different and perhaps simpler argument for the degree 4 case than \cite{KB-2019-Degree4SK-Arxiv}.)

We present our result in terms of another, modified extension result for arbitrary degree~2 pseudomoments.
This extension only reaches degree 6, but allows the flexibility we sought above in $\epsilon_{\pow}$.
We obtain it by inelegant means, using simplifications specific to the diagrams appearing at degree 6 to make some technical improvements in the argument of Theorem~\ref{thm:lifting}.
\begin{definition}[Additional incoherence quantities]
    For $\bM \in \RR^{N \times N}_{\sym}$ and $t > 0$, define the following quantities:
    \begin{align}
      \widetilde{\epsilon}_{\pow}(\bM, t) &\colonequals \max\left\{ \|\bM^{\circ 2} - \bm I_N - t\one_N\one_N^{\top}\|, \, \max_{k \geq 3} \|\bM^{\circ 3} - \bm I_N\|\right\}, \\
      \widetilde{\epsilon}(\bM, t) &\colonequals \epsilon_{\offdiag}(\bM) + \widetilde{\epsilon}_{\pow}(\bM, t) + N^{-1/2}\epsilon_{\err}(\bM; 6).
    \end{align}
\end{definition}

\begin{theorem}
    \label{thm:low-rank-lifting}
    Let $\bM \in \RR^{N \times N}_{\sym}$ with $M_{ii} = 1$ for all $i \in [N]$, and suppose $t_{\pow} > 0$.
    Suppose that
    \begin{equation}
        \lambda_{\min}(\bM) \geq 10^{50}\|\bM\|^{5}\widetilde{\epsilon}(\bM, t_{\pow})^{1/3}.
    \end{equation}
  Define the constant
  \begin{equation}
      c \colonequals 250 t_{\pow}\big(\|\bM\|^6 \|\bM^2\|_F + N \epsilon_{\offdiag}(\bM^2) + N^2 \epsilon_{\offdiag}(\bM^2)^3\big).
  \end{equation}
  Then, there exists a degree 6 pseudoexpectation $\tEE$ with $\tEE[\bx\bx^{\top}] = (1 - c)\bM + c\bm I_N$.
\end{theorem}
\noindent
(The abysmal constant in the first condition could be improved with a careful analysis, albeit one even more specific to degree 6.)
We show as part of the proof that a pseudoexpectation achieving this can be built by adding a correction of sub-leading order to those terms of the pseudoexpectation in Definition~\ref{def:lifting} where $F$ is a perfect matching.
It is likely that to extend this result to degree $\omega(1)$ using our ideas would require somewhat rethinking our construction and the derivation we give in Section~\ref{sec:informal-deriv} to take into account the above obstruction, but this makes it plausible that the result will be some form of small correction added to $\tEE_{\bM}$.

Finally, applying the above with a uniformly random low-rank projection matrix gives the following degree~6 lower bound for the SK Hamiltonian.
As mentioned before, this result is subsumed in the results of the parallel work \cite{GJJPR-2020-SK}, but we include it here to illustrate a situation where the above extension applies quite easily.
\begin{theorem}
    \label{thm:appl-sk}
    For any $\epsilon > 0$, for $\bW \sim \GOE(N)$, $\lim_{N \to \infty} \PP[\SOS_6(\bW) \geq (2 - \epsilon)N] = 1$.
\end{theorem}

\subsection{Proof Techniques}

We give a brief overview here of the ideas playing a role in the proof of our extension results, Theorems~\ref{thm:lifting} and \ref{thm:low-rank-lifting}.
The method for predicting the values of $\tEE_{\bM}$ was suggested in \cite{KB-2019-Degree4SK-Arxiv}: we predict $\tEE_{\bM}[\bx^{\bs} \bx^{\bt}] = \EE[G_{\bs}G_{\bt}]$ for a gaussian symmetric tensor $\bG$, a ``surrogate'' for $\bx^{\otimes d}$, which is endowed with a natural orthogonally-invariant tensor distribution conditional on properties that cause $\tEE_{\bM}$ to satisfy (1) some of the ideal annihilation and normalization constraints, and (2) the ``subspace constraint'' that $\bx$ behaves as if it is constrained to the row space of $\bM$ (recall that our motivation is the case where $\bM$ is roughly a projection matrix, in which case this is just the subspace that $\bM$ projects to).
We call this a \emph{positivity-preserving extension} of $\bM$ because by construction the pseudomoment matrix of $\tEE$ is the degree 2 moment matrix of $\bG$, so if we were to not make any further modifications, $\tEE$ would be guaranteed to be psd.

To carry out an approximate calculation of the mean and covariance of $\bG$, we reframe the task in terms of homogeneous polynomials.
This reveals that the fluctuations of $\bG$ after conditioning are along a subspace of symmetric tensors associated to certain \emph{multiharmonic polynomials}, those $p(\bz)$ satisfying $\bv_i^{\top}\nabla^2 p(\bz) \bv_i = 0$ where the $\bv_i$ are the Gram vectors for which $M_{ij} = \langle \bv_i, \bv_j \rangle$.
To compute the covariance of $\bG$, we must compute orthogonal projections to this subspace with respect to the \emph{apolar inner product}, that inherited by homogeneous polynomials through their correspondence with symmetric tensors.
To compute these projections, we heuristically extend classical but somewhat obscure ideas of Maxwell and Sylvester~\cite{Maxwell-1873-Treatise1,Sylvester-1876-SphericalHarmonics} for projecting to harmonic polynomials, and a generalization of Clerc \cite{Clerc-2000-KelvinTransform} for projecting to certain multiharmonic polynomials, which does not quite capture our situation but allows us to make a plausible prediction.

Using this, we arrive at a closed form for our prediction of $\tEE_{\bM}$, which, after some combinatorial arguments, reduces to the form given above in Definition~\ref{def:lifting} where graphical terms are multiplied by an associated \Mobius\ function.
Identifying the \Mobius\ function in these coefficients allows us to verify that $\tEE_{\bM}$ approximately satisfies \emph{all} ideal annihilation constraints, not just those enforced by the construction of $\bG$ (and that those enforced by the construction have not been lost in our heuristic calculations), as well as the symmetry constraints that $\tEE_{\bM}[\bx^{\bs} \bx^{\bt}]$ is unchanged by joint permutations of $\bs$ and $\bt$.
This combinatorial calculation is the key to the argument, as it shows the relationship between the positivity of $\tEE_{\bM}$ produced by our ``Gramian'' construction in terms of $\bG$ and the entrywise constraints satisfied thanks to the \Mobius\ function's appearance.

Lastly, to actually give a full proof, we work backwards: as we have done above, we \emph{define} $\tEE_{\bM}$ in its final graphical form in terms of the forest CGSs and their \Mobius\ functions, which exactly satisfies all entrywise constraints.
We then show that, up to small error, it admits the Gram factorization inspired by our construction of $\bG$, and therefore also satisfies positivity.

\subsection{Organization}

The remainder of the paper is organized as follows.
In Section~\ref{sec:prelim}, we present preliminary materials.
In Section~\ref{sec:informal-deriv}, we derive the extension formula \eqref{eq:lifting-prediction}, starting with the conjectural construction of~\cite{KB-2019-Degree4SK-Arxiv} and following the sketch above to reach a closed form.
This derivation is informal, but provides an important intuition accounting for the positivity of the pseudoexpectation.
Then, in Section~\ref{sec:poset}, we describe the partial ordering structure associated with the forests giving the terms in the pseudomoments, and show that their coefficients in the extension formula are the \Mobius\ function of this partially ordered set.
Combining these ideas, in Section~\ref{sec:pf:lifting} we give the full proof of Theorem~\ref{thm:lifting}.
In Section~\ref{sec:pf:applications} we prove our applications to Laurent's pseudomoments and random high-rank projection matrices.
In Section~\ref{sec:future}, we give some discussion about what stops our techniques from extending directly to the low-rank case that is more relevant for the SK Hamiltonian and random graph applications.
Finally, in Section~\ref{sec:pf:applications-sk}, we prove the degree~6 extension of Theorem~\ref{thm:low-rank-lifting} and our partial result for the SK Hamiltonian.

\section{Preliminaries}
\label{sec:prelim}

\subsection{Notation}

\paragraph{Sets, multisets, partitions}
For a set $A$, we write $2^A$, $\binom{A}{k}$, and $\binom{A}{\leq k}$ for the sets of all subsets of $A$, subsets of size $k$ of $A$, and subsets of size at most $k$ of $A$, respectively.
We write $\sM(A)$, $\sM_k(A)$, and $\sM_{\leq k}(A)$ for the sets of all multisets (sets with repeated elements allowed) with elements in $A$, all multisets of size $k$ with elements in $A$, and all multisets of size at most $k$ with elements in $A$, respectively.

For $A, B$ multisets, we write $|A|$ for the number of elements in $A$, $A + B$ for the disjoint union of $A$ and $B$, and $A - B$ for the multiset difference of $A$ and $B$ (where the number of occurrences of an element in $A$ is reduced by the number of occurrences in $B$, stopping at zero).
It may be clearer to think of multisets as functions from an alphabet to $\NN$, in which case $A + B$ is ordinary pointwise addition, while $A - B$ is the maximum of the pointwise difference with the zero function.
We write $A \subseteq B$ if each element occurs at most as many times in $A$ as it does in $B$.
We do not use a special notation for explicit multisets; when we write, e.g., $A = \{i, i, j\}$, then it is implied that $A$ is a multiset.

For $A$ a set or multiset, we write $\Part(A)$ for the set or multiset, respectively, of partitions of $A$.
Repeated elements in a multiset are viewed as distinct for generating partitions, making $\Part(A)$ a multiset when $A$ is a multiset.
For example,
\begin{equation}
    \Part(\{i, i, j\}) = \bigg\{\{\{i\}, \{i\}, \{j\}\}, \, \{\{i, i\}, \{j\}\}, \, \underbrace{\{\{i\}, \{i, j\}\}, \, \{\{i\}, \{i, j\}\}}_{\text{repeated}}, \, \{\{i, i, j\}\}\bigg\}.
\end{equation}
We write $\Part(A; \even)$ and $\Part(A; \odd)$ for partitions into only even or odd parts, respectively, and $\Part(A; k)$, $\Part(A; \geq k)$, and $\Part(A; \leq k)$ for partitions into parts of size exactly, at most, and at least $k$, respectively.
We also allow these constraints to be chained, so that, e.g., $\Part(A; \even; \geq k)$ is the set of partitions into even parts of size at least $k$.
Similarly, for a specific partition $\pi \in \Part(A)$, we write $\pi[\even], \pi[\odd], \pi[k], \pi[\geq k], \pi[\leq k], \pi[\even; \geq k]$ and so forth for the parts of $\pi$ with the specified properties.

\paragraph{Linear algebra}
We use bold uppercase letters ($\bA$, $\bB$) for matrices, bold lowercase letters ($\ba$, $\bb$) for vectors, and plain letters for scalars, including for the entries of matrices and vectors ($a$, $a_i$, $A_{ij}$).
We denote by $\bA \circ \bB$ the entrywise or Hadamard product of matrices, and by $\bA^{\circ k}$ the Hadamard powers.

\subsection{Symmetric Tensors and Homogeneous Polynomials}
\label{sec:prelim-symtens-poly}

We first review some facts about symmetric tensors, homogeneous polynomials, and the relationships between their respective Hilbert space structures.

\paragraph{Hilbert space structures}
The vector space of \emph{symmetric $d$-tensors} $\Sym^d(\RR^N) \subset (\RR^N)^{\otimes d}$ is the subspace of $d$-tensors whose entries are invariant under permutations of the indices.
The vector space of \emph{homogeneous degree $d$ polynomials} $\RR[y_1, \dots, y_N]_d^{\hom}$ is the subspace of degree $d$ polynomials whose monomials all have total degree $d$.
Having the same dimension $\binom{N + d - 1}{d}$, these two vector spaces are isomorphic; a natural correspondence between homogeneous polynomials $p(\by)$ and symmetric tensors $\bA$ is
\begin{align}
    \bA &\mapsto p(\by) = \bA[\by, \ldots, \by] = \sum_{\bs \in [N]^d} A_{\bs}\by^{\bs}, \label{eq:correspondence-symtens-poly} \\
    p(\by) &\mapsto A_{\bs} = \binom{d}{\freq(\bs)}^{-1} \cdot [\by^{\bs}](p), \label{eq:correspondence-poly-symtens}
\end{align}
where $\freq(\bs)$ is the sequence of integers giving the number of times different indices occur in $\bs$ (sometimes called the \emph{derived partition}) and $[\by^{\bs}](p)$ denotes the extraction of a coefficient.

The general $d$-tensors $(\RR^N)^{\otimes d}$ may be made into a Hilbert space by equipping them with the \emph{Frobenius inner product},
\begin{equation}
    \langle \bA, \bB \rangle \colonequals \sum_{\bs \in [N]^d} A_{\bs}B_{\bs}.
\end{equation}
The symmetric $d$-tensors inherit this inner product, which when $\bA, \bB \in \Sym^d(\RR^N)$ may be written
\begin{equation}
    \langle \bA, \bB \rangle = \sum_{S \in \sM_d([N])} \binom{d}{\freq(S)}A_{S}B_{S}.
\end{equation}
Perhaps less well-known is the inner product induced on homogeneous degree $d$ polynomials by the Frobenius inner product pulled back through the mapping  \eqref{eq:correspondence-poly-symtens}, which is called the \emph{apolar inner product} \cite{ER-1993-ApolarityCanonicalForms,Reznick-1996-HomogeneousPolynomial,Vegter-2000-ApolarBilinearForm}.\footnote{\cite{ER-1993-ApolarityCanonicalForms} write: ``...the notion of apolarity has remained sealed in the well of oblivion.''}\textsuperscript{,}\footnote{Other names used in the literature for this inner product include the \emph{Bombieri}, \emph{Bombieri-Weyl}, \emph{Fischer}, or \emph{Sylvester} inner product.
  The term \emph{apolar} itself refers to polarity in the sense of classical projective geometry; see~\cite{ER-1993-ApolarityCanonicalForms} for a historical overview in the context of invariant theory.}
For the sake of clarity, we distinguish this inner product with a special notation:
\begin{equation}
    \langle p, q \rangle_{\circ} \colonequals \sum_{S \in \sM_d([N])} \binom{d}{\freq(S)}^{-1} \cdot [\by^{S}] (p) \cdot [\by^{S}] (q).
    \label{eq:apolar-ip-def}
\end{equation}
In the sequel we also follow the standard terminology of saying that ``$p$ and $q$ are apolar'' when $\langle p, q \rangle_{\circ} = 0$; we also use this term more generally to refer to orthogonality under the apolar inner product, speaking of apolar subspaces, apolar projections, and so forth.

\paragraph{Properties of the apolar inner product}
The most important property of the apolar inner product that we will use is that multiplication and differentiation are adjoint to one another.
We follow here the expository note \cite{Reznick-1996-HomogeneousPolynomial}, which presents applications of this idea to PDEs, a theme we will develop further below.
The basic underlying fact is the following.
For $q \in \RR[y_1, \ldots, y_N]$, write $q(\bm\partial) = q(\partial_{y_1}, \dots, \partial_{y_N})$ for the associated differential operator.\footnote{If, for instance, $q(\by) = y_1^2y_2 + y_3^3$, then $q(\bm\partial)f = \frac{\partial^3 f}{\partial y_1^2\partial y_2} + \frac{\partial^3 f}{\partial y_3^3}$.}

\begin{proposition}[Theorem 2.11 of \cite{Reznick-1996-HomogeneousPolynomial}]
    \label{prop:apolar-adjointness}
    Suppose $p, q, r \in \RR[y_1, \dots, y_N]^{\hom}$, with degrees $\deg(p) = a, \deg(q) = b$, and $\deg(r) = a + b$.
    Then,
    \begin{equation}
        \langle pq, r \rangle_{\circ} = \frac{a!}{(a + b)!}\langle p, q(\bm\partial)r \rangle_{\circ}.
    \end{equation}
    In particular, if $\deg(p) = \deg(q) = a$, then $\langle p, q \rangle_{\circ} = p(\bm \partial) q / a!$.
\end{proposition}
\noindent
In fact, it will later be useful for us to define the following rescaled version of the apolar inner product that omits the rescaling above.
\begin{definition}
    \label{def:apolar-partial}
    For $p, q \in \RR[y_1, \dots, y_N]^{\hom}$ with $\deg(p) = \deg(q)$, let $\langle p, q \rangle_{\partial} \colonequals p(\bm \partial) q$.
\end{definition}

Using the preceding formula, we also obtain the following second important property, that of invariance under orthogonal changes of monomial basis.
\begin{proposition}
    \label{prop:apolar-orth-invariant}
    Suppose $p, q \in \RR[y_1, \dots, y_N]^{\hom}_d$ and $\bQ \in \sO(N)$.
    Then,
    \begin{equation}
        \langle p(\by), q(\by) \rangle_{\circ} = \langle p(\bQ \by), q(\bQ \by) \rangle_{\circ}.
    \end{equation}
\end{proposition}

\paragraph{Isotropic gaussians}
Using these Hilbert space structures, we may define the canonical isotropic gaussian random ``vectors'' (tensors or polynomials) in $\Sym^d(\RR^N)$ or $\RR[y_1, \dots, y_N]^{\hom}_d$.
\begin{definition}
    For $\sigma > 0$, $\mathcal{G}^{\tens}_d(N, \sigma^2)$ is the unique centered gaussian measure over $\mathsf{Sym}^d(\RR^N)$ such that, when $\bG \sim \mathcal{G}^{\tens}_d(N, \sigma^2)$, then for any $\bA, \bB \in \mathsf{Sym}^d(\RR^N)$,
    \begin{equation}
        \EE[\langle \bA, \bG \rangle \langle \bB, \bG \rangle] = \sigma^2 \langle \bA, \bB \rangle.
    \end{equation}
    Equivalently, the entries of $\bG$ have laws $G_{\bs} \sim \sN(0, \sigma^2 / \binom{d}{\freq(\bs)})$ and are independent up to equality under permutations.
    Equivalently again, letting $\bG^{(0)} \in (\RR^N)^{\otimes d}$ have i.i.d.\ entries distributed as $\sN(0, \sigma^2)$, $G_{\bs} = \frac{1}{d!}\sum_{\pi \in S_d} G^{(0)}_{s_{\pi(1)}, \ldots , s_{\pi(d)}}$.
\end{definition}
\noindent
For example, the gaussian orthogonal ensemble scaled to have the bulk of its spectrum supported asymptotically in $[-2, 2]$ is $\sG^{\tens}_{2}(N, 2 / N)$.
The tensor ensembles have also been used by \cite{RM-2014-TensorPCA} and subsequent works on tensor PCA under the name ``symmetric standard normal'' tensors.

\begin{definition}
    \label{def:isotropic-poly}
For $\sigma > 0$, $\mathcal{G}^{\poly}_d(N, \sigma^2)$ is unique centered gaussian measure over $\RR[y_1, \dots, y_N]_d^{\hom}$ such that, when $g \sim \mathcal{G}^{\poly}_d(N, \sigma^2)$, then for any $p, q \in \RR[y_1, \dots, y_N]_d^{\hom}$,
\begin{equation}
    \EE[\langle p, g \rangle_{\circ} \langle q, g \rangle_{\circ}] = \sigma^2 \langle p, q \rangle_{\circ}.
\end{equation}
Equivalently, the coefficients of $g$ are independent and distributed as $[\by^{\bs}](g) \sim \sN(0, \sigma^2 \binom{d}{\freq(\bs)})$.
\end{definition}
\noindent
See \cite{Kostlan-2002-SystemRandomPolynomials} for references to numerous works and results on this distribution over polynomials, and justification for why it is ``the most natural random polynomial.''
Perhaps the main reason is that, as a corollary of Proposition~\ref{prop:apolar-orth-invariant}, this polynomial is orthogonally invariant (unlike, say, a superficially simpler-looking random polynomial with i.i.d.\ coefficients).
\begin{proposition}
    \label{prop:G-poly-orth-invariant}
    If $g \sim \mathcal{G}^{\poly}_d(N, \sigma^2)$ and $\bQ \in \sO(N)$, then $g \eqd g \circ \bQ$.
\end{proposition}
\noindent
(Likewise, though we will not use it, $\mathcal{G}^{\tens}_d(N, \sigma^2)$ is invariant under contraction of each index with the same orthogonal matrix, generalizing the orthogonal invariance of the GOE.)

Finally, by the isotropy properties and the isometry of apolar and Frobenius inner products under the correspondences \eqref{eq:correspondence-symtens-poly} and \eqref{eq:correspondence-poly-symtens}, we deduce that these two gaussian laws are each other's pullbacks under those correspondences.
\begin{proposition}
    If $\bG \sim \sG_d^{\tens}(N, \sigma^2)$, then $\bG[\by, \ldots, \by]$ has the law $\sG_d^{\poly}(N, \sigma^2)$.
    Conversely, if $g \sim \sG_d^{\poly}(N, \sigma^2)$, and $\bG$ has entries $G_{\bs} = [\by^{\bs}](g) / \binom{d}{\freq(\bs)}$, then $\bG$ has the law $\sG_d^{\tens}(N, \sigma^2)$.
\end{proposition}

\subsection{Homogeneous Ideals and Multiharmonic Polynomials}
\label{sec:prelim-ideal-harmonic}

We now focus on homogeneous polynomials and the apolar inner product, and describe a crucial consequence of Proposition~\ref{prop:apolar-adjointness}.
Namely, for any homogeneous ideal, any polynomial uniquely decomposes into one part belonging to the ideal, and another part, apolar to the first, that satisfies a certain system of PDEs associated to the ideal.
\begin{proposition}
    \label{prop:apolar-decomp}
    Let $p_1, \dots, p_m \in \RR[y_1, \dots, y_N]^{\hom}$ and $d \geq \max_{i = 1}^m \deg(p_i)$.
    Define two subspaces of $\RR[y_1, \dots, y_N]^{\hom}_d$:
    \begin{align}
        V_{\mathcal{I}} &\colonequals \left\{\sum_{i = 1}^m p_iq_i : q_i \in \RR[y_1, \dots, y_N]^{\hom}_{d - \deg p_i}\right\}, \text{ the \emph{``ideal subspace,''} and}\\
        V_{\mathcal{H}} &\colonequals \left\{q: p_i(\bm\partial)q = 0 \text{ for all } i \in [m]\right\}, \text{ the \emph{``harmonic subspace.''}}
    \end{align}
    Then, $V_{\mathcal{I}}$ and $V_{\mathcal{H}}$ are orthogonal complements under the apolar inner product.
    Consequently, $\RR[y_1, \dots, y_N]^{\hom}_d = V_{\mathcal{I}} \oplus V_{\mathcal{H}}$.
\end{proposition}

Perhaps the most familiar example is the special case of harmonic polynomials, for which this result applies as follows.
\begin{example}
\label{ex:harmonic-polynomials}
Suppose $m = 1$, and $p_1(\by) = \|\by\|_2^2 = y_1^2 + \cdots + y_N^2$.
Then, $p_1(\bm\partial) = \Delta$, so Proposition~\ref{prop:apolar-decomp} implies that any $q \in \RR[y_1, \dots, y_N]^{\hom}_d$ may be written uniquely as $p(\by) = q_d(\by) + \|\by\|^2q_{d - 2}(\by)$ where $q_d$ is harmonic, $\deg(q_d) = d$, and $\deg(q_{d - 2}) = d - 2$.
Repeating this inductively, we obtain the familiar fact from harmonic analysis that we may in fact expand
\begin{equation}
    p(\by) = \sum_{a = 0}^{\lfloor d / 2 \rfloor} \|\by\|^{2a}q_{d - 2a}(\by)
\end{equation}
where each $q_{i}$ is harmonic with $\deg(q_i) = i$ and the $q_i$ are uniquely determined by $p$.
\end{example}
\noindent
This is sometimes called the ``Fischer decomposition;'' see also the ``Expansion Theorem'' in \cite{Reznick-1996-HomogeneousPolynomial} for a generalization of this type of decomposition.

We will be especially interested in computing apolar projections onto $V_{\mathcal{H}}$ (or, equivalently, $V_{\mathcal{I}}$).
We therefore review a few situations where there are direct methods for carrying out such computations.
Again, the clearest case is that of harmonic polynomials.
\begin{proposition}[Theorem 1.7 of \cite{AR-1995-HarmonicPolynomialsDirichlet}; Theorem 5.18 of \cite{ABW-2013-HarmonicFunctionTheory}]
    Suppose $N \geq 3$.\footnote{A variant of this result also holds for $N = 2$; see Section 4 of \cite{AR-1995-HarmonicPolynomialsDirichlet}.}
    Let $V_{\mathcal{H}} \subset \RR[y_1, \dots, y_N]^{\hom}_d$ be the subspace of harmonic polynomials ($q(\by)$ with $\Delta q = 0$), and let $P_{\mathcal{H}}$ be the apolar projection to $V_{\mathcal{H}}$.
    Define
    \begin{align}
        \phi(\by) &\colonequals \|\by\|^{2 - N} \text{ (the \emph{Green's function of $\Delta$}), and} \\
        K[u](\by) &\colonequals \phi(\by)u(\by / \|\by\|^2) \text{ (the \emph{Kelvin transform}),}
    \end{align}
    the latter defined for $u: \RR^N \setminus \{\bm 0\} \to \RR$ a smooth function.
    Let $p \in \RR[y_1, \dots, y_N]^{\hom}_d$.
    Then,
    \begin{equation}
        P_{\mathcal{H}}[q] = \frac{1}{\prod_{i = 0}^{d - 1}(2 - N - 2i)}K\left[q(\bm\partial)\phi\right].
    \end{equation}
\end{proposition}
\noindent
Roughly speaking, the Kelvin transform is a generalization to higher dimensions of inversion across a circle, so this result says that apolar projections to harmonic polynomials may be computed by inverting corresponding derivatives of the Green's function of $\Delta$.

This result has a long history.
At least for $N = 3$, the idea and its application to the expansion of Example~\ref{ex:harmonic-polynomials} were already present in classical physical reasoning of Maxwell~\cite{Maxwell-1873-Treatise1}.
Soon after, Sylvester~\cite{Sylvester-1876-SphericalHarmonics} gave a mathematical treatment, mentioning that an extension to other $N$ is straightforward.
See Section VII.5.5 of \cite{CH-1962-MathematicalPhysics1} on ``The Maxwell-Sylvester representation of spherical harmonics'' for a modern exposition.
These ideas were rediscovered by~\cite{AR-1995-HarmonicPolynomialsDirichlet}; there and in the later textbook treatment~\cite{ABW-2013-HarmonicFunctionTheory} there is greater emphasis on $P_{\sH}$ being a projection, though the fact that the apolar inner product makes it an \emph{orthogonal} projection goes unmentioned.
Some further historical discussion is given in an unpublished note of Gichev~\cite{Gichev-XXXX-HarmonicComponentHomogeneous} as well as Appendix A of the lecture notes \cite{Arnold-1997-LecturesPDE}.

When we seek to apply these ideas in our setting, we will want to project to \emph{multiharmonic} polynomials, which satisfy $p_i(\bm\partial)q = 0$ for several polynomials $p_1, \dots, p_m$.\footnote{Unfortunately, the term \emph{multiharmonic function} is also sometimes used to refer to what is usually called a \emph{pluriharmonic function}, the real or imaginary part of a holomorphic function of several variables, or to what is usually called a \emph{polyharmonic function}, one that satisfies $\Delta^mq = 0$ for some $m \in \NN$.}
In our case the polynomials will be quadratic, but a generalization to arbitrary polynomials is also sensible.
This question has been studied much less.
The main work we are aware of in this direction is due to Clerc \cite{Clerc-2000-KelvinTransform} (whose Green's function construction was suggested earlier in Herz's thesis \cite{Herz-1955-BesselFunctionsMatrix}; see Lemma~1.6 of the latter), where the $p_i$ are quadratic forms with the basis elements of a Jordan subalgebra of $\RR^{r \times r}_{\sym}$.
The following is one, essentially trivial, instance of those results.
\begin{proposition}
    \label{prop:basis-harmonic-proj}
    Let $\bv_1, \dots, \bv_r$ be an orthonormal basis of $\RR^r$.
    Define an associated Green's function and Kelvin transform
    \begin{align}
      \phi(\bz) &\colonequals \prod_{i = 1}^r \langle \bv_i, \bz \rangle, \\
      K[f](\bz) &\colonequals \phi(\bz)f\left(\sum_{i = 1}^r \langle \bv_i, \bz \rangle^{-1} \bv_i\right).
    \end{align}
    Let $p \in \RR[\bx_1, \dots, \bx_r]$.
    Then, the apolar projection of $p$ to the harmonic subspace is $K[p(\bm\partial) \phi]$.
\end{proposition}
\noindent
In this case, it is easy to give a hands-on proof: one may write $p$ in the monomial basis $\langle \bv_i, \bz \rangle$, and in this basis the desired projection is just the multilinear part of $p$.
On the other hand, we have $p(\bm\partial)\phi = q / \phi$, where $q$ is the multilinear part of $p$, and the result follows.
Though this is a simple derivation, we will see that extending it to overcomplete families of vectors $\bv_i$ in fact forms one of the key heuristic steps in our derivation.

\subsection{\Mobius\ Functions of Partially Ordered Sets}

Finally, we review some basic concepts of the combinatorics of partially ordered sets (henceforth \emph{posets}).
Recall that a poset is a set $\sP$ equipped with a relation $\leq$ that satisfies reflexivity ($x \leq x$ for all $x \in \sP$), antisymmetry (if $x \leq y$ and $y \leq x$ then $x = y$), and transitivity (if $x \leq y$ and $y \leq z$, then $x \leq z$).
For the purposes of this paper, we will assume all posets are finite.
The following beautiful and vast generalization of the classical \Mobius\ function of number theory was introduced by Rota in \cite{Rota-1964-Foundations} (the reference's introduction gives a more nuanced discussion of the historical context at the time).

\begin{definition}[Poset \Mobius\ function]
    Let $\sP$ be a poset.
    Then, the \emph{\Mobius\ funcion} of $\sP$, denoted $\mu_{\sP}(x, y)$, is defined over all pairs $x \leq y$ by the relations
    \begin{align}
      \mu_{\sP}(x, x) &= 1, \\
      \sum_{x \leq y \leq z} \mu_{\sP}(x, y) &= 0 \text{ for all } x < z.
    \end{align}
\end{definition}

The key consequence of this definition is the following general inclusion-exclusion principle over posets, again a vast generalization of both the \Mobius\ inversion formula of number theory and the ordinary inclusion-exclusion principle over the poset of subsets of a set.
\begin{proposition}[Poset \Mobius\ inversion]
    If $\sP$ has a minimal element, $f: \sP \to \RR$ is given, and $g(x) \colonequals \sum_{y \leq x} f(y)$, then $f(x) = \sum_{y \leq x}\mu_{\sP}(y, x)g(y)$.
    Similarly, if $\sP$ has a maximal element and $g(x) \colonequals \sum_{y \geq x} f(y)$, then $f(x) = \sum_{y \geq x} \mu_{\sP}(x, y)g(y)$.
\end{proposition}
\noindent
In addition to \cite{Rota-1964-Foundations}, the reader may consult, e.g., \cite{BG-1975-MobiusInversionCombinatorial} for some consequences of this result in enumerative combinatorics.

We give three examples of \Mobius\ functions of posets of partitions that will be useful in our calculations.
The first concerns subsets and corresponds to the classical inclusion-exclusion principle, and the latter two concern partitions of a set.
\begin{example}[Subsets]
    \label{ex:mobius-subset}
    Give $2^{[m]}$ the poset structure of $S \leq T$ whenever $S \subseteq T$.
    Write $\mu_{\mathsf{Subset}}(\cdot, \cdot)$ for the \Mobius\ function of $[m]$.
    Then,
    \begin{equation}
        \mu_{\mathsf{Subset}}(S, T) = (-1)^{|T| - |S|}.
    \end{equation}
\end{example}

\begin{example}[Partitions \cite{Rota-1964-Foundations}]
    \label{ex:mobius-part}
    Let $\Part([m])$ denote the poset of partitions of $[m]$, where $\pi \leq \rho$ whenever $\pi$ is a refinement of $\rho$.
    Write $\mu_{\Part}(\cdot, \cdot)$ for the \Mobius\ function of $\Part([m])$, eliding $m$ for the sake of brevity.
    Then,
    \begin{equation}
        \mu_{\Part}(\pi, \rho) = \prod_{A \in \rho}(-1)^{\#\{B \in \pi: B \subseteq A\} - 1} (\#\{B \in \pi: B \subseteq A\} - 1)!.
    \end{equation}
    In particular, letting $\oslash \colonequals \{\{1\}, \dots, \{m\}\}$ be the unique minimal element of $\Part([m])$, we have
    \begin{equation}
        \mu_{\Part}(\oslash, \rho) = \prod_{A \in \rho}(-1)^{|A| - 1} (|A| - 1)!.
    \end{equation}
\end{example}

\begin{example}[Partitions into even parts \cite{Sylvester-1976-ContinuousSpinIsing}]
    \label{ex:mobius-part-even}
    For $m \geq 2$ even, let $\EvenPart([m])$ denote the poset of partitions of $[m]$ into even parts, where $\pi \leq \rho$ whenever $\pi$ is a refinement of $\rho$, along with the additional formal element $\oslash$ with $\oslash \leq \pi$ for all partitions $\pi$.
    Again write $\mu_{\EvenPart}(\cdot, \cdot)$ for the associated \Mobius\ function, eliding $m$ for the sake of brevity.
    Let the sequence $\nu(k)$ for $k \geq 0$ be defined by the exponential generating function $\log \cosh(x) \equalscolon \sum_{k = 0}^{\infty} \frac{\nu(k)}{k!} x^k$, or equivalently $\tanh(x) \equalscolon \sum_{k = 0}^{\infty} \frac{\nu(k + 1)}{k!}x^k$. Then,
    \begin{equation}
        \mu_{\EvenPart}(\oslash, \rho) = -\prod_{A \in \rho} \nu(|A|).
    \end{equation}
    On the other hand, if $\pi > \oslash$, the $[\pi, \rho]$ is isomorphic to a poset of ordinary partitions, so we recover
    \begin{equation}
        \mu_{\EvenPart}(\pi, \rho) = \prod_{A \in \rho}(-1)^{\#\{B \in \pi: B \subseteq A\} - 1} (\#\{B \in \pi: B \subseteq A\} - 1)!.
    \end{equation}
\end{example}
\noindent
There is no convenient closed form for $\nu(k)$, but a combinatorial interpretation (up to sign) is given by $(-1)^k\nu(2k)$ counting the number of alternating permutations of $2k + 1$ elements.
This fact, as a generating function identity, is a classical result due to Andr\'{e} \cite{Andre-1881-AlternatingPermutations} who used it to derive the asymptotics of $\nu$; see also \cite{Stanley-2010-AlternatingPermutations} for a survey.
The connection with \Mobius\ functions was first observed in Sylvester's thesis \cite{Sylvester-1976-ContinuousSpinIsing}, and Stanley's subsequent work \cite{Stanley-1978-ExponentialStructures} explored further situations where the \Mobius\ function of a poset is given by an exponential generating function.
Some of our calculations in Section~\ref{sec:poset} indicate that the poset defined there, while not one of Stanley's ``exponential structures,'' is still amenable to analysis via exponential generating functions, suggesting that the results of \cite{Stanley-1978-ExponentialStructures} might be generalized to posets having more general self-similarity properties.

\section{Positivity-Preserving Extensions from Surrogate Tensors}
\label{sec:informal-deriv}

We now explain how we arrive at the extension formula \eqref{eq:lifting-prediction} for $\tEE_{\bM}$ (which we will abbreviate simply $\tEE$ in this section) in Definition~\ref{def:lifting}.
Most of the discussion in this section will not be fully mathematically rigorous; however, these heuristic calculations will give important context to our later proof techniques.

In Section~\ref{sec:deg2-assumptions}, we describe the informal assumptions on $\bM$ that we make for these calculations.
In Section~\ref{sec:initial-conj}, we review the conjecture of \cite{KB-2019-Degree4SK-Arxiv} that forms our starting point.
In the remaining sections, we carry out the relevant calculations, showing how we reach the extension formula.
We remind the reader that we will be making extensive use of connections between symmetric tensors and homogeneous polynomials, which we have introduced in Section~\ref{sec:prelim-symtens-poly}.

\subsection{Notations and Assumptions for Degree 2 Pseudomoment Matrix}
\label{sec:deg2-assumptions}

Suppose $\bM \in \RR^{N \times N}_{\sym}$ with $\bM \succeq \bm 0$ and $M_{ii} = 1$ for all $i \in [N]$.
We will assume this matrix is fixed for the remainder of Section~\ref{sec:informal-deriv}.
Since $\bM \succeq \bm 0$, we may further suppose that, for some $\bV \in \mathbb{R}^{r \times N}$ with $r \leq N$ and having full row rank, $\bM = \bV^{\top}\bV$.
In particular then, $\rank(\bM) = r$.
Since this number will come up repeatedly, we denote the ratio between the rank of $\bM$ and the ambient dimension by
\begin{equation}
    \delta \colonequals \frac{r}{N}.
\end{equation}
Writing $\bv_1, \dots, \bv_N \in \RR^r$ for the columns of $\bV$, we see that $\bM$ is the Gram matrix of the $\bv_i$, and since $\diag(\bM) = \one$, $\|\bv_i\|_2 = 1$ for all $i \in [N]$.

We now formulate our key assumption on $\bM$.
For the purposes of our derivations in this section, it will suffice to leave the ``approximate'' statements below vague.

\begin{assumption}[Informal]
    \label{ass:M}
    The following equivalent conditions on $\bM$ hold:
\begin{enumerate}
\item All non-zero eigenvalues of $\bM$, of which there are $r$, are approximately equal.
\item $\bM$ is approximately equal to a projection matrix to an $r$-dimensional subspace of $\RR^N$, multiplied by $\delta^{-1}$.
\item $\bV \bV^{\top} \approx \delta^{-1} \bm I_r$.
\item The vectors $\bv_1, \dots, \bv_N$ approximately form a \emph{unit-norm tight frame} (see, e.g., \cite{Waldron-2018-FiniteTightFrames}).
\end{enumerate}
\end{assumption}
\noindent
We will see that, to derive the extension of $\bM$, we may reason as if the approximate equalities are exact and obtain a sound result.

In light of Condition 3 above, it will be useful to define a normalized version of $\bV$, whose \emph{rows} have approximately unit norm: we let $\what{\bV} \colonequals \delta^{1/2}\bV$, so that $\what{\bV}\what{\bV}^{\top} \approx \bm I_r$.
The use of this matrix will be that it can be extended, by adding rows, to an orthogonal matrix (this is equivalent to the \emph{Naimark complement} construction in frame theory; see Section 2.8 of \cite{Waldron-2018-FiniteTightFrames}).

\subsection{Initial Conjecture of \cite{KB-2019-Degree4SK-Arxiv}: Conditioning Gaussian Symmetric Tensors}
\label{sec:initial-conj}

We first review a construction conjectured in Section 5 of the paper \cite{KB-2019-Degree4SK-Arxiv} of the author's with Bandeira, which gives a way to define a pseudoexpectation that satisfies many of the necessary constraints.
This description is relatively straightforward, but leaves the actual pseudoexpectation values implicit, making it difficult to verify that all constraints are satisfied or to proceed towards a proof.
Our goal in the remainder of this section will be to derive those values from the following more conceptual description.

The initial idea is to build pseudoexpectation values as second moments of the entries of a random symmetric tensor.
That is, for degree $2d$, we build $\tEE$ using a random $\bG^{(d)} \in \mathsf{Sym}^d(\RR^N)$ and taking, for multisets of indices $S, T \in \sM_d([N])$,
\begin{equation}
\label{eq:pseudoexpectation-prediction}
\text{`` }\tEE\left[\prod_{i \in S} x_{i} \prod_{j \in T} x_{j}\right] \colonequals \EE\left[G_{S}^{(d)}G_{T}^{(d)}\right].\text{ ''}
\end{equation}
At an intuitive level, if $\bx$ has the pseudodistribution encoded by $\tEE$, then one should think of identifying
\begin{equation}
    \label{eq:Gk-xk-intuition}
    \text{`` }\bG^{(d)} = \bx^{\otimes d}.\text{ ''}
\end{equation}
The key point is that, while we cannot model the values of $\tEE$ as being the moments of an actual (``non-pseudo'') random vector $\bx$, we can model them as being the moments of the tensors $\bG^{(d)}$, which are surrogates for the tensor powers of $\bx$.

The most immediate problem with \eqref{eq:pseudoexpectation-prediction} is that, \emph{a priori}, $\tEE$ is not well-defined: for that, the quantity $\EE[G_{S}^{(d)}G_{T}^{(d)}]$ must be invariant under permutations that leave $S + T$ unchanged.
Forcing $\bG^{(d)}$ to be a \emph{symmetric} tensor, as we have indicated above, ensures that some of these permutation equalities will hold (namely, those that leave both $S$ and $T$ unchanged).
It is far from obvious, however, how to design $\bG^{(d)}$ so that the other permutation equalities hold (even approximately).
We will only be able to verify that those equalities do hold at the end of our heuristic discussion, in Section~\ref{sec:informal:simplifying}.
To work more precisely with pseudoexpectations-like operators that are not guaranteed to satisfy these symmetry conditions, we use the following notation.
\begin{definition}[Bilinear pseudoexpectation]
    \label{def:bilinear-pe}
    For a bilinear operator, for some ambient degree $d$, $\tEE: \RR[x_1, \dots, x_N]_{\leq d} \times \RR[x_1, \dots, x_N]_{\leq d} \to \RR$, we denote its action by $\tEE(p(\bx), q(\bx))$.
    We use parentheses to distinguish it from a linear operator $\tEE: \RR[x_1, \dots, x_N]_{\leq 2d} \to \RR$, whose action we denote with brackets, $\tEE[p(\bx)]$.
\end{definition}

In exchange for the difficulty of ensuring these permutation equalities, our construction makes it easy to ensure that the other constraints are satisfied.
Most importantly, a pseudoexpectation of the form \eqref{eq:pseudoexpectation-prediction} is, by construction, positive semidefinite.
Also, assuming that the permutation equalities mentioned above hold, the remaining linear constraints on a pseudoexpectation will be satisfied provided that the following conditions hold:
\begin{align}
    \EE[(G_{S}^{(d)})^2] &= 1 \text{ for all } S \in \sM_d([N]), \label{eq:cond-norm} \\
    G^{(0)}_{\emptyset} &= 1 \label{eq:G0}, \\
    G_{S^{\prime} + \{i, i\}}^{(d)} &= G^{(d - 2)}_{S^{\prime}} \text{ for all } S^{\prime} \in \sM_{d -2}([N]) \text{ and } i \in [N]. \label{eq:cond-consistency}
\end{align}
The first constraint directly expresses the hypercube constraints $x_i^2 = 1$, the second the normalization constraint $\tEE[1] = 1$, and the third a ``compatibility'' condition among the $\bG^{(d)}$ that mimics the compatibility conditions that must hold among the $\bx^{\otimes d}$ for $\bx$ a hypercube vector.

Besides these general considerations, we must also build $\tEE$ that extends the prescribed degree~2 values, $\tEE[x_ix_j] = M_{ij}$.
We isolate one specific consequence of fixing these values, which plays an important role when $\bM$ is rank-deficient.
Intuitively, the idea is that if $\bM = \tEE[\bx\bx^{\top}]$ has row space $\mathsf{row}(\bM)$, then $\bx \in \mathsf{row}(\bM)$ with probability 1, and any vector formed by ``freezing'' all but one index of $\bx^{\otimes d}$ should also belong to $\mathsf{row}(\bM)$---indeed, any such vector is just $\pm \bx$.
Importing this constraint to the $\bG^{(d)}$, if we let $\bG^{(d)}[S^{\prime}, :] \colonequals (G^{(d)}_{S^{\prime} + \{i\}})_{i = 1}^N$, then we must have:
\begin{equation}
    \bG^{(d)}[S^{\prime}, :] \in \mathsf{row}(\bM) \text{ for all } S^{\prime} \in \sM_{d - 1}([N]).
    \label{eq:cond-spectral}
\end{equation}
This is a system of linear conditions on $\bG^{(d)}$: it may be written $\langle \bw_i, \bG^{(d)}[S^{\prime}, :]\rangle = \bm 0$ for $\bw_i$ a basis of $\mathsf{row}(\bM)^{\perp} = \ker(\bM)$.

We now have four desiderata for $\bG^{(d)}$: the family of quadratic normalization conditions \eqref{eq:cond-norm}, the ``base case'' \eqref{eq:G0}, and the two families of linear conditions \eqref{eq:cond-consistency} and \eqref{eq:cond-spectral}.
To describe a suitable law for $\bG^{(d)}$, we begin with the isotropic gaussian symmetric tensor, $\mathcal{G}^{\tens}_d(N, \sigma_d^2)$ (see Section~\ref{sec:prelim-symtens-poly}).
We leave a free scaling parameter $\sigma_d^2$ to select later so that \eqref{eq:cond-norm} is satisfied.
Then, we condition this gaussian distribution to satisfy the linear constraints \eqref{eq:cond-consistency} and \eqref{eq:cond-spectral}.
More precisely, we define $\bG^{(d)}$ iteratively for $d \geq 0$ as follows.

\begin{leftbar}
\vspace{-1.35em}
\paragraph{Pseudoexpectation Prediction (Tensors)}
Let $\bw_1, \dots, \bw_{N - r}$ be a basis of $\ker(\bM)$. Define a jointly gaussian collection of tensors $\bG^{(d)} \in \Sym^{d}(\RR^N)$ as follows.
\begin{enumerate}
    \item $\bG^{(0)} \in \mathsf{Sym}^0(\RR^N)$ is a scalar, with one entry $G^{(0)}_{\emptyset} = 1$.
    \item For $d \geq 1$, $\bG^{(d)}$ has the law of $\mathcal{G}_d^{\tens}(N, \sigma_d^2)$, conditional on the following two properties:
    \begin{enumerate}
        \item If $d \geq 2$, then for all $S^{\prime} \in \sM_{d - 2}([N])$ and $i \in [N]$, $G_{S^{\prime} + \{i, i\}}^{(d)} = G_{S^{\prime}}^{(d - 2)}$.
        \item For all $S^{\prime} \in \sM_{d - 1}([N])$ and $i \in [N - r]$, $\langle \bw_i, \bG^{(d)}[S^{\prime}, :]\rangle = 0$.
    \end{enumerate}
\end{enumerate}
Then, for $S, T \in \sM_d([N])$, set
\begin{equation}
\tEE(\bx^{S}, \bx^T) \colonequals \EE\left[G_{S}^{(d)}G_{T}^{(d)}\right].
\label{eq:tEE-prediction-2}
\end{equation}
\vspace{-1em}
\end{leftbar}

If it is possible to choose $\sigma_d^2$ such that the right-hand side is approximately invariant under permutations that preserve $S + T$ and such that $\EE[(G_{S}^{(d)})^2] \approx 1$ for all $S, T \in \sM_{d}([N])$, then this gives an approximate construction of a degree $2d$ pseudoexpectation extending $\bM$.

\subsection{Conditioning by Translating to Homogeneous Polynomials}
\label{sec:informal:condition-poly}

We now apply the linear-algebraic rules for conditioning gaussian vectors to compute the means and covariances of the entries of each $\bG^{(d)}$, which lets us write down the right-hand side of \eqref{eq:tEE-prediction-2} and check if it indeed defines an approximately valid pseudoexpectation.

It turns out that it is easier to interpret this calculation in terms of homogeneous polynomials rather than symmetric tensors.
Passing each $\bG^{(d)}$ through the isometry between $\Sym^d(\RR^N)$ and $\RR[y_1, \dots, y_N]^{\hom}_d$ described in Section~\ref{sec:prelim-symtens-poly}, we find an equivalent construction in terms of random polynomials $g^{(d)} \in \RR[y_1, \dots, y_N]^{\hom}_d$.
Though this may seem unnatural at first---bizarrely, we will be defining $\tEE$ for each degree in terms of the correlations of various \emph{coefficients} of a random polynomial---we will see that viewing the extraction of coefficients in terms of the apolar inner product brings forth an important connection to multiharmonic polynomials, allowing us to use a variant of the ideas in Section~\ref{sec:prelim-ideal-harmonic} to complete the calculation.

\begin{leftbar}
\vspace{-1.35em}
\paragraph{Pseudoexpectation Prediction (Polynomials)}
Let $\bw_1, \dots, \bw_{N - r}$ be a basis of $\ker(\bM)$.
Define a jointly gaussian collection of polynomials $g^{(d)} \in \RR[y_1, \dots, y_N]^{\hom}_d$ as follows.
\begin{enumerate}
    \item $g^{(0)}(\by) = 1$.
    \item For $d \geq 1$, $g^{(d)}$ has the law of $\mathcal{G}_d^{\poly}(N, \sigma_d^2)$, conditional on the following two properties:
    \begin{enumerate}
        \item If $d \geq 2$, then for all $i \in [N]$ and $S^{\prime} \in \sM_{d - 2}([N])$, $\langle g^{(d)}, \by^{S^{\prime}}y_i^2\rangle_{\circ} = \langle g^{(d - 2)}, \by^{S^{\prime}}\rangle_{\circ}$.
        \item For all $S^{\prime} \in \sM_{d - 1}([N])$ and $i \in [N - r]$, $\langle g^{(d)}, \by^{S^{\prime}} \langle \bw_i, \by \rangle\rangle_{\circ} = 0$.
    \end{enumerate}
\end{enumerate}
Then, for $S, T \in \sM_d([N])$, set
    \begin{equation}
\tEE(\bx^{S}, \bx^{T}) \colonequals \EE\left[\langle g^{(d)}, \by^{S}\rangle_{\circ}\, \langle g^{(d)}, \by^{T}\rangle_{\circ}\right].
\label{eq:tEE-prediction-poly}
\end{equation}
\vspace{-1em}
\end{leftbar}

The most immediate advantage of reframing our prediction in this way is that it gives us access to the clarifying concepts of ``divisibility'' and ``differentiation,'' whose role is obscured by the previous symmetric tensor language.
Moreover, these are nicely compatible with the apolar inner product per Proposition~\ref{prop:apolar-adjointness}.
Indeed, thanks to these connections it is now possible to carry out our calculation completely in terms of the concepts from Section~\ref{sec:prelim-symtens-poly}.
We briefly outline the reasoning below and present the final result.
More detailed justification is given in Appendix~\ref{app:pf:lem:poly-cond}.

Roughly speaking, conditioning on Property (b) above projects $g^{(d)}$ to the subspace of polynomials depending only on $\what{\bV}\by$.
On the other hand, by Proposition~\ref{prop:G-poly-orth-invariant}, $g^{(d)}$ is invariant under compositions with orthogonal matrices, and by our assumptions on $\what{\bV}$, it is merely the upper $r \times N$ block of some orthogonal matrix.
From this, one may compute that, if $g_1^{(d)}$ has the law of $\sG^{\poly}_d(N, \sigma_d^2)$ conditional on Property (b), then the collection of coefficients $(\langle g_1^{(d)}, \by^S\rangle_{\circ})_{S \in \sM_d([N])}$ has the same law as the collection $(\langle h^{(d)}, (\what{\bV}^{\top} \bz)^S\rangle_{\circ} )_{S \in \sM_d([N])}$ for $h^{(d)}(\bz) \sim \sG^{\poly}_d(r, \sigma_d^2)$.
(We use $\by = (y_1, \dots, y_N)$ for formal variables of dimension $N$ and $\bz = (z_1, \dots, z_r)$ for formal variables of dimension $r$.)
Thus conditioning on Property (b) is merely a dimensionality reduction of the canonical gaussian polynomial, in a suitable basis.

Conditioning $h^{(d)}$ as above on Property (a) brings in the ideal and harmonic subspaces discussed in Section~\ref{sec:prelim-ideal-harmonic}.
Let us define
\begin{align}
  V_{\sI} &\colonequals \left\{\sum_{i = 1}^N \langle \bv_i, \bz \rangle^2 q_i(\bz) : q_i \in \RR[z_1, \dots, z_r]^{\hom}_{d - 2}\right\}, \\
  V_{\sH} &\colonequals \left\{q \in \RR[z_1, \dots, z_r]^{\hom}_d: \langle \bv_i, \bm\partial \rangle^2 q = 0 \text{ for all } i \in [N]\right\},
\end{align}
instantiations of the constructions from Proposition~\ref{prop:apolar-decomp} for the specific polynomials $\{\langle \bv_i, \bz \rangle^2\}_{i = 1}^N$.
Conditioning on Property (a) fixes the component of $h$ belonging to $V_{\sI}$, leaving a fluctuating part equal to the apolar projection of $h^{(d)}$ to $V_{\sH}$.
Completing this calculation gives the following.

\begin{lemma}
    \label{lem:poly-cond}
    Suppose that the conditions of Assumption~\ref{ass:M} hold exactly.
    Let $P_{\sI}$ and $P_{\sH}$ be the apolar projections to $V_{\sI}$ and $V_{\sH}$, respectively.
    For each $S \in \sM_d([N])$, let $r_S \in \RR[x_1, \dots, x_N]_d^{\hom}$ be a polynomial having
    \begin{align}
      P_{\sI}[(\bV^{\top}\bz)^{S}] &= r_{S}(\bV^{\top}\bz),  \label{eq:r-poly-cond-1} \\
      r_{S}(\bx) &= \sum_{i = 1}^N x_i^{2d_i} r_{S, i}(\bx) \text{ for } d_i \geq 1, r_{S, i} \in \RR[x_1, \dots, x_N]_{d - 2d_i}^{\hom}, \label{eq:r-poly-cond-2}
                   \intertext{and further define}
                   r_{S}^{\downarrow}(\bx) &= \sum_{i = 1}^N r_{S, i}(\bx) \in \RR[x_1, \dots, x_N]_{\leq d - 2},
    \end{align}
    where we emphasize that $r_S^{\downarrow}$ is \emph{not} necessarily homogeneous.
    Set $h_S(\bx) \colonequals \bx^S - r_S(\bx)$, whereby $P_{\sH}[(\bV^{\top}\bz)^S] = h_S(\bV^{\top}\bz)$.
    Then, the right-hand side of \eqref{eq:tEE-prediction-poly} is
    \begin{equation}
        \label{eq:tEE-prediction-decomp}
        \tEE(\bx^S, \bx^T) = \underbrace{\tEE(r_{S}^{\downarrow}(\bx),  r_{T}^{\downarrow}(\bx))}_{\text{``ideal'' term}} + \underbrace{\sigma_d^2\delta^d \cdot \langle h_S(\bV^{\top}\bz), h_T(\bV^{\top}\bz)\rangle_{\circ}}_{\text{``harmonic'' term}}.
    \end{equation}
\end{lemma}
\noindent
Thus our prediction for the degree $2d$ pseudoexpectation values decomposes according to the ideal-harmonic decomposition of the input; the ideal term depends only on the pseudoexpectation values of strictly lower degree, while the harmonic term is a new contribution that is, in a suitable spectral sense, orthogonal to the ideal term.
In this way, one may think of building up the spectral structure of the pseudomoment matrices of $\tEE$ by repeatedly lifting the pseudomoment matrix two degrees lower into a higher-dimensional domain, and then adding a new component orthogonal to the existing one.

\begin{remark}[Multiharmonic basis and block diagonalization]
    \label{rem:heuristic-block-diag}
We also mention a different way to view this result that will be more directly useful in our proofs later.
Defining $h_S^{\downarrow}(\bx) \colonequals \bx^S - r_S^{\downarrow}(\bx)$, note that we expect $\tEE(h_S^{\downarrow}(\bx), r_S^{\downarrow}(\bx)) = \tEE(h_S(\bx), r_S(\bx)) = 0$ since the ideal and harmonic subspaces are apolar.
Thus, we also expect to have
\begin{equation}
    \label{eq:heuristic-gram-mx}
    \tEE(h_S^{\downarrow}(\bx), h_T^{\downarrow}(\bx)) = \sigma_d^2\delta^d \cdot \left\langle h_S(\bV^{\top}\bz), h_T(\bV^{\top}\bz) \right\rangle_{\circ}.
\end{equation}
The $h_S^{\downarrow}(\bx)$ are a basis modulo the idea generated by the constraint polynomials $x_i^2 - 1$, which we call the \emph{multiharmonic basis}.
Since the inner product on the right-hand side is zero unless $|S| = |T|$, this basis achieves a \emph{block diagonalization} of the pseudomoment matrix.
This idea turns out to be easier to use to give a proof of positivity of $\tEE$ than the ideal-harmonic decomposition of \eqref{eq:tEE-prediction-decomp}.
\end{remark}

\subsection{Heuristic for Projecting to Multiharmonic Polynomials}

We have reduced our task to understanding how a polynomial of the form $(\bV^{\top} \bz)^{S} = \prod_{i \in S} \langle \bv_i, \bz \rangle$ decomposes into an ideal part of a linear combination of multiples of $\langle \bv_i, \bz \rangle^2$ and a harmonic part that is a zero of any linear combination of the differential operators $\langle \bv_i, \bm\partial \rangle^2$.
In this section, we develop a heuristic method to compute these projections.
Since $(\bV^{\top}\bz)^{S}$ is the sum of the two projections, it suffices to compute either one.
We will work with the projection to the multiharmonic subspace $V_{\sH}$, because that is where we may use the intuition discussed in Section~\ref{sec:prelim-ideal-harmonic}.
We warn in advance that this portion of the derivation is the least mathematically rigorous; our goal is only to obtain a plausible prediction for the projections in question.

Recall that the basic theme discussed in Section~\ref{sec:prelim-ideal-harmonic} was that the projection $P_{\sH}[q]$ may be computed by applying the differential operator $q(\bm\partial)$ to a suitable ``Green's function'' for the system of PDEs $\langle \bv_i, \bm\partial \rangle^2 f = 0$, and then taking a suitable ``Kelvin transform'' or ``inversion'' of the output.
Moreover, in Proposition~\ref{prop:basis-harmonic-proj}, we saw that when $N = r$ and the $\bv_i$ form an orthonormal basis, then one may take the Green's function $\phi(\bz) = \prod_{i = 1}^r \langle \bv_i, \bz \rangle$ and the Kelvin transform $K[f](\bz) = \phi(\bz)f(\sum_{i = 1}^r \langle \bv_i, \bz \rangle^{-1} \bv_i)$.
Crucially, we were able to build the mapping $\bz \mapsto \sum_{i = 1}^r \langle \bv_i, \bz \rangle^{-1} \bv_i$ which coordinatewise inverts the coefficients $\langle \bv_i, \bz \rangle$.

The difference in our setting is that $N > r$, and the $\bv_i$ form an overcomplete set.
In particular, in the Kelvin transform, it is not guaranteed that, given some $\bz$, there exists a $\bz^{\prime}$ such that $\bV^{\top} \bz^{\prime}$ is the coordinatewise reciprocal of $\bV^{\top}\bz$, so a genuine ``inversion'' like in the case of an orthonormal basis may be impossible.\footnote{We also remark that it is not possible to apply the results of \cite{Clerc-2000-KelvinTransform} mentioned earlier directly, since the  $\bv_i\bv_i^{\top}$ typically do not span a Jordan algebra: $\frac{1}{2}(\bv_i\bv_i^{\top} \bv_j\bv_j^{\top} + \bv_j\bv_j^{\top}\bv_i\bv_i^{\top})$ need not be a linear combination of the $\bv_k\bv_k^{\top}$.}
Nonetheless, let us continue with the first part of the calculation by analogy with Proposition~\ref{prop:basis-harmonic-proj}.
Define
\begin{equation}
    \phi(\bz) \colonequals \prod_{i = 1}^N \langle \bv_i, \bz \rangle.
\end{equation}
Then, one may compute inductively by the product rule that
\begin{equation}
  (\bV^{\top}\bm\partial)^S\phi(\bz) = \left(\sum_{\sigma \in \Part(S)} \prod_{A \in \sigma}(-1)^{|A| - 1}(|A| - 1)!\left\{\sum_{a = 1}^N \prod_{j \in A}M_{aj} \cdot \langle \bv_a, \bz \rangle^{-|A|}\right\}\right) \phi(\bz).
\end{equation}
(That various summations over partitions arise in such calculations is well-known; see, e.g., \cite{Hardy-2006-CombinatoricsPartialDerivatives} for a detailed discussion.)
We now take a leap of faith: despite the preceding caveats, let us suppose we could make a fictitious mapping $\widetilde{F}: \RR^r \to \RR^r$ that would invert the values of each $\langle \bv_i, \bz \rangle$, i.e., $\langle \bv_i, \widetilde{F}(\bz)\rangle = \langle \bv_i, \bz \rangle^{-1}$ for each $i \in [N]$.
Then, we would define a Kelvin transform (also fictitious) by $\widetilde{K}[f](\bz) = f(\widetilde{F}(\bz)) \cdot \phi(\bz)$.
Using this, and noting that $\phi(\widetilde{F}(\bz)) = \phi(\bz)^{-1}$, we predict
\begin{align}
  P_{\sH}\left[(\bV^{\top}\bz)^S\right]
  &= \widetilde{K}\left[(\bV^{\top} \bm\partial)^S\phi\right](\bz) \nonumber \\
  &= \sum_{\sigma \in \Part(S)}\prod_{A \in \sigma}(|A| - 1)!\left\{\sum_{a = 1}^N \prod_{j \in A}M_{aj} \cdot \langle \bv_a, \bz \rangle^{|A|}\right\}.
\end{align}

We make one adjustment to this prediction: when $|A| = 1$ with $A = \{i\}$, then the inner summation is $\sum_{a = 1}^N M_{ai}\langle \bv_a, \bz \rangle = (\bM\bV^{\top}\bz)_i = (\bV^{\top}\bV\bV^{\top}\bz)_i \approx \delta^{-1}\langle \bv_i, \bz \rangle$.
However, the factor of $\delta^{-1}$ here appears to be superfluous; one way to confirm this is to compare this prediction for $|S| = 2$ with the direct calculations of $P_{\sH}$ for $d = 2$ in \cite{BK-2018-GramianDescription} or \cite{KB-2019-Degree4SK-Arxiv}.
Thus we omit this factor in our final prediction.

We are left with the following prediction for the harmonic projection.
First, it will be useful to set notation for the polynomials occuring inside the summation.
\begin{definition}
For $S \in \sM([N])$ with $S \neq \emptyset$, $m \in \NN$, and $\bx \in \RR^N$, define
\begin{align}
  q_{S,m}(\bx) &\colonequals \left\{\begin{array}{ll} x_i^m & \text{if } |S| = 1 \text{ with } S = \{i\}, \\ \sum_{a = 1}^N \prod_{j \in T} M_{aj} x_a^m & \text{otherwise},\end{array}\right. \\
  q_{S}(\bx) &\colonequals q_{S, |S|}(\bx).
\end{align}
\end{definition}
\noindent
We then predict
\begin{equation}
P_{\sH}\left[(\bV^{\top}\bz)^S\right] \approx \sum_{\sigma \in \mathsf{Part}(S)} \prod_{A \in \sigma} (-1)^{|A| - 1}(|A| - 1)! \, q_{A}(\bV^{\top}\bz).
\end{equation}
By the orthogonality of the ideal and harmonic subspaces, we also immediately obtain a prediction for the orthogonal projection to $V_{\mathcal{I}}$:
\begin{equation}
    P_{\sI}\left[(\bV^{\top}\bz)^S \right] = (\bV^{\top}\bz)^S - P_{\sH}\left[(\bV^{\top}\bz)^S\right] \approx -\sum_{\substack{\sigma \in \mathsf{Part}(S) \\ |\sigma| < |S|}} \prod_{A \in \sigma} (-1)^{|A| - 1}(|A| - 1)! \, q_{A}(\bV^{\top}\bz).
\end{equation}

We therefore obtain the following corresponding predictions for the polynomials $h_S(\bx)$ and $r_S(\bx)$ appearing in Lemma~\ref{lem:poly-cond}:
\begin{align}
    h_S(\bx) &\approx \sum_{\substack{\sigma \in \mathsf{Part}(S)}} \prod_{A \in \sigma} (-1)^{|A| - 1}(|A| - 1)! \, q_{A}(\bx), \label{eq:hS-approx} \\
    r_{S}(\bx) &\approx -\sum_{\substack{\sigma \in \mathsf{Part}(S) \\ |\sigma| < |S|}}\prod_{A \in \sigma} (-1)^{|A| - 1}(|A| - 1)! \, q_{A}(\bx).
\end{align}

The ``lowered'' polynomials $r_{S}^{\downarrow}(\bx)$ may also be defined by simply reducing the powers of $x_i$ appearing in $q_{A}(\bx)$.
Here again, however, we make a slight adjustment: when $|A| = 2$ with $A = \{i, j\}$, we would compute $q_{A, 0}(\bx) = \sum_{a = 1}^NM_{ai}M_{aj} = (\bM^2)_{ij} \approx \delta^{-1} M_{ij}$.
This factor of $\delta^{-1}$ again appears to be superfluous, with the same justification as before.
Removing it, we make the following definition.
\begin{definition}
    For $S \in \sM([N])$ with $S \neq \emptyset$ and $|S|$ even, define
    \begin{equation}
        q_S^{\downarrow}(\bx) = q_S^{\downarrow}(\bx; \bM) \colonequals \left\{ \begin{array}{ll} q_{S, 1}(\bx) & \text{if } |S| \text{ is odd}, \\ M_{ij} & \text{if } |S| = 2 \text{ with } S = \{i, j\},\\ q_{S, 0}(\bx) & \text{if } |S| \geq 4 \text{ is even}. \end{array}\right.
    \end{equation}
\end{definition}
\noindent
With this adjustment, we obtain the prediction
\begin{equation}
    r_{S}^{\downarrow}(\bx) \approx -\sum_{\substack{\sigma \in \mathsf{Part}(S) \\ |\sigma| < |S|}} \prod_{A \in \sigma} (-1)^{|A| - 1}(|A| - 1)! \, q_{S^{\prime}}^{\downarrow}(\bx). \label{eq:rS-down-approx}
\end{equation}

\subsection{Simplifying to the Final Prediction}
\label{sec:informal:simplifying}

Substituting the approximations \eqref{eq:hS-approx} and \eqref{eq:rS-down-approx} into the pseudoexpectation expression we obtained in Lemma~\ref{lem:poly-cond}, we are now equipped with a fully explicit heuristic recursion for $\tEE$, up to the choice of the constants $\sigma_d^2$.
Let us demonstrate how, with some further heuristic steps, this recovers our Definition~\ref{def:lifting} for $d = 1$ and $d = 2$.
For $d = 1$ we expect a sanity check recovering that $\tEE[x_ix_j] = M_{ij}$, while for $d = 2$ we expect to recover the formula studied in \cite{BK-2018-GramianDescription,KB-2019-Degree4SK-Arxiv}.

\begin{example}[$d = 1$]
    If $|S| = 1$ with $S = \{i\}$, then there is no partition $\sigma$ of $S$ with $|\sigma| < |S|$, so $h_{\{i\}}(\bx) = x_i$ and $r_{\{i\}}^{\downarrow}(\bx) = 0$.
    Thus, if $S = \{i\}$ and $T = \{j\}$, we are left with simply
    \begin{equation}
      \tEE(x_i, x_j)
      = \sigma_1^2\delta \cdot \bigg\langle \langle \bv_i, \bz \rangle, \langle \bv_j, \bz \rangle \bigg\rangle_{\circ}
      = \sigma_1^2 \delta \cdot \langle \bv_i, \bv_j \rangle
      = \sigma_1^2 \delta \cdot M_{ij}
    \end{equation}
    which upon taking $\sigma_1^2 = \delta^{-1}$ gives $\tEE(x_i, x_j) = M_{ij}$, as expected.
\end{example}

\begin{example}[$d = 2$]
    \label{ex:k-equals-2}
    If $S = \{i, j\}$ then the only partition $\sigma$ of $S$ with $|\sigma| < |S|$ is the partition $\sigma = \{\{i, j\}\}$.
    Therefore,
    \begin{align}
      h_{\{i, j\}}(\bx) &= x_ix_j - (2 - 1)! \cdot q_{\{i, j\}}(\bx) \nonumber \\
      &= x_ix_j - \sum_{a = 1}^N M_{ai}M_{aj} x_a^2, \\
      r_{\{i, j\}}^{\downarrow}(\bx) &= (2 - 1)! \cdot Z^{\{i, j\}} = M_{ij}.
    \end{align}
    If, furthermore, $T = \{k, \ell\}$, then we compute
    \begin{align}
      &\langle h_S(\bV^{\top}\bz), h_T(\bV^{\top}\bz) \rangle_{\circ} \nonumber \\
      &\hspace{1cm} = \left\langle \langle \bv_i, \bz\rangle \langle \bv_j, \bz\rangle - \sum_{a = 1}^N M_{ai}M_{aj}\langle \bv_a, \bz \rangle^2, \langle \bv_k, \bz\rangle \langle \bv_{\ell}, \bz\rangle - \sum_{a = 1}^N M_{ak}M_{a\ell}\langle \bv_a, \bz \rangle^2\right\rangle_{\circ} \nonumber \\
      &\hspace{1cm} = \frac{1}{2}M_{ik}M_{j\ell} + \frac{1}{2}M_{i\ell}M_{jk} - 2\sum_{a = 1}^NM_{ai}M_{aj}M_{ak}M_{a\ell} + \sum_{a = 1}^N\sum_{b = 1}^N M_{ai}M_{aj}M_{bk}M_{b\ell} M_{ab}^2 \nonumber
        \intertext{and if we make the approximation that the only important contributions from the final sum are when $a = b$, then we obtain}
      &\hspace{1cm} \approx \frac{1}{2}M_{ik}M_{j\ell} + \frac{1}{2}M_{i\ell}M_{jk} - \sum_{a = 1}^NM_{ai}M_{aj}M_{ak}M_{a\ell}.
    \end{align}
    \noindent
    Substituting the above into Lemma~\ref{lem:poly-cond}, we compute
    \begin{align}
      \tEE(x_ix_j, x_kx_{\ell}) &= M_{ij}M_{k\ell} + \sigma_2^2\delta^2 \cdot \left\langle \langle \bv_i, \bz\rangle \langle \bv_j, \bz\rangle - \sum_{a = 1}^N M_{ai}M_{aj}x_a^2, \langle \bv_k, \bz\rangle \langle \bv_{\ell}, \bz\rangle - \sum_{a = 1}^N M_{ak}M_{a\ell}x_a^2\right\rangle_{\circ} \nonumber \\
      &= M_{ij}M_{k\ell} + \sigma_2^2\delta^2 \left(\frac{1}{2}M_{ik}M_{j\ell} + \frac{1}{2}M_{i\ell}M_{jk} - \sum_{a = 1}^NM_{ai}M_{aj}M_{ak}M_{a\ell}\right) \nonumber
        \intertext{where we see that the only value of $\sigma_2^2$ that will give both permutation symmetry and the normalization conditions $\tEE[x_ix_jx_ix_j] \approx 1$ is $\sigma_2^2 = 2\delta^{-2}$, which gives}
        &=M_{ij}M_{k\ell} + M_{ik}M_{j\ell} + M_{i\ell}M_{jk} - 2\sum_{a = 1}^NM_{ai}M_{aj}M_{ak}M_{a\ell},
    \end{align}
    the formula obtained for the deterministic structured case of equiangular tight frames in \cite{BK-2018-GramianDescription} and in general using a similar derivation in \cite{KB-2019-Degree4SK-Arxiv}.
\end{example}

One may continue these (increasingly tedious) calculations for larger $d$ to attempt to find a pattern in the resulting polynomials of $\bM$.
This is how we arrive at the values given in Theorem~\ref{thm:lifting}, and we will sketch the basic idea of how to close the recursion below.
We emphasize, though, that in these heuristic calculations it is important to be careful with variants of the step above where we restricted the double summation to indices $a = b$ (indeed, that step in the above derivation is not always valid; as we will detail in Section~\ref{sec:future}, this is actually the crux of the difficulty in applying our method to low-rank rather than high-rank $\bM$).
This type of operation is valid only when the difference between the \emph{matrices} containing the given summations as their entries has negligible operator norm---a subtle condition.
We gloss this point for now, but give much attention to these considerations in the technical proof details in Section~\ref{sec:pf:lifting} and in Appendix~\ref{app:cgm-tools}.

Reasoning diagramatically, increasing the degree generates new diagrams whose CGSs occur in the pseudomoments in two ways.
First, it turns out that all of the CGSs that arise from the inner product $\langle h_S(\bV^{\top}\bz), h_T(\bV^{\top}\bz) \rangle_{\circ}$ may be fully ``collapsed'' to a CGS with only one summation (or a product of summations over subsets of indices), as we have done above.
These contribute diagrams that are forests of \emph{stars}: each connected component is either two $\bullet$ vertices connected to one another, or a $\square$ vertex connected directly to several leaves.
Second, in computing $\tEE[r_S^{\downarrow}(\bx)r_T^{\downarrow}(\bx)]$, we join the diagrams of odd partition parts to the leaves of existing diagrams, a process we illustrate later in Figure~\ref{fig:odd-merge}.
This recursion generates the good forests: any good forest is either a forest of stars, or stars on some subsets of leaves with the odd subsets attached to a smaller good forest.

Thus, assuming the above intuition about collapsing summations is sound, we expect the pseudomoments to be a linear combination of CGSs of good forests.
Taking this for granted, if the pseudomoments are to be symmetric under permutations of the indices, then the coefficients $\mu(F)$ should depend only on the unlabelled shape of the graph $F$, not on the leaf labels $\kappa$.
Making this assumption, each successive $\mu(F)$ may be expressed in a cumbersome but explicit combinatorial form, eventually letting us predict the formula for $\mu(F)$.

We emphasize the pleasant interplay of diagrammatic and analytic ideas here.
As we mentioned after Lemma~\ref{lem:poly-cond}, the decomposition of the pseudomoments into the ideal and harmonic parts expresses the spectral structure of the pseudomoment matrices, which involves a sequence of alternating ``lift lower-degree pseudomoment matrix'' and ``add orthogonal harmonic part'' steps.
These correspond precisely to the sequence of alternating ``compose partitions with old forests'' and ``add new star forests'' steps generating good forests recursively.

\begin{remark}[Setting $\sigma_d^2$]
    The further calculations we have alluded to above confirm the pattern in the examples $d = 1, 2$ that the correct choice of the free scaling factor is $\sigma_d^2 \colonequals d!\, \delta^{-d}$.
    With this choice, we note that the harmonic term may be written more compactly as
    \begin{equation}
        \sigma_d^2 \delta^d \cdot \left\langle h_S(\bV^{\top}\bz), h_T(\bV^{\top}\bz) \right\rangle_{\circ} = \left\langle h_S(\bV^{\top}\bz), h_T(\bV^{\top}\bz) \right\rangle_{\partial},
    \end{equation}
    where $\langle \cdot, \cdot \rangle_{\partial}$ is the rescaled apolar inner product from Definition~\ref{def:apolar-partial}.
\end{remark}

\section{Partial Ordering and \Mobius\ Function of $\sF(m)$}
\label{sec:poset}

In the previous section, we found a way to compute the coefficients $\mu(F)$ attached to each forest diagram $F$ in Definition~\ref{def:lifting}.
Calculating examples, one is led to conjecture the formula given in Definition~\ref{def:mu-F} for these quantities.
We will eventually give the rather difficult justification for that equality in the course of proving Theorem~\ref{thm:lifting} in Section~\ref{sec:pf:lifting}.
For now, we prove the simpler interpretation of these constants that will guide other parts of that proof: the $\mu(F)$ give the \Mobius\ function of a certain partial ordering on the CGS terms of the pseudoexpectation.

\subsection{Compositional Poset Structure}

We first introduce the poset structure that is associated with $\mu(F)$.
We call this a \emph{compositional} structure because it is organized according to which forests are obtained by ``composing'' one forest with another by inserting smaller forests at each $\square$ vertex.
\begin{definition}[Compositional poset]
    \label{def:compositional-poset}
    Suppose $F \in \sF(m)$.
    For each $v \in V^{\square}(F)$, write $E(v)$ for the set of edges incident with $v$, and fix $\kappa_v: E(v) \to [|E(v)|]$ a labelling of $E(v)$.
    Suppose that for each $v \in V^{\square}(F)$, we are given $F_v \in \sF(\deg(v))$.
    Write $F[(F_v)_{v \in V^{\square}(F)}] \in \sF(m)$ for the forest formed as follows.
    Begin with the disjoint union of all $F_v$ for $v \in V^{\square}(F)$ and all pairs in $F$.
    Denote the leaves of $F_v$ in this disjoint union by $\ell_{v, 1}, \dots, \ell_{v, \deg(v)}$.
    Then, merge the edges ending at $\ell_{v, i}$ and $\ell_{w, j}$ whenever $\kappa_v^{-1}(i) = \kappa_w^{-1}(j)$.
    Whenever $\kappa_v^{-1}(i)$ terminates in a leaf $x$ of $F$, give $\ell_{v, i}$ the label that $x$ has in $F$.
    Finally, whenever $x$ belongs to a pair of $F$, give $x$ in the disjoint union the same label that it has in $F$.
    
    Let the \emph{compositional relation} $\leq$ on $\sF(m)$ be defined by setting $F^{\prime} \leq F$ if, for each $v \in V^{\square}(F)$, there exists $F_v \in \sF(\deg(v))$ such that $F^{\prime} = F[(F_v)_{v \in V^{\square}(F)}]$.
\end{definition}
\noindent
It is straightforward to check that this relation does not depend on the auxiliary orderings $\kappa_v$ used in the definition.
While the notation used to describe the compositional relation above is somewhat heavy, we emphasize that it is conceptually quite intuitive, and give an illustration in Figure~\ref{fig:poset}.

We give the following additional definition before continuing to the basic properties of the resulting poset.

\begin{definition}[Star tree]
    For $m \geq 4$ an even number, we denote by $S_m$ the \emph{star tree} on $m$ leaves, consisting of a single $\square$ vertex connected to $m$ $\bullet$ vertices.
    For $m = 2$, denote by $S_2$ the tree with no $\square$ vertices and two $\bullet$ vertices connected to one another.
    Note that all labellings of $S_m$ are isomorphic, so there is a unique labelled star tree in $\sF(m)$.
\end{definition}

\begin{proposition}
    $\sF(m)$ endowed with the relation $\leq$ forms a poset.
    The unique maximal element in $\sF(m)$ is $S_m$, while any perfect matching in $\sF(m)$ is a minimal element.
\end{proposition}

To work with the \Mobius\ function, it will be more convenient to define a version of this poset augmented with a unique minimal element, as follows (this is the same manipulation as is convenient to use, for example, in the analysis of the poset of partitions into sets of even size; see \cite{Stanley-1978-ExponentialStructures}).
\begin{definition}
    Let $\sFbar(m)$ consist of $\sF(m)$ with an additional formal element denoted $\oslash$.
    We extend the poset structure of $\sF(m)$ to $\sFbar(m)$ by setting $\oslash \leq F$ for all $F \in \sFbar(m)$.
    When we wish to distinguish $\oslash$ from the elements of $\sF(m)$, we call the latter \emph{proper forests}.
\end{definition}

\begin{figure}
    \begin{minipage}[b]{0.4\textwidth}
      \includegraphics[scale=0.7]{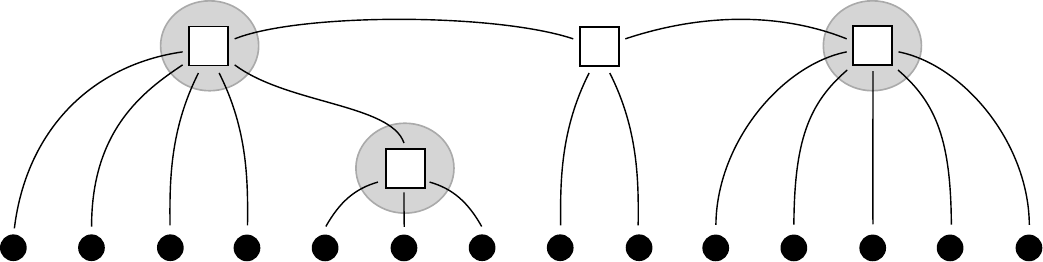}
  \end{minipage}
  \hspace{1.1cm}
  \begin{minipage}[b]{0.055\textwidth}
      $\mathlarger{\mathlarger{\geq}}$
      \vspace{0.7cm}
  \end{minipage}
  \begin{minipage}[b]{0.4\textwidth}
        \includegraphics[scale=0.7]{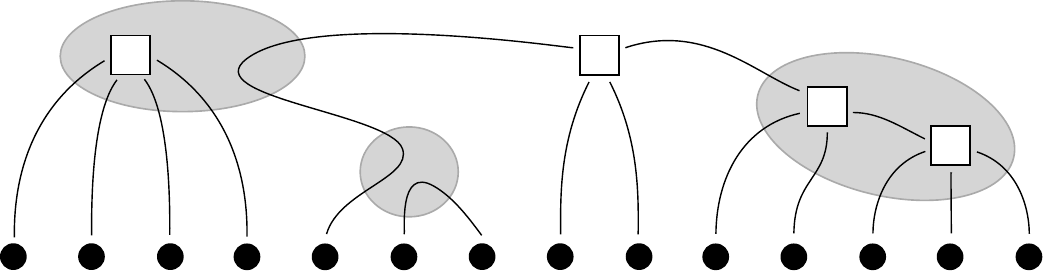}
  \end{minipage} 
    \caption{\textbf{Ordering in the compositional poset $\sF(m)$.} We give an example of the ordering relation between two forests in $\sF(14)$, highlighting the ``composing forests'' at each $\square$ vertex of the greater forest that witness this relation.}
    \label{fig:poset}
\end{figure}

\subsection{\Mobius\ Function Derivation}

The main result of this section is the following.
\begin{lemma}
    \label{lem:mobius-fn}
    Let $\mu_{\sFbar}(\cdot, \cdot)$ be the \Mobius\ function of $\sFbar(m)$ (where we elide $m$ for the sake of brevity, as it is implied by the arguments).
    Then $\mu_{\sFbar}(\oslash, \oslash) = 1$, and for $F \in \sF(m)$,
    \begin{equation}
        \mu_{\sFbar}(\oslash, F) = (-1)^{|V^{\square}(F)| + 1} \prod_{v \in V^{\square}(F)} (\deg(v) - 2)! = -\mu(F),
    \end{equation}
    where $\mu(\cdot)$ on the right-hand side is the quantity from Definition~\ref{def:mu-F}.
\end{lemma}

We proceed in two steps: first, and what is the main part of the argument, we compute the \Mobius\ function of a star tree.
Then, we show that the \Mobius\ function of a general forest factorizes into that of the star trees corresponding to each of its internal vertices.
The following ancillary definition will be useful both here and in a later proof.

\begin{definition}[Rooted odd tree]
    \label{def:T-root}
    For odd $\ell \geq 3$, define the set of \emph{rooted odd trees} on $\ell$ leaves, denoted $\sT^{\root}(\ell)$, to be the set of rooted trees where the number of children of each internal vertex is odd and at least 3, and where the leaves are labelled by $[\ell]$.
    Define a map $e: \sT^{\root}(m - 1) \to \sT(m)$ that attaches the leaf labeled $m$ to the root.
\end{definition}
\noindent
While it is formally easier to express this definition in terms of rooted trees, it may be intuitively clearer to think of a rooted odd tree as still being a good tree, only having one distinguished ``stub'' leaf, whose lone neighbor is viewed as the root.

\begin{proposition}
    \label{prop:mobius-fn-star}
    $\mu_{\sFbar}(\oslash, S_2) = -1$. For all even $m \geq 4$, $\mu_{\sFbar}(\oslash, S_{m}) = (m - 2)!$.
\end{proposition}
\begin{proof}
    We first establish the following preliminary identity.
    \begin{equation}
        \text{\emph{Claim:} }\sum_{T \in \sT(m)} \mu_{\sFbar}(\oslash, T) = -\nu(m) = \mu_{\EvenPart}(\oslash, \{[m]\}).
    \end{equation}
    We proceed using a common idiom of \Mobius\ inversion arguments, similar to, e.g., counting labelled connected graphs (see Example 2 in Section 4 of \cite{BG-1975-MobiusInversionCombinatorial}).
    For $F \in \sF(m)$, let $\conn(F) \in \Part([m]; \even)$ denote the partition of leaves into those belonging to each connected component of $F$.
    For $\pi \in \Part([m]; \even)$, define
    \begin{equation}
        b(\pi) \colonequals -\sum_{\substack{F \in \sF([m]) \\ \conn(F) = \pi}} \mu(\oslash, F),
    \end{equation}
    and $b(\oslash) = 0$.
    Then, the quantity we are interested in is $-b(\{[m]\})$.
    By \Mobius\ inversion,
    \begin{equation}
        b(\{[m]\}) = \sum_{\pi \in \Part(m; \even) \cup \{\oslash\}} \left(\sum_{\oslash \leq \rho \leq \pi} b(\rho)\right) \mu_{\EvenPart}(\pi, \{[m]\}).
    \end{equation}
    The inner summation is zero if $\pi = \oslash$, and otherwise equals
    \begin{equation}
        \sum_{\oslash \leq \rho \leq \pi} b(\rho) = \prod_{A \in \pi}\left(\sum_{F \in \sF([|A|])} -\mu(F)\right) = 1 \text{ if } \pi \neq \oslash.
    \end{equation}
    Therefore, we may continue
    \begin{equation}
        b(\{[m]\}) = \sum_{\pi \in \Part(m; \even)}  \mu_{\EvenPart}(\pi, \{[m]\}) = \sum_{\pi \in \Part(m; \even)} (-1)^{|\pi| - 1}(|\pi| - 1)!.
    \end{equation}
    By the composition formula for exponential generating functions, this means
    \begin{equation}
        \sum_{k \geq 0} \frac{b(\{[2k]\})}{(2k)!}x^{2k} = \log\left(1 + (\cosh(x) - 1)\right) = \log \cosh(x),
    \end{equation}
    and the result follows by equating coefficients.

    Next, we relate the trees of $\sT([m])$ that we sum over in this identity to the rooted odd trees introduced in Definition~\ref{def:T-root}.
    We note that the map $e$ defined there is a bijection between $\sT^{\root}(m - 1)$ and $\sT(m)$ (the inverse map removes the leaf labelled $m$ and sets its single neighbor to be the root).
    The \Mobius\ function composed with this bijection is
    \[ \mu(\oslash, e(T)) = (-1)^{|V^{\square}(T)| + 1}\prod_{v \in V^{\square}(T)} \mu(\oslash, S_{|c(v)| + 1}), \]
    where $c(v)$ gives the number of children of an internal vertex.

    Finally, we combine the recursion associated to the rooted structure of $\sT^{\root}(m - 1)$ (whereby a rooted tree is, recursively, the root and a collection of rooted trees attached to the root) and the identity of the Claim to derive a generating function identity that completes the proof.
    Namely, we may manipulate, for $m \geq 4$,
    \begin{align}
      \nu(m)
      &= -\sum_{\substack{T \in \sT(m)}} \mu(\oslash, T) \nonumber \\
      &= -\sum_{T \in \sT^{\root}(m - 1)} \mu(\oslash, e(T)) \nonumber \\
      &= \sum_{T \in \sT^{\root}(m - 1)} (-1)^{|V^{\square}(T)|}\prod_{v \in V^{\square}(G)} \mu(\oslash, S_{|c(v)| + 1}) \nonumber \\
      &= \sum_{\substack{\pi \in \mathsf{Part}([m - 1]; \mathsf{odd}) \\ |\pi| > 1}}  (-\mu(\oslash, S_{|\pi| + 1})) \prod_{S \in \pi}\left(\sum_{T \in \sT^{\root}(S)} (-1)^{|V^{\square}(T)|} \prod_{v \in V^{\square}(T)} \mu(\oslash, S_{|c(v)| + 1})\right) \nonumber \\
      &= \sum_{\substack{\pi \in \mathsf{Part}([m - 1]; \mathsf{odd}) \\ |\pi| > 1}} (-\mu(\oslash, S_{|\pi| + 1})) \prod_{S \in \pi}\left(-\sum_{T \in \sT^{\root}(|S|)} \mu(\oslash, e(T))\right) \nonumber \\
      &= \sum_{\substack{\pi \in \mathsf{Part}([m - 1]; \mathsf{odd}) \\ |\pi| > 1}} (-\mu(\oslash, S_{|\pi| + 1})) \prod_{S \in \pi}\nu(|S| + 1)
    \end{align}

We now have a relatively simple identity connecting $\mu(\oslash, S_m)$ with $\nu(m)$.
To translate this into a relation of generating functions, we remove the condition $|\pi| > 1$ and correct to account for the case $m = 2$, obtaining, for any even $m \geq 2$,
\begin{equation}
    \label{eq:mobius-fn-stars-final-id}
    2\nu(m) = \One\{m = 2\} + \sum_{\pi \in \mathsf{Part}([m - 1]; \mathsf{odd})} (-\mu(\oslash, S_{|\pi| + 1}))\prod_{S \in \pi}\nu(|S| + 1).
\end{equation}
Now, let $F(x) \colonequals \sum_{k \geq 1} \frac{\mu(\oslash, S_{2k})}{(2k)!} x^{2k}$.
Multiplying by $x^{2k - 1} / (2k - 1)!$ on either side of \eqref{eq:mobius-fn-stars-final-id} and summing over all $k \geq 1$, we find, by the composition formula for exponential generating functions,
\begin{equation}
    2\tanh(x) = x - F^{\prime}(\tanh(x)).
\end{equation}
Equivalently, taking $y = \tanh(x)$, we have
\begin{equation}
    F^{\prime}(y) = \tanh^{-1}(y) - 2y.
\end{equation}
Recalling
\begin{equation}
    \tanh^{-1}(y) = \frac{1}{2}\left(\log(1 + y) - \log(1 - y)\right) = \sum_{k \geq 0} \frac{y^{2k + 1}}{2k + 1},
\end{equation}
we have
\begin{equation}
    F(y) = -\frac{1}{2}x^2 + \sum_{k \geq 2} \frac{x^{2k}}{2k(2k - 1)} = -\frac{1}{2!}x^2 + \sum_{k \geq 2} \frac{(2k - 2)!}{(2k)!}x^{2k},
\end{equation}
and the result follows.
\end{proof}

Before completing the proof of Lemma~\ref{lem:mobius-fn}, we give the following preliminary result, describing the interval lying below a forest as a product poset.
This follows immediately from the definition of the compositional relation, since the set of forests smaller than $F$ corresponds to a choice of a local ``composing forest'' at each $v \in V^{\square}(F)$.
\begin{proposition}
    \label{prop:poset-star-factorization}
    Let $F = ((V^{\bullet}, V^{\square}), E) \in \sF(m)$.
    Then, we have the isomorphism of posets
    \begin{equation}
        (\oslash, F] \cong \prod_{v \in V^{\square}(F)} (\oslash, S_{\deg(v)}].
    \end{equation} 
\end{proposition}
\noindent
We now complete the proof of the main Lemma.
\begin{proof}[Proof of Lemma~\ref{lem:mobius-fn}]
    Let $\widehat{\mu}(\oslash, \cdot)$ be the putative \Mobius\ function from the statement,
    \begin{equation}
        \widehat{\mu}(\oslash, F) = (-1)^{|V^{\square}(F)| + 1} \prod_{v \in V^{\square}(F)} (\deg(v) - 2)!.
    \end{equation}
    We proceed by induction on $m$.
    For $m = 2$, the result holds by inspection.
    Suppose $m \geq 4$.
    Since $\widehat{\mu}(\oslash, \oslash) = 1$ by definition, and since by Proposition~\ref{prop:mobius-fn-star} we know that $\widehat{\mu}(\oslash, S_m) = \mu(\oslash, S_m)$, it suffices to show that, for all $F \in \sF(m)$ with $F \neq S_m$, we have
    \begin{equation}
        \sum_{F^{\prime} \in [\oslash, F]} \widehat{\mu}(\oslash, F^{\prime}) = 0.
    \end{equation}
    Let $F \in \sF(m)$ with $F \neq S_m$.
    We then compute:
    \begin{align}
      \sum_{F^{\prime} \in [\oslash, F]} \widehat{\mu}(\oslash, F^{\prime})
      &= 1 - \sum_{F^{\prime} \in (\oslash, F]} \prod_{v \in V^{\square}(F^{\prime})} (-(\deg(v) - 2)!) \nonumber \\
      &= 1 - \prod_{v \in V^{\square}(F)}\left(\sum_{F^{\prime} \in (\oslash, S_{\deg(v)}]} \prod_{w \in V^{\square}(F^{\prime})} (-(\deg(w) - 2)!)\right) \tag{by Proposition~\ref{prop:poset-star-factorization}} \\
      &= 1 - \prod_{v \in V^{\square}(F)}\left(-\sum_{F^{\prime} \in (\oslash, S_{\deg(v)}]} \mu(\oslash, F^{\prime})\right) \nonumber \\
      &= 1 - \prod_{v \in V^{\square}(F)} \mu(\oslash, \oslash) \tag{by inductive hypothesis} \nonumber \\
      &= 0,
    \end{align}
    completing the proof.
\end{proof}

\section{Extension Formula: Proof of Theorem~\ref{thm:lifting}}
\label{sec:pf:lifting}

\subsection{Pseudomoment and Contractive Graphical Matrices}
\label{sec:cgm}

We first outline the general approach of our proof and introduce the main objects involved.
By construction, $\tEE$ as given in Definition~\ref{def:lifting} satisfies Conditions 1 and 2 of the pseudoexpectation properties from Definition~\ref{def:pe} (normalization and ideal annihilation); therefore, it suffices to prove positivity.
Moreover, positivity may be considered in any suitable basis modulo the ideal generated by the constraint polynomials, and given any fixed basis positivity may be written in linear-algebraic terms as the positive semidefiniteness of the associated \emph{pseudomoment matrix}.
We state this explicitly below, in an application of standard reasoning in the SOS literature (see, e.g., \cite{Laurent-2009-SOS}).
\begin{proposition}
    \label{prop:positivity-any-basis}
    Let $\tEE: \RR[x_1, \dots, x_N]_{\leq 2d} \to \RR$ be a linear operator satisfying the normalization and ideal annihilation properties of Definition~\ref{def:pe}.
    Let $\sI \subset \RR[x_1, \dots, x_N]$ be the ideal generated by $x_i^2 - 1$ for $i = 1, \dots, N$, and let $p_1, \dots, p_{\binom{N}{\leq d}} \in \RR[x_1, \dots, x_N]_{\leq d}$ be a collection of coset representatives for a basis of $\RR[x_1, \dots, x_N]_{\leq d} \, /\,  \sI$.
    Define the associated pseudomoment matrix $\bZ \in \RR^{\binom{N}{\leq d} \times \binom{N}{\leq d}}$ with entries
    \begin{equation}
        Z_{s, t} = \tEE[p_s(\bx)p_t(\bx)].
    \end{equation}
    Then, $\tEE$ satisfies the positivity property of Definition~\ref{def:pe} if and only if $\bZ \succeq \bm 0$.
\end{proposition}

If we were to take the standard multilinear monomial basis for the $p_s(\bx)$, we would wind up with $\bZ$ being a sum of CGSs of different diagrams in each entry, with the CGS indices corresponding to the set indexing of $\bZ$.
While we will ultimately work in a different basis, this general observation will still hold, so we define the following broad formalism for the matrices that will arise.

The following enhancement of the diagrams introduced in Definition~\ref{def:cgs} is the analogous object to what is called a \emph{shape} in the literature on graphical matrices \cite{AMP-2016-GraphMatrices,BHKKMP-2019-PlantedClique}.
We prefer to reserve the term \emph{diagram} for any object specifying some contractive calculation, to use that term unadorned for the scalar version, and to add \emph{ribbon} to indicate the specification of ``sidedness'' that induces a matrix structure. 
\begin{definition}[Ribbon diagram]
    Suppose $G = (V, E)$ is a graph with two types of vertices, which we denote $\bullet$ and $\square$ visually and whose subsets we denote $V = V^{\bullet} \sqcup V^{\square}$.
    Suppose also that $V^{\bullet}$ is further partitioned into two subsets, which we call ``left'' and ``right'' and denote $V^{\bullet} = \sL \sqcup \sR$.
    Finally, suppose that each of $\sL$ and $\sR$ is equipped with a labelling $\kappa_{\sL}: \sL \to [|\sL|]$ and $\kappa_{\sR}: \sR \to [|\sR|]$.
    We call such $G$ together with the labellings $\kappa_{\sL}$ and $\kappa_{\sR}$ a \emph{ribbon diagram}.
\end{definition}

\begin{definition}[Good forest ribbon diagram]
    We write $\sF(\ell, m)$ for the set of good forests on $\ell + m$ vertices, equipped with a partition of the leaves $V^{\bullet} = \sL \sqcup \sR$ with $|\sL| = \ell$ and $|\sR| = m$.
\end{definition}

\begin{definition}[Contractive graphical matrix]
    \label{def:cgm}
    Suppose $G$ is a ribbon diagram with labellings $\kappa_{\sL}$ and $\kappa_{\sR}$.
    Define $\kappa: V^{\bullet} \to [|V^{\bullet}|]$ by $\kappa(\ell) = \kappa_{\sL}(\ell)$ for $\ell \in \sL$ and $\kappa(r) = |\sL| + \kappa_{\sR}(r)$ for $r \in \sR$.
    With this labelling, we interpret $G$ as a CGS diagram.

    For $\bM \in \RR^{N \times N}_{\sym}$, we then define the \emph{contractive graphical matrix (CGM)} of $G$ to be the matrix $\bZ^G \in \RR^{\binom{[N]}{|\sL|} \times \binom{[N]}{|\sR|}}$ with entries
    \begin{equation}
        Z^G_{S,T} = Z^G_{S,T}(\bM) \colonequals Z^G(\bM; (s_1, \dots, s_{|\sL|}, t_1, \dots, t_{|\sR|}))
    \end{equation}
    where $S = \{s_1, \dots, s_{|\sL|}\}$ and $T = \{t_1, \dots, t_{|\sR|}\}$ with $s_1 < \cdots < s_{|\sL|}$ and $t_1 < \cdots < t_{|\sR|}$.
\end{definition}
\noindent
We note that the restriction to set-valued indices in this definition is rather artificial; the most natural indexing would be by $[N]^{|\sL|} \times [N]^{|\sR|}$.
However, as the set-indexed submatrix of this larger matrix is most relevant for our application, we use this definition in the main text; we present several technical results with the more general tuple-indexed CGMs in Appendix~\ref{app:cgm-tools}.

\begin{remark}[Multiscale spectrum]
    \label{rem:multiscale-spectrum}
    As in calculations involving graphical matrices \cite{AMP-2016-GraphMatrices,RSS-2018-EstimationSOS,BHKKMP-2019-PlantedClique}, the scale of the norm of a CGM may be read off of its ribbon diagram.
    We emphasize the following general principle: if $\|\bM\| = O(1)$ and $F \in \sF(2d)$, then
    \textbf{$\|\bZ^F\| = \omega(1)$ if and only if some connected components of $F$ have leaves in only $\sL$ or only $\sR$}.
    We call such components \emph{sided}.
    CGMs tensorize over connected components (Proposition~\ref{prop:cgm-tt-tensorization}), so the norm of a CGM is the product of the norms of CGMs of its diagram's components.
    In the case of $\bM$ a rescaled random low-rank projection matrix, where $\|\bM\| = O(1)$ and $\|\bM\|_F = \Theta(N^{1/2})$, components that are not sided give norm $O(1)$ (Proposition~\ref{prop:cgm-generic-norm-bound}), while each sided component gives norm roughly $\widetilde{\Theta}(N^{1/2})$, which follows from calculating the sizes of the individual CGM entries assuming square root cancellations.
    Thus the norm of a $CGM$ is $\widetilde{\Theta}(N^{\#\{\text{sided components}\} / 2})$.

    In particular, we will encounter the same difficulty as in other SOS lower bounds that the pseudomoment matrices we work with have a \emph{multiscale spectrum}, meaning simply that the scaling of different $\bZ^G$ with $N$ can be very different.
    For the main part of our analysis we will be able to ignore this, since by working in the multiharmonic basis from Remark~\ref{rem:heuristic-block-diag}, we will be able to eliminate all ribbon diagrams with sided components, leaving us with only terms of norm $O(1)$.
    Unfortunately, this issue returns when handling various error terms, so some of our choices below will still be motivated by handling the multiscale difficulty correctly.
\end{remark}

\subsection{Main and Error Terms}
\label{sec:main-error-terms}

We recall that our pseudoexpectation was constructed in Section~\ref{sec:informal-deriv} as a sum of $Z^F(\bM; S)$ for $S$ a \emph{multiset}, and had the property of being approximately unchanged by adding pairs of repeated indices to $S$.
While in Definition~\ref{def:lifting} we have forced these to be exact equalities to produce a pseudoexpectation satisfying the ideal annihilation constraints exactly, the approximate version of the pseudoexpectation, which is better suited for the diagrammatic reasoning that will justify positivity, will still be very useful.
Therefore, we decompose $\tEE$ into a ``main term,'' which is the actual result of our heuristic calculation but only approximately satisfies the hypercube constraints, and an ``error term'' that implements the remaining correction, as follows.
\begin{definition}[Main and error pseudoexpectations]
    Define $\tEE^{\main}, \tEE^{\err}: \RR[x_1, \dots, x_N] \to \RR$ to be linear operators with values on monomials given by
    \begin{align}
      \tEE^{\main}[\bx^S] &\colonequals \sum_{F \in \sF(|S|)} \mu(F) \cdot Z^F(\bM; S), \\
      \tEE^{\err}[\bx^S] &\colonequals \tEE[\bx^S] - \tEE^{\main}[\bx^S],
    \end{align}
    for all multisets $S \in \sM([N])$.
    Note that for $S$ a multiset, $\tEE[\bx^S] = \tEE[\bx^{S^{\prime}}]$ where $S^{\prime}$ is the (non-multi) set of indices occurring an odd number of times in $S$.
\end{definition}

In the remainder of this section, we show how the presence of the \Mobius\ function in our pseudomoment values implies that $\tEE^{\err}[\bx^S]$ is small.
It is not difficult to see that $\tEE$ and $\tEE^{\main}$ are equal to leading order, since if, for instance, only one index is repeated two times in $S$, then the dominant terms of $\tEE^{\main}$ will be those from diagrams where the two occurrences of this index are paired and there is an arbitrary forest on the remaining indices; various generalizations thereof hold for greater even and odd numbers of repetitions.
This kind of argument shows, for example, that for $\bM$ a rescaled random low-rank projection matrix, we have $\tEE^{\main}[\bx^S] = (1 + O(N^{-1/2}))\tEE[\bx^S]$ as $N \to \infty$.
However, due to the multiscale spectrum of the pseudomoments as discussed in Remark~\ref{rem:multiscale-spectrum}, it turns out that this does not give sufficient control of $\tEE^{\err}$.

We must go further than this initial analysis and take advantage of cancellations among even the sub-leading order terms of $\tEE^{\err}$, a fortunate side effect of the \Mobius\ function coefficients.
These cancellations generalize the following observation used in \cite{KB-2019-Degree4SK-Arxiv} for the degree 4 case.
If we take $S = \{i, i, j, k\}$ (the simple situation mentioned above), then we have
\begin{align}
  \tEE^{\err}[\bx^S]
  &= \underbrace{M_{jk}\vphantom{\sum_{a = 1}^N}}_{\mathclap{\tEE[\bx^S]}} - \Bigg(\underbrace{M_{ii}M_{jk} + M_{ij}M_{ik} + M_{ik}M_{ij} - 2\sum_{a = 1}^N M_{ai}^2 M_{aj}M_{ak}}_{\tEE^{\main}[\bx^S]}\Bigg) \nonumber \\
  &= 2\sum_{a \in [N] \setminus \{i\}} M_{ai}^2 M_{aj}M_{ak},
    \label{eq:err-cancellation-deg-4}
\end{align}
where the term $2 M_{ij}M_{ik}$ in $\tEE^{\main}[\bx^S]$ has cancelled.
For $\bM$ a rescaled random low-rank projection matrix, this makes a significant difference: the term that cancels is $\Theta(N^{-1})$, while the remaining error term after the cancellation is only $\Theta(N^{-3/2})$ (assuming square root cancellations).

Surprisingly, a similar cancellation obtains at all degrees and for any combination of repeated indices.
The general character of the remaining error terms is that, as in the above simple example the $\square$ vertex connecting two equal leaves labelled $i$ was not allowed to have its index equal $i$, so in general the minimal spanning subtree of a collection of leaves with the same label cannot ``collapse'' by having all of its internal vertices have that same label.

The collections of spanning subtrees with respect to which we will study this cancellation are precisely the forests $\mathsf{MaxSpan}(F, \bs)$, as defined earlier in Definition~\ref{def:repetitions-forest}.
Below we record the important properties of the subgraphs that result from this construction.
\begin{proposition}[Properties of $\MaxSpan$]
    \label{prop:MaxSpan-properties}
    For any $F \in \sF(m)$ and $\bs \in [N]^m$, $\MaxSpan(F, \bs)$ satisfies the following.
    \begin{enumerate}
    \item (Components) For every connected component $C$ of $\MaxSpan(F, \bs)$, there is some $i \in [N]$ and $C_j$ a connected component of $F$ such that $|\kappa^{-1}(i) \cap V^{\bullet}(C_j)| \geq 2$ and $C$ is the minimal spanning tree of $\kappa^{-1}(i) \cap V^{\bullet}(C_j)$.
    \item (Maximality) $\MaxSpan(F, \bs)$ is the union of a maximal collection of vertex-disjoint spanning trees of the above kind.
    \item (Independence over connected components) If the connected components of $F$ are $C_1, \dots, C_k$, then $\MaxSpan(F, \bs) = \MaxSpan(C_1, \bs|_{C_1}) \, \sqcup \, \cdots \, \sqcup \, \MaxSpan(C_k, \bs|_{C_k})$. (We write $\bs|_{C_i}$ for the restriction of $\bs$ to the indices that appear as labels of the leaves of $C_i$.)
    \item (Priority of small indices) Whenever $i < j$, $|\kappa^{-1}(i) \cap V^{\bullet}(C_k)| \geq 2$, $|\kappa^{-1}(j) \cap V^{\bullet}(C_k)| \geq 2$, and $\MaxSpan(F, \bs)$ contains the minimal spanning tree of $\kappa^{-1}(j) \cap V^{\bullet}(C_k)$, then it also contains the minimal spanning tree of $\kappa^{-1}(i) \cap V^{\bullet}(C_k)$.
    \end{enumerate}
\end{proposition}

We are now prepared to express our generalization of the cancellation that we observed above in \eqref{eq:err-cancellation-deg-4}, which amounts to the cancellation of all summation terms where the entire subgraph $\MaxSpan(F, \bs)$ collapses in the sense discussed previously.
\begin{definition}[Graphical error terms]
    Let $F \in \sF(m)$ and $\bs \in [N]^m$.
    Recall that we say $\ba \in [N]^{V^{\square}(F)}$ is \emph{$(F, \bs)$-tight} if, for all connected components $C$ of $\MaxSpan(F, \bs)$, if $s_{\kappa(x)} = i$ for all leaves $x$ of $C$, then $a_v = i$ for all $v \in V^{\square}(C)$ as well.
    Otherwise, we say that $\ba$ is \emph{$(F, \bs)$-loose}.
    With this, we define
    \begin{equation}
        \Delta^F(\bM; \bs) \colonequals \sum_{\substack{\ba \in [N]^{V^{\square}} \\ \ba \text{ is } (F, \bs)\text{-loose}}} \prod_{\{v, w\} \in E} M_{f_{\bs, \ba}(v)f_{\bs,\ba}(w)}.
    \end{equation}
    As in Definition~\ref{def:cgs}, we also extend the definition to allow sets or multisets in the second argument of $\Delta^F$ by replacing them with the corresponding tuple of elements in ascending order.
\end{definition}
\noindent
The following preliminary definition, building off of the rooted odd trees from Definition~\ref{def:T-root}, will be useful in the argument.
\begin{definition}[Good forest with rooted components]
    \label{def:F-root}
    For $m \geq 2$, let $\sF^{\root}(m)$ be the set of forests on $m$ leaves where every connected component is either a good tree (per Definition~\ref{def:good-forest}) or a rooted odd tree (per Definition~\ref{def:T-root}), and where the leaves are labelled by $[m]$.
    Note that some but not all components of such a forest may have distinguished roots.
    For $F \in \sF^{\root}(m)$, let $\odd(F)$ denote the set of rooted odd tree components of $F$, and let $\mu(F) \colonequals \mu(F^{\prime})$ for $F^{\prime}$ formed by attaching an extra leaf to the root of every tree in $\odd(F)$.
\end{definition}

\begin{lemma}[Graphical error pseudomoments]
    \label{lem:E0-Err-decomp}
    For any $S \in \sM([N])$,
    \begin{equation}
        \tEE^{\err}[\bx^S] = -\sum_{F \in \sF(m)} \mu(F) \cdot \Delta^F(\bM; S).
    \end{equation}
\end{lemma}

\begin{proof}
    Our result will follow from the following, purely combinatorial, result.
    For $F \in \sF(m)$ and $A \subseteq [m]$, let us say that $F$ is \emph{$A$-dominated} if, for every connected component $C$ of $F$, every $\square$ vertex of $C$ is contained in the minimal spanning tree of the leaves $\kappa^{-1}(A) \cap V^{\bullet}(C)$.
    Then,
    \begin{equation}
        \text{\emph{Claim:}}\sum_{\substack{F \in \mathcal{F}(m) \\ F \text{ is } A\text{-dominated}}} \mu(F) = \begin{cases} 1 & \text{if } m \in \{|A|, |A| + 1\}, \\ 0 & \text{otherwise}. \end{cases}
    \end{equation}

    We first prove the Claim.
    Let $\ell = |A|$; without loss of generality we take $A = [\ell]$.
    Let us write
    \begin{equation}
        c(\ell, m) \colonequals \sum_{\substack{F \in \mathcal{F}(m) \\ F \text{ is } [\ell]\text{-dominated}}} \mu(F).
    \end{equation}
    For each fixed $\ell$, we will proceed by induction on $m \geq \ell$.
    For the base case, we have $c(\ell, \ell) = 1$ by the defining property of the \Mobius\ function, since in this case the summation is over all $F \in \sF(\ell)$.
    
    Let $r_{\ell}: \sF(m) \to \sF^{\root}(m - \ell)$ return the rooted forest formed by deleting the minimal spanning trees of the elements of $[\ell]$ in each connected component, where upon deleting part of a tree, we set any vertex with a deleted neighbor to be the root of the new odd tree connected component thereby formed.
    Then, we have
    \begin{align}
      1
      &= \sum_{F \in \sF(m)} \mu(F) \nonumber \\
      &= \sum_{R \in \sF^{\root}(m - \ell)} \sum_{\substack{F \in \sF(m) \\ r_{\ell}(F) = R}} \mu(F) \nonumber
      \intertext{and, factoring out $\mu(R)$ from $\mu(F)$ with $r_{\ell}(F) = R$, we note that what is left is a sum of $\mu(F)$ over $[\ell]$-dominated forests $F$ on $[\ell + |\odd(R)|]$, whereby}
      &= \sum_{R \in \sF^{\root}(m - \ell)} \mu(R) c(\ell, \ell + |\odd(R)|).
    \end{align}
    Now, we consider two cases.
    First, if $m = \ell + 1$ for $\ell$ odd, then there is only one $R$ in the above summation, having one leaf connected to a root, which has $\mu(R) = 1$.
    Therefore, $c(\ell, \ell + 1) = 1$.

    Otherwise, supposing $m \geq \ell + 2$ and continuing the induction, if we assume the Claim holds for all smaller $m$, then we find
    \begin{equation}
        1 = \sum_{\substack{R \in \sF^{\root}(m - \ell) \\ |\odd(R)| = \One\{\ell \text{ odd}\}}} \mu(R)  + c(\ell, m).
    \end{equation}
    If $\ell$ is even, then the first term is a sum over $R \in \sF(m - \ell)$.
    If $\ell$ is odd, then the sum may be viewed as a sum over $R \in \sF(m - \ell + 1)$ by viewing the single root vertex as an additional leaf.
    In either case, this sum equals 1 by the definition of the \Mobius\ function, whereby $c(\ell, m) = 0$, completing the proof.

    We now return to the proof of the statement.
    Suppose $S \in \sM([N])$, and let $\bs \in [N]^{|S|}$ be the tuple of the elements of $S$ in ascending order.
    Given $F \in \sF(|S|)$, let us write $\ind(F, S)$ for the multiset with one occurrence of each $i \in [N]$ for each connected component $C$ of $\MaxSpan(F, \bs)$ with $s_{\kappa(\ell)} = i$ for all leaves $\ell$ in $C$, and a further occurrence of each $i \in [N]$ for each leaf $\ell$ not belonging to $\MaxSpan(F; \bs)$ with $s_{\kappa(\ell)} = i$.
    And, write $\coll(F, S) \in \sF(|\ind(F, S)|)$ for the good forest obtained by deleting from $F$ each component of $\MaxSpan(F; \bs)$ and replacing each incidence between $\MaxSpan(F; \bs)$ and $F \setminus \MaxSpan(F; \bs)$ with a new leaf, labelled such that
    \begin{equation}
        Z^{\coll(F, S)}(\bM; \ind(F, S)) = \sum_{\substack{\ba \in [N]^{V^{\square}(F)} \\ \ba \text{ } (F, S)\text{-tight}}} \prod_{\{v, w\} \in E(F)} M_{f_{S, \ba}(v)f_{S, \ba}(w)}.
    \end{equation}
    Intuitively, these definitions describe the version of $F$ where all tight subtrees in $\MaxSpan(F, \bs)$ have been collapsed, with extra occurrences of their indices added as labels on leaves in the new ``fragmented'' tree.
    
    Using these definitions, we may rewrite the quantity that we need to compute as follows.
    \begin{align}
      \tEE^{\main}[\bx^S] - \sum_{F \in \sF(|S|)} \mu(F) \cdot \Delta^F(\bM; S)
      &= \sum_{F \in \sF(|S|)} \mu(F) \sum_{\substack{\ba \in [N]^{V^{\square}(F)} \\ \ba \text{ } (F, S)\text{-tight}}} \prod_{\{v, w\} \in E(F)} M_{f_{S, \ba}(v)f_{S, \ba}(w)} \nonumber \\
      &= \sum_{F \in \sF(|S|)} \mu(F) Z^{\coll(F, S)}(\bM; \ind(F, S)) \nonumber \\
      &= \sum_{S^{\prime} \in \sM([N])} \sum_{F^{\prime} \in \sF(|S^{\prime}|)} \Bigg(\underbrace{\sum_{\substack{F \in \sF(|S|) \\ \coll(F, S) = F^{\prime} \\ \ind(F, S) = S^{\prime}}} \mu(F)}_{\equalscolon \, \zeta(F^{\prime}, S^{\prime}, S)}\Bigg) Z^{F^{\prime}}(\bM; S^{\prime}).
    \end{align}
    We claim that the inner coefficient $\zeta(F^{\prime}, S^{\prime}, S)$ is zero unless $S^{\prime}$ is the (non-multi) set of indices occurring an odd number of times in $S$, in which case it is $\mu(F^{\prime})$.
    This will complete the proof, since we will then have that the above equals $\tEE[\bx^S]$ (by definition of the latter).

    Since $\ind(F, S)$ for any $F$ only contains indices occurring in $S$, we will have $\zeta(F^{\prime}, S^{\prime}, S) = 0$ unless $S^{\prime}$ only contains indices also occurring in $S$.
    In other words, we have $\set(S^{\prime}) \subseteq \set(S)$; note, however, that a given index can occur more times in $S^{\prime}$ than in $S$.

    Let $C_1^{\prime}, \dots, C_m^{\prime}$ be the connected components of $F^{\prime}$, let $\kappa^{\prime}$ be the function labelling the leaves of $F^{\prime}$, and let $\bs^{\prime}$ be the tuple of elements of $S^{\prime}$ in ascending order.
    Let $S_i^{\prime} \colonequals \{s^{\prime}_{\kappa^{\prime}(\ell)}: \ell \in V^{\bullet}(C_i^{\prime})\}$, \emph{a priori} a multiset.
    In fact, no index can occur twice in any $S_i^{\prime}$: if $j$ is the least such index, then by construction the minimal spanning tree on all $\ell \in V^{\bullet}(C_i^{\prime})$ with $s^{\prime}_{\kappa^{\prime}(\ell)} = j$ would have been included in $\MaxSpan(F, \bs)$ and would have been collapsed in forming $F^{\prime}$.
    Therefore, each $S_i^{\prime}$ is a set, and $S = S_1^{\prime} + \cdots + S_m^{\prime}$.

    Now, we define the subsets of connected components containing a leaf labelled by each index: for $j \in [N]$, let $A_j = \{i \in [m]: j \in S_i^{\prime}\}$.
    Also, let $n_j$ equal the number of occurrences of $j$ in $S$.
    Then, every $F$ with $\coll(F, S) = F^{\prime}$ and $\ind(F, S) = S^{\prime}$ is obtained by composing with $F^{\prime}$ forests $F_j$ for $j \in [N]$ whose leaves are $\kappa^{\prime^{-1}}(j)$, together with some $n_j - |A_j|$ further leaves $\ell_{j, 1}, \dots, \ell_{j, n_j}$, such that $F_j$ is dominated (in the sense above) by these further leaves, which all have $s_{\kappa(\ell_{j, k})} = j$, and such that $F_j$ does not connect any $C_{i_1}^{\prime}, C_{i_2}^{\prime}$ for $i_1, i_2 \in A_{j^{\prime}}$ with $j^{\prime} < j$.
    It is easier to understand the description in reverse: the $F_j$ are precisely the forests added to $\MaxSpan(F, \bs)$ for index $j$.

    Using this description of the $F$ appearing in $\zeta(F^{\prime}, S^{\prime}, S)$ and the fact that $\mu(F)$ factorizes over $\square$ vertices, we may factorize
    \begin{equation}
        \zeta(F^{\prime}, S^{\prime}, S) = \mu(F^{\prime})\prod_{j = 1}^N \Bigg(\sum_{\substack{F \in \sF(n_j) \\ F \text{ } [n_j - |A_j|]\text{-dominated} \\ F \text{ does not connect } C_{i_1}^{\prime}, C_{i_2}^{\prime} \\ \text{for } i_1, i_2 \in A_{j^{\prime}}, j^{\prime} < j}} \mu(F)\Bigg).
    \end{equation}
    Now, suppose for the sake of contradiction that $\zeta(F^{\prime}, S^{\prime}, S) \neq 0$ for some $F^{\prime}$, and $|A_j| \neq \One\{n_j \text{ odd}\}$ for some $j$ (remembering that $A_j$ are defined in terms of $F^{\prime}$).
    Choose the smallest such $j$.
    Then, the connectivity property on $F$ in the $j$th factor above is vacuous since $|A_{j^{\prime}}| \leq 1$ for all $j < j^{\prime}$, so it may be removed in the summation.
    By the Claim, that factor is then zero, whereby $\zeta(F^{\prime}, S^{\prime}, S) = 0$, unless $|A_j| = \One\{n_j \text{ odd}\}$, so we reach a contradiction.
    Finally, if indeed $|A_j| = \One\{n_j \text{ odd}\}$ for all $j$, then the connectivity condition is vacuous for all terms, so it may always be removed, whereupon by the Claim the product above is 1 and $\zeta(F^{\prime}, S^{\prime}, S) = \mu(F^{\prime})$ as desired.
\end{proof}

Lastly, we prove the following additional result on $\tEE^{\err}$ that will be useful later, showing that it decomposes into a sum over a choice of some ``main term trees'' and some ``error trees'' to apply to subsets of $S$.
\begin{proposition}[Error term factorizes over connected components]
    \label{prop:E-err-factorization}
    For any $S \in \sM([N])$,
    \begin{equation}
        \tEE^{\err}[\bx^{S}] = \sum_{\substack{A \subseteq S \\ A \neq \emptyset}} \tEE^{\main}[\bx^{S - A}] \sum_{\pi \in \Part(A; \even)} \prod_{R \in \pi}\left( -\sum_{T \in \sT(|R|)} \mu(T) \cdot \Delta^T(\bM; R)\right).
    \end{equation}
\end{proposition}
\begin{proof}
    We begin from the definition,
    \begin{align}
      \Delta^F(\bM; \bs)
      &= \sum_{\substack{\ba \in [N]^{V^{\square}} \\ \ba \text{ is } (F, \bs)\text{-loose}}} \prod_{\{v, w\} \in E} M_{f_{\bs, \ba}(v) f_{\bs, \ba}(v)} \nonumber
      \intertext{Now, we observe from Proposition~\ref{prop:MaxSpan-properties} that $\ba$ is $(F, \bs)$-loose if and only if $\ba|_{T}$ is $(F, \bs|_{T})$-loose for some $T \in \conn(F)$. Therefore, by the inclusion-exclusion principle, we may write}
      &= \sum_{\substack{A \subseteq \conn(F) \\ A \neq \emptyset}} (-1)^{|A| - 1}\sum_{\substack{\ba \in [N]^{V^{\square}} \\ \ba|_{T} \text{ is } (F, \bs|_{T})\text{-loose for } T \in A}} \prod_{\{v, w\} \in E} M_{f_{\bs, \ba}(v) f_{\bs, \ba}(v)} \nonumber \\
      &= -\sum_{\substack{A \subseteq \conn(F) \\ A \neq \emptyset}} \prod_{T \in A}\big(-\Delta^{T}(\bM; \bs|_{T})\big) \prod_{T \notin A}Z^{T}(\bM; \bs|_{T}).
    \end{align}
    Now, we use that $\mu(F) = \prod_{i = 1}^k \mu(T_i)$, so by definition of $\tEE^{\err}$, we have
    \begin{align}
      \tEE^{\err}[\bx^{\bs}]
      &= -\sum_{F \in \sF(m)} \mu(F) \cdot \Delta^F(\bM; \bs) \nonumber \\
      &= \sum_{F \in \sF(m)} \sum_{\substack{A \subseteq \conn(F) \\ A \neq \emptyset}} \prod_{T \in A}\big(-\mu(T) \cdot \Delta^{T}(\bM; \bs|_{T})\big) \prod_{T \notin A}\big(\mu(T) \cdot Z^{T}(\bM; \bs|_{T})\big).
    \end{align}
    Reversing the order of summation and then reorganizing the inner sum according to the partition $\pi$ of the leaves of $F$ lying in each connected component then gives the result.
\end{proof}

\subsection{Proof Outline: Spectral Analysis in the Harmonic Basis}

Our basic strategy for proving positivity is to invoke Proposition~\ref{prop:positivity-any-basis} with the multiharmonic basis discussed in Remark~\ref{rem:heuristic-block-diag}.
As our heuristic calculations there suggested, this will attenuate the multiscale spectrum of the pseudomoment matrix written in the standard monomial basis, making the analysis of the spectrum much simpler.
It will also let us use the heuristic Gram matrix expression \eqref{eq:heuristic-gram-mx} as a tool for proving positivity.

In this section, we describe the objects that arise after writing the pseudomoments in this basis, and state the main technical results that lead to the proof of Theorem~\ref{thm:lifting}.
First, we recall the definition of the basis.

\begin{definition}[Multiharmonic basis polynomials]
    For $S \subseteq [N]$, we define
    \begin{align}
      q_S^{\downarrow}(\bx; \bM) &\colonequals \left\{\begin{array}{ll} x_i & \text{if } |S| = 1 \text{ with } S = \{i\}, \\ M_{ij} & \text{if } |S| = 2 \text{ with } S = \{i, j\}, \\ \sum_{a = 1}^N \prod_{i \in S} M_{ia} \cdot x_a & \text{if } |S| \geq 3 \text{ is odd}, \\ \sum_{a = 1}^N \prod_{i \in S} M_{ia} & \text{if } |S| \geq 4 \text{ is even}, \end{array}\right.
                                                                              \\
      h_S^{\downarrow}(\bx; \bM) &\colonequals \sum_{\sigma \in \Part(S)} \prod_{A \in \sigma} (-1)^{|A| - 1}(|A| - 1)!\, q_{A}^{\downarrow}(\bx; \bM).
    \end{align}
    For the sake of brevity, we will usually omit the explicit dependence on $\bM$ below, abbreviating $h_S^{\downarrow}(\bx) = h_S^{\downarrow}(\bx; \bM)$.
\end{definition}
\noindent
Next, we write the pseudomoments in this basis, separating the contributions of the main and error terms.
\begin{definition}[Main and error pseudomoments]
    \label{def:main-error}
    Define matrices $\bZ^{\main}, \bZ^{\err}, \bZ \in \RR^{\binom{[N]}{\leq d} \times \binom{[N]}{\leq d}}$ to have entries
\begin{align}
  Z^{\main}_{S,T} &\colonequals \tEE^{\main}[h_S^{\downarrow}(\bx)h_T^{\downarrow}(\bx)], \\
  Z^{\err}_{S,T} &\colonequals \tEE^{\err}[h_S^{\downarrow}(\bx)h_T^{\downarrow}(\bx)], \\
  Z_{S,T} &\colonequals \tEE[h_S^{\downarrow}(\bx)h_T^{\downarrow}(\bx)] \\
                  &= Z^{\main}_{S,T} + Z^{\err}_{S,T}.
\end{align}
\end{definition}
\noindent
We assign a technical lemma to the analysis of each of the two terms.
\begin{lemma}[Positivity of main term]
    \label{lem:Z-main-positive}
    Under the assumptions of Theorem~\ref{thm:lifting},
    \begin{equation}
        \lambda_{\min}(\bZ^{\main}) \geq \lambda_{\min}(\bM)^d - (6d)^{10d}\|\bM\|^{3d}(\epsilon_{\tree}(\bM; d) + \epsilon_{\pow}(\bM) + \epsilon_{\offdiag}(\bM) + \epsilon_{\corr}(\bM)).
    \end{equation}
\end{lemma}

\begin{lemma}[Bound on error term]
    \label{lem:Z-err-bound}
    Under the assumptions of Theorem~\ref{thm:lifting},
    \begin{equation}
        \|\bZ^{\err}\| \leq (12d)^{32d}\|\bM\|^{5d}\epsilon_{\err}(\bM; 2d).
    \end{equation}
\end{lemma}

\noindent
Given these statements, it is straightforward to prove our main theorem.

\begin{proof}[Proof of Theorem~\ref{thm:lifting}]
    Since the only multilinear monomial in $h_S^{\downarrow}(\bx)$ is $\bx^S$, the $h_S^{\downarrow}(\bx)$ for $S \in \binom{[N]}{\leq d}$ form a basis for $\RR[x_1, \dots, x_N]_{\leq d} / \sI$ for $\sI$ the ideal generated by $\{x_i^2 - 1\}_{i = 1}^N$.
    Thus by Proposition~\ref{prop:positivity-any-basis} it suffices to show $\bZ \succeq \bm 0$.
    Since $\bZ = \bZ^{\main} + \bZ^{\err}$, we have $\lambda_{\min}(\bZ) \geq \lambda_{\min}(\bZ^{\main}) - \|\bZ^{\err}\|$.
    Substituting the results of Lemmata~\ref{lem:Z-main-positive} and \ref{lem:Z-err-bound} then gives the result.
\end{proof}

\subsection{Approximate Block Diagonalization of the Main Term}

As a first step towards showing the positivity of $\bZ^{\main}$, we show that our choice of writing the pseudomoments of $\tEE^{\main}$ in the multiharmonic basis makes $\bZ^{\main}$ approximately block-diagonal.
This verifies what we expect based on the informal argument leading up to Remark~\ref{rem:heuristic-block-diag}.

\paragraph{Reducing to stretched ribbon diagrams}
We first describe an important cancellation in $\bZ^{\main}$.
Writing the pseudomoments in the multiharmonic basis in fact leaves only the following especially well-behaved type of forest ribbon diagram.
Below we call a $\square$ vertex \emph{terminal} if it is incident to any leaves.
\begin{definition}[Stretched forest ribbon diagram]
    We say that $F \in \sF(\ell, m)$ is \emph{stretched} if it satisfies the following properties:
    \begin{enumerate}
    \item Every terminal $\square$ vertex of $F$ has a neighbor in both $\sL$ and $\sR$.
    \item No connected component of $F$ is a \emph{sided pair}: a pair of connected $\bullet$ vertices both lying in $\sL$ or both lying in $\sR$.
    \item No connected component of $F$ is a \emph{skewed star}: a star with one vertex in $\sL$ and more than one vertex in $\sR$, or one vertex in $\sR$ and more than one vertex in $\sL$.
    \end{enumerate}
\end{definition}
\noindent
A fortunate combinatorial cancellation shows that, in the multiharmonic basis, the pseudomoment terms of stretched forest ribbon diagrams retain their initial coefficients, while non-stretched forest ribbon diagrams are eliminated.
\begin{proposition}
    \label{prop:Z-main-harmonic-basis}
    For any $S, T \subseteq [N]$,
    \begin{equation}
        Z^{\main}_{S,T} = \tEE^{\main}[h_S^{\downarrow}(\bx)h_T^{\downarrow}(\bx)] = \sum_{\substack{F \in \sF(|S|, |T|) \\ F \text{ stretched}}} \mu(F) \cdot Z^F_{S, T}(\bM).
    \end{equation}
\end{proposition}
\begin{proof}
    We expand directly:
    \begin{align*}
      &\hspace{-0.2cm}\tEE^{\main}[h_S^{\downarrow}(\bx)h_T^{\downarrow}(\bx)] \\
      &= \tEE^{\main}\Bigg[\Bigg(\sum_{\sigma \in \Part(S)} \prod_{A \in \sigma}(-1)^{|A| - 1}(|A| - 1)! \, q_A^{\downarrow}(\bx)\Bigg) \Bigg(\sum_{\tau \in \Part(T)} \prod_{B \in \tau}(-1)^{|B| - 1}(|B| - 1)! \, q_B^{\downarrow}(\bx)\Bigg) \Bigg] \nonumber
      \intertext{}
      &= \sum_{\substack{\sigma \in \Part(S) \\ \tau \in \Part(T)}} \prod_{R \in \sigma + \tau} (-1)^{|R| - 1}(|R| - 1)!
                                                                                                                                                       \prod_{R \in \sigma[\even] + \tau[\even]} q_R^{\downarrow} \nonumber \\
      &\hspace{1.35cm} \sum_{\substack{\ba \in [N]^{\sigma[\odd; \geq 3]} \\ \bb \in [N]^{\tau[\odd; \geq 3]}}} \prod_{A \in \sigma[\odd; \geq 3]} \prod_{i \in A} M_{a(A), i} \, \cdot \, \prod_{B \in \tau[\odd; \geq 3]} \prod_{j \in B} M_{b(B), j} \, \, \cdot \nonumber \\
      &\hspace{4.1cm}  \tEE^{\main}\left[\prod_{\{i\} \in \sigma[1]} x_i \prod_{A \in \sigma[\odd; \geq 3]} x_{a(A)} \prod_{\{j\} \in \tau[1]} x_j \prod_{B \in \tau[\odd; \geq 3]} x_{b(B)}\right]
        \intertext{Let us write $f_{\ba}: \sigma[\odd] \to [N]$ to map $A = \{i\} \mapsto i$ when $|A| = 1$ and to map $A \mapsto a(A)$ when $|A| \geq 3$, and likewise $g_{\bm b}: \tau[\odd] \to [N]$. Then, expanding the pseudoexpectation, we have}
      &= \sum_{\substack{\sigma \in \Part(S) \\ \tau \in \Part(T)}} \prod_{R \in \sigma + \tau} (-1)^{|R| - 1}(|R| - 1)!
                                                                                                                                                       \prod_{R \in \sigma[\even] + \tau[\even]} q_R^{\downarrow} \nonumber \\
      &\hspace{1.35cm} \sum_{\substack{\ba \in [N]^{\sigma[\odd; \geq 3]} \\ \bb \in [N]^{\tau[\odd; \geq 3]}}} \prod_{A \in \sigma[\odd; \geq 3]} \prod_{i \in A} M_{a(A), i} \, \cdot \, \prod_{B \in \tau[\odd; \geq 3]} \prod_{j \in B} M_{b(B), j} \, \, \cdot \nonumber \\
      &\hspace{1.48cm} \sum_{F \in \sF(|\sigma[\odd]|, |\tau[\odd]|)} \mu(F) \cdot Z^F_{(f_{\ba}(A))_{A \in \sigma[\odd]}, (g_{\bb}(B))_{B \in \tau[\odd]}}
        \intertext{We say that $F \in \sF(|S|, |T|)$ is an \emph{odd merge of $(\sigma, \tau)$ through $F^{\prime} \in \sF(|\sigma[\odd]|, |\tau[\odd]|)$} if $F$ consists of even stars on the even parts of $\sigma$ and $\tau$, and even stars on the odd parts of $\sigma$ and $\tau$ with one extra leaf added to each, composed in the sense of the compositional ordering of Definition~\ref{def:compositional-poset} with $F^{\prime}$.
        See Figure~\ref{fig:odd-merge} for an example.
        When $F$ is an odd merge of $(\sigma, \tau)$ through $F^{\prime}$, then $F^{\prime}$ is uniquely determined by $\sigma, \tau$, and $F$.
        Using this notion, we may rewrite the above as}
      &= \sum_{F \in \sF(|S|, |T|)}\Bigg(\sum_{\substack{\sigma \in \Part(S) \\ \tau \in \Part(T) \\ F^{\prime} \in \sF(|\sigma[\odd]|, |\tau[\odd]|) \\ F \text{ is an odd merge} \\ \text{of } (\sigma, \tau) \text{ through } F^{\prime}}} \prod_{R \in \sigma + \tau} (-1)^{|R| - 1}(|R| - 1)!
                                                                                                                                                                                                                                                                                                                                 \cdot \mu(F^{\prime})\Bigg)Z^F_{S,T}
                                                                                                                                                                                                                                                                                                                                 \intertext{We make two further simplifying observations. First, the factors of $(-1)^{|R|}$ multiply to $(-1)^{|S| + |T|}$, and $|S| + |T|$ must be even in order for $\sF(|S|, |T|)$ to be non-empty, so we may omit the $(-1)^{|R|}$ factors. Second, by the factorization $\mu(F) = \prod_{v \in V^{\square}(F)}(-(\deg(v) - 2)!)$, when $F$ is an odd merge of $(\sigma, \tau)$ through $F^{\prime}$ then we have $\mu(F) = \mu(F^{\prime}) \prod_{|R| \geq 3 \text{ is odd}}(-(|R| - 1)!) \prod_{|R| \geq 4 \text{ is even}}(-(|R| - 2)!)$, where both products are over $R \in \sigma + \tau$ satisfying the given conditions. It may again be helpful to consult Figure~\ref{fig:odd-merge} to see why this formula holds. Using this, we may extract $\mu(F)$ and continue}
      &= \sum_{F \in \sF(|S|, |T|)}\Bigg(\underbrace{\sum_{\substack{\sigma \in \Part(S) \\ \tau \in \Part(T) \\ F \text{  odd merge} \\ \text{of } (\sigma, \tau)}} (-1)^{|\sigma[\leq 2]| + |\tau[\leq 2]|} \prod_{R \in \sigma[\even; \geq 4] + \tau[\even; \geq 4]}(R - 1)}_{\equalscolon \, \eta(F)}
                                                                      \Bigg)\mu(F) \cdot Z^F_{S,T}.
    \end{align*}
    \begin{figure}
    \begin{center}
        \includegraphics[scale=0.7]{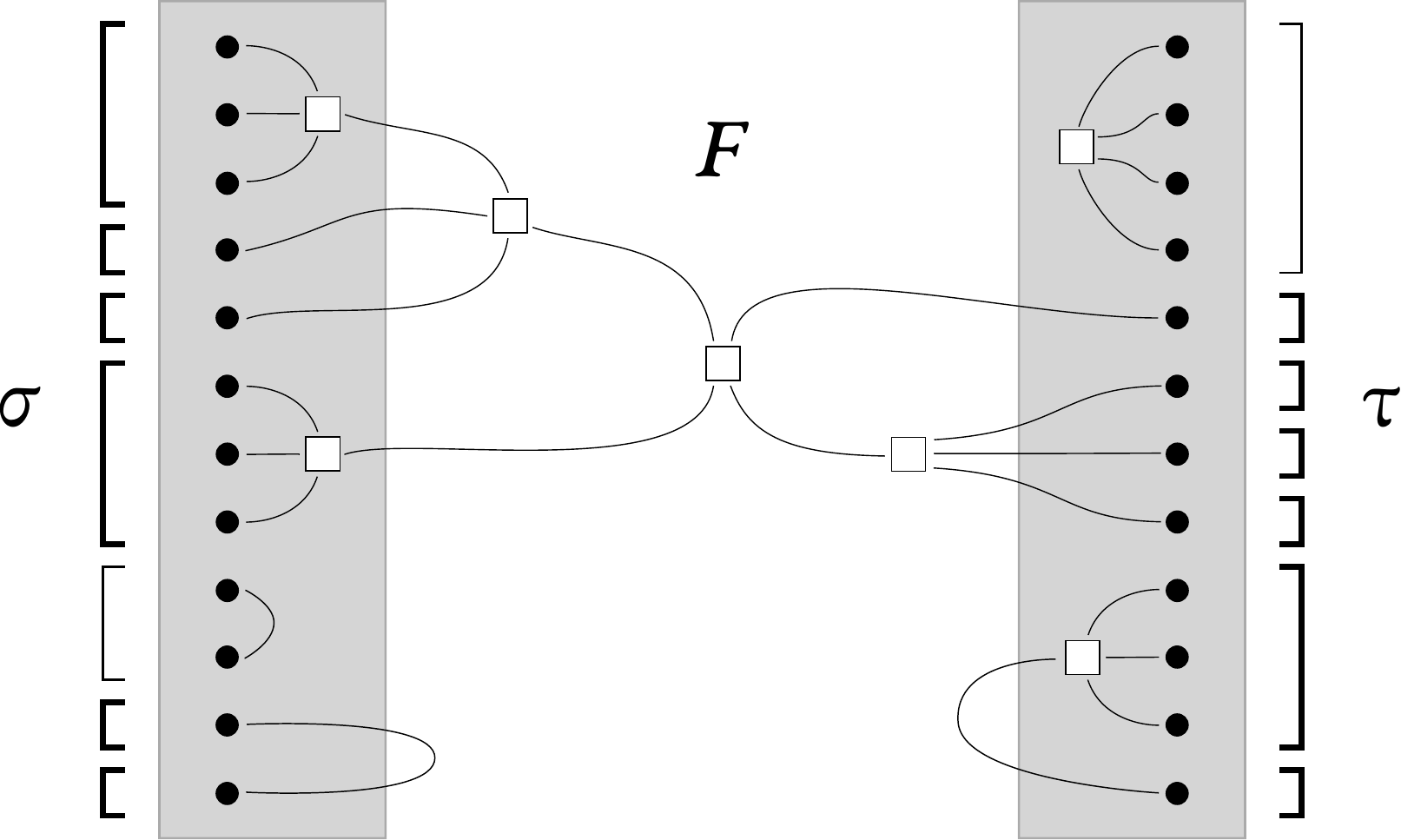}  
    \end{center}
    \caption{\textbf{Odd merge of partitions.} We illustrate an odd merge of two partitions $\sigma$ and $\tau$ through a forest ribbon diagram $F$, as used in the proof of Proposition~\ref{prop:Z-main-harmonic-basis}. The gray boxes show the components of the resulting diagram arising from the partitions (and, one step before, arising from terms in the multiharmonic basis polynomials $h_S^{\downarrow}(\bx)$), while the remainder is the forest ribbon diagram that merges the odd parts of the partitions. The odd parts that are merged by $F$ are highlighted with bold brackets.}
    \label{fig:odd-merge}
\end{figure}

    It remains to analyze the inner coefficient $\eta(F)$.
    To do this, it suffices to enumerate the pairs of partitions $(\sigma, \tau)$ of which $F$ is an odd merge.
    We describe the possible partitions below.
    \begin{itemize}
    \item If $v \in V^{\square}(F)$ is the only $\square$ vertex of a skewed star connected component with leaves $i_1, \dots, i_k \in \sL$ (for $k$ odd) and $j \in \sR$, then $j$ must be a singleton in $\tau$ while $i_1, \dots, i_k$ can either (1)~all be singletons in $\sigma$ or (2)~constitute one part $\{i_1, \dots, i_k\}$ of $\sigma$.
        A symmetric condition holds if there is more than one leaf in $\sR$ and one leaf in $\sL$.
    \item If $v \in V^{\square}(F)$ is the only $\square$ vertex of a sided star connected component with leaves $i_1, \dots, i_k \in \sL$ (for $k$ even), then the $i_1, \dots, i_k$ can either (1)~all be singletons in $\sigma$, (2)~constitute one part $\{i_1, \dots, i_k\}$ of $\sigma$, or (3)~be divided into an odd part $\{i_1, \dots, i_k\} \setminus \{i_{k^{\star}}\}$ and a singleton $\{i_{k^{\star}}\}$ for any choice of $k^{\star} \in 1, \dots, k$.
        A symmetric condition holds if the leaves are all in $\sR$.
    \item If $i_1, i_2 \in \sL$ form a sided pair in $F$, then $i_1, i_2$ can either (1)~both be singletons in $\sigma$, or (2)~constitute one part $\{i_1, i_2\}$ of $\sigma$.
        A symmetric condition holds if the two leaves are in $\sR$.
    \item If $v \in V^{\square}(F)$ is terminal, is not the only $\square$ vertex of its connected component, and has leaf neighbors $i_1, \dots, i_k \in \sL$ (for $k$ odd), then the $i_1, \dots, i_k$ can either (1)~all be singletons in $\sigma$, or (2)~constitute one part $\{i_1, \dots, i_k\}$ of $\sigma$.
        A symmetric condition holds if the leaves are all in $\sR$.
    \item If $v \in V^{\square}(F)$ is terminal, is not the only $\square$ vertex of its connected component, and has leaf neighbors in both $\sL$ and $\sR$, then all leaves attached to $v$ must be singletons in $\sigma$ and $\tau$ (according to whether they belong to $\sL$ or $\sR$, respectively).
    \item If $i \in \sL$ and $j \in \sR$ form a non-sided pair in $F$, then $i$ and $j$ must be singletons in $\sigma$ and $\tau$, respectively.
    \end{itemize}
    Factorizing $\eta(F)$ according to which terminal $\square$ vertex or pair connected component each leaf of $F$ is attached to using these criteria, we find
    \begin{align}
      \eta(F)
      &= \prod_{\substack{C \in \conn(F) \\ C \text{ sided pair}}}(1 - 1) \prod_{\substack{v \text{ terminal in } V^{\square}(F), \\ \text{all leaf neighbors of } v \text{ in } \sL \\ \text{or all leaf neighbors of } v \text{ in } \sR}}(1 - 1) \prod_{\substack{C \in \conn(F) \\ C \text{ sided even star} \\ \text{on } k \geq 4 \text{ leaves}}} ((k - 1) - k + 1) \, \cdot \nonumber \\
      &\hspace{0.3cm}\prod_{\substack{C \in \conn(F) \\ C \text{ skewed star} \\ \text{on } \geq 4 \text{ leaves}}}(1 - 1) \nonumber \\
      &= \One\{F \text{ is stretched}\},
    \end{align}
    completing the proof.
\end{proof} 

\paragraph{Tying stretched ribbon diagrams}
The above does not appear to give the block diagonalization we promised---there exist stretched ribbon diagrams in $\sF(\ell, m)$ even when $\ell \neq m$, so the off-diagonal blocks of $\bZ^{\main}$ are non-zero.
To find that this is an \emph{approximate} block diagonalization, we must recognize that the CGMs of many stretched ribbon diagrams are approximately equal (up to a small error in operator norm) and then observe another combinatorial cancellation in these ``standardized'' diagrams.

Specifically, we will show that all stretched forest ribbon diagrams' CGMs can be reduced to the following special type of stretched forest ribbon diagram.
\begin{definition}[Bowtie forest ribbon diagram]
    \label{def:bowtie-forest-ribbon}
    We call $F \in \sF(\ell, m)$ a \emph{bowtie forest} if every connected component of $F$ is a pair or star with at least one leaf in each of $\sL$ and $\sR$.
    Each connected component of such $F$ is a \emph{bowtie}.
    We call a bowtie or bowtie forest \emph{balanced} if all components have an equal number of leaves in $\sL$ and in $\sR$.
\end{definition}
\noindent
Note that there are no balanced bowtie forests in $\sF(\ell, m)$ unless $\ell = m$; thus, since in our final expression for the pseudomoments below only balanced bowtie forests will remain, we will indeed have an approximate block diagonalization.

Next, we show that any stretched forest ribbon diagram can be ``tied'' to form a bowtie forest ribbon diagram by collapsing every non-pair connected component to have a single $\square$ vertex, while incurring only a small error in the associated CGMs.
\begin{lemma}[Stretched forest ribbon diagrams: tying bound]
    \label{lem:tying-stretched-ribbon}
    Suppose that $\epsilon_{\tree}(\bM; (\ell + m) / 2) \leq 1$.
    Let $F \in \sF(\ell, m)$ be stretched.
    Let $\tie(F) \in \sF(\ell, m)$ be formed by replacing each connected component of $F$ that is not a pair with the bowtie of a single $\square$ vertex attached to all of the leaves of that connected component.
    Then, $\tie(F)$ is a bowtie forest, and
    \[ \|\bZ^F - \bZ^{\tie(F)}\| \leq (\ell + m) (2\|\bM\|)^{\frac{3}{2}(\ell + m)} \epsilon_{\tree}(\bM; (\ell + m) / 2). \]
\end{lemma}
\noindent
We give the proof in Appendix~\ref{app:pf:lem:tying-stretched-ribbon}.
The basic intuition behind the result is that, since every terminal $\square$ vertex of a stretched ribbon diagram is connected to both $\sL$ and $\sR$, the corresponding CGM may be factorized into $\bZ^F = \bA^{\sL} \bD \bA^{\sR}$, where $\bA^{\sL}$ and $\bA^{\sR}$ correspond to the contributions of edges attaching $\sL$ and $\sR$ respectively to the internal vertices, while $\bD$ is CGM of the induced ``inner''  ribbon diagram on the $\square$ vertices.
Thanks to the diagram being stretched, $\bD$ is actually diagonal, and the result essentially states that its only significant entries are those corresponding to all $\square$ vertices in each connected component having the same index.
That is the origin of the $\epsilon_{\tree}$ incoherence quantity here.

We next define the result of tying all of the ribbon diagrams in $\bZ^{\main}$:
\begin{align}
  Z^{\tied}_{S,T}
  &\colonequals \sum_{\substack{F \in \sF(|S|, |T|) \\ F \text{ stretched}}} \mu(F) \cdot Z^{\tie(F)}_{S, T}, \label{eq:def:Z-tied}
  \intertext{where since each bowtie forest can be formed by tying multiple stretched forests, we rewrite to isolate the resulting coefficient of each bowtie forest,}
  &= \sum_{\substack{F \in \sF(|S|, |T|) \\ F \text{ bowtie forest}}} \Bigg(\underbrace{\sum_{\substack{F^{\prime} \in \sF(|S|, |T|) \\ F^{\prime} \text{ stretched} \\ \tie(F^{\prime}) = F}} \mu(F^{\prime})}_{\equalscolon \, \xi(F)} \Bigg) Z^{F}_{S, T}.
\end{align}
The following result gives the combinatorial analysis of the coefficients $\xi(F)$ appearing here, which yields a surprising cancellation that verifies that $\bZ^{\main}$ is approximately block diagonal.

\begin{lemma}[Stretched forest ribbon diagrams: combinatorial reduction]
    \label{lem:stretched-ribbon-combinatorial}
    Let $F \in \sF(\ell, m)$ be a bowtie forest.
    Then,
    \begin{equation}
        \xi(F) \colonequals \sum_{\substack{F^{\prime} \in \sF(|S|, |T|) \\ F^{\prime} \text{ stretched} \\ \tie(F^{\prime}) = F}} \mu(F^{\prime}) = \One\{F \text{ balanced}\} \prod_{\substack{C \in \conn(F) \\ C \text{ balanced bowtie} \\ \text{on } 2k \text{ leaves}}} (-1)^{k - 1} (k - 1)!\, k!.
    \end{equation}
\end{lemma}
\noindent
We give the proof, a rather involved calculation with exponential generating functions, in Appendix~\ref{app:pf:lem:stretched-ribbon-combinatorial}.
We leave open the interesting problem of finding a more conceptual combinatorial proof of this result, especially in light of the appearance of $\xi(F)$ again in Lemma~\ref{lem:partition-transport-ribbon-combinatorial} below.

Equipped with these results, let us summarize our analysis below, giving a combined error bound between $\bZ^{\main}$ and $\bZ^{\tied}$ as well as the final form of the latter.
\begin{corollary}
    \label{cor:Z-main-Z-tied}
    Suppose that $\epsilon_{\tree}(\bM; d) \leq 1$.
    We have
    \begin{align}
      Z^{\tied}_{S, T}
      &= \sum_{\substack{F \in \sF(|S|, |T|) \\ F \text{ balanced bowtie forest}}} \xi(F) \cdot Z^F_{S,T},
    \end{align}
    and this matrix satisfies
    \begin{equation}
        \|\bZ^{\main} - \bZ^{\tied}\| \leq (6d)^{10d} \|\bM\|^{3d} \epsilon_{\tree}(\bM; d).
    \end{equation}
\end{corollary}
\begin{proof}
    The first formula follows from Lemma~\ref{lem:stretched-ribbon-combinatorial}.
    For the norm bound, we apply Lemma~\ref{lem:tying-stretched-ribbon} to each CGM term in each block of $\bZ^{\main}$ and use the norm bound of Proposition~\ref{prop:block-matrix-norm}, which gives
    \begin{align}
      \|\bZ^{\main} - \bZ^{\tied}\| &\leq \sum_{\ell, m = 0}^d (\ell + m) (2\|\bM\|)^{\frac{3}{2}(\ell + m)} \epsilon_{\tree}(\bM; (\ell + m) / 2) \sum_{F \in \sF(\ell, m)} |\mu(F)| \nonumber \\
      &\leq (2d)(2\|\bM\|)^{3d} \epsilon_{\tree}(\bM; d) \sum_{\ell, m = 0}^d |\sF(\ell, m)| \max_{F \in \sF(\ell, m)} |\mu(F)| \nonumber
        \intertext{and using Propositions~\ref{prop:size-F} and \ref{prop:bound-mu} to bound $|\sF(\ell, m)|$ and $|\mu(F)|$ respectively, we finish}
      &\leq (2d)(2\|\bM\|)^{3d} \epsilon_{\tree}(\bM; d) \cdot d^2 (6d)^{3d} (6d)^{6d},
    \end{align}
    and the remaining bound follows from elementary manipulations.
\end{proof}

\subsection{Positivity of the Main Term: Proof of Lemma~\ref{lem:Z-main-positive}}

To prove a lower bound on $\lambda_{\min}(\bZ^{\main})$, our strategy will be to justify the equality
\begin{equation}
    \tEE^{\main}[h_S^{\downarrow}(\bx)h_T^{\downarrow}(\bx)] \approx \langle h_S(\bV^{\top}\bz), h_T(\bV^{\top}\bz) \rangle_{\partial},
\end{equation}
that was suggested by our calculations in Section~\ref{sec:informal-deriv}.
The right-hand side is block-diagonal (since homogeneous polynomials of different degrees are apolar), so our block-diagonalization of the left-hand side is a useful start.
To continue, we follow the same plan for the right-hand side in this section as we did for the left-hand side in the previous section: we (1) express the Gram matrix as a linear combination of CGMs, (2) show that the corresponding ribbon diagrams may be simplified to the same bowtie forests from Definition~\ref{def:bowtie-forest-ribbon}, and (3) perform a combinatorial analysis of the coefficients attached to each bowtie forest after the simplification.

We first describe the class of ribbon diagrams that will arise in expanding the inner products above, which may be viewed as encoding the following kind of combinatorial object.
\begin{definition}[Partition transport plan]
    For a pair of partitions $\sigma, \tau \in \Part([d])$, we write $\Plans(\sigma, \tau)$ for the set of matrices $\bD \in \NN^{\sigma \times \tau}$ where the sum of each row indexed by $A \in \sigma$ is $|A|$, and the sum of each column indexed by $B \in \tau$ is $|B|$.
\end{definition}
\noindent
We borrow the terms ``transport'' and ``plan'' from the theory of optimal transport \cite{Villani-2008-OT}, since $\bD$ may be seen as specifying a protocol for moving masses corresponding to the part sizes of $\sigma$ and $\tau$.
These same matrices also play a crucial role in the Robinson-Schensted-Knuth correspondence of representation theory and the combinatorics of Young tableaux \cite{Fulton-1997-YoungTableaux}.
It is an intriguing question for future investigation whether this connection can shed light on our use of $\Plans(\sigma, \tau)$.

We encode a pair of partitions and a partition transport plan between them into a ribbon diagram in the following way.
\begin{definition}[Partition transport ribbon diagram]
    Let $\sigma, \tau \in \Part([d])$ and $\bD \in \Plans(\sigma, \tau)$.
    Then, let $G = G(\sigma, \tau, \bD)$ be the associated \emph{partition transport ribbon diagram}, with graphical structure defined as follows:
    \begin{itemize}
    \item $\sL(G)$ and $\sR(G)$ are two disjoint sets, each labelled by $1, \dots, d$.
    \item $V^{\square}(G)$ contains one vertex for each part of $\sigma[\geq 2]$ and each part of $\tau[\geq 2]$.
    \item Whenever $i \in A \in \sigma$, then the vertex labelled $i$ in $\sL$ and the $\square$ vertex corresponding to $A$ have an edge between them.
        Likewise, whenever $j \in B \in \tau$, then the vertex labelled $j$ in $\sR$ and the $\square$ vertex corresponding to $B$ have an edge between them.
    \item When $A \in \sigma[\geq 2]$ and $B \in \tau[\geq 2]$, then there are $D_{A,B}$ parallel edges between the corresponding $\square$ vertices.
    \item When $A = \{i\} \in \sigma[1]$, $B \in \tau[\geq 2]$, and $D_{A, B} = 1$, then there is an edge between the vertex labelled $i$ in $\sL$ and the $\square$ vertex corresponding to $B$.
        Likewise, when $B = \{j\} \in \tau[1]$, $A \in \sigma[\geq 2]$, and $D_{A,B} = 1$, then there is an edge between the vertex labelled $j$ in $\sR$ and the $\square$ vertex corresponding to $A$.
    \item When $A = \{i\} \in \sigma[1]$, $B = \{j\} \in \tau[1]$, and $D_{A,B} = 1$, then there is an edge between the vertex labelled $i$ in $\sL$ and the vertex labelled $j$ in $\sR$.
    \end{itemize}
\end{definition}
\noindent
See Figure~\ref{fig:partition-transport-ribbons} for an example featuring all of these situations that may be clearer than the written description.

\begin{figure}
    \begin{minipage}[c]{0.25\textwidth}
        \vspace{-1em}
        \begin{align*}
          \sigma &= \{\{1, 2, 3, 4\}, \\
                 &\hspace{0.72cm}\{5, 6, 7, 8\}, \\
                 &\hspace{0.72cm}\{9\}, \\
                 &\hspace{0.72cm}\{10, 11, 12\} \} \\[0.5em]
          \tau &= \{ \{1, 2, 3, 4, 5\}, \\
                 &\hspace{0.715cm}\{6, 7, 8\}, \\
                 &\hspace{0.715cm}\{9\}, \\
                 &\hspace{0.715cm}\{10\}, \\
                 &\hspace{0.715cm}\{11, 12\}\} \\[0.5em]
          \bD &= \left[\begin{array}{ccccc}
                         2 & 2 & 0 & 0 & 0 \\
                         3 & 1 & 0 & 0 & 0 \\
                         0 & 0 & 1 & 0 & 0 \\
                       0 & 0 & 0 & 1 & 2\end{array}\right]
        \end{align*}
    \end{minipage}
    \begin{tabular}{c:c}
      & \\ & \\ & \\ & \\ & \\ & \\ & \\ & \\ & \\ & \\ & \\ & \\ & \\ & \\ & \\ & \\ & \\
    \end{tabular}
    \begin{minipage}[c]{0.7\textwidth}
        \vspace{0.5em}
        \includegraphics[scale=0.7]{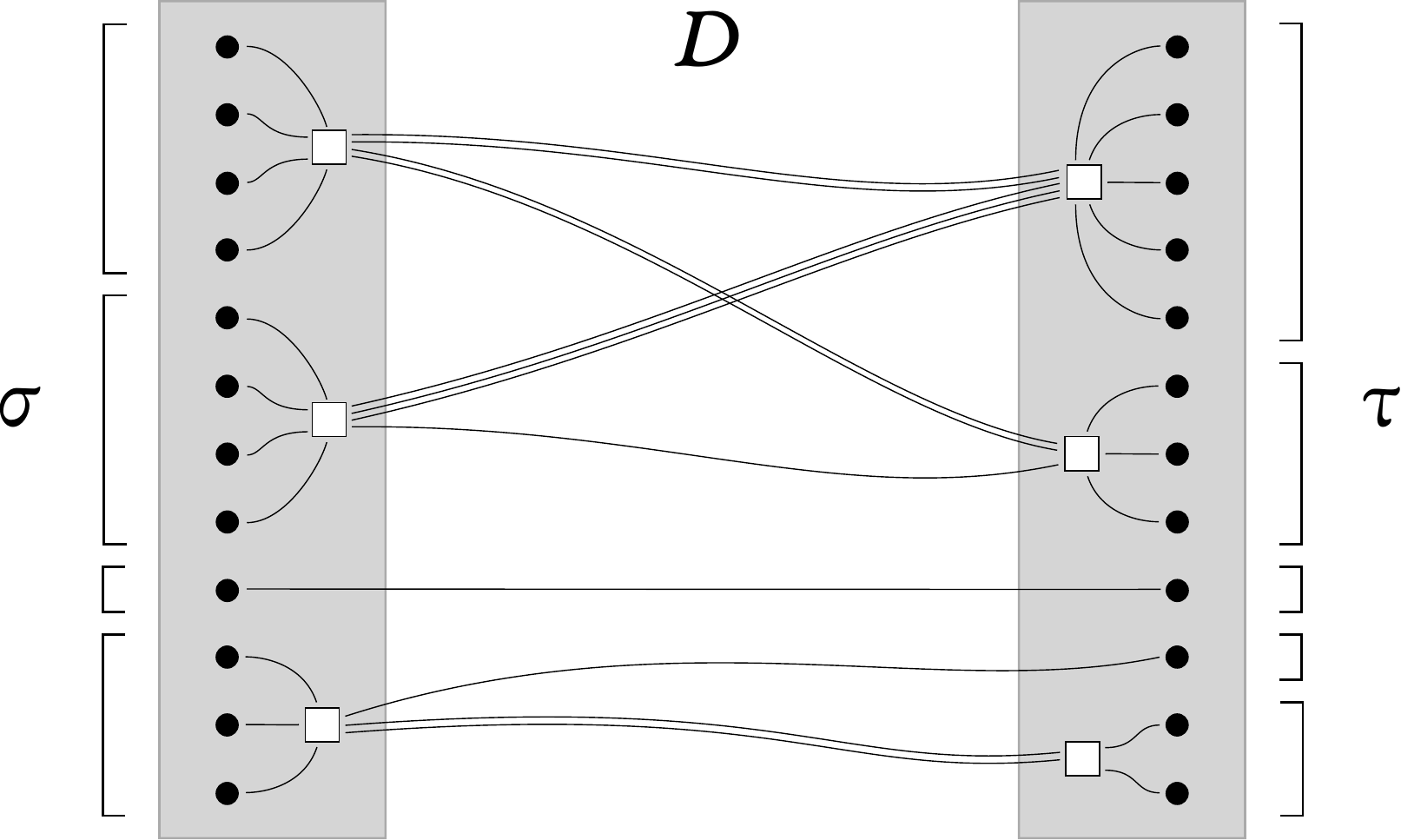}
    \end{minipage}
    \vspace{0.1em}
    \caption{\textbf{Partition transport plan and ribbon diagram.} We illustrate an example of two partitions $\sigma, \tau \in \Part([12])$, a partition transport plan $\bD \in \Plans(\sigma, \tau)$, and the associated partition transport ribbon diagram $G^{(\sigma, \tau, \bD)}$.}
    \label{fig:partition-transport-ribbons}
\end{figure}

This formalism allows us to make the following compact CGM description of the Gram matrix of the non-lowered multiharmonic polynomials.
\begin{proposition}[CGM expansion of multiharmonic Gram matrix]
    \label{prop:gram-cgm-expansion}
    Define $\bY \in \RR^{\binom{[N]}{\leq d} \times \binom{[N]}{\leq d}}$ by
    \begin{equation}
        Y_{S, T} = \langle h_S(\bV^{\top}\bz), h_T(\bV^{\top}\bz) \rangle_{\partial}.
    \end{equation}
    Then, $\bY$ is block-diagonal, with diagonal blocks $\bY^{[d, d]} \in \RR^{\binom{[N]}{d} \times \binom{[N]}{d}}$ given by
    \begin{align}
      \bY^{[d, d]} = \sum_{\sigma, \tau \in \Part([d])} \prod_{R \in \sigma + \tau} (-1)^{|R| - 1}(|R| - 1)! (|R|)! \sum_{\bD \in \Plans(\sigma, \tau)} \frac{1}{\bD!} \bZ^{G(\sigma, \tau, \bD)}(\bM),
    \end{align}
    where $\bD! \colonequals \prod_{A \in \sigma} \prod_{B \in \tau} D_{A,B}!$.
\end{proposition}
\begin{proof}
    We expand directly by linearity:
    \begin{align}
      &\langle h_S(\bV^{\top}\bz), h_T(\bV^{\top}\bz) \rangle_{\partial} \nonumber \\
      &= \left\langle \sum_{\sigma \in \mathsf{Part}(S)} \prod_{A \in \sigma} (-1)^{|A| - 1}(|A| - 1)! \, q_{A}(\bV^{\top}\bz), \sum_{\tau \in \mathsf{Part}(T)} \prod_{B \in \tau} (-1)^{|B| - 1}(|B| - 1)! \, q_{B}(\bV^{\top}\bz) \right\rangle_{\partial} \nonumber \\
      &= \sum_{\substack{\sigma \in \mathsf{Part}(S) \\ \tau \in \mathsf{Part}(T)}} \prod_{R \in \sigma + \tau} (-1)^{|R| - 1}(|R| - 1)! \left\langle \prod_{A \in \sigma} q_{A}(\bV^{\top}\bz), \prod_{B \in \tau} q_{B}(\bV^{\top}\bz)\right\rangle_{\partial} \label{eq:cgm-gram-expand-1}
    \end{align}
    The remaining polynomials we may further expand
    \begin{align}
      \prod_{A \in \sigma} q_{A}(\bV^{\top} \bz)
      &= \prod_{A = \{i\} \in \sigma[1]} \langle \bv_{i}, \bz \rangle \cdot \prod_{A \in \sigma[\geq 2]}\left\{\sum_{a = 1}^N \prod_{i \in A} M_{ai} \cdot \langle \bv_a, \bz \rangle^{|A|}\right\} \nonumber \\
      &=\prod_{A = \{i\} \in \sigma[1]} \langle \bv_{i}, \bz \rangle \sum_{a \in [N]^{\sigma[\geq 2]}} \prod_{A \in \sigma[\geq 2]} \prod_{i \in A} M_{i, a(A)} \cdot \langle \bv_{a(A)}, \bz \rangle^{|A|} \nonumber \\
      &= \sum_{\ba \in [N]^{\sigma[\geq 2]}} \prod_{A \in \sigma} \prod_{i \in A} M_{i, f_{\ba}(A)} \cdot  \langle \bv_{f_{\ba}(A)}, \bz \rangle^{|A|}, \label{eq:cgm-gram-expand-2}
    \end{align}
    where we define $f_{\ba}(A) = a(A)$ if $|A| \geq 2$, and $f_{\ba}(A) = i$ if $A = \{i\}$.
    Likewise, for $\bb \in [N]^{\tau[\geq 2]}$, as will arise in $q_{B}$, we set $g_{\bb}(B) = b(B)$ if $|B| \geq 2$, and $g_{\bb}(B) = j$ if $B = \{j\}$.
    Thus we may expand the remaining inner product from before as
    \begin{align}
      &\left\langle \prod_{A \in \sigma} q_{A}(\bV^{\top}\bz), \prod_{B \in \tau} q_{B}(\bV^{\top}\bz)\right\rangle_{\partial} \nonumber \\
      &\hspace{0.5cm} = \sum_{\substack{\ba \in [N]^{\sigma[\geq 2]} \\ \bb \in [N]^{\tau[\geq 2]}}} \prod_{A \in \sigma} \prod_{i \in A} M_{i,f_{\ba}(A)} \cdot \prod_{B \in \tau} \prod_{j \in B} M_{j, g_{\bb}(B)} \cdot \left\langle \prod_{A \in \sigma} \langle \bv_{f_{\ba}(A)}, \bz \rangle^{|A|}, \prod_{B \in \tau} \langle \bv_{f_{\bb}(B)}, \bz \rangle^{|B|}\right\rangle_{\partial}.
      \label{eq:cgm-gram-expand-3}
    \end{align}

    Finally, this remaining inner product we expand by the product rule, executing which gives rise to the summation over partition transport plans that arises in our result (this calculation is easy to verify by induction, or may be seen as an application of the more general Fa\`{a} di Bruno formula; see, e.g., \cite{Hardy-2006-CombinatoricsPartialDerivatives}):
    \begin{align}
      \bigg\langle \prod_{A \in \sigma} &\langle \bv_{f(A)}, \bz \rangle^{|A|}, \prod_{B \in \tau} \langle \bv_{f(B)}, \bz \rangle^{|B|}\bigg\rangle_{\partial} \nonumber \\
      &= \prod_{A \in \sigma} \langle \bv_{f(A)}, \bm \partial  \rangle^{|A|} \prod_{B \in \tau} \langle \bv_{f(B)}, \bz \rangle^{|B|} \nonumber \\
      &= \prod_{B \in \tau}(|B|)! \sum_{\bD \in \mathsf{Plans}(\sigma, \tau)} \prod_{A \in \sigma} \binom{|A|}{ D_{A, B_1} \cdots D_{A, B_{|\tau|}}} \prod_{A \in \sigma}\prod_{B \in \tau} (M_{f(A),f(B)})^{D_{A, B}} \nonumber \\
      &= \prod_{R \in \sigma + \tau}(|R|)! \sum_{\bD \in \mathsf{Plans}(\sigma, \tau)} \prod_{A \in \sigma} \prod_{B \in \tau} \frac{(M_{f(A),f(B)})^{D_{A,B}}}{D_{A,B}!}, \label{eq:cgm-gram-expand-4}
    \end{align}
    where we remark that in the final expression here we restore the symmetry between $\sigma$ and $\tau$, which was briefly broken to perform the calculation.
    The final result then follows from combining the preceding equations \eqref{eq:cgm-gram-expand-1}, \eqref{eq:cgm-gram-expand-3}, and \eqref{eq:cgm-gram-expand-4}, identifying the summation occurring as precisely that defined by the partition transport ribbon diagram corresponding to $(\sigma, \tau, \bD)$.
\end{proof}

We now describe the ``tied'' version of a partition transport ribbon diagram and bound the error in operator norm incurred by the tying procedure.
\begin{lemma}[Partition transport ribbon diagrams: tying bound]
    \label{lem:tying-partition-transport-ribbon}
    Let $\sigma, \tau \in \Part([d])$ and $\bD \in \Plans(\sigma, \tau)$.
    Let $G = G(\sigma, \tau, \bD) \in \sF(d, d)$ be the associated partition transport ribbon diagram.
    Let $\tie(G)$ denote the diagram obtained from $G$ by contracting all connected components that are not pairs to a bowtie.
    Then, $\tie(G)$ is a balanced bowtie forest, and
    \begin{equation}
      \|\bZ^G - \bZ^{\tie(G)}\| \leq d^2 \|\bM\|^{3d}(\epsilon_{\pow}(\bM) + \epsilon_{\offdiag}(\bM) + \epsilon_{\corr}(\bM))
    \end{equation}
\end{lemma}
\noindent
We provide the proof in Appendix~\ref{app:pf:lem:tying-partition-transport-ribbon}.
The proof considers various cases depending on the graph structure of $G$, but is similar in principle to the proof of Lemma~\ref{lem:tying-stretched-ribbon}, the tying bound for stretched forest diagrams---we again factorize CGMs and argue that the ``inner'' ribbon diagrams may be collapsed without incurring a substantial error in operator bound.

\begin{figure}
    \begin{center}
        \includegraphics[scale=0.7]{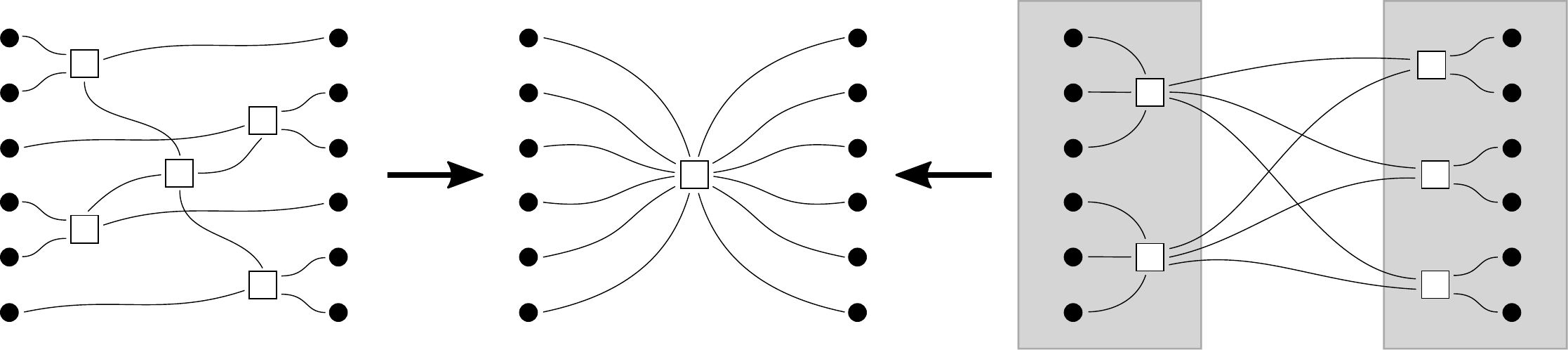}
    \end{center}
    \caption{\textbf{Tying stretched forest and partition transport ribbon diagrams.} We illustrate the key diagrammatic idea of the argument proving Lemma~\ref{lem:Z-main-positive}, that both the stretched forest ribbon diagrams appearing in $\bZ^{\main}$ and the partition transport ribbon diagrams appearing in $\bY$ may be ``tied'' to form the same bowtie forest ribbon diagrams.}
    \label{fig:tying-results}
\end{figure}

Next, we describe the result of tying every ribbon diagram in $\bY$.
As before, this involves a combinatorial calculation to sum over all diagrams that produce a given bowtie forest upon tying.
As we would hope, the resulting coefficients are the same as those arising in $\bZ^{\tied}$.

\begin{lemma}[Partition transport ribbon diagrams: combinatorial reduction]
    \label{lem:partition-transport-ribbon-combinatorial}
    For $d \in \NN$, define $\bY^{\tied[d, d]} \in \RR^{\binom{[N]}{d} \times \binom{[N]}{d}}$ by
    \begin{equation}
        \bY^{\tied[d, d]} \colonequals \sum_{\sigma, \tau \in \Part([d])} \prod_{A \in \sigma + \tau} (-1)^{|A| - 1}(|A| - 1)! (|A|)! \sum_{\bD \in \Plans(\sigma, \tau)} \frac{1}{\bD!} \bZ^{\tie(G(\sigma, \tau, \bD))}(\bM).
    \end{equation}
    Then,
    \begin{equation}
        \bY^{\tied[d,d]} = \sum_{\substack{F \in \sF(d, d) \\ F \text{ balanced bowtie forest}}} \xi(F) \cdot \bZ^F,
    \end{equation}
    where $\xi(F)$ is the same coefficient from Lemma~\ref{lem:stretched-ribbon-combinatorial}, given by
    \begin{equation}
        \xi(F) = \One\{F \text{ balanced}\} \prod_{\substack{C \in \conn(F) \\ C \text{ balanced bowtie} \\ \text{on } 2k \text{ leaves}}} (-1)^{k - 1} (k - 1)!\, k!.
    \end{equation}
    In particular, the direct sum of $\bY^{\tied[d^{\prime}, d^{\prime}]}$ over $0 \leq d^{\prime} \leq d$ equals $\bZ^{\tied}$ as defined in \eqref{eq:def:Z-tied}.
\end{lemma}
\noindent

\begin{remark}
    We note that the fact that the combinatorial quantities in Lemma~\ref{lem:stretched-ribbon-combinatorial} earlier, sums of \Mobius\ functions of stretched forests, and those in Lemma~\ref{lem:partition-transport-ribbon-combinatorial} above, sums of combinatorial coefficients of partition transport plans, are \emph{equal} is quite surprising.
    We isolate this fact to emphasize its unusual form:
    \begin{equation}
        \sum_{\substack{F \in \sF(\ell, m) \\ F \text{ stretched}}} \mu(F) = \sum_{\substack{\sigma \in \Part([\ell]) \\ \tau \in \Part([m])}} \prod_{A \in \sigma + \tau} (-1)^{|A| - 1}(|A| - 1)! (|A|)! \sum_{\bD \in \Plans(\sigma, \tau)} \frac{1}{\bD!}.
    \end{equation}
    Our proofs, unfortunately, give little insight as to why this should be the case, instead showing, in an especially technical manner for the left-hand side, that both sides equal another quantity.
    It would be interesting to find a combinatorial or order-theoretic argument explaining this coincidence, which is crucial to our argument, more directly.
\end{remark}

We again summarize our findings below.
\begin{corollary}
    \label{cor:Y-Z-tied}
    $\|\bY - \bZ^{\tied}\|
    \leq d^{5d}\|\bM\|^{3d}(\epsilon_{\pow}(\bM) + \epsilon_{\offdiag}(\bM) + \epsilon_{\corr}(\bM))$.
\end{corollary}
\begin{proof}
    Since $\bY$ and $\bZ^{\tied}$ are both block-diagonal, it suffices to consider a diagonal block indexed by $\binom{[N]}{d}$.
    Since by Lemma~\ref{lem:partition-transport-ribbon-combinatorial} this block of $\bZ^{\tied}$ is $\bY^{\tied[d, d]}$, this amounts to bounding $\|\bY^{[d, d]} - \bY^{\tied[d, d]}\|$.
    Applying triangle inequality and Lemma~\ref{lem:tying-partition-transport-ribbon}, we find
    \begin{align}
      \|\bY^{[d, d]} &- \bY^{\tied[d, d]}\| \nonumber \\
      &\leq \sum_{\sigma, \tau \in \Part([d])} \prod_{A \in \sigma + \tau} (|A| - 1)! (|A|)! \sum_{\bD \in \Plans(\sigma, \tau)} \frac{1}{\bD!} \|\bZ^{G(\sigma, \tau, \bD)} - \bZ^{\tie(G(\sigma, \tau, \bD))}\| \nonumber \\
      &\leq d^2\|\bM\|^{3d}(\epsilon_{\pow}(\bM) + \epsilon_{\offdiag}(\bM) + \epsilon_{\corr}(\bM)) \cdot (d - 1)! \, d! \cdot \sum_{\sigma, \tau \in \Part([d])} |\Plans(\sigma, \tau)| \nonumber
        \intertext{and bounding $|\Part([d])|$ and $|\Plans(\sigma, \tau)|$ using Propositions~\ref{prop:size-Part} and \ref{prop:size-Plans} respectively and noting that $d \cdot d! \leq d^d$ we have}
      &\leq \|\bM\|^{3d}(\epsilon_{\pow}(\bM) + \epsilon_{\offdiag}(\bM) + \epsilon_{\corr}(\bM)) \cdot d^{2d} \cdot d^{2d}d^d,
    \end{align}
    and the result follows.
\end{proof}

Finally, what will be the crucial feature of $\bY$ for us is that its spectrum is bounded below.
Since $\bY$ is formed by definition as a Gram matrix we will certainly have $\bY \succeq 0$; however, as we show below, we moreover have $\bY \succeq \lambda \bm I$ for some small but positive $\lambda > 0$, depending on the smallest eigenvalue of $\bM$.
Intuitively, before any technical reasoning we would already expect some quantitative statement of this kind, since whenever $\bM$ is non-singular the $h_S(\bV^{\top}\bz)$ are linearly independent and their conditioning should depend mostly on the conditioning of $(\bV^{\top}\bz)^S$.
\begin{proposition}
    \label{prop:Y-positive}
    $\lambda_{\min}(\bY) \geq \lambda_{\min}(\bM)^d$.
\end{proposition}
\begin{proof}
    Since $\bY$ is block-diagonal, it suffices to show the result for each $\bY^{[d, d]}$.
    Let $\bU \in \RR^{\sM_d([N]) \times \binom{[N]}{d}}$ have as its columns the monomial coefficients of the polynomials $h_S(\bx)$ (noting that the entries of the columns are indexed by multisets in $[N]$, compatible with this interpretation), and let $\bA \in \RR^{\sM_{d}([N]) \times \sM_{d}([N])}$ have as its entries $A_{S,T} = \langle (\bV^{\top}\bz)^S, (\bV^{\top}\bz)^T \rangle_{\partial}$.
    We then have $\bY^{[d, d]} = \bU^{\top} \bA \bU$.

    Write the singular value decomposition $\bV = \bQ_1 \bm\Sigma \bQ_2^{\top}$ for $\bQ_i \in \sO(N)$ and $\bm\Sigma$ diagonal containing the singular values of $\bV$.
    Then, applying the orthogonal invariance of Proposition~\ref{prop:apolar-orth-invariant}, we have $A_{S, T} = \langle (\bQ_2\bm\Sigma \bz)^S, (\bQ_2\bm\Sigma\bz)^T \rangle_{\partial}$.
    Now, letting $\bU^{\prime} \in \RR^{\sM_{d}([N]) \times \sM_{d}([N])}$ have as its columns the monomial coefficients of $(\bQ_2\bz)^S$ and letting $\bA^{\prime} \in \RR^{\sM_{d}([N]) \times \sM_{d}([N])}$ have entries $A_{S,T}^{\prime} = \langle (\bm\Sigma \bz)^S, (\bm\Sigma \bz)^T \rangle_{\partial}$, we have $\bA = \bU^{\prime\top} \bA^{\prime} \bU^{\prime}$, and $\bA^{\prime}$ is a diagonal matrix with entries equal to $d!$ multiplied by various products of $2d$ of the singular values of $\bV$.
    In particular, letting $\sigma_{\min}$ be the smallest such singular value, we have $\bA^{\prime} \succeq \sigma_{\min}^{2d} \bm I_{\sM_{d}([N])}$.
    Therefore, $\bA \succeq \sigma_{\min}^{2d} \cdot d!\, \bU^{\prime^{\top}} \bU^{\prime}$.
    But the entries of the matrix $d!\,\bU^{\prime^{\top}} \bU^{\prime}$ are the inner products $\langle (\bQ_2 \bz)^S, (\bQ_2\bz)^T \rangle_{\partial} = \langle \bz^S, \bz^T \rangle_{\partial}$, again by orthogonal invariance.
    Thus this matrix is merely $d! \, \bm I_{\sM_{d}([N])}$, so we find $\bA \succeq \sigma_{\min}^{2d} d! \bm I_{\sM_{d}([N])}$.

    Therefore, $\bY^{[d, d]} \succeq \sigma_{\min}^{2d} d! \, \bU^{\top} \bU$.
    Since the only multilinear monomial appearing in $h_S(\bx)$ is $\bx^{\bS}$, and this occurs with coefficient 1, the block of $\bU^{\top}$ indexed by all columns and rows corresponding to multisets with no repeated elements is the identity.
    In particular then, $\bU^{\top}\bU \succeq \bm I_{\binom{[N]}{d}}$, so $\bY^{[d, d]} \succeq \sigma_{\min}^{2d} d! \, \bm I_{I_{\binom{[N]}{d}}}$.
    Finally, since $\bM = \bV^{\top} \bV$, we have $\lambda_{\min}(\bM) = \sigma_{\min}^2$, and the result follows.
\end{proof}

Combining our results, we are now prepared to prove Lemma~\ref{lem:Z-main-positive}.
\begin{proof}[Proof of Lemma~\ref{lem:Z-main-positive}]
    We need only recall the main results from the last two sections:
    \begin{align}
      \|\bZ^{\main} - \bZ^{\tied}\| &\leq (6d)^{10d} \|\bM\|^{3d} \epsilon_{\tree}(\bM; d), \tag{\text{by Corollary~\ref{cor:Z-main-Z-tied}}}  \\
      \|\bY - \bZ^{\tied}\| &\leq d^{5d}\|\bM\|^{3d}(\epsilon_{\pow}(\bM) + \epsilon_{\offdiag}(\bM) + \epsilon_{\corr}(\bM)), \tag{\text{by Corollary~\ref{cor:Y-Z-tied}}}\\
      \lambda_{\min}(\bY) &\geq \lambda_{\min}(\bM)^d, \tag{\text{by Proposition~\ref{prop:Y-positive}}}
    \end{align}
    where we note that the assumption $\epsilon_{\tree}(\bM; d) \leq 1$ in Corollary~\ref{cor:Z-main-Z-tied} follows from the condition of Theorem~\ref{thm:lifting}, which is assumed in the statement.
    The result follows then follows by the eigenvalue inequality $\lambda_{\min}(\bZ^{\main}) \geq \lambda_{\min}(\bY) - \|\bZ^{\main} - \bZ^{\tied}\| - \|\bZ^{\tied} - \bY\|$.
\end{proof}

\subsection{Bound on the Error Term: Proof of Lemma~\ref{lem:Z-err-bound}}

Our first step in analyzing the error term is, as for the main term, to evaluate it in the multiharmonic basis, giving the entries of $\bZ^{\err}$.
We recall that, in Proposition~\ref{prop:E-err-factorization}, we found that $\tEE^{\err}$ decomposes as an application of $\tEE^{\main}$ to some of the input indices and a combination of error trees applied to the other indices.
Thus, part of the result of this calculation will be the familiar stretched forest terms from the calculations in Proposition~\ref{prop:Z-main-harmonic-basis}, while the remainder will consist of the $\Delta^F$ error components from Section~\ref{sec:main-error-terms}, applied to some of the partition components of the multiharmonic basis polynomials.
To make it easier to describe and eventually bound the latter, we define the following type of error matrix.
\begin{definition}[Partition-error matrices]
    \label{def:partition-error-matrix}
    Suppose $\sigma = \{A_1, \dots, A_n\} \in \Part([\ell]; \odd)$, $\tau = \{B_1, \dots, B_p\} \in \Part([m]; \odd)$, and $T \in \sT(n + p)$.
    We then define the \emph{partition-error matrix} $\bm\Delta^{(\sigma, \tau, T)} \in \RR^{\binom{[N]}{\ell} \times \binom{[N]}{m}}$ associated to this triplet to have entries
    \begin{align}
      &\Delta^{(\sigma, \tau, T)}_{\bs\bt} \colonequals \sum_{\substack{\ba \in [N]^{\sigma[\geq 3]} \\ \bb \in [N]^{\tau[\geq 3]}}} \prod_{A \in \sigma}\prod_{i \in A}M_{i,f_{\bs, \ba}(A)} \cdot \prod_{B \in \tau}\prod_{j \in B} M_{j,g_{\bt, \bb}(B)} \, \, \cdot \nonumber \\
      &\hspace{3.55cm} \Delta^T(\bM; (f_{\bs, \ba}(A_1), \dots, f_{\bs, \ba}(A_n), g_{\bt, \bb}(B_1), \dots, g_{\bt, \bb}(B_p))).
    \end{align}
\end{definition}

\begin{proposition}
    \label{prop:Z-err-harmonic-basis}
    For any $S, T \subseteq [N]$,
    \begin{align}
      Z^{\err}_{S,T}
      &= \tEE^{\err}[h_S^{\downarrow}(\bx)h_T^{\downarrow}(\bx)] \nonumber \\
      &= \sum_{\substack{A \subseteq S \\ B \subseteq T \\ \mathclap{A + B \neq S + T}}}
      \Bigg(\sum_{\substack{F \in \sF(|A|, |B|) \\ F \text{ stretched}}}\mu(F) \cdot Z^F_{A,B}\Bigg) \Bigg(\sum_{\pi \in \Part((S \setminus A) + (T \setminus B); \even)} (-1)^{|\pi|}\prod_{R \in \pi}\sum_{\substack{\sigma \in \Part([|R \cap S|]; \odd) \\ \tau \in \Part([|R \cap T|]; \odd)}}  \nonumber \\
      &\hspace{2cm}  \prod_{A \in \sigma + \tau}(-1)^{|A| - 1}(|A| - 1)!\sum_{F \in \sT(|\sigma| + |\tau|)} \mu(F) \cdot \Delta^{(\sigma, \tau, F)}_{R \cap S, R \cap T}\Bigg).
    \end{align}
\end{proposition}
\begin{proof}
    As in Proposition~\ref{prop:Z-main-harmonic-basis}, we begin by expanding directly:
    \begin{align*}
      &\tEE^{\err}[h_S^{\downarrow}(\bx)h_T^{\downarrow}(\bx)] \\
      &= \sum_{\substack{\sigma \in \Part(S) \\ \tau \in \Part(T)}} \prod_{R \in \sigma + \tau} (-1)^{|R| - 1}(|R| - 1)!
                                                                                                                                                       \prod_{R \in \sigma[\even] + \tau[\even]} q_R^{\downarrow} \, \, \cdot \nonumber \\
      &\hspace{1.35cm} \sum_{\substack{\ba \in [N]^{\sigma[\odd; \geq 3]} \\ \bb \in [N]^{\tau[\odd; \geq 3]}}} \prod_{A \in \sigma[\odd]} \prod_{i \in A} M_{f_{\ba}(A), i} \prod_{B \in \tau[\odd]} \prod_{j \in B} M_{f_{\bb}(B), j} \cdot \tEE^{\err}\left[ \prod_{A \in \sigma[\odd]}x_{f_{\ba}(A)} \prod_{B \in \tau[\odd]} x_{g_{\bb}(B)}\right]
      \intertext{where $f_{\ba}, g_{\bb}$ are defined as in Proposition~\ref{prop:Z-main-harmonic-basis}. Now, expanding the pseudoexpectation according to Proposition~\ref{prop:E-err-factorization}, we have}
      &= \sum_{\substack{\sigma \in \Part(S) \\ \tau \in \Part(T)}} \prod_{R \in \sigma + \tau} (-1)^{|R| - 1}(|R| - 1)!
                                                                                                                                                       \prod_{R \in \sigma[\even] + \tau[\even]} q_R^{\downarrow} \, \, \cdot \nonumber \\
      &\hspace{1.35cm} \sum_{\substack{\ba \in [N]^{\sigma[\odd; \geq 3]} \\ \bb \in [N]^{\tau[\odd; \geq 3]}}} \prod_{A \in \sigma[\odd]} \prod_{i \in A} M_{f_{\ba}(A), i} \prod_{B \in \tau[\odd]} \prod_{j \in B} M_{f_{\bb}(B), j} \sum_{\substack{\pi \subseteq \sigma[\odd] \\ \rho \subseteq \tau[\odd] \\ \mathclap{|\pi| + |\rho| < |\sigma[\odd]| + |\tau[\odd]|}}} \tEE^{\main}\left[\prod_{A \in \pi}x_{f_{\ba}(A)} \prod_{B \in \rho} x_{g_{\bb}(B)}\right] \\
      &\hspace{1.2cm}  \sum_{\substack{\beta \in \Part((\sigma[\odd] \setminus \pi) + (\tau[\odd] \setminus \rho); \even)}}\prod_{\gamma \in \beta} \left( -\sum_{T \in \sT(|R|)} \mu(T) \cdot \Delta^T(\bM; (f_{\ba}(A))_{A \in \gamma \cap \sigma} \circ (g_{\bb}(B))_{B \in \gamma \cap \tau})\right)
        \intertext{Here, we swap the order of summation and reorganize the sum according to the union of all parts of $\pi$ and $\sigma[\even]$, which we call $J$, and the union of all parts of $\rho$ and $\tau[\even]$, which we call $K$.
        Recognizing after this manipulation the intermediate result from Proposition~\ref{prop:Z-main-harmonic-basis}, we continue
        }
      &= \sum_{\substack{J \subseteq S \\ K \subseteq T \\ J + K \neq S + T}} \tEE^{\main}[h_J^{\downarrow}(\bx)h_K^{\downarrow}(\bx)] \sum_{\substack{\sigma \in \Part(S \setminus J; \odd) \\ \tau \in \Part(T \setminus K; \odd)}} \prod_{R \in \sigma + \tau}(-1)^{|R| - 1}(|R| - 1)! \\
      &\hspace{1.55cm}\sum_{\substack{\ba \in [N]^{\sigma[\odd; \geq 3]} \\ \bb \in [N]^{\tau[\odd; \geq 3]}}} \prod_{A \in \sigma}\prod_{i \in A} M_{f_{\ba}(A),i} \prod_{B \in \tau}\prod_{j \in B} M_{f_{\bb}(B),i} \\
      &\hspace{2.6cm}\sum_{\beta \in \Part(\sigma + \tau; \even)} \prod_{\gamma \in \beta} \left( -\sum_{T \in \sT(|R|)} \mu(T) \cdot \Delta^T(\bM; (f_{\ba}(A))_{A \in \gamma \cap \sigma} \circ (g_{\bb}(B))_{B \in \gamma \cap \tau})\right)
    \end{align*}
    and again exchanging the order of summation and letting $\pi$ be the partition formed by taking the union of the sets in every part of $\beta$, by a similar manipulation to that in Proposition~\ref{prop:Z-main-harmonic-basis} we complete the proof.
\end{proof}

We now develop a few tools to bound the norms of partition-error matrices.
The following is a minor variant of Proposition~\ref{prop:factorization}, a diagrammatic factorization of CGMs that we use at length in the deferred technical proofs.
This shows how $\bm\Delta^{(\sigma, \tau, T)}$ can be factorized into two outer factors that are similar to CGMs with no $\square$ vertices, and an inner factor that consists of values of $\Delta^T$ arranged in a matrix of suitable shape.
\begin{proposition}[Factorizing partition-error matrices]
    \label{prop:partition-error-factorization}
    Let $\sigma = \{A_1, \dots, A_n\} \in \Part([\ell]; \odd)$, $\tau = \{B_1, \dots, B_p\} \in \Part([m]; \odd)$, and $T \in \sT(n + p)$.
    Define $\bZ^{\sigma} \in \RR^{\binom{[N]}{\ell} \times [N]^{n}}$ to have entries
    \begin{equation}
        Z^{\sigma}_{\bs\ba} = \prod_{A_q = \{i\} \in \sigma[1]}\One\{s_i = a_q\}\prod_{A_q \in \sigma[\geq 3]}\prod_{i \in A_q}M_{s_i,a_q},
    \end{equation}
    and similarly $\bZ^{\tau} \in \RR^{\binom{[N]}{m} \times [N]^{p}}$.
    Let $\bF^{(T, n, p)} \in \RR^{[N]^{n} \times [N]^p}$ have entries
    \begin{equation}
        F^{(T, n, p)}_{\ba\bb} = \Delta^T(\bM; (a_1, \dots, a_n, b_1, \dots, b_p)).
    \end{equation}
    Then, $\bm\Delta^{(\sigma, \tau, T)} = \bm Z^{\sigma} \bm F^{(T, n, p)} \bZ^{\tau^{\top}}$.
\end{proposition}
\begin{proof}
    The result follows from expanding the definition of the matrix multiplication and comparing with Definition~\ref{def:partition-error-matrix}.
\end{proof}

Next, we show how the norm of the inner matrix can be controlled; in fact, we give a stronger statement bounding the Frobenius norm.
\begin{proposition}[Error matrix Frobenius norm bound]
    \label{prop:err-frob-norm-bound}
    For any $d^{\prime} \leq d$ and $T \in \sT(2d^{\prime})$,
    \begin{equation}
        \left(\sum_{\bs \in [N]^{2d^{\prime}}} (\Delta^T(\bM; \bs))^2\right)^{1/2} \leq (2d)^d \, \epsilon_{\err}(\bM; 2d)
    \end{equation}
\end{proposition}
\begin{proof}
    Recall that we denote by $\set(\bs)$ the set of distinct indices appearing in $\bs$.
    By definition of $\epsilon_{\err}$, we have
    \begin{equation}
        |\Delta^T(\bM; \bs)| \leq N^{-|\set(\bs)| / 2} \epsilon_{\err}(\bM; 2d).
    \end{equation}
    Therefore, we find
    \begin{align}
      \left(\sum_{\bs \in [N]^{2d^{\prime}}} (\Delta^T(\bM; \bs))^2\right)^{1/2}
      &\leq \left(\sum_{\bs \in [N]^{2d^{\prime}}} N^{-|\set(\bs)|}\right)^{1/2}\epsilon_{\err}(\bM; 2d) \\
      &\leq \left(\sum_{k = 1}^{2d^{\prime}} N^{-k}\cdot\#\{\bs \in [N]^{2d^{\prime}}: |\set(\bs)| = k\}\right)^{1/2}\epsilon_{\err}(\bM; 2d) \\
      &\leq \left(\sum_{k = 1}^{2d^{\prime}} \frac{k^{2d^{\prime}}}{k!}\right)^{1/2}\epsilon_{\err}(\bM; 2d) \\
      &\leq (2d^{\prime})^{d^{\prime}}\epsilon_{\err}(\bM; 2d),
    \end{align}
    and the result follows since $d^{\prime} \leq d$.
\end{proof}

Combining these results with an ancillary result from the Appendix gives the following bound.
\begin{corollary}
    \label{cor:partition-error-norm-bound}
    $\|\bm\Delta^{(\sigma, \tau, T)}\| \leq (2(\ell + m))^{(\ell + m)}\|\bM\|^{\ell + m} \epsilon_{\err}(\bM; 2d)$.
\end{corollary}
\begin{proof}
    By norm submultiplicativity, $\|\bm\Delta^{(\sigma, \tau, T)}\| \leq \|\bZ^{\sigma}\| \cdot \|\bF^{(T, n, p)}\| \cdot \|\bZ^{\tau}\|$.
    By Proposition~\ref{prop:cgm-generic-norm-bound}, we have $\|\bZ^{\sigma}\| \leq \|\bM\|^{\ell}$ and $\|\bZ^{\tau}\| \leq \|\bM\|^m$, and by Proposition~\ref{prop:err-frob-norm-bound}, we have $\|\bF^{(T, n, p)}\| \leq \|\bF^{(T, n, p)}\|_F \leq (2(n + p))^{n + p}\epsilon_{\err}(\bM; n + p)$, and the result then follows after noting $n + p \leq \ell + m$ since $n = |\sigma|$ and $p = |\tau|$.
\end{proof}

\begin{proof}[Proof of Lemma~\ref{lem:Z-err-bound}]
    First, we note that under the assumptions of Theorem~\ref{thm:lifting}, which we have also assumed in the statement of the Lemma, we have $\epsilon_{\err}(\bM; 2d) \leq 1$.

    We then follow the same manipulations as in Lemma~\ref{lem:Z-main-positive}, using Proposition~\ref{prop:block-matrix-norm} to bound the norm of $\bZ^{\err}$ by the sum of all block norms.
    Also, since products over subsets of indices correspond to tensorizations of terms in this sum (see Proposition~\ref{prop:cgm-tt-tensorization}), we may bound such products by corresponding products of matrix norms.
    We therefore find
    \begin{align*}
      \|\bZ^{\err}\|
      &\leq \sum_{\ell, m = 0}^d \sum_{\substack{a \in [\ell] \\ b \in [m] \\ a + b < \ell + m}} \Bigg(\sum_{\substack{F \in \sF(a, b) \\ F \text{ stretched}}} |\mu(F)| \cdot \|\bZ^F\| \Bigg)\\
      &\hspace{1.6cm}\sum_{\pi \in \Part([(\ell - a) + (m - b)]; \even)}\prod_{R \in \pi}\Bigg(\sum_{\substack{\sigma \in \Part(R \cap [\ell - a]; \odd) \\ \tau \in \Part(R \cap \{\ell - a + 1, \dots, (\ell - a) + (m - b)\}; \odd)}} \prod_{A \in \sigma + \tau} (|A| - 1)! \\
      &\hspace{8cm}\sum_{T \in \sT(|\sigma|, |\tau|)} |\mu(T)| \cdot \|\bm \Delta^{(\sigma, \tau, T)}\|\Bigg)
        \intertext{In the sum over stretched forest ribbon diagrams, by Proposition~\ref{prop:bound-mu} we have $|\mu(F)| \leq (3(a + b))! \leq (3(a + b))^{3(a + b)} \leq 3((\ell + m))^{3(\ell + m)} \leq (6d)^{6d}$, by Proposition~\ref{prop:size-F} we have $|\sF(a, b)| \leq (a + b)^{3(a + b)} \leq (2d)^{6d}$, and by Proposition~\ref{prop:cgm-generic-norm-bound} and Corollary~\ref{cor:sF-vertex-edge-bound} we have $\|\bZ^F\| \leq \|\bM\|^{3d}$.
        In the second term, by Corollary~\ref{cor:partition-error-norm-bound} we have $\|\bm\Delta^{(\sigma, \tau, T)}\| \leq (2|R|)^{|R|}\|\bM\|^{|R|}\epsilon_{\err}(\bM; 2d)$, by Proposition~\ref{prop:bound-mu} we have $|\mu(T)| \leq (3|R|)! \leq (3|R|)^{3|R|}$, and by Proposition~\ref{prop:size-F} we have $|\sT(|\sigma|, |\tau|)| \leq |\sF(|\sigma|, |\tau|)| \leq (|\sigma| + |\tau|)^{3(|\sigma| + |\tau|)} \leq (|R|)^{3|R|}$.
        We also have $\prod_{A \in \sigma + \tau}(|A| - 1)! \leq (|\sigma| + |\tau|)! \leq (|R|)! \leq |R|^{|R|}$.
        Finally, by Proposition~\ref{prop:size-Part}, viewing $\sigma$ and $\tau$ taken together as a partition of $R$, the number of choices of $\sigma$ and $\tau$ is at most $|R|^{|R|}$.
        Combining these observations, we continue}
      &\leq (6d)^{12d}\|\bM\|^{3d}\epsilon_{\err}(\bM; 2d)\sum_{\ell, m = 0}^d\sum_{\substack{a \in [\ell] \\ b \in [m] \\ a + b < \ell + m}} \sum_{\pi \in \Part([(\ell - a) + (m - b)]; \even)}\prod_{R \in \pi} (3|R|)^{9|R|} \|\bM\|^{|R|}
      \intertext{where since $\sum_{R \in \pi}|R| \leq 2d$ we continue}
      &\leq (6d)^{30d}\|\bM\|^{5d}\epsilon_{\err}(\bM; 2d) \cdot d^2 \cdot d^2 \cdot |\Part(2d)|
        \intertext{and by Proposition~\ref{prop:size-Part} again we may finish}
      &\leq (12d)^{32d}\|\bM\|^{5d}\epsilon_{\err}(\bM; 2d),
    \end{align*}
    completing the proof.
\end{proof}

\section{Proofs of High-Rank Applications}
\label{sec:pf:applications}

Before proceeding with the individual proofs, we introduce the following generally useful bookkeeping formalism.
\begin{definition}[Pattern diagram]
    Suppose $F = (V^{\bullet} \sqcup V^{\square}, E)$ is a diagram, $\bs \in [N]^{|V^{\bullet}|}$, and $\ba \in [N]^{V^{\square}}$.
    Let the associated \emph{pattern diagram}, denoted $\pat(F, \bs, \ba)$, be the graph $G$ with two types of vertices, $\bullet$ and $\square$ (as for diagrams), formed by starting with $G$ and identifying all $v$ whose value of $f_{\bs, \ba}(v)$ is equal, where if we identify a $\bullet$ vertex with either a $\bullet$ or a $\square$ vertex then the result is a $\bullet$ vertex, but if we identify two $\square$ vertices then the result is again a $\square$ vertex.
    We then remove all self-loops from $\pat(F, \bs, \ba)$ (but allow parallel edges).
    The graph $G$ is also equipped with a natural labelling inherited from $\bs$ and $\ba$, which we denote $f$, sometimes writing $(G, f) = \pat(F, \bs, ba)$.

    Finally, if $F_1, \dots, F_m$ are diagrams, and $\bs_i \in [N]^{|V^{\bullet}(F_i)|}$ and $\ba_i \in [N]^{V^{\square}(F_i)}$, then we let $\pat((F_1, \bs_1, \ba_1), \dots, (F_m, \bs_m, \ba_m))$ be the graph formed by applying the above identification procedure to the disjoint union of the $F_i$, each labelled by $\bs_i$ and $\ba_i$.
\end{definition}
\begin{definition}
    Let $\Pat^{\leq 2d}(m)$ be the set of unlabelled $\pat((T_1, \bs, ba_1), \dots, (T_m, \bs, \ba_m))$ that occur for $T_i \in \sT(2d^{\prime})$, $\bs \in [N]^{2d^{\prime}}$, and $\ba_i \in [N]^{V^{\square}(T)}$ for some choice of $1 \leq d^{\prime} \leq d$.
    We emphasize that we force $\bs$ to be the same in all inputs here.
\end{definition}
\noindent
The way we will use this is by considering the pattern diagram $G$ of any term in any CGS quantity, whose magnitude scales, depending on the behavior of the entries of $\bM$, as $N^{-\gamma|E(G)|}$.
On the other hand, the number of terms sharing a given pattern diagram is essentially $N^{|V^{\square}(G)|}$.
Grouping terms by pattern diagram allows us to take advantage of the tradeoff between these two quantities.

In particular, we will want to use this to analyze the quantities $\epsilon_{\tree}$ and $\epsilon_{\err}$, so we define the following subsets of pattern diagrams.
\begin{definition}
    Let $\Pat_{\tree}^{\leq 2d}(m) \subseteq \Pat^{\leq 2d}(m)$ be the set of unlabelled $\pat((T_1, \bs, \ba_1), \dots, (T_m, \bs, \ba_m))$ that occur for $T_i \in \sT(2d^{\prime})$, $\bs \in [N]^{2d^{\prime}}$, and $\ba_i \in [N]^{V^{\square}(T)}$, such that either the entries of $\bs$ are not all equal, or $s_1 = \cdots = s_{2d^{\prime}} = j$ but not all of the entries of $\ba_i$ equal $j$ for all $i$, for some $1 \leq d^{\prime} \leq d$.
\end{definition}

\begin{definition}
    Let $\Pat_{\err}^{\leq 2d}(m) \subseteq \Pat^{\leq 2d}(m)$ be the set of unlabelled $\pat((T_1, \bs, ba_1), \dots, (T_m, \bs, \ba_m))$ that occur for $T_i \in \sT(2d^{\prime})$, $\bs \in [N]^{2d^{\prime}}$, and $\ba_i \in [N]^{V^{\square}(T)}$, such that $\ba_i$ is $(T_i, \bs)$-loose for all $i$, for some choice of $1 \leq d^{\prime} \leq d$.
\end{definition}

The following two simple facts will be useful throughout; we will introduce other combinatorial properties as needed in our arguments.

\begin{proposition}
    \label{prop:pat-conn}
    All diagrams in $\Pat^{\leq 2d}(m)$ are connected for any $d \geq 1$ and $m \geq 1$.
\end{proposition}

\begin{proposition}
    \label{prop:size-Pat}
    $|\Pat^{\leq 2d}(m)| \leq (3md)^{9md}$.
\end{proposition}
\begin{proof}
    Every $G \in \Pat^{\leq 2d}(m)$ is connected, and has at most $3md$ vertices and $3md$ edges by Corollary~\ref{cor:sF-vertex-edge-bound} since this holds for each $T_i \in \sT(2d^{\prime})$ for any $d^{\prime} \leq d$ and $G$ can have only fewer vertices and edges than the disjoint union of the $T_i$.
    Generally, the number of connected graphs on at most $m \geq 2$ vertices, with at most $n$ edges for $n \geq m$, and equipped with a partition of the vertices into two parts is at most $2^m \cdot (m^2)^{n} \leq m^{3n}$, where we ignore that there may be fewer vertices or edges by allowing ``excess'' vertices and edges to be added to different connected components that we may ignore.
\end{proof}

\subsection{Laurent's Pseudomoments: Proof of Theorem~\ref{thm:appl-laurent-approx}}

We will use the following more detailed bounds on pattern diagrams.
\begin{proposition}
    \label{prop:pat-tree-bound}
    If $G = (V^{\bullet} \sqcup V^{\square}, E) \in \Pat_{\tree}^{\leq 2d}(1)$ for any $d$, then $|E| \geq |V^{\bullet}| + |V^{\square}| - \One\{|V^{\bullet}| > 1\}$.
\end{proposition}
\begin{proof}
    We consider two cases.
    If $|V^{\bullet}| > 1$, then the result follows since $G$ is connected by Proposition~\ref{prop:pat-conn}.
    If $|V^{\bullet}| = 1$, then since the initial diagram $G$ is formed from by identifying vertices is a tree, all leaves of that tree are identified in forming $G$, and $G$ has more than one vertex, $G$ must have a cycle.
    Therefore, in this case, $|E| \geq |V^{\bullet}| + |V^{\square}|$, completing the proof.
\end{proof}

\begin{proposition}
    \label{prop:pat-err-bound}
    If $G = (V^{\bullet} \sqcup V^{\square}, E) \in \Pat_{\err}^{\leq 2d}(1)$ for any $d$, then $|E| \geq |V^{\bullet}| + |V^{\square}|$.
\end{proposition}
\begin{proof}
    Suppose for the sake of contradiction that this is not the case.
    Since $T$ is connected, by Proposition~\ref{prop:pat-conn} $G$ is connected as well, and if $|E| \leq |V^{\bullet}| + |V^{\square}| - 1$ then in fact equality holds and $G$ is a tree, in particular having no parallel edges.
    On the other hand, if $\ba$ is $(T, \bs)$-loose, then there exists some index $i \in [N]$ and a $\square$ vertex $v$ in the minimal spanning subtree of $\{w \in V^{\bullet}: s_{\kappa(w)} = i\}$ such that $a_v \neq i$.
    Thus there must exist some $\square$ vertex $v^{\prime}$ in this minimal spanning subtree with at least two neighbors $w_1, w_2$ such that $f_{\bs, \ba}(w_1) = f_{\bs, \ba}(w_2) = i$ but $a_{v^{\prime}} \neq i$.
    The $\square$ vertex of $G$ to which $v^{\prime}$ is identified will then be incident with a pair of parallel edges, giving a contradiction.
\end{proof}

\begin{proof}[Proof of Theorem~\ref{thm:appl-laurent-approx}]
    We will set
    \begin{equation}
        \bM \colonequals \left(1 + \frac{1 - \alpha}{N - 1}\right)\bm I_N - \frac{1 - \alpha}{N - 1}\one_N\one_N^{\top}
    \end{equation}
    and take $\tEE = \tEE_{\bM}$.
    We first use Theorem~\ref{thm:lifting} to show that $\tEE$ is a degree $2d$ pseudoexpectation.

    For the simpler incoherence quantities, we directly bound
    \begin{align}
      \epsilon_{\offdiag}(\bM)
      &\leq \frac{1}{N - 1}, \\
      \epsilon_{\corr}(\bM)
      &\leq \left(2\left(\frac{1}{N - 1}\right)^2 + (N - 2) \left(\frac{1}{N - 1}\right)^4\right)^{1/2} \leq \frac{2}{N - 1}, \\
      \epsilon_{\pow}(\bM) &\leq (N - 1) \left(\frac{1}{N - 1}\right)^2 = \frac{1}{N - 1}.
    \end{align}
    
    For $\epsilon_{\tree}$, we group terms according to their pattern diagram.
    A given $G = ((V^{\bullet}, V^{\square}), E)$ can occur in at most $N^{|V^{\square}|}$ terms, and each term contributes at most $(N - 1)^{-|E|}$.
    We then have
    \begin{align}
      \epsilon_{\tree}(\bM; 2d)
      &= \max_{0 \leq d^{\prime} \leq d} \max_{T \in \sT(2d^{\prime})} \max_{\bs \in [N]^{2d^{\prime}}} \left|Z^T(\bM; \bs) - \One\{s_1 = \cdots = s_N\}\right| \nonumber \\
      &\leq \sum_{\substack{G = ((V^{\bullet}, V^{\square}), E) \in \Pat_{\tree}^{\leq 2d}(1)}} N^{|V^{\square}|} (N - 1)^{-|E|} \nonumber \\
      &\leq \sum_{\substack{G = ((V^{\bullet}, V^{\square}), E) \in \Pat_{\tree}^{\leq 2d}(1)}} N^{|V^{\square}|} (N - 1)^{-|V^{|\square|} - |V^{\bullet}| + \One\{|V^{\bullet}| > 1\}} \tag{by Proposition~\ref{prop:pat-tree-bound}} \\
      &\leq \frac{2^{3d}}{N - 1}\sum_{\substack{G = ((V^{\bullet}, V^{\square}), E) \in \Pat_{\tree}^{\leq 2d}(1)}} 1 \tag{by Corollary~\ref{cor:sF-vertex-edge-bound}} \\
      &\leq \frac{(3d)^{12d}}{N - 1}. \tag{by Proposition~\ref{prop:size-Pat}}
    \end{align}

    For $\epsilon_{\err}$, we follow the same strategy.
    The main additional observation is that $|\set(\bs)|$ is always simply the number of $\bullet$ vertices in $\pat(T, \bs, \ba)$, regardless of $T$ or $\ba$.
    Thus we find
    \begin{align}
      \epsilon_{\err}(\bM; 2d)
      &= \max_{0 \leq d^{\prime} \leq d} \max_{T \in \sT(2d^{\prime})} \max_{\bs \in [N]^{2d^{\prime}}} N^{|\set(\bs)| / 2} \Bigg|\sum_{\substack{\ba \in [N]^{V^{\square}} \\ \ba \text{ } (T, \bs)\text{-loose}}} \prod_{(v, w) \in E(T)} M_{f_{\bs, \ba}(v)f_{\bs, \ba}(w)}\Bigg| \nonumber \\
      &\leq \sum_{G = (V^{\bullet} \sqcup V^{\square}, E) \in \Pat^{\leq 2d}_{\err}(1)}N^{|V^{\bullet}| / 2} \cdot N^{|V^{\square}|} (N - 1)^{-|E|} \nonumber \\
      &\leq \sum_{G = (V^{\bullet} \sqcup V^{\square}, E) \in \Pat^{\leq 2d}_{\err}(1)} N^{|V^{\bullet}| / 2 + |V^{\square}|} (N - 1)^{-|V^{\bullet}| - |V^{\square}|} \tag{by Proposition~\ref{prop:pat-err-bound}} \\
      &\leq \frac{2^{6d}}{\sqrt{N}} \sum_{G = (V^{\bullet} \sqcup V^{\square}, E) \in \Pat^{\leq 2d}_{\err}(1)} 1 \tag{by Corollary~\ref{cor:sF-vertex-edge-bound}} \\
      &\leq \frac{(3d)^{15d}}{\sqrt{N}}. \tag{by Proposition~\ref{prop:size-Pat}}
    \end{align}
    (This may be sharpened to $O(N^{-1})$ by considering the case $k = 1$ separately, but that would not change the final result significantly.)

    Combining these results, we find $\epsilon(\bM; 2d) \leq 2(3d)^{15d} N^{-1/2}$.
    Thus, since $\lambda_{\min}(\bM) \geq \alpha$ and $\|\bM\| \leq 2$, the result will follow so long as $\alpha \geq 64 (12d)^{47} N^{-1/2d}$.
    If $d = \log N / 100 \log\log N$, then we have
    \begin{equation}
        64(12d)^{47} N^{-1 / 2d} = 64\left(\frac{12 \log N}{100 \log\log N}\right)^{47} (\log N)^{-100 / 2} \leq 64\left(\frac{3}{20 \log\log N}\right)^{50},
    \end{equation}
    so with $\alpha = (\log\log N)^{-50}$ it follows that $\tEE_{\bM}$ is a degree $2d$ pseudoexpectation.
    
    It remains to verify \eqref{eq:laurent-approx-1}, which gives the leading order behavior of the pseudomoments:
    \begin{equation}
        \tEE\left[\prod_{i \in S} x_i\right] = \One\{|S| \text{ even}\} \cdot \left(\frac{(-1)^{|S| / 2}(|S| - 1)!!}{(N / (1 - \alpha))^{|S| / 2}} + O_{|S|}\left(\frac{1}{N^{|S| / 2 +  1}}\right)\right).
    \end{equation}
    We claim that the leading order part is exactly the sum of the terms $\mu(F) \cdot Z^F(\bM; S)$ for $F$ a forest where every connected component is two $\bullet$ vertices connected by an edge (a ``pair''), since there are $(|S| - 1)!!$ such perfect matchings.
    Thus it suffices to bound the contributions of all other terms.

    We use pattern graphs once again, now noting that, since we are assuming $S$ is a \emph{set}, no $\bullet$ vertices will be identified with one another.
    Suppose $F$ is a good forest and $G = \pat(F, \bs, \ba)$ where all indices of $\bs$ are distinct.
    We then have $|E(G)| \geq 2|V^{\square}(G)| + \frac{1}{2}|V^{\bullet}(G)|$, because $|E(G)| = \frac{1}{2}\sum_{v} \deg(v)$, and every $\square$ vertex in $G$ after the identification procedure will still have degree at least 4 and every $\bullet$ vertex will still have degree at least 1 since no $\bullet$ vertices are identified.
    Moreover, if $|V^{\square}(G)| = 0$, then the above inequality is tight if and only if $F$ is a perfect matching to begin with.
    Therefore, if $F$ is \emph{not} a perfect matching, then, writing for the moment $\Pat_{\sF}^{2d}$ for the pattern graphs arising as any $\pat(F, \bs, \ba)$ for $F \in \sF(2d)$ (with forests rather than trees), we find
    \begin{align}
      |Z^F(\bM; S)|
      &\leq \sum_{\substack{G = (V^{\bullet} \sqcup V^{\square}, E) \in \Pat_{\sF}^{|S|} \\ |V^{\bullet}| = |S|}} N^{|V^{\square}|} (N - 1)^{-|E|} \nonumber \\
      &\leq \sum_{\substack{G = (V^{\bullet} \sqcup V^{\square}, E) \in \Pat_{\sF}^{|S|} \\ |V^{\bullet}| = |S|}} N^{|V^{\square}|} (N - 1)^{-2|V^{\square}| - |V^{\bullet}| / 2 - \One\{|V^{\square}| = 0\}} \nonumber \\
      &\leq N^{-|V^{\bullet}| / 2 - 1}\sum_{\substack{G = (V^{\bullet} \sqcup V^{\square}, E) \in \Pat_{\sF}^{|S|} \\ |V^{\bullet}| = |S|}} 1.
    \end{align}
    Finally, since the remaining counting coefficient, $\max_{F \in \sF(|S|)} |\mu(F)|$, and $|\sF(|S|)|$ all depend only on $|S|$, the result follows.
    (Bounding the remaining combinatorial coefficient and using Propositions~\ref{prop:bound-mu} and \ref{prop:size-F} here can give a weak quantitative dependence on $|S|$ as well.)
\end{proof}

\subsection{Random High-Rank Projectors: Proof of Theorem~\ref{thm:appl-high-rank}}

To handle the random case, we will need some more involved tools that we introduce now.
The key analytic tool for controlling the more complicated incoherence quantities is the family of \emph{hypercontractive} concentration inequalities, which state (in the case we will use) that low-degree polynomials of independent gaussian random variables concentrate well.
The underlying fact is the following norm inequality for these polynomials.
\begin{proposition}[Gaussian hypercontractivity, Theorem 5.10 of \cite{Janson-1997-GaussianHilbertSpaces}]
    \label{prop:gauss-hc}
    Let $p \in \RR[x_1, \dots, x_N]$ be a polynomial with $\deg(p) \leq D$.
    Then, for all $q \geq 2$,
    \begin{equation}
        (\EE |p(\bg)|^q)^{1/q} \leq (q - 1)^{D / 2} \cdot (\EE |p(\bg)|^2)^{1/2}
    \end{equation}
\end{proposition}
\noindent
The consequence we will be interested in is the following very convenient tail bound, which reduces analyzing the concentration of a polynomial to computing its second moment.
\begin{corollary}
    \label{cor:gauss-hc-tail}
    Let $p \in \RR[x_1, \dots, x_N]$ be a polynomial with $\deg(p) \leq D$.
    Then, for all $t \geq (2e^2)^{D / 2}$,
    \begin{equation}
        \PP\left[ \|\bp(\bg)\|_2^2 \geq t \cdot \EE \|\bp(\bg)\|^2_2\right] \leq \exp\left(-\frac{D}{e^2}t^{2/D}\right).
    \end{equation}
\end{corollary}
\begin{proof}
    By Proposition~\ref{prop:gauss-hc}, for any $q \geq 2$,
    \begin{align}
      \PP[|p(\bg)| \geq t (\EE|p(\bg)|^2)^{1/2}]
      &= \PP[|p(\bg)|^q \geq t^q(\EE|p(\bg)|^2)^{q/2}] \nonumber \\
      &\leq t^{-q} (\EE|p(\bg)|^2)^{-q/2}\EE |p(\bg)|^q \nonumber \\
      &\leq t^{-q} (q - 1)^{qD / 2} \nonumber \\
      &\leq (q^{D / 2} / t)^{q}, \nonumber 
        \intertext{and setting $q \colonequals t^{2/D} / e^2 \geq 2$ we have}
      &= \exp\left(-\frac{D}{e^2}t^{2/D}\right),
    \end{align}
    completing the proof.
\end{proof}

We will also need some combinatorial preliminaries, which describe how to compute expectations of gaussian polynomials like those that will wind up associated with pattern graphs in our calculations.
\begin{definition}
    \label{def:cycle-cover}
    A \emph{cycle cover} of a graph $G$ is a partition of the edges into edge-disjoint cycles.
    We denote the number of cycles in the largest cycle cover of $G$ by $c_{\max}(G)$.
\end{definition}

\begin{proposition}
    \label{prop:cycle-cover}
    Let $G = (V, E)$ be a graph, and for each $v \in V$ draw $\bg_v \sim \sN(\bm 0, \bm I_k)$ independently.
    Then,
    \begin{equation}
        \EE\left[\prod_{(v, w) \in E} \langle \bg_v, \bg_w \rangle\right] = \sum_{C \text{ cycle cover of } G} k^{|C|} \leq |E|^{|E|} k^{c_{\max}(G)}.
    \end{equation}
\end{proposition}
\begin{proof}
    The first equality is proved in Section 4 of \cite{MR-2011-GraphIntegralCircuit}.
    The inequality follows from the fact that the number of cycle covers of $G$ is at most the number of partitions of the edges of $G$, which is at most $|E|^{|E|}$ by Proposition~\ref{prop:size-Part}.
\end{proof}

\begin{proposition}
    \label{prop:cycle-cover-edges}
    Suppose $G = (V, E)$ is a connected graph with no self-loops, but possibly with parallel edges.
    Then, $|V| + c_{\max}(G) - 1 \leq |E|$, with equality if and only if $G$ is an \emph{inflated tree}---a tree where every edge has been replaced with a cycle.
\end{proposition}
\begin{proof}
    Let $C$ be a maximum cycle cover of $G$.
    Let $G^{\prime}$ be the graph formed by removing an arbitrary edge from every cycle in $C$.
    Then, $|E(G^{\prime})| = |E| - c_{\max}(G)$ and $|V(G^{\prime})| = |V|$.
    Moreover, $G^{\prime}$ is connected, since there is a path in $G^{\prime}$ between the endpoints of each edge that was removed (along the remaining edges of the corresponding cycle).
    Thus, $|E(G^{\prime})| \geq |V(G^{\prime})| - 1$, and substituting gives $|E| - c_{\max}(G) \geq |V| - 1$.

    Equality holds if and only if $G^{\prime}$ is a tree.
    If $G$ is an inflated tree, this will clearly be the case.
    Suppose now that $G^{\prime}$ is a tree; we want to show that $G$ is an inflated tree.
    If two edge-disjoint cycles intersect in more than one vertex, then after one edge is removed from each cycle, there still exists a cycle among their edges.
    Therefore, if $G^{\prime}$ is a tree, then any two cycles of $C$ can intersect in at most one vertex.
    Moreover, again because $G^{\prime}$ is a tree, there can exist no further cycle of $G$ including edges from more than one of the cycles of $C$.
    Therefore, the graph formed by collapsing each cycle of $C$ to an edge must be a tree, whereby $G$ is an inflated tree.
\end{proof}

\begin{proof}[Proof of Theorem~\ref{thm:appl-high-rank}]
    Let $\bg_1, \dots, \bg_n \sim \sN(\bm 0, \bm I_N)$ be a collection of independent gaussian vectors coupled to $V^{\perp}$ such that $V^{\perp} = \mathsf{span}(\bg_1, \dots, \bg_n)$.
    Let us define
    \begin{equation}
        \alpha \colonequals (\log\log N)^{-32},
    \end{equation}
    which will play a similar role here to that of $\alpha$ in the proof of Theorem~\ref{thm:appl-laurent-approx}.
    Define $\bM^{(0)} \colonequals (1 - \alpha / 2)\frac{1}{N}\sum_{i = 1}^n \bg_i\bg_i^{\top}$, let $\bD$ be the diagonal matrix with $\diag(\bD) = \diag(\bM^{(0)})$, and define $\bM \colonequals \bm I + \bD - \bM^{(0)}$.
    We will then take $\tEE = \tEE_{\bM}$ for this choice of $\bM$ (which we note satisfies $\diag(\bM) = \one$ by construction).

    We first establish a preliminary asymptotic on the eigenvalues of $\bM^{(0)}$.
    Let $\lambda_1(\bM^{(0)}) \geq \lambda_2(\bM^{(0)}) \geq \cdots \geq \lambda_N(\bM^{(0)}) \geq 0$ be the ordered eigenvalues of $\bM^{(0)}$.
    Then, $\lambda_{n + 1}(\bM^{(0)}) = \cdots = \lambda_N(\bM^{(0)}) = 0$ almost surely.
    We note that $\sqrt{n / N} \ll 1 / \sqrt{\log N} \ll \alpha$, whereby the concentration inequality of Proposition~\ref{prop:rectangular-gaussian-conc} implies that, with high probability as $N \to \infty$,
    \begin{equation}
        1 - \alpha \leq \lambda_n(\bM^{(0)}) \leq \cdots \leq \lambda_1(\bM^{(0)}) \leq 1 - \frac{1}{3}\alpha.
    \end{equation}

    We next control the entries of $\bD$.
    These are
    \begin{equation}
        D_{ii} = \left(1 - \frac{\alpha}{2}\right) \frac{1}{N}\sum_{j = 1}^n (\bg_j)_i^2,
    \end{equation}
    where the law of the inner sum is $\chi^2(n)$.
    By the concentration inequality of Proposition~\ref{prop:chi-squared-conc}, we have that $\PP[D_{ii} \geq 4\frac{n}{N}] \leq \exp(-n)$, and since $n \gg \log N$ by assumption, upon taking a union bound we have that, with high probability, $\bm 0 \preceq \bD \preceq 4\frac{n}{N}\bm I_N$.
    
    We next establish the projection-like behavior of $\bM$.
    Suppose first that $\bv \in V$.
    Since the row space of $\bM^{(0)}$ is $V^{\perp}$, we have, on the event that the bound for $\bD$ above holds,
    \begin{equation}
      \|\bv\|^2 \leq \bv^{\top}\bM\bv = \|\bv\|^2 + \bv^{\top}\bD \bv \leq \left(1 + 4\frac{n}{N}\right)\|\bv\|^2
    \end{equation}
    Now, suppose $\bv \in V^{\perp}$.
    Then, on the event that the bound for $\bM^{(0)}$ above holds, we have that since $\bv$ is in the subspace spanned by the top $n$ eigenvectors of $\bM^{(0)}$,
    \begin{equation}
        \bv^{\top}\bM\bv = \|\bv\|^2 + \bv^{\top}\bD \bv - \bv^{\top}\bM^{(0)}\bv^{\top} \leq \|\bv\|^2 + 4\frac{n}{N}\|\bv\|^2 - (1 - \alpha)\|\bv\|^2 = \left(4\frac{n}{N} + \alpha\right)\|\bv\|^2.
    \end{equation}

    We now take up the main task of showing that $\tEE_{\bM}$ is a pseudoexpectation of the required degree.
    Note that the above results imply that $\lambda_{\min}(\bM) \geq 1 - \lambda_{\max}(\bM^{(0)}) \geq \alpha / 3$, giving the necessary control of the smallest eigenvalue.
    It remains to control the incoherence quantities.
    
    Writing $\bh_1, \dots, \bh_N \in \RR^n$ for the vectors $\bh_i = ((\bg_j)_i)_{j = 1}^N$, we note that, for $i \neq j$, we have $M_{ij} = -(1 - \alpha)\frac{1}{N}\langle \bh_i, \bh_j \rangle$, and the $\bh_i$ are independent and identically distributed with law $\sN(\bm 0, \bm I_n)$.
    Since for any fixed $i \neq j$ the law of $\langle \bh_i, \bh_j \rangle$ is the same as that of $\|\bh_i\|_2(\bh_j)_1$ (by orthogonal invariance of gaussian vectors), using Proposition~\ref{prop:chi-squared-conc} again we may bound
    \begin{align}
      \PP\left[\frac{1}{N}|\langle \bh_i, \bh_j \rangle| \geq t\right]
      &\leq \PP[\|\bh_i\|_2 \geq 2\sqrt{n}] + \PP\left[|(\bh_j)_1| \geq \frac{tN}{2\sqrt{n}}\right] \nonumber \\
      &\leq \exp\left(-\frac{n}{2}\right) + \exp\left(-\frac{t^2N^2}{8n}\right)
    \end{align}
    Recall that we have assumed $n \gg \log N$.
    Therefore, taking a union bound over these events for $\{i, j\} \in \binom{[N]}{2}$ we find that, with high probability as $N \to \infty$, the simpler incoherence quantities will satisfy
    \begin{align}
      \epsilon_{\offdiag}(\bM)
      &\leq 5\sqrt{\frac{n \log N}{N^2}}, \\
      \epsilon_{\corr}(\bM)
      &\leq 25\left(2\frac{n \log N}{N^2} + (N - 2) \frac{n^2\log^2 N}{N^4}\right)^{1/2} \leq 50 \sqrt{\frac{n \log N}{N^2}}.
    \end{align}
    For $\epsilon_{\pow}$, we observe that by the above reasoning with high probability $\bM \succeq \bm 0$, and thus $|M_{ij}| \leq 1$ for all $i \neq j$.
    On this event, we have
    \begin{equation}
        \epsilon_{\pow}(\bM) \leq \max_{i \in [N]} \sum_{j \neq i} M_{ij}^2 \leq \frac{1}{N^2} \max_{i \in [N]} \bh_i^{\top}\left(\sum_{j \neq i} \bh_j\bh_j^{\top}\right)\bh_i \leq \frac{1}{N^2} \left\|\sum_{i = 1}^N \bh_i\bh_i^{\top}\right\|\max_{i \in [N]} \|\bh_i\|^2_2.
    \end{equation}
    By the calculations above, with high probability we have both $\|\bh_i\|_2^2 \leq 4n$ for all $i \in [N]$ and $\|\sum_{i = 1}^N\bh_i\bh_i^{\top}\| = \|\sum_{i = 1}^n \bg_i\bg_i^{\top}\| = \frac{N}{1 - \alpha}\|\bM^{(0)}\| \leq 2N$.
    Thus we find that, with high probability,
    \begin{equation}
        \epsilon_{\pow}(\bM) \leq 8\frac{n}{N}.
    \end{equation}

    Finally, for $\epsilon_{\tree}$ and $\epsilon_{\err}$ we will use pattern diagrams together with hypercontractivity.
    We begin with $\epsilon_{\tree}$.
    Examining one term in the maximization, for a given $T \in \sT(2d^{\prime})$ and $\bs \in [N]^{2d^{\prime}}$, we have
    \begin{align}
      Z^T(\bM; \bs) &- \One\{s_1 = \cdots = s_N\} \nonumber \\
      &= \sum_{\substack{\ba \in [N]^{V^{\square}} \text{with } \bs, \ba \text{ not all equal}}} \prod_{(v, w) \in E(T)} M_{f_{\bs, \ba}(v)f_{\bs, \ba}(w)} \nonumber \\
      &= \sum_{\substack{\ba \in [N]^{V^{\square}} \text{with } \bs, \ba \text{ not all equal} \\ (G, f) = \pat(T, \bs, \ba)}} \prod_{(v, w) \in E(G)} M_{f(v),f(w)} \nonumber 
      \intertext{and since the pattern diagram is constructed to have all edges between vertices with different indices, we may expand this in terms of the $\bh_i$,}
      &= \sum_{\substack{\ba \in [N]^{V^{\square}} \text{with } \bs, \ba \text{ not all equal} \\ (G, f) = \pat(T, \bs, \ba)}} \left(-\frac{1 - \alpha}{N}\right)^{|E(G)|}\prod_{(v, w) \in E(G)} \langle \bh_{f(v)}, \bh_{f(w)} \rangle.
    \end{align}
    Towards applying the hypercontractive inequality, we compute the second moment:
    \begin{align}
      \EE[&(Z^T(\bM; \bs) - \One\{s_1 = \cdots = s_N\})^2] \nonumber  \\
      &= \sum_{\substack{\ba_1, \ba_2 \in [N]^{V^{\square}} \text{with } \bs, \ba_1 \text{ not all equal} \\ \text{ and } \bs, \ba_2 \text{ not all equal} \\ (G, f) = \pat((T, \bs, \ba_1), (T, \bs, \ba_2))}} \left(-\frac{1 - \alpha}{N}\right)^{|E(G)|}\EE\left[\prod_{(v, w) \in E(G)} \langle \bh_{f(v)}, \bh_{f(w)} \rangle\right] \nonumber
      \intertext{and simplifying the remaining expectation using Proposition~\ref{prop:cycle-cover} and bounding the first term,}
                        &\leq (6d)^{6d}\sum_{\substack{\ba_1, \ba_2 \in [N]^{V^{\square}} \text{with } \bs, \ba_1 \text{ not all equal} \\ \text{ and } \bs, \ba_2 \text{ not all equal} \\ G = \pat((T, \bs, \ba_1), (T, \bs, \ba_2))}} \One\{G \text{ has a cycle cover}\}N^{-|E(G)|} n^{c_{\max}(G)}, \nonumber
      \intertext{where we note that the expression does not depend on the labelling $f$ of the vertices of $G$ anymore.
      Now, as before, we group terms according to the graph $G$, using that each occurs at most $N^{|V^{\square}(G)|}$ times in the sum, and that each $G$ arising is connected, contains at most $6d$ vertices and $6d$ edges, and at least one $\bullet$ vertex:}
                        &\leq (6d)^{6d}\sum_{G = (V^{\bullet} \sqcup V^{\square}, E) \in \Pat^{\leq 2d}_{\tree}(2)} \One\{G \text{ has a cycle cover}\} N^{|V^{\square}| - |E|} n^{c_{\max}(G)} \nonumber \\
                        &\leq (6d)^{6d} \frac{n}{N} \sum_{G = (V^{\bullet} \sqcup V^{\square}, E) \in \Pat^{\leq 2d}_{\tree}(2)} N^{|V^{\square}| + c_{\max}(G) - |E|}. \nonumber
                          \intertext{Now, by Proposition~\ref{prop:cycle-cover-edges} we have $|V^{\bullet}| + |V^{\square}| + c_{\max}(G) - |E| - 1 \leq 0$. If $|V^{\bullet}| \geq 2$, then this yields $|V^{\square} + c_{\max}(G) - |E| \leq -1$.
                          If $|V^{\bullet}| = 1$, we argue slightly more carefully and note that in this case, since all $\bullet$ vertices in both underlying trees collapsed to a single vertex, in fact all vertices in $G$ have degree at least 4, so $G$ cannot be an inflated tree as in the only case of equality for Proposition~\ref{prop:cycle-cover-edges}.
                          Therefore, in this case we have $|V^{\square}| + c_{\max}(G) - |E| \leq -1$ again, whereby}
      &\leq (6d)^{6d} \frac{n}{N^2} \sum_{G = (V^{\bullet} \sqcup V^{\square}, E) \in \Pat^{\leq 2d}_{\tree}(2)} 1 \nonumber \\
      \intertext{and concluding with Proposition~\ref{prop:size-Pat}, we find}
          &\leq (6d)^{24d} \frac{n}{N^2}. \nonumber \\
      &\leq \frac{(6d)^{24d}}{N}.
    \end{align}
    (Here we have been slightly more precise than strictly necessary, in anticipation of referring to our results when discussing the SK Hamiltonian below.)
    
    Now, we observe that $Z^T(\bM; \bs) - \One\{s_1 = \cdots = s_{2d^{\prime}}\}$ is a polynomial of degree at most $2|E(T)| \leq 6d$ (by Corollary~\ref{cor:sF-vertex-edge-bound}) in the entries of the $\bh_i$, which are i.i.d.\ standard gaussians.
    Thus we can apply the hypercontractive tail bound of Corollary~\ref{cor:gauss-hc-tail} to find, taking $t = N^{1/4} \geq (2e^2)^{3d}$ for $N$ sufficiently large,
    \begin{align}
      \PP\left[|Z^T(\bM; \bs) - \One\{s_1 = \cdots = s_{2d^{\prime}}\}| \geq (6d)^{12d} N^{-1/4}\right] \leq \exp(-6dN^{1 / 12d}).
    \end{align}
    Taking a union bound, since the number of choices of $d^{\prime}$, $T$, and $\bs$ is at most $d \cdot 2(3d)^{3d} \cdot N^{2d} \leq N^{6d}$ for $N$ sufficiently large, we have
    \begin{align}
      \PP\left[\epsilon_{\tree}(\bM; 2d) \geq (6d)^{12d} N^{-1/4}\right]
      &\leq N^{6d} \exp(-6dN^{1 / 12d}) \nonumber \\
      &\leq \exp\left(6d(\log N - N^{1 / 12d})\right),
    \end{align}
    and recalling that $d \leq \frac{1}{300}\log (N / n) / \log\log N$ from our assumption we find that the event above holds with high probability.
    Also, from this same assumption we find $(6d)^{12d} \leq N^{1/8}$ for $N$ sufficiently large, whereby with high probability
    \begin{equation}
        \epsilon_{\tree}(\bM; 2d) \leq N^{-1/8}.
    \end{equation}

    We now perform the same analysis for $\epsilon_{\err}(\bM; 2d)$.
    Again examining one term with a given $T \in \sT(2d^{\prime})$ and $\bs \in [N]^{2d^{\prime}}$, manipulating as before, and computing the second moment, we find
    \begin{align}
      \EE\bigg[\Bigg(\sum_{\substack{\ba \in [N]^{V^{\square}} \\ \ba \text{ } (T, \bs)\text{-loose}}} &\prod_{(v, w) \in E(T)} M_{f_{\bs, \ba}(v)f_{\bs, \ba}(w)}\Bigg)^2\Bigg] \nonumber \\
      &= \sum_{\substack{\ba_1, \ba_2 \in [N]^{V^{\square}} \\ \ba_1, \ba_2 \text{ } (T, \bs)\text{-loose} \\ (G, f) = \pat((T, \bs, \ba_1), (T, \bs, \ba_2))}} \left(-\frac{1 - \alpha}{N}\right)^{|E(G)|}\tEE\left[\prod_{(v, w) \in E(G)} \langle \bh_{f(v)}, \bh_{f(w)} \rangle\right] \nonumber \\
      &\leq (6d)^{6d} \sum_{\substack{\ba_1, \ba_2 \in [N]^{V^{\square}} \\ \ba_1, \ba_2 \text{ } (T, \bs)\text{-loose} \\ G = \pat((T, \bs, \ba_1), (T, \bs, \ba_2))}} \One\{G \text{ has a cycle cover}\}N^{-|E(G)|} n^{c_{\max}(G)} \nonumber \\
      &\leq (6d)^{6d} \sum_{G = (V^{\bullet} \sqcup V^{\square}, E) \in \Pat^{\leq 2d}_{\err}(2)} \One\{G \text{ has a cycle cover}\}N^{|V^{\square}|-|E|} n^{c_{\max}(G)} \nonumber
      \intertext{We recall that $|V^{\bullet}| = |\set(\bs)|$ and $|V^{\bullet}| + |V^{\square}| = |V|$, so we may rewrite this by ``forgetting'' the vertex types as}
      &= N^{-|\set(\bs)|} (6d)^{6d} \sum_{G = (V, E) \in \Pat^{\leq 2d}_{\err}(2)} \One\{G \text{ has a cycle cover}\}N^{|V|-|E|} n^{c_{\max}(G)} \nonumber \\
      &\leq N^{-|\set(\bs)|} (6d)^{6d} \frac{n}{N}\sum_{G = (V, E) \in \Pat^{\leq 2d}_{\err}(2)} N^{|V| + c_{\max}(G) - |E|}.
    \end{align}
    Now, by Proposition~\ref{prop:cycle-cover-edges}, the inner term is at most 1 unless $G$ is an inflated tree.
    We claim that, when $G = \pat((T, \bs, \ba_1), (T, \bs, \ba_2))$ where $\ba_1$ and $\ba_2$ are both $(T, \bs)$-loose, then $G$ cannot be an inflated tree.
    To prove this, we consider two cases.

    \emph{Case 1: $|\set(\bs)| = 1$.} In this case, as we have argued above, since the total number of $\bullet$ vertices in the two initial trees taken together is at least 4 and neither of these trees is a pair, every vertex in $G$ will have degree at least 4, whereby $G$ cannot be an inflated tree.

    \emph{Case 2: $|\set(\bs)| > 1$.} Suppose, more specifically, that $G = \pat((T_1, \bs, \ba_1), (T_2, \bs, \ba_2))$.
    Since $\ba_1$ is $(T_1, \bs)$-loose, there exists some $i \in [N]$ and some $v \in V^{\square}(T_1)$ such that $v$ belongs to the minimal spanning tree of leaves $\ell$ with $s_{\kappa_{T_1}(\ell)} = i$, but $(a_1)_v \neq i$.
    In particular, there must exist two such leaves $\ell_1, \ell_2$ such that $v$ is along the path from $\ell_1$ to $\ell_2$.
    In $G$, the vertex that $v$ is identified to---call it $x$---is different from the vertex that $\ell_1$ and $\ell_2$ are identified to---call it $y$.
    Since $|\set(\bs)| > 1$, there is some $j \neq i$ and a leaf $\ell^{\prime}$ with $s_{\kappa_{T_1}(\ell^{\prime})} = j$.
    Suppose $\ell^{\prime}$ is identified to a vertex $z$ in $G$.
    Then, there is a path from $x$ to $z$ in $G$, so there are two different paths from $y$ to $z$ consisting only of edges coming from $T_1$.

    On the other hand, since $T_2$ is a tree and $\ba_2$ is $(T_2, \bs)$-loose, there is another path in $G$ from $y$ to $z$ and consisting of edges different from the first two paths, coming from $T_2$.
    Therefore, in $G$ there exist three different paths between $y$ and $z$; put differently, a triple edge can be obtained as a minor of $G$ (after discarding self-loops).
    On the other hand, any minor of an inflated tree is still an inflated tree (again after discarding self-loops), and a triple edge is not an inflated tree.
    Thus, $G$ cannot be an inflated tree.
    This concludes the proof of our intermediate claim.
    
    We then conclude the main argument using Proposition~\ref{prop:size-Pat}:
    \begin{align}
      \EE\bigg[\Bigg(\sum_{\substack{\ba \in [N]^{V^{\square}} \\ \ba \text{ } (T, \bs)\text{-loose}}} \prod_{(v, w) \in E(T)} M_{f_{\bs, \ba}(v)f_{\bs, \ba}(w)}\Bigg)^2\Bigg] &\leq N^{-|\set(\bs)|} (6d)^{6d} \frac{n}{N}\sum_{G = (V, E) \in \Pat^{\leq 2d}_{\err}(2)} 1 \nonumber \\
      &\leq N^{-|\set(\bs)|} (6d)^{24d} \frac{n}{N}.
    \end{align}
    Similarly to before, we apply Corollary~\ref{cor:gauss-hc-tail} with $t = (N / n)^{1/4}$, finding
    \begin{align}
      &\PP\Bigg[N^{|\set(\bs)| / 2}\Bigg|\sum_{\substack{\ba \in [N]^{V^{\square}} \\ \ba \text{ } (T, \bs)\text{-loose}}} \prod_{(v, w) \in E(T)} M_{f_{\bs, \ba}(v)f_{\bs, \ba}(w)}\Bigg| \geq (6d)^{12d}(n / N)^{1/4}\Bigg] \nonumber \\
      &\hspace{2cm}\leq \exp(-6d(N / n)^{1/12d}),
    \end{align}
    and performing the same union bound calculation over all choices of $d^{\prime}$, $T$, and $\bs$ shows that, with high probability,
    \begin{equation}
        \epsilon_{\err}(\bM; 2d) \leq (n / N)^{1/8}.
    \end{equation}
    
    Thus, combining the results on the incoherence quantities, with high probability we have
    \begin{equation}
        \epsilon(\bM; 2d) \leq 55 \sqrt{\frac{n \log N}{N^2}} + 8\frac{n}{N} + \left(\frac{1}{N}\right)^{1/8} + \left(\frac{n}{N}\right)^{1/8} \leq 65 \left(\frac{n}{N}\right)^{1/8}.
    \end{equation}
    On this event, we work with the condition of Theorem~\ref{thm:lifting}, for $N$ sufficiently large:
    \begin{equation}
        (12d)^{32}\|\bM\|^5\epsilon(\bM; 2d)^{1/d} \leq 64\left(\frac{12d}{\log N}\right)^{32} \leq \frac{1}{3}(\log\log N)^{-32} = \frac{1}{3}\alpha \leq \lambda_{\min}(\bM),
    \end{equation}
    concluding the proof.
\end{proof}

\section{Extension to Rank $N - \Theta(N)$: Obstacles and Strategies}
\label{sec:future}

A careful reading of the proof of Theorem~\ref{thm:appl-high-rank} reveals the two major obstacles at hand in attempting to execute our strategy on low-rank projection matrices.
These correspond to the two incoherence quantities that are no longer $o(1)$ as $N \to \infty$ once $n = \Omega(N)$: $\epsilon_{\pow}$ and $\epsilon_{\err}$ (we note that a more naive proof technique that would still work in the setting of Theorem~\ref{thm:appl-high-rank} would have also led us to include $\epsilon_{\tree}$ on this list, but that is simple to circumvent with a small adjustment as we have made in that argument).
As we describe below, the former seems to represent a fundamental barrier requiring us to reconsider our derivation of the pseudomoments, while the latter can probably be handled without much difficulty.

\subsection{Failure of Order-2 Tensor Orthonormality: The $\epsilon_{\pow}$ Obstacle}
\label{sec:eps-pow-obstacle}

The more substantial issue arises, in the course of our proof of Theorem~\ref{thm:lifting}, in the collapse of partition transport diagrams.
As illustrated in Figure~\ref{fig:2-2-issue}, when $n = \Omega(N)$ then we no longer have merely $\bM^{\circ 2} \approx \bm I_N$, but rather $\bM^{\circ 2} \approx \bm I_N + t\one_N\one_N^{\top}$ for $t = \Theta(N^{-1})$, whereby one particular diagram contributes a non-negligible new term.
(To see explicitly that this occurs, we may compute $N^{-1}\one_N^{\top}\bM^{\circ 2}\one_N = \|\bM\|_F^2 / N = \Theta(1)$ in this scaling.)
This happens already at degree $2d = 4$, where this diagram has coefficient $+2$, and therefore our approximate Gramian factorization $\bY$ of $\bZ^{\main}$ has an extra positive term of the form of the right-most diagram in the Figure, two sided pairs with edges labelled by $\bM^2$ instead of $\bM$.
The associated CGM, after multiplying by $t$, has spectral norm $O(1)$ (see the scaling discussed in Remark~\ref{rem:multiscale-spectrum}), so we have ``$\bY \gg \bZ^{\main}$'' in psd ordering---in this case, our approximate Gramian factorization is simply false.
Moreover, when $\bM^2 \approx \lambda \bM$ (as for $\bM$ a rescaled projector), this additional diagram is the same as the diagram that is ``orthogonalized away'' by writing the pseudomoments in the multiharmonic basis.
Therefore, the other diagrams in $\bZ^{\main}$ cannot compensate for this negative term, whereby $\bZ^{\main} \not \succeq \bm 0$ and it is not merely the Gramian approximation that fails but rather the pseudomoment construction itself.

Some technical tricks can work around this issue at low degrees: in \cite{KB-2019-Degree4SK-Arxiv}, we made an adjustment to the pseudomoments before passing to the multiharmonic basis, and also adjusted the basis so that $\bZ^{\main}$ written in this basis has a similar extra term, whereby we can restore the equality $\bZ^{\main} \approx \bY$.\footnote{The change of basis is expressed there as a Schur complement, which turns out to be equivalent.}
In Section~\ref{sec:pf:low-rank-lifting} we will take a different approach, simply adding extra terms to $\tEE_{\bM}$ itself that achieve the same thing.
This is slightly more flexible, working up to $2d = 6$.

It seems, however, that to resolve this issue for arbitrarily high degrees would require rethinking much of our derivation to include a second-order correction.
Namely, our very initial construction of the approximate Green's function to the system of PDEs $\langle \bv_i, \bm\partial \rangle^2 p = 0$ in Section~\ref{sec:informal-deriv} already has built in the premise that the matrices $\bv_i\bv_i^{\top}$ are nearly orthonormal, when compared to the actual Green's function derivations in similar situations by \cite{Clerc-2000-KelvinTransform}.
Since this same assumption eventually leads our construction astray, it seems plausible that the correct resolution will come from building a more nuanced approximate Green's function and deriving the prediction anew.

\begin{figure}
    \begin{center}
        \includegraphics[scale=0.7]{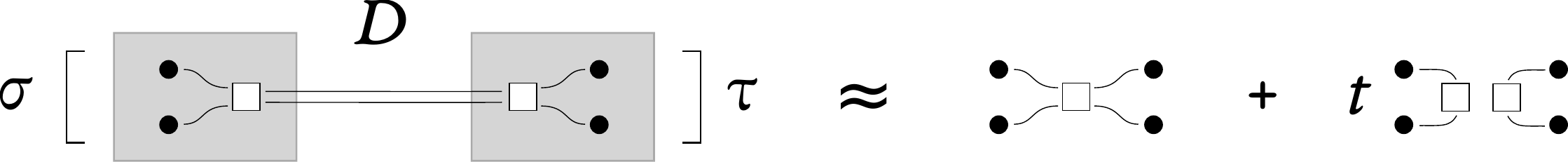}
    \end{center}
    \caption{\textbf{Extra term from $\bD = [\, 2 \, ]$ partition transport ribbon diagram.} We illustrate the issue discussed in Section~\ref{sec:eps-pow-obstacle}, that a partition transport ribbon diagram with associated plan $\bD = [\, 2 \, ]$ produces an extra term consisting of two sided pairs when collapsed. The two parallel edges in the diagram on the left represent the matrix $\bM^{\circ 2}$, and this expansion corresponds to an approximation $\bM^{\circ 2} \approx \bm I_N + t \one_N\one_N^{\top}$.}
    \label{fig:2-2-issue}
\end{figure}

\subsection{Sided Error Terms: The $\epsilon_{\err}$ Obstacle}

Another, milder difficulty that arises when $n = \Omega(N)$ is that we have $\epsilon_{\err}(\bM) = \Theta(1)$.
This is less of a problem because we have used a very coarse approach to bounding the error contribution $\bZ^{\err}$, bounding related matrix norms by Frobenius norms in the proof of Lemma~\ref{lem:Z-err-bound}.
Since, assuming randomly behaving entries, we only expect such a bound to be tight when a matrix is actually a vector, this should not be a problem except in situations where the diagram-like objects of error terms appearing in the proof of Lemma~\ref{lem:Z-err-bound} are \emph{sided}, applying only to left or right indices in a diagram.
To be more specific, we recall that, per Proposition~\ref{prop:E-err-factorization}, the full pseudoexpectation may be written as a sum over diagrams where ``main term trees'' are applied to some subsets of indices and ``error term trees'' are applied to others.
The multiharmonic basis eliminates sided main term trees, but does not have any useful effect on sided error term trees.

Below, we give a simple result showing that, in quite general pseudoexpectations like this, one may build a basis that eliminates sided terms of any kind.
It is possible to combine this construction with the multiharmonic basis to eliminate sided error terms, though this complicates other parts of the analysis and has no direct benefit in our applications, so we do not pursue it here.

\begin{proposition}
    Let $\tEE: \RR[x_1, \dots, x_N] \to \RR$ be linear and given by
    \begin{equation}
        \tEE[\bx^S] = \sum_{\pi \in \mathsf{Part}(S)} \prod_{A \in \pi} E_A
    \end{equation}
    for all $S \in \sM([N])$, where $E_A$ are arbitrary multiset-indexed quantities.
    Define the polynomials
    \begin{equation}
        p_S(\bx) \colonequals \sum_{T \subseteq S} \left(\sum_{\pi \in \Part(S - T)} (-1)^{|\pi|}\prod_{A \in \pi} E_A\right)\bx^T.
    \end{equation}
    for all $S \in \sM([N])$.
    The pseudomoments written in this basis are then
    \begin{equation}
        \tEE[p_S(\bx)p_T(\bx)] \sum_{\substack{\pi \in \mathsf{Part}(S + T) \\ A \cap S \neq \emptyset, A \cap T \neq \emptyset \\ \text{for all } A \in \pi}}\prod_{A \in \pi} E_A.
    \end{equation}
\end{proposition}    
\begin{proof}
    We calculate directly:
    \begin{align}
\tEE[p_S(\bx)p_T(\bx)]
&= \sum_{\substack{S^{\prime} \subseteq S \\ T^{\prime} \subseteq T}}\sum_{\substack{\sigma \in \mathsf{Part}(S^{\prime}) \\ \tau \in \mathsf{Part}(T^{\prime}) \\ \pi \in \mathsf{Part}((S - S^{\prime}) + (T - T^{\prime})) \\ A \cap S \neq \emptyset, A \cap T \neq \emptyset \\ \text{for all } A \in \pi}}\Bigg(\sum_{\substack{\sigma^{\prime} \subseteq \sigma \\ \tau^{\prime} \subseteq \tau}}(-1)^{|\sigma^{\prime}| + |\tau^{\prime}|}\Bigg) \prod_{A \in \sigma + \tau + \pi} E_A \nonumber \\
&= \sum_{\substack{S^{\prime} \subseteq S \\ T^{\prime} \subseteq T}}\sum_{\substack{\sigma \in \mathsf{Part}(S^{\prime}) \\ \tau \in \mathsf{Part}(T^{\prime}) \\ \pi \in \mathsf{Part}((S - S^{\prime}) + (T - T^{\prime})) \\ A \cap S \neq \emptyset, A \cap T \neq \emptyset \\ \text{for all } A \in \pi}}\Bigg(\sum_{\sigma^{\prime} \subseteq \sigma} (-1)^{|\sigma^{\prime}|}\Bigg)\Bigg(\sum_{\tau^{\prime} \subseteq \tau}(-1)^{|\tau^{\prime}|}\Bigg) \prod_{A \in \sigma + \tau + \pi} E_A,
    \end{align}
    and the remaining coefficients are 1 if $|S^{\prime}| = |T^{\prime}| = 0$ and 0 otherwise, completing the proof.
\end{proof}
Here we have used implicitly the simple \Mobius\ function of the subset poset from Example~\ref{ex:mobius-subset} in building our basis.
There would appear to be an analogy between this feature and the appearance of the \Mobius\ function of partitions from Example~\ref{ex:mobius-part} in the multiharmonic basis.
It would be interesting to develop more general techniques for ``orthogonalizing away'' the terms of pseudoexpectations that contribute to a multiscale spectrum in the pseudomoment matrix using bases that incorporate poset combinatorics.

\section{Proofs of Low-Rank Applications}
\label{sec:pf:applications-sk}

Using the ideas above, especially from Section~\ref{sec:eps-pow-obstacle}, we now prove our applications to extending low-rank matrices, which are restricted to degree~6.

\subsection{Degree 6 Low-Rank Extension: Proof of Theorem~\ref{thm:low-rank-lifting}}
\label{sec:pf:low-rank-lifting}

\begin{proof}[Proof of Theorem~\ref{thm:low-rank-lifting}]
    The construction of $\tEE$ is a minor variant of $\tEE_{\bM}$ from Theorem~\ref{thm:lifting}, modified to counteract the negative terms discussed in Section~\ref{sec:eps-pow-obstacle}.
    \begin{align}
      \tEE^{\pairs}_{\bM}\left[\prod_{i \in S} x_i\right] &\colonequals \One\{|S| \in \{4, 6\}\} \sum_{\substack{F \in \sF(|S|) \\ F \text{ all pairs}}} Z^F(\bM^2; S) \text{ for all } S \subseteq [N], \\
      \tEE^{\id}\left[\prod_{i \in S} x_i\right] &\colonequals \One\{S = \emptyset\}, \\
      \tEE &\colonequals (1 - c)\left(\tEE_{\bM} + 2t_{\pow}\tEE^{\pairs}_{\bM}\right) + c\, \tEE^{\id}.
    \end{align}
    We remark that here we use a combination of the two strategies for attenuating the spectrum of the error terms that were discussed in Remark~\ref{rem:nudging}.
    We also emphasize the detail that the matrix used in the CGSs in $\tEE^{\pairs}$ is the \emph{square} of $\bM$.
    (For $\bM$ a rescaled random projection we expect $\bM^2 \approx \lambda \bM$ for some $\lambda$, but we do not require such a relation to hold, nor is this taken into account in the incoherence quantities used in the statement.)

    Let us moreover decompose $\tEE^{\pairs}_{\bM}$ into three terms, as follows.
    Note that the first two are bilinear operators on polynomials of degree at most $d$, in the sense of Definition~\ref{def:bilinear-pe}, while the last has the additional symmetry making it a linear operator on polynomials of degree at most $2d$.
    \begin{align}
      \tEE^{\pairs:\main:1}_{\bM}(\bx^S, \bx^T) &\colonequals \One\{|S| + |T| \in \{4, 6\}\} \sum_{\substack{F \in \sF(|S|, |T|) \\ F \text{ all pairs} \\ \mathclap{F \text{ has 2 sided pairs}}}} Z^F_{S,T}(\bM^2) \text{ for all } S, T \in \sM([N]), \\
      \tEE^{\pairs:\main:2}_{\bM}(\bx^S, \bx^T) &\colonequals \One\{|S| + |T| \in \{4, 6\}\} \sum_{\substack{F \in \sF(|S|, |T|) \\ F \text{ all pairs} \\ \mathclap{F \text{ has } \leq 1 \text{ sided pair}}}} Z^F_{S,T}(\bM^2) \text{ for all } S \in \sM([N]), \\
      \tEE^{\pairs:\err}_{\bM}(\bx^S, \bx^T) &\colonequals \tEE^{\pairs}_{\bM}[\bx^{S + T}] - \tEE^{\pairs:\main:1}_{\bM}(\bx^S, \bx^T) - \tEE^{\pairs:\main:2}_{\bM}(\bx^S, \bx^T) \\
      &= \tEE^{\pairs:\err}[\bx^{S + T}]. \nonumber
    \end{align}
    The point here is that, since we expect $t_{\pow} \|\bM^2\|_F^2 = O(1)$, only the ribbon diagrams with two sided pairs, those in $\tEE^{\pairs:\main:1}$, will contribute significantly.
    The further decomposition between $\tEE^{\pairs:\main:1} + \tEE^{\pairs:\main:2}$ and $\tEE^{\pairs:\err}$ is precisely the same as that between $\tEE^{\main}$ and $\tEE^{\err}$ in Theorem~\ref{thm:lifting}, the former being simpler to work with in terms of CGMs and the latter being a small correction.

    Our result will then follow from the following three claims:
    \begin{align}
      \tEE_{\bM} + 2t_{\pow}\tEE^{\pairs:\main:1}_{\bM} &\succeq \bm 0, \label{eq:low-rank-lifting-main} \\
      2(1 - c)t_{\pow}\tEE^{\pairs:\main:2}_{\bM} + \frac{c}{2} \tEE^{\id} &\succeq \bm 0, \label{eq:low-rank-lifting-err-1} \\
      2(1 - c)t_{\pow}\tEE^{\pairs:\err}_{\bM} + \frac{c}{2} \tEE^{\id} &\succeq \bm 0. \label{eq:low-rank-lifting-err-2}
    \end{align}
    For \eqref{eq:low-rank-lifting-main} we will argue by adjusting the proof of Theorem~\ref{thm:lifting}, arguing for positivity in the harmonic basis, and using that the additional term counteracts the negative terms discussed in Section~\ref{sec:eps-pow-obstacle}.
    For \eqref{eq:low-rank-lifting-err-1} and \eqref{eq:low-rank-lifting-err-2}, we will make simpler arguments in the standard monomial basis.

    \vspace{1em}
    
    \emph{Proof of \eqref{eq:low-rank-lifting-main}:}
    We will be quite explicit about the calculations in this section, essentially recapitulating this special case of Theorem~\ref{thm:lifting} with adjustments as needed.
    We notice first that the only cases where $\tEE^{\pairs:\main:1}_{\bM}(\bx^S, \bx^T) \neq 0$ are where either $|S| = |T| = 2$ or $|S| = |T| = 3$.
    In the former case there is only a single diagram, with one sided pair in each $\sL$ and $\sR$, while in the latter case there are 9 such diagrams, with one additional non-sided pair (there are $3 \cdot 3 = 9$ ways to choose the leaves belonging to this pair).

    Let us enumerate explicitly the multiharmonic basis polynomials $h_S^{\downarrow}(\bx)$ for $|S| \leq 3$:
    \begin{align}
      h_{\emptyset}^{\downarrow}(\bx) &= 1, \\
      h_{\{i\}}^{\downarrow}(\bx) &= x_i, \\
      h_{\{i, j\}}^{\downarrow}(\bx) &= x_ix_j - M_{ij}, \\
      h_{\{i, j, k\}}^{\downarrow}(\bx) &= x_ix_jx_k - M_{ij}x_k - M_{ik}x_j - M_{jk}x_i + 2\sum_{a = 1}^N M_{ai}M_{aj}M_{ak} x_a.
    \end{align}
    We see therefore that the only cases with $\tEE^{\pairs:\main:1}(h_{S}^{\downarrow}(\bx), h_T^{\downarrow}(\bx) \neq 0$ will again be those with either $|S| = |T| = 2$ or $|S| = |T| = 3$, and in these two cases we just have $\tEE^{\pairs:\main:1}(h_{S}^{\downarrow}(\bx), h_T^{\downarrow}(\bx)) = \tEE^{\pairs:\main:1}(\bx^S, \bx^T)$; thus, the more complicated terms in the harmonic basis polynomials are in fact entirely ``invisible'' to the corrective term $\tEE^{\pairs:\main:1}$.
    That this does not happen anymore once $2d \geq 8$ seems to be one of the main obstructions to applying a similar adjustment technique there.

    Following the proof of Theorem~\ref{thm:lifting} but adding an extra detail, we define the pseudomoment matrices $\bZ^{\main:1}, \bZ^{\main:2}, \bZ^{\err}, \bZ \in \RR^{\binom{[N]}{\leq 3} \times \binom{[N]}{\leq 3}}$ to have entries
    \begin{align}
      Z^{\main:1}_{S,T} &\colonequals \tEE^{\main}_{\bM}[h_S^{\downarrow}(\bx)h_T^{\downarrow}(\bx)], \\
      Z^{\main:2}_{S,T} &\colonequals 2t_{\pow}\tEE^{\pairs:\main:1}_{\bM}[h_S^{\downarrow}(\bx)h_T^{\downarrow}(\bx)], \\
      Z^{\err}_{S,T} &\colonequals \tEE^{\err}_{\bM}[h_S^{\downarrow}(\bx)h_T^{\downarrow}(\bx)], \\
      \bZ &\colonequals \bZ^{\main:1} + \bZ^{\main:2} + \bZ^{\err}.
    \end{align}
    It then suffices to prove $\bZ \succeq \bm 0$.

    By Proposition~\ref{prop:Z-main-harmonic-basis}, $\bZ^{\main:1}$ is block-diagonal (note that in this small case there are no stretched forest ribbon diagrams with unequal numbers of leaves in $\sL$ and $\sR$, so the block diagonalization is exact).
    Define $\bZ^{\tied}$ as in Corollary~\ref{cor:Z-main-Z-tied}.
    By Corollary~\ref{cor:Z-main-Z-tied}, we have
    \begin{align}
      \|\bZ^{\main:1} - \bZ^{\tied}\| &\leq 10^{38}\|\bM\|^9 \epsilon_{\tree}(\bM; 3), \nonumber 
                                        \intertext{where we can note that $\epsilon_{\tree}(\bM; 3) = \epsilon_{\tree}(\bM; 2) = \epsilon_{\offdiag}(\bM)$, since the only good tree on two leaves is a pair, allowing us to eliminate $\epsilon_{\tree}$,}
                                        &= 10^{38}\|\bM\|^9 \epsilon_{\offdiag}(\bM). 
        \label{eq:low-rank-pf-Z-main-1-Z-tied}
    \end{align}
    (The constants can be improved with more careful analysis and diagram counting specific to $2d = 6$, which we do not pursue here.)

    We claim that the same approximate Gram factorization that held for $\bZ^{\tied}$ in Theorem~\ref{thm:lifting} holds for $\bZ^{\tied} + \bZ^{\main:2}$ in this case.
    Namely, as in Theorem~\ref{thm:lifting}, we define $\bY \in \RR^{\binom{[N]}{\leq 3} \times \binom{[N]}{\leq 3}}$ to have entries $Y_{S,T} = \langle h_S(\bV^{\top}\bz), h_T(\bV^{\top}\bz) \rangle_{\partial}$.
    By Proposition~\ref{prop:Y-positive}, we have $\bY \succeq \lambda_{\min}(\bM)^3$, and we will bound $\|\bZ^{\tied} + \bZ^{\main:2} - \bY\|$ below.

    By construction $\bY$ is block diagonal.
    We let $\bY^{[d, d]} \in \RR^{\binom{[N]}{d} \times \binom{[N]}{d}}$ be the diagonal block indexed by sets of size $d$.
    Likewise, $\bZ^{\tied}$ and $\bZ^{\main:2}$ are block diagonal, by Corollary~\ref{cor:Z-main-Z-tied} and our remark above, respectively.
    Denote $\bZ^{\tied[d, d]}$ and $\bZ^{\main:2[d, d]}$ for their respective diagonal blocks.
    We have $\bZ^{\main:2[0, 0]}$ and $\bZ^{\main:2[1, 1]}$ both identically zero, and $\bZ^{\tied[0, 0]} = \bY^{[0, 0]} = [1]$ and $\bZ^{\tied[1, 1]} = \bY^{[1, 1]} = \bM$, so it suffices to bound $\|\bZ^{\tied[d, d]} + \bZ^{\main:2[d, d]} - \bY^{[d, d]}\|$ for $d \in \{2, 3\}$.

    For $d = 2$, as mentioned above, there is only one diagram consisting of two sided pairs occurring in $\bZ^{\main:2[2, 2]}$. 
    There are three bowtie forest ribbon diagrams in $\bZ^{\tied[2, 2]}$: two consist of two pairs, while the third consists of all four leaves connected to a $\square$ vertex.
    For a given choice of $\sigma, \tau \in \Part([2])$, if either $|\sigma| = 1$ or $|\tau| = 1$ then there is a unique $\bD \in \Plans(\sigma, \tau)$; otherwise, there are two plans corresponding to the two matchings of two copies of $\{1, 2\}$.
    Enumerating all of these terms, we find that most of them cancel in $\bZ^{\tied[2, 2]} + \bZ^{\main:2[2, 2]} - \bY^{[2, 2]}$: see Figure~\ref{fig:low-rank-d-2} for a graphical depiction of the calculation.
    We are left only with the final diagrams illustrated there.
    The sum of their CGMs' norm we can bound using the factorization of Proposition~\ref{prop:factorization} and the norm bound of Proposition~\ref{prop:cgm-generic-norm-bound}:
    \begin{equation}
        \|\bZ^{\tied[2, 2]} + \bZ^{\main:2[2, 2]} - \bY^{[2, 2]}\| \leq 2\|\bM\|^4\|\bM + t_{\pow}\one\one^{\top} - \bM^{\circ 2}\| = 2\|\bM\|^4 \widetilde{\epsilon}_{\pow}(\bM).
        \label{eq:low-rank-pf-d2}
    \end{equation}

    \begin{figure} 
        \begin{center}
            \includegraphics[scale=0.7]{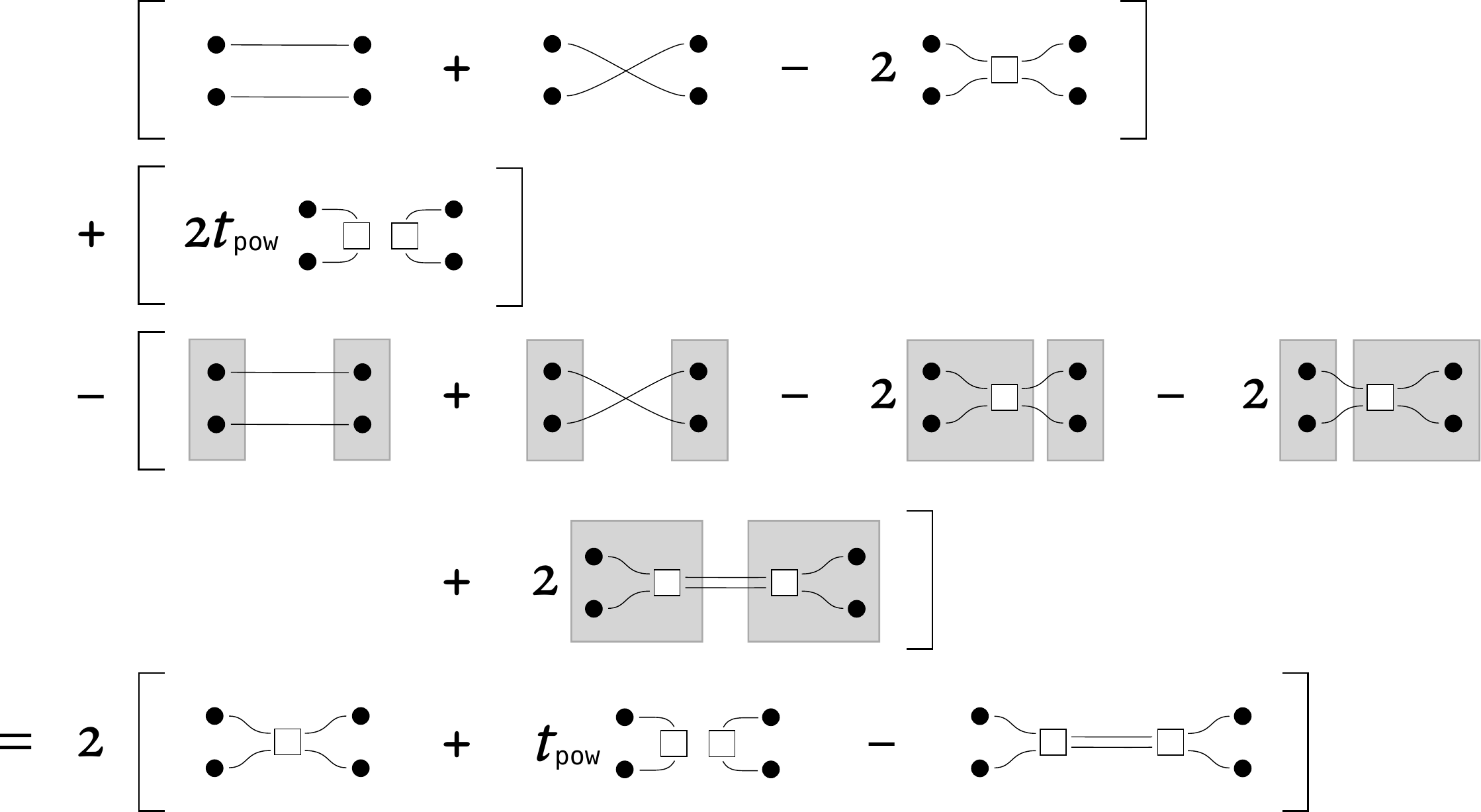}
        \end{center}
        \caption{\textbf{Diagrammatic manipulation of degree~4 adjustment.} We show the calculation of $\bZ^{\tied[2, 2]} + \bZ^{\main:2[2, 2]} - \bY{[2, 2]}$ (listed here in the same order they appear in the figure) and the role of the additional diagram of $\bZ^{\main:2[2, 2]}$ in adjusting the result. The reader may compare with Figure~\ref{fig:2-2-issue} to see why the right-hand side is a desirable outcome whose norm we are able to bound.} 
        \label{fig:low-rank-d-2}
    \end{figure}
    
    For $d = 3$, again as mentioned before, there are 9 diagrams in $\bZ^{\main:2[3, 3]}$, each consisting of two sided pairs and one non-sided pair.
    There are 16 bowtie forest ribbon diagrams in $\bZ^{\tied[3, 3]}$: 6 where every connected component is a pair, 9 where one connected component is a pair and another is a star on 4 leaves, and 1 star on all 6 leaves.
    We first apply Lemma~\ref{lem:tying-partition-transport-ribbon}, which controls the norm error incurred in tying any partition transport ribbon diagram.
    We use the following minor variant of this bound, which follows upon examining the proof in Appendix~\ref{app:pf:lem:tying-partition-transport-ribbon} for the particular case $d = 3$: for $\sigma, \tau \in \Part([3])$ and $\bD \in \Plans(\sigma, \tau)$, if $\bD$ does not equal a permutation of the matrix $\bD^{\star} \colonequals \left[\begin{array}{cc} 1 & 0 \\ 0 & 2 \end{array}\right]$, then
    \begin{equation}
        \|\bZ^{G(\sigma, \tau, \bD)} - \bZ^{\tie(G(\sigma, \tau, \bD))}\| \leq 9\|\bM\|^9(\epsilon_{\offdiag}(\bM) + \widetilde{\epsilon}_{\pow}(\bM)).
    \end{equation}
    This follows simply by observing that, in this case, either $\sigma$ or $\tau$ contains a singleton whereby Case~2 of that proof applies (giving the first term above), or $\sigma = \tau = \{\{1, 2, 3\}\}$, in which case $\bD = [3]$ and Case 3 applies and immediately yields the tied diagram after the first step, incurring error of only $\|\bM\|^6\|\bM^{\circ 3} - \bm I_N\|$.
    (We do this carefully to avoid incurring a cost of $\|\bM^{\circ 2} - \bm I_N\|$, which we are no longer assuming we have good control over.)

    Applying this bound to all partition transport ribbon diagrams in $\bY$ whose matrix $\bD$ is not a permutation of $\bD^{\star}$ as above, we may form $\bY^{\tied[3, 3]}$ where all of these diagrams are tied which, by the same counting as in Corollary~\ref{cor:Y-Z-tied}, will satisfy
    \begin{equation}
        \|\bY^{[3, 3]} - \bY^{\tied[3, 3]}\| \leq 10^8 \|\bM\|^9 (\epsilon_{\offdiag}(\bM) + \widetilde{\epsilon}_{\pow}(\bM)).
    \end{equation}
    Now, in $\bZ^{\tied[3, 3]} + \bZ^{\main:2[3, 3]} - \bY^{\tied[3, 3]}$, all ribbon diagrams cancel (as in the $d = 2$ case) except for 9 copies of the situation illustrated in Figure~\ref{fig:low-rank-d-2}, each with an extra non-sided pair.
    Therefore, applying the same argument, we obtain the same bound, multiplied by 9 for the number of diagrams and by $\|\bM\|$ for the extra non-sided pair.
    Thus:
    \begin{equation}
        \|\bZ^{\tied[3, 3]} + \bZ^{\main:2[3, 3]} - \bY^{\tied[3, 3]}\| \leq 18\|\bM\|^5 \widetilde{\epsilon}_{\pow}(\bM),
        \label{eq:low-rank-pf-d3}
    \end{equation}
    and by triangle inequality, combining this with the previous inequality we find
    \begin{equation}
        \|\bZ^{\tied[3, 3]} + \bZ^{\main:2[3, 3]} - \bY^{[3, 3]}\| \leq 10^9 \|\bM\|^9 (\epsilon_{\offdiag}(\bM) + \widetilde{\epsilon}_{\pow}(\bM)).
    \end{equation}

    Combining \eqref{eq:low-rank-pf-Z-main-1-Z-tied}, \eqref{eq:low-rank-pf-d2}, and \eqref{eq:low-rank-pf-d3}, we have:
    \begin{align}
      \lambda_{\min}(\bZ^{\main:1} + \bZ^{\main:2})
      &\geq \lambda_{\min}(\bY) - \lambda_{\max}(\bZ^{\tied} + \bZ^{\main:2} - \bY) - \lambda_{\max}(\bZ^{\main:1} - \bZ^{\tied}) \\
      &\geq \lambda_{\min}(\bM)^3 - 10^{39}\|\bM\|^9\big(\widetilde{\epsilon}_{\pow}(\bM) + \epsilon_{\offdiag}(\bM)\big). \label{eq:low-rank-pdf-pos-term}
    \end{align}

    It remains to bound $\|\bZ^{\err}\|$.
    Here we again use a small improvement on the general strategy, this time that of Corollary~\ref{cor:partition-error-norm-bound} for bounding the inner error matrices $\bm\Delta^{(\sigma, \tau, T)}$, specific to degree~6.
    Recall that here $\sigma \in \Part([\ell]; \odd), \tau \in \Part([m]; \odd)$, and $T \in \sT(|\sigma| + |\tau|)$, for $0 \leq \ell, m \leq 3$.
    $\bm\Delta^{(\sigma, \tau, T)}$ is indexed by $\binom{[N]}{\ell} \times \binom{[N]}{m}$, and contains terms $\Delta^T(\bM; \cdot)$.
    We note that if $|\sigma| + |\tau| < 4$, then $\Delta^T(\bM; \bs) = 0$ identically, so $\bm\Delta^{(\sigma, \tau, T)} = \bm 0$ for any $T$ in this case.
    But, when $\ell, m \leq 3$, the only way this can be avoided and $|\sigma|$ and $|\tau|$ can have the same parity is in one of three cases: (1) $\sigma$ and $\tau$ both consist of two singletons, (2) $\sigma$ and $\tau$ both consist of three singletons, or (3) one has a single part of size 3 and the other consists of three singletons.
    In any case, the matrix $\bF$ from Proposition~\ref{prop:partition-error-factorization} may be viewed as \emph{sparse}, along whichever of the rows or columns is indexed by a partition consisting only of singletons.
    Therefore, our norm bound which naively bounded $\|\bF\| \leq \|\bF\|_F$ can be improved by applying Proposition~\ref{prop:rectangular-gershgorin}.
    Repeating the argument of Corollary~\ref{cor:partition-error-norm-bound} in this way, we find
    \begin{align}
      \|\bm\Delta^{(\sigma, \tau, T)}\|
      &\leq \|\bM\|^6\sqrt{\left(N^{-1}\epsilon_{\err}(\bM; 6) + N^{-3/2}\epsilon_{\err}(\bM; 6) \cdot N\right)^2} \nonumber \\
      &\leq 2\|\bM\|^6N^{-1/2}\epsilon_{\err}(\bM; 6), \text{ if } \sigma = \tau = \{\{1\}, \{2\}\}, \\
      \|\bm\Delta^{(\sigma, \tau, T)}\| &\leq \|\bM\|^6\sqrt{\left(N^{-3/2}\epsilon_{\err}(\bM; 6) + N^{-4/2}\epsilon_{\err}(\bM; 6) \cdot N + N^{-5/2}\epsilon_{\err}(\bM; 6) \cdot N^2\right)^2} \nonumber \\
      &\leq 3\|\bM\|^6N^{-1/2}\epsilon_{\err}(\bM; 6), \text{ if } \sigma = \tau = \{\{1\}, \{2\}, \{3\}\}, \\
      \|\bm\Delta^{(\sigma, \tau, T)}\|
      &\leq \sqrt{(N^{-3/2}\epsilon_{\err}(\bM; 6) \cdot N)(N^{-3/2}\epsilon_{\err}(\bM; 6) \cdot 3)} \nonumber \\
      &\leq 3\|\bM\|^6N^{-1/2}\epsilon_{\err}(\bM; 6), \text{ if } \sigma = \{\{1, 2, 3\}\}, \tau = \{\{1\}, \{2\}, \{3\}\} \text{ or vice-versa}.
    \end{align}

    In effect, we are able to scale $\epsilon_{\err}(\bM; 6)$ down by an additional factor of $N^{-1/2}$ using this technique.
    Following the remainder of the proof of Lemma~\ref{lem:Z-err-bound} with this improvement then gives
    \begin{equation}
        \|\bZ^{\err}\| \leq 10^{150}\|\bM\|^{15} N^{-1/2}\epsilon_{\err}(\bM; 6).
    \end{equation}
    Combining this with \eqref{eq:low-rank-pdf-pos-term} then gives
    \begin{align}
      \lambda_{\min}(\bZ)
      &\geq \lambda_{\min}(\bZ^{\main:1} + \bZ^{\main:2}) - \|\bZ^{\err}\| \nonumber \\
      &\geq \lambda_{\min}(\bM)^3 - 10^{150}\|\bM\|^{15}\big(\widetilde{\epsilon}_{\pow}(\bM) + \epsilon_{\offdiag}(\bM) + N^{-1/2}\epsilon_{\err}(\bM; 6)\big) \nonumber \\
      &\geq 0
    \end{align}
    by our assumption in the statement.
    Thus, $\bZ \succeq \bm 0$.
    
    \vspace{1em}
    
    \emph{Proof of \eqref{eq:low-rank-lifting-err-1}:}
    We consider the pseudomoment matrix of the left-hand side, written in the standard monomial basis.
    The term arising from $\tEE^{\id}$ is then simply the identity matrix.
    Let us write $\widehat{\bZ}^{\pairs:\main:2} \in \RR^{\binom{[N]}{\leq 3} \times \binom{[N]}{\leq 3}}$ for the matrix arising from $\tEE^{\pairs:\main:2}$ (the hat serving as a reminder that this is a pseudomoment in the standard basis rather than the harmonic basis).
    Then, $\widehat{\bZ}^{\pairs:\main:2}$ is blockwise a sum of CGM terms corresponding to diagrams $F$, each of which consists only of pairs and has at most one sided pair.
    By and Propositions~\ref{prop:norm-tensorization} and \ref{prop:cgm-generic-norm-bound}, the norm of such a CGM is at most $\|\bM^2\|^3\|\bM^2\|_F \leq \|\bM\|^6 \|\bM^2\|_F$, the operator norm terms accounting for the non-sided pairs and the Frobenius norm term for the sided pair.
    The total number of such CGM terms across all blocks is 14: $2! = 2$ from the block with $|S| = |T| = 2$, 3 from the block with $|S| = 1, |T| = 3$, 3 from the block with $|S| = 3, |T| = 1$, $3! = 6$ from the block with $|S| = |T| = 3$.
    Therefore, we have
    \begin{equation}
        \|2(1 - c)t_{\pow}\widehat{\bZ}^{\pairs:\main:2}\| \leq 28 t_{\pow}\|\bM\|^6\|\bM^2\|_F \leq \frac{c}{2}
    \end{equation}
    by our choice of $c$, concluding the proof of the claim.

    \vspace{1em}
    
    \emph{Proof of \eqref{eq:low-rank-lifting-err-2}:}
    We again consider the pseudomoment matrix of the left-hand side, written in the standard monomial basis, and the arising from $\tEE^{\id}$ is again the identity matrix.
    Let us write $\widehat{\bZ}^{\pairs:\err} \in \RR^{\binom{[N]}{\leq 3} \times \binom{[N]}{\leq 3}}$ for the matrix arising from $\tEE^{\pairs:\err}$.
    The entries of this matrix are as follows.
    First, $\widehat{Z}^{\pairs:\err}_{S,T} = 0$ whenever $|S| + |T| < 4$, $|S|$ and $|T|$ have different parity, or $|S \cap T| = 0$.
    Also, $\widehat{Z}^{\pairs:\err}$ is symmetric.
    The remaining entries are given, writing $\bH = \bM^2$ here to lighten the notation, by
    \begin{align}
      \widehat{Z}^{\pairs:\err}_{\{i\}\{i, j, k\}} = \widehat{Z}^{\pairs:\err}_{\{i,j\}\{i, k\}} &= -H_{jk} - 2H_{ij}H_{ik}, \\
      \widehat{Z}^{\pairs:\err}_{\{i, j, k\},\{i, \ell, m\}} &= -2H_{ij}H_{ik}H_{\ell m} - 2H_{ij}H_{i\ell}H_{km} - 2H_{ij}H_{im}H_{k\ell} \nonumber\\
                                                                                                 &\hspace{1.1em} - 2H_{ik}H_{i\ell}H_{jm} - 2H_{ik}H_{im}H_{j\ell} - 2H_{i\ell}H_{im}H_{jk} \nonumber \\
                                                                                                 &\hspace{2em} \text{ for } j, k, \ell, m \text{ distinct}, \\
      \widehat{Z}^{\pairs:\err}_{\{i,j, k\}\{i, j,\ell\}} &= -H_{k\ell} - 2H_{ij}^2H_{k\ell} - 2H_{ik}H_{i\ell} - 2H_{jk}H_{j\ell} \nonumber \\
      &\hspace{1.1em} - 4H_{ij}H_{ik}H_{j\ell}- 4H_{ij}H_{i\ell}H_{jk}. 
    \end{align}
    Accordingly, we find the entrywise bounds
    \begin{align}
      |\widehat{Z}^{\pairs:\err}_{\{i\}\{i,j,k\}}| &\leq 3\epsilon_{\offdiag}(\bM^2) \\
      |\widehat{Z}^{\pairs:\err}_{\{i, j\}\{i,k\}}| &\leq 3\epsilon_{\offdiag}(\bM^2) \text{ if } j, k \text{ distinct}, \\
      |\widehat{Z}^{\pairs:\err}_{\{i, j\}\{i,j\}}| &\leq 3 \\
      |\widehat{Z}^{\pairs:\err}_{\{i, j, k\},\{i, \ell, m\}}| &\leq 12\epsilon_{\offdiag}(\bM^2)^3 \text{ if } j, k, \ell, m \text{ distinct}, \\
      |\widehat{Z}^{\pairs:\err}_{\{i, j, k\},\{i, j, \ell\}}| &\leq 15\epsilon_{\offdiag}(\bM^2) \text{ if } k,\ell \text{ distinct}, \\
      |\widehat{Z}^{\pairs:\err}_{\{i, j, k\},\{i, j, k\}}| &\leq 15.
    \end{align}
    Let us write $\widehat{\bZ}^{\pairs:\err[\ell, m]}$ for the submatrix of $\widehat{\bZ}^{\pairs:\err}$ indexed by $|S| = \ell$ and $|T| = m$.
    By the Gershgorin circle theorem, we then find the bounds
    \begin{align}
      \|\widehat{\bZ}^{\pairs:\err[2, 2]}\| &\leq 3 + 6N\epsilon_{\offdiag}(\bM^2) \\
      \|\widehat{\bZ}^{\pairs:\err[3, 3]}\| &\leq 15 + 45 N \epsilon_{\offdiag}(\bM^2) + 36 N^2 \epsilon_{\offdiag}(\bM^2)^3,
                                              \intertext{and by the ``rectangular Gershgorin'' bound we prove in Proposition~\ref{prop:rectangular-gershgorin}, we find}
     \|\widehat{\bZ}^{\pairs:\err[1, 3]}\| &\leq \sqrt{3N^2\epsilon_{\offdiag}(\bM^2) \cdot 3\epsilon_{\offdiag}(\bM^2)} = 3N\epsilon_{\offdiag}(\bM^2).
    \end{align}
    Finally, by Proposition~\ref{prop:block-matrix-norm}, we combine these bounds to find
    \begin{align}
      \|\widehat{\bZ}^{\pairs:\err}\|
      &\leq \|\widehat{\bZ}^{\pairs:\err[2, 2]}\| + \|\widehat{\bZ}^{\pairs:\err[3, 3]}\| + 2\|\widehat{\bZ}^{\pairs:\err[2, 2]}\| \nonumber \\
      &\leq 18 + 57 N \epsilon_{\offdiag}(\bM^2) + 36N^2 \epsilon_{\offdiag}(\bM^2)^3. 
    \end{align}
    Therefore,
    \begin{equation}
        \|2(1 - c)t_{\pow}\widehat{\bZ}^{\pairs:\err}\| \leq 114 t_{\pow}(1 + N \epsilon_{\offdiag}(\bM^2) + N^2 \epsilon_{\offdiag}(\bM^2)^3) \leq \frac{c}{2}
    \end{equation}
    by the definition of $c$, concluding the proof of the claim.
\end{proof}

\subsection{Sherrington-Kirkpatrick Hamiltonian: Proof of Theorem~\ref{thm:appl-sk}}

\begin{proof}[Proof of Theorem~\ref{thm:appl-sk}]
    We start out following similar steps to Theorem~\ref{thm:appl-high-rank}.
    Fix some small $\delta > 0$, and let $V$ be the eigenspace spanned by the $\delta N$ leading eigenvectors of $\bW$.
    Let us define $r \colonequals \delta N$, following the notation from earlier.
    Let $\bg_1, \dots, \bg_r \in \sN(\bm 0, \bm I_N)$ be a collection of independent gaussian vectors coupled to $V$ such that $V = \mathsf{span}(\bg_1, \dots, \bg_r)$.
    Define $\bM^{(0)} \colonequals (1 - \alpha / 2)\frac{\delta^{-1}}{N}\sum_{i = 1}^r \bg_i\bg_i^{\top}$.
    Let $\bD$ be the diagonal matrix with $\diag(\bD) = \diag(\bM^{(0)})$, and define $\bM \colonequals \bm I - \bD + \bM^{(0)}$.

    We first control the entries of $\bD$.
    These are
    \begin{equation}
        D_{ii} = \left(1 - \frac{\alpha}{2}\right) \frac{\delta^{-1}}{N}\sum_{j = 1}^r (\bg_j)_i^2 = \left(1 - \frac{\alpha}{2}\right) \frac{1}{r}\sum_{j = 1}^r (\bg_j)_i^2.
    \end{equation}
    Applying the concentration inequality of Proposition~\ref{prop:chi-squared-conc} and a union bound, we find that with high probability $(1 - \alpha)\bm I_N \preceq \bD \preceq (1 - \alpha / 3) \bm I_N$.
    Since $\bM^{(0)} \succeq \bm 0$, on this event we have $\bM \succeq (\alpha / 3)\bm I_N$, and also $\|\bM - \bM^{(0)}\| = \|\bm I_N - \bD\| \leq \alpha$.

    We now show that $\bM$ satisfies the conditions of Theorem~\ref{thm:low-rank-lifting}.
    We note that, again writing $\bh_1, \dots, \bh_N \in \RR^r$ for the vectors $\bh_i = ((\bg_j)_i)_{j = 1}^N$, for $i \neq j$ we have $M_{ij} = (1 - \alpha / 2)\frac{1}{r}\langle \bh_i, \bh_j \rangle$.
    Following the same calculations as in Theorem~\ref{thm:appl-high-rank} (noting that we may follow them truly verbatim, since the setting is identical except for the constant in front of each $M_{ij}$ with $i \neq j$), we find that we have, with high probability
    \begin{align}
      \epsilon_{\offdiag}(\bM) &\leq K\sqrt{\frac{\log N}{N}}, \\
      \epsilon_{\err}(\bM; 6) &\leq K.
    \end{align}
    Here and in the rest of this proof, we adopt the convention that $K = K(\delta) > 0$ is a constant that may change from line to line.

    We now control $\epsilon_{\pow}(\bM)$.
    For $k \geq 3$, by the Gershgorin circle theorem and substituting in $\epsilon_{\offdiag}(\bM)$, we note that we have
    \begin{equation}
        \|\bM^{\circ k} - \bm I_N\| \leq N \left(K\sqrt{\frac{\log N}{N}}\right)^3 \leq K^3 \frac{\log^2 N}{\sqrt{N}}.
    \end{equation}
    We will take $t_{\pow} = (1 - \alpha / 2)^2\frac{1}{r}$.
    Let us define vectors $\ba_i \colonequals \isovec(\bh_i\bh_i^{\top} - \bm I_r)$.
    We note that
    \begin{equation}
        \langle \ba_i, \ba_j \rangle = \langle \bh_i, \bh_j \rangle^2 - \|\bh_i\|^2 - \|\bh_j\|^2 + r.
    \end{equation}
    Thus, writing $\bA$ for the matrix with the $\ba_i$ as its columns and $\bb$ for the vector with entries $\|\bh_i\|^2$, we have
    \begin{align}
      \bM^{\circ 2} - \bm I_N
      &= \frac{(1 - \alpha / 2)^2}{r^2} \bigg(\bA^{\top}\bA + \bb\one_N^{\top} + \one_N\bb^{\top} - r\one_N\one_N^{\top} - \diag(\bb)^2\bigg) \\
      &= \frac{(1 - \alpha / 2)^2}{r^2} \bigg(\bA^{\top}\bA + (\bb - r\one_N)\one_N^{\top} + \one_N(\bb - r\one_N)^{\top} + r\one_N\one_N^{\top} - \diag(\bb)^2\bigg),
    \end{align}
    whereby with our choice of $t_{\pow}$ we may bound
    \begin{align}
      &\|\bM^{\circ 2} - \bm I_N - t_{\pow}\one_N\one_N^{\top}\| \nonumber \\
      &\hspace{2cm} \leq \frac{(1 - \alpha / 2)^2}{r^2} \bigg(\|\bA^{\top}\bA - r^2\bm I_N\| + 2N^{1/2}\|\bb - r\one_N\|_2 + \| r^2\bm I_N - \diag(\bb)^2\| \bigg) \nonumber 
        \intertext{By Proposition~\ref{prop:chi-squared-conc}, with high probability for all $i \in [N]$ we have $|\|\bh\|_i^2 - r| \leq r^{3/4}$, on which event we have $\|\bb - r\one_N\|_2 \leq r^{3/4}N^{1/2}$ and $\|r^2\bm I_N - \diag(\bb)^2\| \leq 3r^{7/4}$, which leaves only the more interesting term,}
      &\hspace{2cm} \leq \left\|\frac{1}{r^2}\bA^{\top}\bA - \bm I_N\right\|^2 + (3 + \delta^{-1})r^{-1/4}
        \intertext{Finally, to bound the remaining quantity, we use the general result of Theorem~3.3 of \cite{ALPTJ-2011-RIPIndependentColumns} which controls low-rank projection-like matrices like $\bA\bA^{\top}$ with independent random columns, subject to modest control of their distribution.
        A very similar situation to ours is also treated as an application of the above result by \cite{FJ-2019-RIPKhatriRao}.
        Applied in our setting, this shows that the remaining term is at most $K\sqrt{\log^2 N / N}$ with high probability, whereby we find with high probability}
      &\hspace{2cm}\leq KN^{-1/4}.
    \end{align}
    Combining these results, we see that $\widetilde{\epsilon}(\bM, t_{\pow}) \leq KN^{-1/4}$.\footnote{I thank Ramon van Handel for suggesting that the line of work in \cite{ALPTJ-2011-RIPIndependentColumns} would be applicable here.} 
    Since by our earlier calculations $\lambda_{\min}(\bM) \geq \alpha / 3$, the condition of Theorem~\ref{thm:low-rank-lifting} will hold with high probability.

    Now, we consider the constant $c$ appearing in the theorem.
    Recall that we chose $t_{\pow} \leq KN^{-1}$.
    We have $\|\bM\| \leq 1 + \|\bM^{(0)}\| \leq K$ with high probability by Proposition~\ref{prop:rectangular-gaussian-conc}, and $\|\bM^2\|_F \leq \sqrt{ r \|\bM^2\|^2} \leq r^{1/2} \|\bM\|^2 \leq KN^{1/2}$.
    The only quantity it remains to control is $\epsilon_{\offdiag}(\bM^2)$.
    We have
    \begin{equation}
        \bM^2 = (\bm I - \bD)^2 + (\bm I - \bD)\bM^{(0)} + \bM^{(0)}(\bm I - \bD) + \bM^{(0)^2},
    \end{equation}
    and thus, for $i \neq j$,
    \begin{align}
      |(\bM^2)_{ij}| &\leq |2 - D_{ii} - D_{jj}| \, |M^{(0)}_{ij}| + |(\bM^{(0)^2})_{ij}| \\
      &\leq K\sqrt{\frac{\log N}{N}} + |(\bM^{(0)^2})_{ij}| 
    \end{align}
    with high probability for all $i \neq j$ by our previous reasoning.
    For the remaining term, let $\bG \in \RR^{N \times r}$ have the $\bg_i$ as its columns and the $\bh_i$ as its rows, so that $\bM^{(0)} = (1 - \alpha / 2) \frac{\delta^{-1}}{N}\bG\bG^{\top}$.
    Then, we may write
    \begin{align}
      (\bM^{(0)^2})_{ij}
      &= (1 - \alpha / 2)^2 \frac{\delta^{-2}}{N^2}\be_i^{\top}\bG\bG^{\top}\bG\bG^{\top}\be_j \\
      &= (1 - \alpha / 2)^2\delta^{-2}\frac{1}{N^2} \bh_i^{\top}\left(\sum_{k = 1}^N \bh_k\bh_k^{\top}\right)\bh_j \\
      &= (1 - \alpha / 2)^2\delta^{-2}\frac{1}{N^2} \left((\|\bh_i\|^2_2 + \|\bh_j\|_2^2) \langle \bh_i, \bh_j \rangle + \bh_i^{\top}\left(\sum_{k \in[N] \setminus \{i, j\}}\bh_k\bh_k^{\top}\right)\bh_j\right).
    \end{align}
    In the last factor, by our previous reasoning with high probability the first summand is, in magnitude, at most $KN^{3/2}\log N$ for all $i \neq j$.
    In the second summand, note that for each fixed $i \neq j$, the vectors $\bh_i, \bh_j$, and the matrix $\sum_{k \in [N] \setminus \{i, j\}} \bh_k\bh_k^{\top} \equalscolon \bB^{\sim \{i, j\}}$ are independent.
    By Proposition~\ref{prop:rectangular-gaussian-conc}, for all $i \neq j$, $\|\bB^{\sim \{i, j\}}\| \leq KN$, and therefore also $\|\bB^{\sim\{i, j\}}\|_F^2 \leq N^3$, with high probability.
    Conditioning on the value of this matrix and applying the Hanson-Wright inequality \cite{RV-2013-HansonWright} for a fixed $i \neq j$ then shows, after a union bound, that with high probability, $|\bh_i^{\top} \bB^{\sim\{i, j\}} \bh_j| \leq KN^{13/8}$ for all $i \neq j$ (indeed, this will hold with any exponent larger than $3/2$).
    Thus we find $\epsilon_{\offdiag}(\bM^2) \leq KN^{-3/8}$ with high probability (this, in turn, will hold with any exponent smaller than $1/2$, which coincides with our expectation that we should have $\epsilon_{\offdiag}(\bM^2) \lesssim K \epsilon_{\offdiag}(\bM)$ for $\bM$ close to a rescaled random projection matrix).

    Therefore, the constant $c$ from the statement of Theorem~\ref{thm:low-rank-lifting} will satisfy
    \begin{equation}
        c \leq \frac{K}{N}\left(\sqrt{N} + N \cdot N^{-3/8} + N^2 \cdot N^{-9/8}\right) \leq KN^{-1/8}.
    \end{equation}
    The theorem produces $\tEE$ a degree 6 pseudoexpectation with $\tEE[\bx\bx^{\top}] = (1 - c)\bM + c\bm I_N$.
    Suppose we have chosen $\delta$ small enough that, with high probability, the largest $\delta N$ eigenvalues of $\bW$ are at least $2 - \epsilon / 2$.
    Then we have, with high probability,
    \begin{align}
      N^{-1}\SOS_{6}(\bW)
      &\geq N^{-1}\langle \bW, (1 - c)\bM + c\bm I_N \rangle \\
      &\geq (1 - c)N^{-1}\langle \bW, \bM^{(0)} \rangle - \alpha \|\bW\| - c|\Tr(\bW)|
        \intertext{and since we have, with high probability, $\|\bW\| \leq 2 + \alpha$, $|\Tr(\bW)| \leq \log N$, and $\langle \bW, \bM^{(0)} \rangle \geq \lambda_{r}(\bM^{(0)}) \lambda_{\delta N}(\bW) \geq (1 - \alpha)(2 - \epsilon / 2)$, we find on this event that}
      &\geq (1 - KN^{-1/8})(1 - \alpha)(2 - \epsilon / 2) - \alpha(2 + \alpha) - KN^{-1/8}\log N 
        \intertext{and choosing $\alpha$ sufficiently small, depending on our choice of $\delta$ above, we will have for sufficiently large $N$}
        &\geq 2 - \epsilon,
    \end{align}
    completing the proof.
\end{proof}

\section*{Acknowledgements}

I thank Afonso Bandeira for many discussions and comments on an early version of the manuscript, Alex Wein for helpful discussions about hypercontractivity and tensor networks, and Ramon van Handel and Aida Khajavirad for suggesting references mentioned in the text.

\bibliographystyle{alpha}
\bibliography{main}

\clearpage
\appendix

\section{Linear Algebra Tools}

The following matrix norm inequality, a simple application of the Gershgorin circle theorem, is quite effective for sparse matrices, as we describe after the statement.
Recall that the $\infty$-norm of a matrix is defined as $\|\bA\|_{\infty} = \max_{\bx \neq 0} \|\bA\bx\|_{\infty} / \|\bx\|_{\infty} = \max_i \sum_j |A_{ij}|$.
\begin{proposition}
    \label{prop:rectangular-gershgorin}
    Let $\bA \in \RR^{m \times n}$.
    Then,
    \begin{equation}
        \|\bA\| \leq \sqrt{\|\bA\|_{\infty}\|\bA^{\top}\|_{\infty}} = \sqrt{\left(\max_{i = 1}^m \sum_{j = 1}^n |A_{ij}|\right)\left(\max_{j = 1}^n \sum_{i = 1}^m |A_{ij}|\right)}.
    \end{equation}
\end{proposition}
\begin{proof}
    We have $\|\bA\| = \sigma_{\max}(\bA) = \sqrt{\lambda_{\max}(\bA\bA^{\top})}$.
    By the Gershgorin circle theorem, we may bound $\lambda_{\max}(\bA\bA^{\top}) \leq \|\bA\bA^{\top}\|_{\infty}$.
    Since the $\infty$-norm on matrices is induced as an operator norm by the $\ell^{\infty}$ vector norm, it is submultiplicative, whereby $\|\bA\bA^{\top}\|_{\infty} \leq \|\bA\|_{\infty} \|\bA^{\top}\|_{\infty}$.
\end{proof}
\begin{remark}
    This inequality is tight for any $\bA \in \{0, 1\}^{m \times n}$ where every row has exactly one 1, but every column can have arbitrary numbers of 1's.
    The extremes in this class of matrices are the identity on the one hand, and $\bm 1 \bm e_k^\top$ on the other (for a standard basis vector $\bm e_k$).
\end{remark}

The following other relative of the Gershgorin circle theorem gives a straightforward bound on block matrix norms.
\begin{proposition}
    \label{prop:block-matrix-norm}
    Suppose $\bA \in \RR^{m \times n}$ is partitioned into blocks $\bA^{[k, \ell]} \in \RR^{m_k \times n_{\ell}}$ where $\sum m_k = m$ and $\sum n_{\ell} = n$.
    Then, $\|\bA\| \leq \sum_{k, \ell} \|\bA^{[k, \ell]}\|$.
\end{proposition}
\begin{proof}
    If $\bx \in \RR^m$ and $\by \in \RR^n$ are partitioned into vectors $\bx_k$ and $\by_{\ell}$ with compatible sizes to the blocks of $\bA$, then we have
    \begin{equation}
        |\bx^{\top}\bA \by| = \left|\sum_{k, \ell} \bx_k^{\top}\bA^{[k, \ell]} \by_{\ell}\right| \leq \sum_{k, \ell}\|\bx_k\|_2 \|\by_{\ell}\|_2 \|\bA^{[k, \ell]}\|.
    \end{equation}
    Thus, letting $\bA^{\prime}$ be the matrix of norms of the blocks of $\bA$, we have $\|\bA\| \leq \|\bA^{\prime}\|$.
    The result then follows since if $\|\bv\| = \|\bw\| = 1$, then $\bv^{\top}\bA^{\prime}\bw \leq (\max_k |v_k|)(\max_{\ell}|w_{\ell}|)(\sum_{k,\ell}|A^{\prime}_{k\ell}|) \leq \sum_{k,\ell}|A^{\prime}_{k\ell}|$.
\end{proof}

\section{Probability Tools}

We use the following standard concentration results on gaussian random vectors and matrices repeatedly in our applications.
\begin{proposition}[Lemma 1 of \cite{LM-2000-ChiSquaredConc}]
    \label{prop:chi-squared-conc}
    Let $\bg \sim \sN(\bm 0, \bm I_n)$.
    Then,
    \begin{equation}
        \PP\left[ \big|\|\bg\|_2 - \sqrt{n}\big| \geq t \right] \leq 2\exp\left(-\frac{t^2}{2}\right).
    \end{equation}
\end{proposition}

\begin{proposition}[Corollary 7.3.3 and Exercise 7.3.4 of \cite{Vershynin-2018-HDP}]
    \label{prop:rectangular-gaussian-conc}
    Let $\bG \in \RR^{m \times n}$ with $m \geq n$ have i.i.d.\ entries distributed as $\sN(0, 1)$.
    Let $\sigma_{1}(\bG) \geq \cdots \geq \sigma_n(\bG) \geq 0$ denote the ordered singular values of $\bG$.
    Then,
    \begin{equation}
        \PP\left[ \sqrt{m} - \sqrt{n} - t \leq \sigma_n(\bG) \leq \sigma_1(\bG) \leq \sqrt{m} + \sqrt{n} + t\right] \geq 1 - 4\exp(-Ct^2)
    \end{equation}
    for a universal constant $C > 0$.
\end{proposition}

\section{Combinatorial Bounds}

We prove several coarse bounds on combinatorial quantities arising in our arguments.
We begin with bounds on the coefficients of forests that arise in our calculations.
The only tool required for these is that $(a + b)! \geq a! \, b!$, which follows from observing that $\binom{a + b}{a} \geq 1$.
\begin{proposition}
    \label{prop:bound-mu}
    For any $F \in \sF(d)$, $|\mu(F)| \leq (3d)!$.
\end{proposition}
\begin{proof}
    We have
    \begin{align}
      |\mu(F)|
      &= \prod_{v \in V^{\square}} (\deg(v) - 2)! \nonumber \\
      &\leq \left(\sum_{v \in V} \deg(v)\right)! \nonumber \\
      &\leq \left(2 \cdot \frac{3}{2}d\right)!,
    \end{align}
    the last step following by Corollary~\ref{cor:sF-vertex-edge-bound}.
\end{proof}

\begin{proposition}
    \label{prop:bound-xi}
    For any $F \in \sF(d, d)$ a balanced bowtie forest ribbon diagram, $|\xi(F)| \leq (d - 1)! \, d! \leq d^{2d}$.
\end{proposition}
\begin{proof}
    We have
    \begin{align}
      |\xi(F)|
      &= \prod_{\substack{C \in \conn(F) \\C \text{ balanced bowtie} \\ \text{on } 2k \text{ leaves}}} (k - 1)! k! \nonumber \\
      &\leq \Bigg(\sum_{\substack{C \in \conn(F) \\C \text{ balanced bowtie} \\ \text{on } 2k \text{ leaves}}} (k - 1)\Bigg)! \, \Bigg(\sum_{\substack{C \in \conn(F) \\C \text{ balanced bowtie} \\ \text{on } 2k \text{ leaves}}} k\Bigg)! \nonumber \\
      &= (d - |\conn(F)|)!\,  d! \nonumber \\
      &\leq (d - 1)!\, d!,
    \end{align}
    completing the proof.
\end{proof}

We next give some bounds on the cardinalities of various sets of combinatorial objects arising in our analysis.

\begin{proposition}
    \label{prop:size-Part}
    $|\Part([d])| \leq d^d$.
\end{proposition}
\begin{proof}
    For $d \leq 5$, the inequality may be verified directly.
    Assuming $d \geq 6$, we begin with a ``stars and bars'' argument: all partitions of $[d]$ may be obtained by writing the numbers $1, \dots, d$ in some order, and then choosing some subset of the $d - 1$ possible positions between two numbers where a boundary between parts of the partition may be placed.
    Therefore, $|\Part([d])| \leq 2^dd!$.
    For $d \geq 6$, we have $d! \leq (d / 2)^d$, and the result follows.
\end{proof}
\noindent
We note that the numbers $|\Part([d])|$ are known as the \emph{Bell numbers}, for which many more precise asymptotics are known \cite{BT-2010-BoundsBellNumbers}.
We prefer to give a simple hands-on proof here, which matches the correct coarse scaling $\log |\Part([d])| = \Theta(d\log d)$.

\begin{proposition}
    \label{prop:size-Plans}
    For any $\sigma, \tau \in \Part([d])$, $|\Plans(\sigma, \tau)| \leq d!$.
\end{proposition}
\begin{proof}
    For every $\bD \in \Plans(\sigma, \tau)$, there exists a bijection $f: [d] \to [d]$ for which $D_{A, B} = \#\{i \in A: f(i) \in B\}$.
    Therefore, the total number of such $\bD$ is at most the total number of bijections of $[d]$ with itself, which is $d!$.
\end{proof}

\noindent
The following simple and general result gives a bound on the number of vertices in a tree if the degrees of internal vertices are bounded below.
\begin{proposition}
    \label{prop:tree-vertex-bound}
    Suppose $T = (V, E)$ is a tree with $\ell$ leaves.
    Assume that, for any internal vertex $v$ of $T$, $\deg(v) \geq k \geq 3$.
    Then,
    \begin{equation}
        |V| < \frac{k - 1}{k - 2}\ell.
    \end{equation}
\end{proposition}
\begin{proof}
    We count the number of edges $|E|$ in two ways, and then apply the degree bound:
    \begin{align}
      |E|
      &= |V| - 1 \nonumber \\
      &= \frac{1}{2}\sum_{v \in V} \deg(v) \nonumber \\
      &\geq \frac{1}{2}\left(\ell + k(|V| - \ell)\right) \nonumber \\
      &= \frac{k}{2}|V| - \frac{k - 1}{2}\ell.
    \end{align}
    Solving for $|V|$, we have
    \begin{equation}
        |V| \leq \frac{k - 1}{k - 2}\ell - \frac{2}{k - 2},
    \end{equation}
    and the result follows.
\end{proof}
\noindent

\begin{corollary}
    \label{cor:sF-vertex-edge-bound}
    For any $F \in \sF(d)$, the number of vertices and edges in $F$ are both at most $\frac{3}{2}d$, and the number of $\square$ vertices is at most $\frac{1}{2}d$.
\end{corollary}
\noindent
This allows us to bound the number of good forests, as follows.

\begin{proposition}
    \label{prop:size-F}
    For $d$ even, $(d / 2)! \leq |\sF(d)| \leq 2(\frac{3}{2}d)^{\frac{3}{2}d}$.
\end{proposition}
\begin{proof}
    For the lower bound, we simply note that any matching of $1, \dots, d / 2$ with $d / 2 + 1, \dots, d$ corresponds to a distinct element of $\sF(d)$ consisting of the corresponding forest all of whose connected components are pairs.
    Since there are $(d / 2)!$ such matchings, the bound follows.

    For the upper bound, Theorem 4.1 of \cite{Moon-1970-CountingLabelledTrees}, attributed there to \Renyi, gives the following explicit formula for the number of labelled forests on $n$ nodes with $k$ connected components, which we denote $f_{n, k}$ (this is a generalization of Cayley's well-known formula counting labelled trees, which is the case $k = 1$):
    \begin{align}
      f_{n, k}
      &= \binom{n}{k} \sum_{i = 0}^k \left(-\frac{1}{2}\right)^i (k + i) \,  i!\binom{k}{i}\binom{n - k}{i} n^{n - k - i - 1}
        \intertext{from which by coarse bounds we find}
      &\leq \frac{n^k}{k!} \sum_{i = 0}^k \frac{k + i}{2^i} \frac{k!}{(k - i)!} \frac{n^i}{i!} n^{n - k - i - 1} \\
      &= n^{n - 1} \sum_{i = 0}^k \frac{k + i}{2^i(k - i)!\, i!} \\
      &\leq 2n^{n - 1},
    \end{align}
    where the final inequality may be checked by verifying by hand for $k \leq 4$, and for $k \geq 5$ bounding the inner term by $2k / (\lfloor k / 2 \rfloor)! (\lceil k / 2 \rceil)! \leq 2$.
    Thus the number of labelled forests on $n$ nodes with any number of connected components, which equals $\sum_{k = 1}^n f_{n, k}$, is at most $2n^n$.
    By Corollary~\ref{cor:sF-vertex-edge-bound}, the total number of vertices in $F \in \sF(d)$ is at most $\frac{3}{2}d$.
    Since there are no isolated vertices in such $F$, we may always view $F$ as embedded uniquely into a fully labelled forest on exactly $\frac{3}{2}d$ vertices by adding isolated vertices and labelling internal vertices with any deterministic procedure.
    Thus, $|\sF(d)|$ is at most the number of labelled forests on $\frac{3}{2}d$ vertices, and the result follows.
\end{proof}
\noindent
We include the lower bound above to emphasize that the upper bound correctly identifies the coarse behavior $\log |\sF(d)| = \Theta(d\log d)$.
This suggests that, without much more careful proof techniques, the $d^d$ behavior of the leading factor in the condition of Theorem~\ref{thm:lifting} cannot be improved.

\section{Gaussian Conditioning Calculations: Proof of Lemma~\ref{lem:poly-cond}}
\label{app:pf:lem:poly-cond}

\begin{proof}[Proof of Lemma~\ref{lem:poly-cond}]
    Recall that we must compute the distribution of $g^{(d)}(\by) \sim \sG_d^{\poly}(N, \sigma_d^2)$, given $g^{(0)}(\by) = 1$, conditional on (a) having, if $d \geq 2$, for all $i \in [N]$ and $S^{\prime} \in \sM_{d - 2}([N])$ that $\langle g^{(d)}, \by^{S^{\prime}}y_i^2 \rangle_{\circ} = \langle g^{(d - 2)}, \by^{S^{\prime}} \rangle_{\circ}$, and (b) having for all $S^{\prime} \in \sM_{d - 1}([N])$ and $i \in [N - r]$ that $\langle g^{(d)}, \by^{S^{\prime}} \langle \bw_i, \by \rangle_{\circ} = 0$.
    
Working first with Property (a), we see after extending by linearity that it is equivalent to $\langle g^{(d)}, q(\by)y_i^2\rangle_{\circ} = \langle g^{(d - 2)}, q(\by)\rangle_{\circ}$ for all $i \in [N]$ and $q \in \RR[y_1, \dots, y_N]^{\hom}_{d - 2}$.
On the other hand, by the adjointness property from Proposition~\ref{prop:apolar-adjointness}, we have $\langle g^{(d)}, q(\by)y_i^2\rangle_{\circ} = \frac{1}{d(d - 1)} \langle \partial_{y_i}^2g^{(d)}, q(\by) \rangle_{\circ}$, and thus Property (a) is equivalent to the simpler property:
\begin{enumerate}
    \item[(a$^{\prime}$)] If $d \geq 2$, then for all $i \in [N]$, $\partial_{y_i}^2 g^{(d)}(\by) = d(d - 1) \cdot g^{(d - 2)}(\by)$.
\end{enumerate}

Similarly, we see after extending Property~(b) by linearity that it is equivalent to having $\langle g^{(d)}, q(\by) \langle \bw, \by \rangle\rangle_{\circ} = 0$ for all $q \in \RR[y_1, \dots, y_N]^{\hom}_{d - 1}$ and $\bw \in \ker(\bM)$.
Again by Proposition~\ref{prop:apolar-adjointness}, $\langle g^{(d)}, q(\by) \langle \bw, \by \rangle\rangle_{\circ} = d\langle \langle \bw, \bm\partial \rangle g^{(d)}, q(\by) \rangle_{\circ}$, so Property (b) is equivalent to:
\begin{enumerate}
    \item[(b$^{\prime}$)] For all $\bw \in \ker(\bM)$, $\langle \bw, \bm\partial \rangle g^{(d)}(\by) = 0$.
\end{enumerate}
Polynomials $p \in \RR[y_1, \dots, y_N]^{\hom}_d$ with $\langle \bw, \bm\partial \rangle p = 0$ for all $\bw \in \ker(\bM)$ form a linear subspace, which admits the following simple description.
Changing monomial basis to one which extends $z_i = (\what{\bV} \by)_i$ for $i \in [r]$ and invoking Proposition~\ref{prop:apolar-orth-invariant}, we may define
\begin{equation}
    V_{\sB} \colonequals \left\{p \in \RR[y_1, \dots, y_N]_d^{\hom}: \text{there exists } q \in \RR[z_1, \dots, z_r]^{\hom}_d \text{ such that } p(\by) = q(\what{\bV} \by) \right\},
\end{equation}
and with this definition Property~(b) is equivalent to:
\begin{enumerate}
    \item[(b$^{\prime\prime}$)] $g^{(d)} \in V_{\sB}$.
\end{enumerate}

Now, let $\bQ \in \sO(N)$ be an orthogonal matrix formed by adding rows to $\what{\bV}$ (see Section~\ref{sec:deg2-assumptions} for why this is possible under our assumptions on $\bM$).
Letting $a_{S} \sim \sN(0, \sigma_d^2\binom{d}{\freq(S)})$ for each $S \in \sM_d([N])$ independently, we set $g^{(d)}_0(\bm y) \colonequals \sum_{S \in \sM_d([N])} a_{S}\cdot (\bQ\bm y)^{S}$; then, the law of $g^{(d)}_0$ is $\sG_d^{\poly}(N, \sigma_d^2)$ by Proposition~\ref{prop:G-poly-orth-invariant}.
Thus we may form the law of $g^{(d)}$ by conditioning $g^{(d)}_0$ on Properties~(a$^{\prime}$) and (b$^{\prime\prime}$).

Conveniently, conditioning $g^{(d)}_0$ on Property~(b$^{\prime\prime}$) amounts to merely setting those $a_{S}$ with $S \cap \{r + 1, \ldots, N\} \neq \emptyset$ to equal zero.
Denoting the resulting random polynomial by $g^{(d)}_1$, we see that $g^{(d)}_1(\by) = \sum_{S \in \sM_d([r])} a_{S}\cdot (\bQ\bm y)^{S} = \sum_{S \in \sM_d([r])} a_{S}\cdot (\what{\bV}\bm y)^{S}$.
To extract the coefficients of $g^{(d)}_1$ in the standard monomial basis, we compute
\begin{align}
  \langle g_1^{(d)}, \by^{S} \rangle_{\circ}
  &= \sum_{T \in M_d([r])} a_{T} \langle (\what{\bV}\by)^{T}, \by^{S} \rangle_{\circ} \nonumber \\
  &= \sum_{T \in M_d([r])} a_{T} \langle \by^{T}, (\bQ^{\top}\by)^{S}\rangle_{\circ} \tag{\text{by Proposition~\ref{prop:apolar-orth-invariant}}}
    \intertext{and noting that no term involving $y_{r + 1}, \dots, y_N$ will contribute, we may define the truncation $\bz = (y_1, \dots, y_r)$ and continue}
  &= \sum_{T \in M_d([r])} a_{T} \langle \bz^{T}, (\what{\bV}^{\top}\bz)^{S}\rangle_{\circ} \nonumber \\
  &= \left\langle \sum_{T \in M_d([r])} a_{T}\bz^{T}, (\what{\bV}^{\top}\bz)^{S}\right\rangle_{\circ}.
\end{align}
Letting $h_0^{(d)}(\bz) \colonequals \sum_{T \in M_d([r])} a_{T}\bz^{T}$, we see that the law of $h_0^{(d)}$ is $\sG^{\poly}_d(r, \sigma_d^2)$.

Thus we may rewrite the result of the remaining conditioning on Property~(a$^{\prime}$) by letting $h^{(d)}$ have the law of $h_0^{(d)}$ conditioned on $\langle \bv_i, \bm\partial \rangle^2h^{(d)}(\bz) = \delta^{-1} d(d - 1) \cdot h^{(d - 2)}(\bz)$.
Then,
\begin{align}
  \tEE(\bx^{S}, \bx^{T})
  &= \EE\left[\langle h^{(d)}, (\what{\bV}^\top \bz)^{S}\rangle_{\circ}\, \langle h^{(d)}, (\what{\bV}^\top \bz)^{T}\rangle_{\circ}\right] \nonumber \\
  &= \delta^d\cdot \EE\left[\langle h^{(d)}, (\bV^{\top} \bz)^{S}\rangle_{\circ}\, \langle h^{(d)}, (\bV^{\top} \bz)^{T}\rangle_{\circ}\right]
    \label{eq:tEE-prediction-poly-h}
\end{align}

The result of this remaining conditioning is simple to write down in a more explicit form, since now we have just a single family of linear constraints to condition on and $h_0^{(d)}$ is isotropic with respect to the apolar inner product (per Definition~\ref{def:isotropic-poly}).
We define
\begin{align}
  V_{\mathcal{I}} &\colonequals \left\{\sum_{i = 1}^N \langle \bv_i, \bz \rangle^2 q_i(\bz) : q_i \in \RR[z_1, \dots, z_r]^{\hom}_{d - 2}\right\}, \\
  V_{\mathcal{H}} &\colonequals \left\{q \in \RR[z_1, \dots, z_r]^{\hom}_d: \langle \bv_i, \bm\partial \rangle^2 q = 0 \text{ for all } i \in [N]\right\},
\end{align}
instantiations of the ``ideal subspace'' and ``harmonic subspace'' constructions from Proposition~\ref{prop:apolar-decomp} for the specific polynomials $\{\langle \bv_i, \bx \rangle^2\}_{i = 1}^N$.
We let $P_{\sI}$ and $P_{\sH}$ be the orthogonal projections to $V_{\sI}$ and $V_{\sH}$, respectively.
We also define the ``least-squares lifting operator'' $L_{\sI}: \RR[z_1, \dots, z_r]_{d - 2}^{\hom} \to V_{\sI}$ by
\begin{equation}
  L_{\mathcal{I}}[h] \colonequals \mathrm{argmin}\left\{\|f\|_{\circ}^2 : f \in V_{\mathcal{I}}, \langle \bv_i, \bm\partial\rangle^2 f = h \text{ for all } i \in [N]\right\}.
\end{equation}
We then obtain that the law of $h^{(d)}$ is
\begin{equation}
    h^{(d)} \eqd \delta^{-1} d(d - 1) \cdot L_{\mathcal{I}}[h^{(d - 2)}] + P_{\mathcal{H}}[h_0^{(d)}]
\end{equation}
Note that the first summand above belongs to $V_{\sI}$ and the second to $V_{\sH}$, so this is also a decomposition for $h^{(d)}$ precisely of the kind provided by Proposition~\ref{prop:apolar-decomp}.

Now, we work towards substituting this into \eqref{eq:tEE-prediction-poly-h}.
To do that, we must compute inner products of the form $\langle h^{(d)}, (\bV^{\top}\bz)^{S} \rangle_{\circ}$.
We introduce a few further observations to do this: for each $S$, the projection of $(\bV^{\top}\bz)^{S}$ to $V_{\sI}$ is a linear combination of multiples of the $\langle \bv_i, \bz \rangle^2$.
Moreover, the $\langle \bv_i, \bz \rangle$ are an overcomplete system of linear polynomials, so \emph{any} polynomial in $\bz$ may be written as a polynomial in these variables instead.
Combining these facts, we see that there exists a polynomial $r_{S} \in \RR[x_1, \dots, x_N]_d^{\hom}$ such that
\begin{align}
  P_{\sI}[(\bV^{\top}\bz)^{S}] &= r_{S}(\bV^{\top}\bz),  \label{eq:r-poly-cond-1} \\
  r_{S}(\bx) &= \sum_{i = 1}^N x_i^{2d_i} r_{S, i}(\bx) \text{ for } d_i \geq 1, r_{S, i} \in \RR[x_1, \dots, x_N]_{d - 2d_i}^{\hom}.
  \label{eq:r-poly-cond-2}
\end{align}
The polynomial $r_{S}(\bx)$ is not unique, since the vectors $\bv_i$ and therefore the polynomials $\langle \bv_i, \bz \rangle$ are linearly dependent.
Nor is the decomposition of $r_S$ into the $x_i^{2d_i}r_{S, i}$ unique, since for instance if some $d_i \geq 2$ then we may move a factor of $x_i^2$ into $r_{S, i}$.
However, any choice of $r_{S}$ satisfying \eqref{eq:r-poly-cond-1} and $r_{S, i}$ satisfying \eqref{eq:r-poly-cond-2} suffices for our purposes.
In the main text we also give a heuristic calculation of a polynomial that approximately satisfies these conditions for each $S$, so we leave the requirements loosely specified here.
Note also that we reuse the pseudodistribution variables $\bx = (x_1, \dots, x_N)$ intentionally here because of the role $r_{S}(\bx)$ will play below.

With this definition, we compute
\begin{align}
  \langle h^{(d)}, (\bV^{\top}\bz)^{S} \rangle_{\circ}
  &= \delta^{-1}d(d - 1) \cdot \langle L_{\sI}[h^{(d - 2)}], P_{\sI}[(\bV^{\top}\bz)^{S}] \rangle_{\circ} + \langle P_{\sH}[h_0^{(d)}], (\bV^{\top}\bz)^{S} \rangle_{\circ} \nonumber \\
  &= \delta^{-1}d(d - 1) \sum_{i = 1}^N \langle L_{\sI}[h^{(d - 2)}], \langle \bv_i, \bm\bz\rangle^{2d_i} r_{S, i}(\bV^{\top}\bz) \rangle_{\circ} + \langle h_0^{(d)}, P_{\sH}[(\bV^{\top}\bz)^{S}] \rangle_{\circ} \nonumber \\
  &= \delta^{-1} \left\langle h^{(d - 2)}, \sum_{i = 1}^N\langle \bv_i, \bm\bz\rangle^{2d_i - 2} r_{S, i}(\bV^{\top}\bz) \right\rangle_{\circ} + \langle h_0^{(d)}, P_{\sH}[(\bV^{\top}\bz)^{S}] \rangle_{\circ}
\end{align}
Finally, we note that $h^{(d - 2)}$ and $h_0^{(d)}$ are independent, and $h_0^{(d)}$ is isotropic with variance $\sigma_d^2$.
Therefore, we may finally substitute into \eqref{eq:tEE-prediction-poly-h}, obtaining
\begin{align}
  &\tEE(\bx^{S}, \bx^{T}) \nonumber \\
  &= \delta^d\cdot \EE\left[\langle h^{(d)}, (\bV^{\top} \bz)^{S}\rangle_{\circ}\langle h^{(d)}, (\bV^{\top} \bz)^{T}\rangle_{\circ}\right] \nonumber \\
  &= \delta^{d - 2} \cdot \EE\left[\left\langle h^{(d - 2)}, \sum_{i = 1}^N\langle \bv_i, \bm\bz\rangle^{2d_i - 2} r_{S, i}(\bV^{\top}\bz) \right\rangle_{\circ}\left\langle h^{(d - 2)}, \sum_{i = 1}^N\langle \bv_i, \bm\bz\rangle^{2d_i - 2} r_{T, i}(\bV^{\top}\bz) \right\rangle_{\circ}\right] \nonumber \\
  &\hspace{2cm} + \delta^d \cdot \EE\left[\langle h_0^{(d)}, P_{\sH}[(\bV^{\top}\bz)^{S}]\rangle_{\circ}\langle h_0^{(d)}, P_{\sH}[(\bV^{\top}\bz)^{T}]\rangle_{\circ}\right]. \nonumber
    \intertext{While this expression appears complicated, each term simplifies substantially: the first term is an evaluation of $\tEE$ at degree $2d - 4$, per \eqref{eq:tEE-prediction-poly-h}, while the second term, by the isotropy of $h_0^{(d)}$, is an apolar inner product:}
  &= \tEE\left(\sum_{i = 1}^Nx_i^{2d_i - 2} r_{S, i}(\bx), \sum_{i = 1}^Nx_i^{2d_i - 2} r_{T, i}(\bx)\right) + \sigma_d^2 \delta^d \cdot \left\langle P_{\sH}[(\bV^{\top}\bz)^{S}], P_{\sH}[(\bV^{\top}\bz)^{T}]\right\rangle_{\circ}. \nonumber
    \intertext{Finally, we note that the factors $x_i^{2d_i - 2}$ are irrelevant in the evaluation of $\tEE$ (by the pseudoexpectation ideal property from Definition~\ref{def:pe} at degree $2d - 4$), so we obtain our final expression}
  &= \tEE\left(\sum_{i = 1}^N r_{S, i}(\bx), \sum_{i = 1}^Nr_{T, i}(\bx)\right) + \sigma_d^2 \delta^d \cdot \left\langle P_{\sH}[(\bV^{\top}\bz)^{S}], P_{\sH}[(\bV^{\top}\bz)^{T}]\right\rangle_{\circ},
\end{align}
which is the result in the statement.
\end{proof}

\section{Tools for Contractive Graphical Matrices}
\label{app:cgm-tools}

In this appendix, we present some general tools for working with contractive graphical matrices (henceforth CGMs, as in the main text).
In fact, to make some technicalities easier to work around, we use a more general definition, which allows different edges of a ribbon diagram to be labelled with different matrices and also allows for the left and right index subsets to overlap.

\begin{definition}[Generalized ribbon diagram]
    Let $G = ((\sL \cup \sR) \sqcup V^{\square}, E)$ be a graph with two types of vertices, $\bullet$ and $\square$, whose subsets are $V^{\bullet} = \sL \cup \sR$ and $V^{\square}$.
    Suppose that $G$ is equipped with labellings $\kappa_{\sL}: \sL \to [|\sL|]$ and $\kappa_{\sR}: \sR \to [|\sR|]$.
    Suppose that for every $e = \{x, y\} \in E$, we have a matrix $\bM^{(x, y)} \in \RR^{N(x) \times N(y)}$, for some $N: V \to \NN$, which satisfies $\bM^{(y, x)} = \bM^{(x, y)^{\top}}$.
    Note that $N$ is determined by the collection of matrices $\bM^{(x, y)}$, provided that their dimensions satisfy the appropriate equalities.
    We call such a collection of matrices a \emph{compatible matrix labelling} of the edges of $G$, noting that a matrix is associated to each \emph{oriented} edge in such a labelling.
    For the course of this appendix, we call such labelled $G$ a \emph{ribbon diagram}, instead of the weaker definition from the main text.
\end{definition}

\begin{definition}[Generalized CGM]
    Let $G$ be a ribbon diagram.
    Given $\ba \in \prod_{v \in V^{\square}}[N(v)]$, $\bs \in \prod_{i \in [|\sL|]} [N(\kappa_{\sL}^{-1}(i))]$, and $\bt \in \prod_{j \in |\sR|} [N(\kappa_{\sR}^{-1}(j))]$ such that, for all $i \in \sL \cap \sR$, $s(\kappa_{\sL}(i)) = t(\kappa_{\sR}(j))$, define $f_{\ba, \bs, \bt}: V \to \NN$ by
    \begin{equation}
        f_{\ba, \bs, \bt}(x) = \left\{\begin{array}{ll} s(\kappa_{\sL}(x)) & \text{if } x \in \sL, \\ t(\kappa_{\sR}(x)) & \text{if } x \in \sR, \\ a(x) & \text{if } x \in V^{\square}.\end{array}\right.
    \end{equation}
    Note that we always have $f_{\ba, \bs, \bt}(x) \in [N(x)]$.
    Then, the \emph{contractive graphical matrix} associated to the ribbon diagram $G$ labelled by the matrices $\bM^{(x, y)}$ has entries
    \begin{equation}
        Z^{G}_{\bs, \bt} = \prod_{i \in \sL \cap \sR} \One\{s(\kappa_{\sL}(i)) = t(\kappa_{\sR}(i))\} \sum_{\ba \in \prod_{v \in V^{\square}}[N(v)]} \prod_{\{x, y\} \in E} M^{(x, y)}_{f_{\ba, \bs, \bt}(x), f_{\ba, \bs, \bt}(y)}.
        \label{eq:app-cgm-def}
    \end{equation}
\end{definition}
\noindent
We note that, for the purposes of this appendix, we always will think of CGMs as being labelled by tuples rather than sets.
Since set-indexed CGMs are submatrices of tuple-indexed ones, all norm bounds will immediately be inherited by the set-indexed CGMs encountered in the main text.

The general goal we pursue in the following sections is to develop some general tools for connecting the graphical structure of a ribbon diagram and the matrix structure of its CGM.

\subsection{Connected Components and Tensorization}

We first consider the effect of a diagram being disconnected on the CGM.
In this case, it is simple to see that the expression \eqref{eq:app-cgm-def} factorizes, and therefore the CGM decomposes as a tensor product.
We give a precise statement below, taking into account the ordering of indices.
\begin{proposition}
    \label{prop:cgm-tt-tensorization}
    Let $G = ((\sL \cup \sR) \sqcup V^{\square}, E)$ be a ribbon diagram with connected components $G_1, \dots, G_m$, where $V(G_{\ell}) \cap \sL \equalscolon \sL_{\ell}$ and $V(G_{\ell}) \cap \sR \equalscolon \sR_{\ell}$.
    Define $\kappa_{\sL_{\ell}}: \sL_{\ell} \to [|\sL_{\ell}|]$ and $\kappa_{\sR_{\ell}}: \sR_{\ell} \to [|\sR_{\ell}|]$ to be the labellings inherited from $\kappa_{\sL}$ and $\kappa_{\sR}$, i.e.,
    \begin{align}
  \kappa_{\sL_{\ell}}(i) &\colonequals \#\{i^{\prime} \in \sL_{\ell}: \kappa_{\sL}(i^{\prime}) \leq \kappa_{\sL}(i)\}, \\
  \kappa_{\sR_{\ell}}(j) &\colonequals \#\{j^{\prime} \in \sR_{\ell}: \kappa_{\sR}(j^{\prime}) \leq \kappa_{\sR}(j)\}.
\end{align}
Equipped with these labellings, view $G_1, \dots, G_m$ as ribbon diagrams.
Let $\pi_{\sL} \in \mathsf{Sym}(|\sL|)$ and $\pi_{\sR} \in \mathsf{Sym}(|\sR|)$ be the permutations with
\begin{align*}
(\pi_{\sL}(1), \cdots, \pi_{\sL}(|\sL|)) &= \bigg(\kappa_{\sL}(\kappa_{\sL_1}^{-1}(1)), \dots, \kappa_{\sL}(\kappa_{\sL_1}^{-1}(|\sL_1|)), \dots, \kappa_{\sL}(\kappa_{\sL_m}^{-1}(1)), \dots, \kappa_{\sL}(\kappa_{\sL_m}^{-1}(|\sL_m|))\bigg) \\
(\pi_{\sR}(1), \cdots, \pi_{\sR}(|\sR|)) &= \bigg(\kappa_{\sR}(\kappa_{\sR_1}^{-1}(1)), \dots, \kappa_{\sR}(\kappa_{\sR_1}^{-1}(|\sR_1|)), \dots, \kappa_{\sR}(\kappa_{\sR_m}^{-1}(1)), \dots, \kappa_{\sR}(\kappa_{\sR_m}^{-1}(|\sR_m|))\bigg)
\end{align*}
Let $\sigma_{\sL} \in \mathsf{Sym}([N]^{|\sL|})$ map $(a_1, \dots, a_{|\sL|}) \mapsto (a_{\pi_{\sL}^{-1}(1)}, \dots, a_{\pi_{\sL}^{-1}(|\sL|)})$ and $\sigma_{\sR} \in \mathsf{Sym}([N]^{|\sR|})$ map $(a_1, \dots, a_{|\sR|}) \mapsto (a_{\pi_{\sR}^{-1}(1)}, \dots, a_{\pi_{\sR}^{-1}(|\sR|)})$.
Finally, let $\Pi_{\sL} \in \mathbb{R}^{[N]^{|\sL|} \times [N]^{|\sL|}}$ and $\Pi_{\sR} \in \mathbb{R}^{[N]^{|\sR|} \times [N]^{|\sR|}}$ be the permutation matrices of $\sigma_{\sL}$ and $\sigma_{\sR}$, respectively.
Then,
\[ \bZ^{G} = \Pi_{\sL}\bigg( \bigotimes_{\ell = 1}^m \bZ^{G_{\ell}} \bigg) \Pi_{\sR}^{\top}. \]
\end{proposition}

This fact will be most useful when bounding the difference in operator norm incurred by replacing each connected component $G_i$ by some other diagram in terms of the differences of the smaller CGMs corresponding to each connected component taken in isolation.
\begin{proposition}
    \label{prop:cgm-conn-bd}
    Let $G$ be a ribbon diagram with connected components $G_1, \dots, G_m$, where $V(G_{\ell}) \cap \sL = \sL_{\ell}$ and $V(G_{\ell}) \cap \sR = \sR_{\ell}$.
    Suppose $H_1, \dots, H_m$ are other ribbon diagrams on $(\sL_{\ell}, \sR_{\ell}, V^{\square}(G_{\ell}))$, and write $H$ for the union diagram of the $H_i$.
    Then,
    \begin{equation}
        \|\bZ^{G} - \bZ^{H}\| \leq \sum_{\ell = 1}^m \|\bZ^{G_{\ell}} - \bZ^{H_{\ell}}\| \prod_{\ell^{\prime} = 1}^{\ell - 1}\|\bZ^{H_{\ell}}\| \prod_{\ell^{\prime} = \ell + 1}^m \|\bZ^{G_{\ell}}\|.
    \end{equation}
\end{proposition}
\begin{proof}
    Following the notation of Proposition~\ref{prop:cgm-tt-tensorization}, we can write a telescoping sum,
    \begin{equation}
        \bZ^{G} - \bZ^{H} =  \Pi_{\sL}\left(\sum_{\ell = 1}^m \bigotimes_{\ell^{\prime} = 1}^{\ell - 1} \bZ^{H_{\ell^{\prime}}} \otimes (\bZ^{G_{\ell}} - \bZ^{H_{\ell}}) \otimes \bigotimes_{\ell^{\prime} = \ell + 1}^m \bZ^{G_{\ell^{\prime}}}\right)\Pi_{\sR}
    \end{equation}
    The bound then follows by the triangle inequality and the tensorization of the operator norm.
\end{proof}

\subsection{Splitting}
\label{app:splitting}

We describe two operations on a ribbon diagram, which we call ``splittings,'' that add edges to the diagram without changing the associated CGM.
Later this will allow us to perform some regularizing operations on a ribbon diagram's graph structure when making arguments about its CGM.

The first type of splitting lets us expand a diagram and thereby eliminate the intersection of $\sL$ and $\sR$ by adding redundant vertices and suitable adjacencies.

\begin{proposition}[Intersection splitting]
    Let $G = ((\sL \cup \sR) \sqcup V^{\square}, E)$ be a ribbon diagram with $(\bM^{(x, y)})$ a compatible labelling of its edges.
    Write $\sL \cap \sR = \{v_1, \dots, v_n\}$.
    Let $G^{\prime}$ be another labelled ribbon diagram, formed by adding new vertices $\{v_1^{\prime}, \dots, v_n^{\prime}\}$ to $G$, setting $V^{\square}(G^{\prime}) = V^{\square}(G)$, $\sL(G^{\prime}) = \sL(G)$, and $\sR(G^{\prime}) = \sR(G) \setminus \sL(G) \cup \{v_1^{\prime}, \dots, v_n^{\prime}\}$, and adding edges $\{v_i, v_i^{\prime}\}$ labelled by the matrix $\bm I_{N(v_i)}$ for each $i \in [n]$.
    Then, $\bZ^{G} = \bZ^{G^{\prime}}$.
\end{proposition}

The second type of splitting lets us represent a factorization of a matrix labelling an edge by subdividing that edge with intermediate vertices.

\begin{proposition}[Edge splitting]
    \label{prop:edge-splitting}
    Let $G = ((\sL \cup \sR) \sqcup V^{\square}, E)$ be a ribbon diagram with $(\bM^{(x, y)})$ a compatible labelling of its edges.
    Suppose $x \sim z$ in $G$, and there exist matrices $\bA \in \RR^{N(x) \times n}, \bB \in \RR^{n \times N(z)}$ such that $\bM^{(x, z)} = \bA \bB$.
    Let $G^{\prime}$ be another labelled ribbon diagram, with $\sL(G^{\prime}) = \sL(G)$, $\sR(G^{\prime}) = \sR(G)$, $V^{\square}(G^{\prime}) = V^{\square}(G) \cup \{y\}$ for a new vertex $y$, and $E(G^{\prime}) = E(G) \cup \{\{x, y\}, \{y, z\}\}$.
    Let $\{x, y\}$ in $G^{\prime}$ be labelled with the matrix $\bM^{(x, y)} = \bA$, and $\{y, z\}$ be labelled with the matrix $\bM^{(y, z)} = \bB$.
    Then, $\bZ^{G} = \bZ^{G^{\prime}}$.
\end{proposition}
\noindent
Note that, as an especially useful special case, we may always take $n = N(x)$, $\bA = \bm I_{N(x)}$, and $\bm B = \bM^{(x, z)}$.
This particular technique allows us to adjust the graph of the ribbon diagram without needing to find any special factorization of $\bM^{(x, z)}$.

\subsection{Pinning, Cutting, and Direct Sum Decomposition}
\label{app:pinning}

We next explore special operations that may be performed on the following special type of $\square$ vertex in a ribbon diagram.
\begin{definition}[Pinned vertex]
    In a ribbon diagram $G = ((\sL \cup \sR) \sqcup V^{\square}, E)$ with dimension labels $N: V \to \NN$, we call a vertex $v \in V^{\square}$ \emph{pinned} if $N(v) = 1$.
    In graphical representations, we show a pinned vertex with an unfilled circle.
\end{definition}
\noindent
Any edge one of whose endpoints is a pinned vertex must be labelled with a vector, and in the formula \eqref{eq:app-cgm-def} there is effectively no summation corresponding to a pinned vertex, since $[N(v)] = \{1\}$ whereby the vertex's assigned index is always the same---this is the reason for the term ``pinned.''

In terms of manipulations of the ribbon diagram $G$, the important special property of a pinned vertex is that it allows the diagram to be ``cut'' at that vertex without changing the resulting CGM.
\begin{proposition}[Cutting]
    \label{prop:pinned-vx-cutting}
    Let $G = ((\sL \cup \sR) \sqcup V^{\square}, E)$ be a ribbon diagram, and let $v \in V^{\square}$ be pinned.
    Suppose $\deg(v) = m$, enumerate the neighbors of $v$ as $w_1, \dots, w_m$, and suppose these edges are labelled with vectors $\bm m_i \in \RR^{N(w_i)}$ for $i = 1, \ldots, m$.
    Let $G^{\prime}$ be another ribbon diagram, formed by removing $\bv$ from $G$ and adding new vertices $\bv_1^{\prime}, \dots, \bv_m^{\prime}$ to $V^{\square}$, with $N(v_{i}^{\prime}) = 1$, $v_i^{\prime}$ adjacent to only $\bw_i$, and this edge labelled by $\bm m_i$.
    Then, $\bZ^{G} = \bZ^{G^{\prime}}$.
\end{proposition}
\noindent
Note that, after splitting, every pinned vertex has degree 1.
In our case, when we work with tree ribbon diagrams, this means that every pinned vertex is a leaf of the resulting forest, a property that will be important in our analysis.

Finally, we show two ways that pinned vertices arise naturally from matrix-labelled ribbon diagrams where no vertex has dimension label 1 to begin with.
The first, simpler situation is where an edge is labelled with a rank 1 matrix.
\begin{proposition}
    \label{prop:pinned-vx-rank-one}
    Let $G = ((\sL \cup \sR) \sqcup V^{\square}, E)$ be a ribbon diagram, and suppose $\{v, w\} \in E$ is labelled with $\bM^{(v, w)} = \bx\by^{\top}$.
    Let $G^{\prime}$ be formed by adding a vertex $x$ along this edge between $v$ and $w$, setting $N(x) = 1$, and setting $\bM^{(v, x)} = \bx$ and $\bM^{(x, w)} = \by$.
    Then, $\bZ^{G} = \bZ^{G^{\prime}}$.
\end{proposition}
\begin{proof}
    The result follows from applying Proposition~\ref{prop:edge-splitting} to the edge $\{v, w\}$ using the given rank one factorization.
\end{proof}

The second, perhaps more natural, situation is that the CGM of any ribbon diagram with $\sL \cap \sR \neq \emptyset$ may be written as a direct sum of CGMs with those vertices pinned (that is, as a block-diagonal matrix with these CGMs as the diagonal blocks).
\begin{proposition}[Direct sum decomposition]
    \label{prop:intersection-direct-sum}
    Let $G = ((\sL \cup \sR) \sqcup V^{\square}, E)$ be a ribbon diagram, and enumerate $\sL \cap \sR = \{v_1, \dots, v_m\}$.
    Given $\ba \in \prod_{i = 1}^m [N(v_i)]$, let $G[\ba]$ be the diagram formed by moving each $v_i$ to $V^{\square}$, letting $N(v_i) = 1$, and labelling each edge incident with $v_i$, say $\{v_i, x\}$, with the vector $\bm^{(v_i, x)}$ equal to the $a_i$th row of $\bM^{(v_i, x)}$.
    If there is an edge between $v_i$ and $v_j$, then in $G[\ba]$ it is labelled with the constant $M^{(v_i, v_j)}_{a_ia_j}$.
    Then, there exist permutation matrices $\Pi_{\sL}$ and $\Pi_{\sR}$ such that
    \begin{equation}
        \bZ^{G} = \Pi_{\sL}\left(\bigoplus_{\ba \in \prod_{i = 1}^m [N(v_i)]} \bZ^{G(a)}\right)\Pi_{\sR}^{\top}.
    \end{equation}
\end{proposition}
\noindent
While formally a pinned vertex is a $\square$ vertex, in this setting we see that it behaves more like a $\bullet$ vertex whose index is fixed instead of varying with the matrix indices---this ambiguity is the reason for the unfilled circle notation.

\subsection{Factorization}
We now arrive at perhaps the most useful and important manipulation of CGMs via ribbon diagrams.
Namely, certain graphical decompositions of ribbon diagrams correspond to factorizations of CGMs into products of simpler CGMs.

\begin{proposition}
    \label{prop:factorization}
    Let $G = (V, E)$ be a ribbon diagram with $V = (\sL \cup \sR) \sqcup V^{\square}$.
    Suppose that $V$ admits a partition $V = A \sqcup B \sqcup C$, such that the following properties hold:
    \begin{enumerate}
    \item $\sL \subseteq A$.
    \item $\sR \subseteq C$.
    \item $\partial_{\tout} A, \partial_{\tout}C \subseteq B$.
    \end{enumerate}
    (Here $\partial_{\tout}$ denotes the ``outer boundary'' of a set, those vertices not in the set but with a neighbor in the set.)
    Suppose also that the edges within $B$ admit a partition $E(B, B) = E^{A} \sqcup E^{B} \sqcup E^C$, where $E^{A} \subseteq E(\partial_{\tout}A, \partial_{\tout}A)$ and $E^{C} \subseteq E(\partial_{\tout}C, \partial_{\tout}C)$.
    Define the following ancillary ribbon diagrams:
    \begin{enumerate}
    \item $G[A]$ with vertex triple $(\sL, \partial_{\tout}A, A \setminus \sL)$ and edges $E(A, A \cup \partial_{\tout}A) \cup E^A$;
    \item $G[B]$ with vertex triple $(\partial_{\tout}A, \partial_{\tout}C, B \setminus \partial_{\tout}A \setminus \partial_{\tout}C)$ and edges $E^B$; and
    \item $G[C]$ with vertex triple $(\partial_{\tout}C, \sR, C \setminus \sR)$ and edges $E(C, C \cup \partial_{\tout}C) \cup E^C$.
    \end{enumerate}
    In these diagrams, $\partial_{\tout}A$ and $\partial_{\tout}C$ are given arbitrary labellings, but the same labelling each time that they appear in different diagrams.
    Then,
    \begin{equation}
        \bZ^{G} = \bZ^{G[A]} \bZ^{G[B]} \bZ^{G[C]}.
    \end{equation}
\end{proposition}
\begin{proof}
    The proof is a direct verification by expanding the matrix multiplications and definitions of the CGMs involved.
    Note that, by assumption, since $A$ and $C$ are disjoint and $\sL \subseteq A$ and $\sR \subseteq C$, we must have $\sL \cap \sR = \emptyset$.
    We have
    \begin{equation}
        (\bZ^{G[A]} \bZ^{G[B]} \bZ^{G[C]})_{\bs, \bt} = \sum_{\substack{\ba \in [N]^{\partial_{\tout}A} \\ \bc \in [N]^{\partial_{\tout}C}}} Z^{G[A]}_{\bs, \ba} Z^{G[B]}_{\ba, \bc} Z^{G[C]}_{\bc, \bt}.
    \end{equation}

    Given $\ba \in [N]^{\partial_{\tout}A}, \ba^{\prime} \in [N]^{A \setminus \sL}, \bm b \in [N]^{B \setminus \partial_{\tout} A \setminus \partial_{\tout} C}, \bc \in [N]^{\partial_{\tout}C}, \bc^{\prime} \in [N]^{C \setminus \sR}$ such that for $x \in \partial_{\tout}A \cap \partial_{\tout} C$ we have $a(x) = c(x)$, let us define $g = g_{\ba, \ba^{\prime}, \bm b, \bc, \bc^{\prime}, \bs, \bt}: V^{\square} \to [N]$ by
    \begin{equation}
        g(x) \colonequals\left\{\begin{array}{ll}s(\kappa_{\sL}(x)) & \text{if } x \in \sL, \\ t(\kappa_{\sR}(x)) & \text{if } x \in \sR, \\ a(x) & \text{if } x \in \partial_{\tout}A, \\ a^{\prime}(x) & \text{if } x \in A \setminus \sL, \\ b(x) & \text{if } x \in B \setminus \partial_{\tout} A \setminus \partial_{\tout} C, \\ c(x) & \text{if } x \in \partial_{\tout} C, \\ c^{\prime}(x) & \text{if } x \in C \setminus R.  \end{array}\right.
    \end{equation}
    We then compute
    \begin{align}
      (\bZ^{G[A]} \bZ^{G[B]} \bZ^{G[C]})_{\bs, \bt} &= \sum_{\substack{\ba \in [N]^{\partial_{\tout}A} \\ \ba^{\prime} \in [N]^{A \setminus \sL} \\ \bc \in [N]^{\partial_{\tout}C} \\ \bc^{\prime} \in [N]^{C \setminus \sR}}} \prod_{\{x, y\} \in E(A, A \cup \partial_{\tout} A) \cup E^A} M^{(x, y)}_{g(x), g(y)} \prod_{\{x, y\} \in E(C, C \cup \partial_{\tout}C) \cup E^C} M^{(x, y)}_{g(x), g(y)}\nonumber \\ &\hspace{2.73cm} \prod_{x \in \partial_{\tout}A \cap \partial_{\tout}C} \One\{a(x) = c(x)\} \prod_{\{x, y\} \in E^{B}} M^{(x, y)}_{g(x), g(y)} \nonumber \\
                              &= \sum_{\ba \in [N]^{V^{\square}}} \prod_{\{x, y\} \in E} M^{(x, y)}_{f_{\ba, \bs, \bt}(x),f_{\ba, \bs, \bt}(y)} \nonumber \\
                              &= Z^{G}_{\bs, \bt},
    \end{align}
    completing the proof.
\end{proof}

\subsection{General-Purpose Norm Bounds}

Our first application is to prove general-purpose bounds on the norms of CGMs based on ribbon diagram structure and the norms of constituent labelling matrices.

First, we show that norms multiply over connected components, as we have alluded to in Remark~\ref{rem:multiscale-spectrum} in the main text.
This is a direct application of Proposition~\ref{prop:cgm-tt-tensorization}.
\begin{proposition}
    \label{prop:norm-tensorization}
    Let $G$ be a ribbon diagram with connected components $G_1, \dots, G_{m}$, as in Proposition~\ref{prop:cgm-tt-tensorization}.
    Then, $\bZ^{G} = \prod_{\ell = 1}^m \|\bZ^{G_{\ell}}\|$.
\end{proposition}
\noindent
The following bound is less trivial and is used repeatedly in our arguments.
\begin{proposition}
    \label{prop:cgm-generic-norm-bound}
    Let $G = (V, E)$ be a ribbon diagram with $V = (\sL \cup \sR) \sqcup V^{\square}$.
    Suppose that $V$ admits a partition $V = V_1 \sqcup \cdots \sqcup V_m$ with $m \geq 2$, such that $V_1 = \sL$, $V_m = \sR$, and the following properties hold:
    \begin{enumerate}
    \item For every $v \in V_1 = \sL$, there exists some $k > 1$ such that $v$ has a neighbor in $V_k$.
    \item Every $v \in V_m = \sR$, there exists some $k < m$ such that $v$ has a neighbor in $V_k$.
    \item For every $1 < j < m$ and every $v \in V_j$, there exist $i < j < k$ such that $v$ has a neighbor in $V_i$ and a neighbor in $V_k$.
    \end{enumerate}
    Then,
    \begin{equation}
        \|\bZ^{G}\| \leq \prod_{\{x, y\} \in E} \|\bM^{(x, y)}\|.
        \label{eq:cgm-generic-norm-bound}
    \end{equation}
\end{proposition}
\begin{proof}
    Note that, by repeatedly applying Proposition~\ref{prop:edge-splitting} with edges labelled by an identity matrix, we may furthermore assume without loss of generality that every edge of $G$ is either between two vertices of $V_i$, or between one vertex of $V_i$ and one vertex of $V_{i + 1}$ for some $i$.
    Under this assumption, the three conditions in the statement may be rewritten as follows:
    \begin{enumerate}
    \item For every $v \in V_1 = \sL$, $v$ has a neighbor in $V_2$.
    \item Every $v \in V_m = \sR$, $v$ has a neighbor in $V_{m - 1}$.
    \item For every $1 < j < m$ and every $v \in V_j$, $v$ has a neighbor in $V_{j - 1}$ and a neighbor in $V_{j + 1}$.
    \end{enumerate}
    
    Next, we proceed by induction on $m$.
    Suppose first that $m = 2$.
    Then, the assumptions imply that $V^{\square} = \emptyset$ and $\sL \cap \sR = \emptyset$.
    Let us enumerate the edges within $\sL$, within $\sR$, and between $\sL$ and $\sR$ as follows:
    \begin{align}
      E(\sL, \sL) &= \{\{i^{(1)}_1, i^{(2)}_1\}, \dots, \{i^{(1)}_a, i^{(2)}_a\}\} \text{ for } i^{(k)}_{\ell} \in \sL, \\
      E(\sR, \sR) &= \{\{j^{(1)}_1, j^{(2)}_1\}, \dots, \{j^{(1)}_b, j^{(2)}_b\}\} \text{ for } j^{(k)}_{\ell} \in \sR, \\
      E(\sL, \sR) &= \{\{i_1, j_1\}, \dots, \{i_c, j_c\}\} \text{ for } i_{\ell} \in \sL \text{ and } j_{\ell} \in \sR.
    \end{align}
    Then, we have
    \begin{equation}
        Z^{G}_{\bs, \bt} = \prod_{\ell = 1}^a M^{(i^{(1)}_{\ell}, i^{(2)}_{\ell})}_{s(\kappa_{\sL}(i_{\ell}^{(1)})),s(\kappa_{\sL}(i_{\ell}^{(2)}))} \prod_{\ell = 1}^b M^{(j^{(1)}_{\ell}, j^{(2)}_{\ell})}_{t(\kappa_{\sR}(j_{\ell}^{(1)})),t(\kappa_{\sR}(j_{\ell}^{(2)}))} \prod_{\ell = 1}^c M^{(i_{\ell}, j_{\ell})}_{s(\kappa_{\sL}(i_{\ell})),t(\kappa_{\sR}(j_{\ell}))}.
    \end{equation}
    Let us define an ancillary matrix
    \begin{equation}
        \widehat{Z}^G_{\bs, \bt} \colonequals \prod_{\ell = 1}^c M^{(i_{\ell}, j_{\ell})}_{s(\kappa_{\sL}(i_{\ell})),t(\kappa_{\sR}(j_{\ell}))}.
    \end{equation}
    Then, we may write $\bZ^G = \bD^{\sL} \widehat{\bZ}^G \bD^{\sR}$, for suitable diagonal matrices $\bD^{\sL}$ and $\bD^{\sR}$ having entries equal to the first two products above, respectively.
    Since every entry of a matrix is bounded by the matrix norm, we then have
    \begin{equation}
        \|\bZ^G\| \leq \|\bD^{\sL}\| \cdot \|\widehat{\bZ}^G\| \cdot \|\bD^{\sR}\| \leq \|\widehat{\bZ}^G\|\prod_{\{x, y\} \in E(\sL, \sL) \cup E(\sR, \sR)} \|\bM^{(x, y)}\|.
    \end{equation}
    
    For the remaining factor, by taking a singular value decomposition, we can factorize each labelling matrix as $\bM^{(x, y)} = \bU^{(x, y)^{\top}}\bV^{(x, y)}$ such that $\|\bM^{(x, y)}\| = \|\bU^{(x, y)}\| \cdot \|\bV^{(x, y)}\|$ (by including the factor of the singular values in either of the singular vector matrices).
    Writing $\bu^{(x, y, i)}$ for the columns of $\bU^{(x, y)}$ and $\bv^{(x, y, i)}$ for the columns of $\bV^{(x, y)}$, we then have $M^{(x, y)}_{ij} = \langle \bu^{(x, y, i)}, \bv^{(x, y, j)} \rangle$.
    Therefore,
    \begin{equation}
        \widehat{Z}^G_{\bs, \bt} = \prod_{\ell = 1}^c \left\langle \bu^{(i_{\ell}, j_{\ell},s(\kappa_{\sL}(i_{\ell})))}, \bv^{(i_{\ell}, j_{\ell}, t(\kappa_{\sR}(j_{\ell})))} \right\rangle = \left\langle \bigotimes_{\ell = 1}^c\bu^{(i_{\ell}, j_{\ell},s(\kappa_{\sL}(i_{\ell})))}, \bigotimes_{\ell = 1}^c\bv^{(i_{\ell}, j_{\ell}, t(\kappa_{\sR}(j_{\ell})))}\right\rangle.
    \end{equation}
    This writes $\widehat{\bZ}^G = \bU^{G^{\top}}\bV^G$, so we have $\|\widehat{\bZ}^G\| \leq \|\bU^G\| \cdot \|\bV^G\|$.
    We then compute
    \begin{align}
      (\bU^{G^{\top}}\bU^G)_{\bs, \bs^{\prime}}
      &= \prod_{\ell = 1}^c (\bU^{(i_{\ell}, j_{\ell})^{\top}} \bU^{(i_{\ell}, j_{\ell})})_{s(\kappa_{\sL}(i_{\ell})), s^{\prime}(\kappa_{\sL}(i_{\ell}))} \nonumber \\
      &= \prod_{i \in \sL} \prod_{j \sim i} (\bU^{(i, j)^{\top}} \bU^{(i, j)})_{s(\kappa_{\sL}(i)), s^{\prime}(\kappa_{\sL}(i))}.
    \end{align}
    Thus, $\bU^{G^{\top}}\bU^G$ is the tensor product over $i \in \sL$ of the Hadamard products over $j \sim i$ of $\bU^{(i, j)^{\top}} \bU^{(i, j)}$.
    Since the operator norm is multiplicative over tensor products and submultiplicative over Hadamard products, and every $i \in \sL$ has a neighbor $j \in \sR$, we find
    \begin{equation}
        \|\bU^{G^{\top}}\bU^G\| \leq \prod_{i \in \sL} \prod_{j \sim i} \|\bU^{(i, j)^{\top}} \bU^{(i, j)}\|,
    \end{equation}
    whereby
    \begin{equation}
        \|\bU^G\| \leq \prod_{i \in \sL} \prod_{j \sim i} \|\bU^{(i, j)}\| = \prod_{\{i, j\} \in E} \|\bU^{(i, j)}\|.
    \end{equation}
    Repeating the same argument for $\bV^G$ and multiplying the results together, we have
    \begin{equation}
        \|\bZ^G\| \leq \|\bU^G\| \cdot \|\bV^G\| \leq \prod_{\{i, j\} \in E} \|\bU^{(i, j)}\| \cdot \|\bV^{(i, j)}\| = \prod_{\{i, j\} \in E} \|\bM^{(i, j)}\|,
    \end{equation}
    completing the argument for $m = 2$.

    For the inductive step, if we have the result for $m$ and are given a decomposition of $G$ into $m + 1$ sets, the result follows by applying the factorization of Proposition~\ref{prop:factorization} with $A = V_1$, $B = V_2$, and $C = V_3 \sqcup \cdots \sqcup V_m$, and using that $\|\bZ^{G[A]}\|$ and $\|\bZ^{G[B]}\|$ may be bounded using the $m = 2$ case.
\end{proof}

\noindent
To see that the connectivity requirements are important for this argument, one may consider the simple case where $G$ has isolated vertices in $\sL$ or $\sR$: if so, then the associated CGM is the tensor product of an all-ones matrix with the CGM associated to $G$ with the isolated vertices removed.
The norm of this all-ones matrix is polynomial in $N$, whereby the best bound of the type \eqref{eq:cgm-generic-norm-bound} that we could hope for would depend on $N$, spoiling many of our applications.
Other cases where the connectivity requirements fail reduce to a similar situation after sufficiently many applications of the factorization of Proposition~\ref{prop:factorization}.

\section{Tying Ribbon Diagrams: Norm Bounds}

\subsection{Stretched Forest Ribbon Diagrams: Proof of Lemma~\ref{lem:tying-stretched-ribbon}}
\label{app:pf:lem:tying-stretched-ribbon}

\begin{proof}[Proof of Lemma~\ref{lem:tying-stretched-ribbon}]
    We recall the statement of the result.
    Let $F \in \sF(\ell, m)$ be a stretched forest ribbon diagram, with all edges labelled with $\bM$.
    Let $\tie(F)$ be the forest ribbon diagram constructed by tying every connected component of $F$ that is not a pair into a bowtie.
    By construction, $\tie(F)$ is a bowtie forest ribbon diagram.
    Our goal is then to show the bound
    \begin{equation}
        \|\bZ^F - \bZ^{\tie(F)}\| \leq (\ell + m) (2\|\bM\|)^{\frac{3}{2}(\ell + m)} \epsilon_{\tree}(\bM; (\ell + m) / 2).
    \end{equation}

    Let us first suppose that $T \in \sT(\ell, m)$ is connected.
    We will then show the bound
    \begin{equation}
        \|\bZ^T - \bZ^{\tie(T)}\| \leq 2^{\frac{3}{2}(\ell + m)}\epsilon_{\tree}(\bM; (\ell + m) / 2),
    \end{equation}
    the same as the above but without the leading factor of $(\ell + m)$.
    Since the norm of the CGM of a connected component of $F$ that has $k$ leaves is at most $\|\bM\|^{\frac{3}{2}k}$ by Proposition~\ref{prop:cgm-generic-norm-bound} and Corollary~\ref{cor:sF-vertex-edge-bound}, the bound on arbitrary $F$ will then follow by applying Proposition~\ref{prop:cgm-conn-bd}.

    Write $U \subseteq V^{\square}(T)$ for the set of terminal vertices of $T$.
    By Corollary~\ref{cor:sF-vertex-edge-bound}, $|U| \leq |V^{\square}(T)| \leq (\ell + m) / 2$.
    Recall that, since $T$ is stretched, every vertex of $U$ is adjacent to some vertex of $\sL$ and some vertex of $\sR$.
    Let $A$ be the ribbon diagram on vertex triplet $((\sL, U), \emptyset)$ induced by $T$, let $B$ be the ribbon diagram on vertex triplet $((U, U), \emptyset)$ obtained by deleting $\sL$ and $\sR$ from $T$, and let $C$ be the ribbon diagram on vertex triplet $((U, \sR), \emptyset)$ induced by $T$.
    Then, by Proposition~\ref{prop:factorization}, we may factorize $\bZ^T = \bZ^{G[A]} \bZ^{G[B]} \bZ^{G[C]}$.
    Moreover, let $G^{\prime}[B]$ be the ribbon diagram obtained by relabeling every edge in $G[B]$ with the identity matrix.
    Then we have $\bZ^{\tie(T)} = \bZ^{G[A]} \bZ^{G^{\prime}[B]} \bZ^{G[C]}$.
    Therefore, by norm submultiplicativity and Proposition~\ref{prop:cgm-generic-norm-bound} applied to $G[A]$ and $G[C]$, we may bound
    \begin{equation}
        \|\bZ^T - \bZ^{\tie(T)}\| \leq \|\bZ^{G[A]}\| \cdot \|\bZ^{G[C]}\| \cdot \|\bZ^{G[B]} - \bZ^{G^{\prime}[B]}\| \leq \|\bM\|^{\ell + m} \|\bZ^{G[B]} - \bZ^{G^{\prime}[B]}\|.
    \end{equation}

    Since $\sL(G[B]) = \sR(G[B]) = \sL(G^{\prime}[B]) = \sR(G^{\prime}[B]) = U$, the matrices $\bZ^{G[B]}$ and $\bZ^{G^{\prime}[B]}$ are diagonal.
    Let us view $G[B]$ as being partitioned into edge-disjoint (but not necessarily vertex disjoint) subtrees $B_1, \dots, B_n$, such that every leaf of every $B_i$ belongs to $U$.
    Since the $B_i$ are edge-disjoint, $n$ is at most the number of edges in $T$, which by Corollary~\ref{cor:sF-vertex-edge-bound} is at most $\frac{3}{2}(\ell + m)$.
    Since $\sL(G[B]) = \sR(G[B])$, the labelings $\kappa_{\sL}$ and $\kappa_{\sR}$ must be equal, so let us simply write $\kappa$ for this single labeling.
    Suppose that the leaves of $B_i$ are $\ell_{i, 1}, \dots, \ell_{i, a_i} \in U$.
    Then, we have
    \begin{align}
      Z^{G[B]}_{\bs, \bs} = \prod_{i = 1}^n Z^{B_i}\bigg(\bM; \big(s(\kappa(\ell_{i, 1})), \dots, s(\kappa(\ell_{i, a_i}))\big)\bigg).
    \end{align}
    By the definition of $\epsilon_{\tree}$, we have
    \begin{align}
      \Bigg| \One\{s(\kappa(\ell_{i, 1}))& = \cdots = s(\kappa(\ell_{i, a_i}))\} - Z^{B_i}\bigg(\bM; \big(s(\kappa(\ell_{i, 1})), \dots, s(\kappa(\ell_{i, a_i}))\big)\bigg)\Bigg| \nonumber \\
                                         &\leq \epsilon_{\tree}(\bM; a_i) \nonumber \\
      &\leq \epsilon_{\tree}(\bM; (\ell + m) / 2).
    \end{align}
    Since the $B_i$ form an edge partition of $T$ into subtrees, we have
    \begin{equation}
        \prod_{i = 1}^n \One\{s(\kappa(\ell_{i, 1})) = \cdots = s(\kappa(\ell_{i, a_i}))\} = \One\{s(i) = s(j) \text{ for all } i, j \in [|U|]\} = Z^{G^{\prime}[B]}_{\bs, \bs}.
    \end{equation}
    Therefore, substituting and expanding, we find
    \begin{align}
      \bigg|Z^{G[B]}_{\bs, \bs} &- Z^{G^{\prime}[B]}_{\bs, \bs} \bigg| \nonumber \\
      &= \left| \prod_{i = 1}^n Z^{B_i}\bigg(\bM; \big(s(\kappa(\ell_{i, 1})), \dots, s(\kappa(\ell_{i, a_i}))\big)\bigg) - \prod_{i = 1}^n \One\{s(\kappa(\ell_{i, 1})) = \cdots = s(\kappa(\ell_{i, a_i}))\} \right| \nonumber \\
      &= \Bigg| \sum_{\substack{A \subseteq [n] \\ A \neq \emptyset}} \prod_{i \in A} \left(Z^{B_i}\bigg(\bM; \big(s(\kappa(\ell_{i, 1})), \dots, s(\kappa(\ell_{i, a_i}))\big)\bigg) - \One\{s(\kappa(\ell_{i, 1})) = \cdots = s(\kappa(\ell_{i, a_i}))\}\right) \nonumber \\
      &\hspace{1.55cm} \prod_{i \notin A} \One\{s(\kappa(\ell_{i, 1})) = \cdots = s(\kappa(\ell_{i, a_i}))\} \Bigg| \nonumber \\
      &\leq \sum_{\substack{A \subseteq [n] \\ A \neq \emptyset}} \epsilon_{\tree}(\bM; (\ell + m) / 2)^{|A|} \nonumber \\
      &\leq 2^n \epsilon_{\tree}(\bM; (\ell + m) / 2) \nonumber \\
      &\leq 2^{\frac{3}{2}(\ell + m)}\epsilon_{\tree}(\bM; (\ell + m) / 2),
    \end{align}
    completing the proof.
\end{proof}

\subsection{Partition Transport Ribbon Diagrams: Proof of Lemma~\ref{lem:tying-partition-transport-ribbon}}
\label{app:pf:lem:tying-partition-transport-ribbon}

\begin{proof}[Proof of Lemma~\ref{lem:tying-partition-transport-ribbon}]
    We recall the statement: let $\sigma, \tau \in \Part([d])$, let $\bD \in \Plans(\sigma, \tau)$, and let $G = G(\sigma, \tau, \bD)$ be the associated partition transport ribbon diagram.
    We will show the slightly stronger bound,
    \begin{equation}
      \|\bZ^G - \bZ^{\tie(G)}\| \leq d\|\bM\|^{3d}\bigg(\epsilon_{\pow}(\bM) + d\epsilon_{\offdiag}(\bM) + d\epsilon_{\corr}(\bM)\bigg).
    \end{equation}

    As in the previous proof, let us first suppose that $G$ is connected.
    We will then show the bound
    \begin{align}
      &\|\bZ^G - \bZ^{\tie(G)}\| \leq \|\bM\|^{3d}\bigg(\epsilon_{\pow}(\bM) + d\epsilon_{\offdiag}(\bM) + d\epsilon_{\corr}(\bM)\bigg),
    \end{align}
    the same as the above but without the leading factor of $d$.
    Since there are exactly $3d$ edges in any connected component of a partition transport ribbon diagram such that the component has $2d$ leaves, the final result will then follow by applying Propositions~\ref{prop:cgm-generic-norm-bound} and \ref{prop:cgm-conn-bd}.

    Let us write $\partial^{\square} \sL, \partial^{\square} \sR \subset V^{\square}$ for the sets of $\square$ vertices with a neighbor in $\sL$ and $\sR$, respectively.
    We have $\partial^{\square} \sL \cup \partial^{\square} \sR = V^{\square}$, and we observe that $\partial^{\square} \sL \cap \partial^{\square} \sR = \emptyset$ if and only if both $\sigma$ and $\tau$ contain no singletons.

    \vspace{1em}

    \emph{Case 1: $d = 1$.} This is only possible if $\sigma = \tau = \{\{1\}\}$ and $\bD = [1]$.
    In this case, $G = \tie(G)$ (both consist of two leaves connected to one another), so $\|\bZ^G - \bZ^{\tie(G)}\| = 0$.

    \vspace{1em}

    \emph{Case 2: $\partial^{\square} \sL \cap \partial^{\square} \sR \neq \emptyset$ and $d \geq 2$.}
    We argue by induction on $|V^{\square}(G)|$ that the following bound holds in this case:
    \begin{equation}
        \|\bZ^{G} - \bZ^{\tie(G)}\| \leq \|\bM\|^{3d}\cdot |V^{\square}(G)| \cdot \epsilon_{\offdiag}(\bM)
        \label{eq:partition-transport-tying-bound-2-intermediate}
    \end{equation}
    If $|V^{\square}(G)| = 1$, then $G = \tie(G)$, so $\|\bZ^G - \bZ^{\tie(G)}\| = 0$.
    Now, suppose we have established the result for all partition transport ribbon diagrams satisfying the assumptions of this case with $|V^{\square}| \leq m$, and have $|V^{\square}(G)| = m + 1 > 1$.
    Let $v \in \partial^{\square} \sL \cap \partial^{\square} \sR$.
    We apply the factorization of Proposition~\ref{prop:factorization} with $A = \sL$, $B = V^{\square}$, and $C = \sR$, which factorizes $\bZ^G = \bZ^{G[A]}\bZ^{G[B]}\bZ^{G[C]}$.
    Since $v \in \sL(G[B]) \cap \sR(G[B])$, $\bZ^{G[B]}$ is the direct sum of $N$ further CGMs where $v$ is pinned to each possible value in $[N]$.
    Let $G^{\prime}[B]$ denote the ribbon diagram formed from $G[B]$ by relabelling all edges between $v$ and all of its neighbors with the identity matrix.
    When $v$ is pinned, the factors contributed by these edges are constant factors in each direct summand CGM, so we have $\|\bZ^{G[B]} - \bZ^{G^{\prime}[B]}\| \leq \|\bM\|^{|E(G[B])|}\epsilon_{\offdiag}(\bM)$.
    Now, let $G^{\prime}$ denote the ribbon diagram formed from $G$ by contracting all edges between $v$ and all of its neighbors in $V^{\square}$.
    Then, we have $\bZ^{G^{\prime}} = \bZ^{G[A]}\bZ^{G^{\prime}[B]} \bZ^{G[C]}$, so $\|\bZ^{G^{\prime}} - \bZ^G\| \leq \|\bM^{|E(G)|}\epsilon_{\offdiag}(\bM) \leq \|\bM\|^{3d}\epsilon_{\offdiag}(\bM)$ by Proposition~\ref{prop:cgm-generic-norm-bound} and Corollary~\ref{cor:sF-vertex-edge-bound}.
    Since $\tie(G^{\prime}) = \tie(G)$, we find
    \begin{align}
      \|\bZ^G - \bZ^{\tie(G)}\|
      &\leq \|\bZ^G - \bZ^{G^{\prime}}\| + \|\bZ^{G^{\prime}} - \bZ^{\tie(G^{\prime})}\| \nonumber \\
      &\leq \|\bM\|^{3d}\epsilon_{\offdiag}(\bM) + \|\bM^{3d}\| \cdot |V^{\square}(G^{\prime})| \cdot \epsilon_{\offdiag}(\bM) \tag{by inductive hypothesis}
    \end{align}
    Since $|V^{\square}(G^{\prime})| < |V^{\square}(G)|$ by construction, the proof of \eqref{eq:partition-transport-tying-bound-2-intermediate} is complete.

    Lastly, since each $\square$ vertex of $G$ corresponds to a part of $\sigma$ or $\tau$ having size at least 2, we have $|V^{\square}(G)| \leq d$, so we obtain the simpler version
    \begin{equation}
        \|\bZ^{G} - \bZ^{\tie(G)}\| \leq \|\bM\|^{3d}d\epsilon_{\offdiag}(\bM).
        \label{eq:partition-transport-tying-bound-2}
    \end{equation}

    \emph{Case 3: $\partial^{\square} \sL \cap \partial^{\square} \sR = \emptyset$, $d \geq 2$, and a row or column of $\bD$ has only one non-zero entry.}
    Let us suppose without loss of generality that it is a row of $\bD$ that has the specified property, which corresponds to a part $S \in \sigma$.
    Let $v \in V^{\square}$ be the associated $\square$ vertex in $G$.
    The given condition means that $v$ has only one neighbor in $\partial^{\square} R$, which we call $w$, and that there are $|S| \geq 2$ parallel edges between $v$ and $w$.
    Let $G^{\prime}$ denote the diagram where $v$ and $w$ are identified.
    Then, by Proposition~\ref{prop:cgm-generic-norm-bound}, we have
    \begin{equation}
        \|\bZ^G - \bZ^{G^{\prime}}\| \leq \|\bM\|^{3d} \epsilon_{\pow}(\bM).
    \end{equation}
    We note that $\tie(G^{\prime}) = \tie(G)$, and Case 2 applies to $G^{\prime}$ (indeed, $G^{\prime}$ is the partition transport ribbon diagram formed by replacing the part $S$ of $\sigma$ with $|S|$ singletons that are all transported to the part of $\tau$ corresponding to $w$).
    Therefore, using that result, we conclude
    \begin{equation}
        \|\bZ^G - \bZ^{\tie(G)}\| \leq \|\bM\|^{3d}(\epsilon_{\pow}(\bM) + d\epsilon_{\offdiag}(\bM)).
        \label{eq:partition-transport-tying-bound-3}
    \end{equation}

    \vspace{1em}

    \emph{Case 4: $\partial^{\square} \sL \cap \partial^{\square} \sR = \emptyset$, $d \geq 2$, and no row or column of $\bD$ has only one non-zero entry.}
    We argue by induction on $|V^{\square}(G)|$ that the following bound holds:
    \begin{equation}
        \|\bZ^G - \bZ^{\tie(G)}\| \leq \|\bM\|^{3d}(\epsilon_{\pow}(\bM) + |V^{\square}(G)|\epsilon_{\corr}(\bM)).
    \end{equation}
    We cannot have $|V^{\square}(G)| = 1$, since then the single $\square$ vertex would need to belong to $\partial^{\square} \sL \cap \partial^{\square} \sR$.
    Thus the base case is $|V^{\square}(G)| = 2$.
    In this case, $G$ consists of two $\square$ vertices, each connected to $d$ leaves, and with $d$ parallel edges between them.
    We apply the factorization of Proposition~\ref{prop:factorization} with $A = \sL$, $B = V^{\square}$, and $C = \sR$.
    Then, $G[B]$ consists of two vertices, one in $\sL(G[B])$ and one in $\sR(G[B])$, connected by $d$ parallel edges.
    Writing $G^{\prime}[B]$ for the diagram where these edges are replaced by a single one labelled with the identity, we then have $\|\bZ^{G[B]} - \bZ^{G^{\prime}[B]}\| = \|\bM^{\circ d} - \bm I_N\| \leq \epsilon_{\pow}(\bM)$.
    And, since $\bZ^{\tie(G)} = \bZ^{G[A]} \bZ^{G^{\prime}[B]} \bZ^{G[C]}$, using Proposition~\ref{prop:cgm-generic-norm-bound} we find
    \begin{equation}
        \|\bZ^G - \bZ^{\tie(G)}\| \leq \|\bM\|^{3d}\epsilon_{\pow}(\bM).
    \end{equation}
    (We could reduce the exponent to $2d$ here, but accept the insignificant additional slack to make the final expression cleaner.)

    Now, suppose we have established the result for all partition transport ribbon diagrams $G$ with $|V^{\square}(G)| \leq m$, and have $G$ with $|V^{\square}(G)| = m + 1 > 2$.
    Since $G$ is connected, has $\partial^{\square} \sL \cap \partial^{\square} \sR = \emptyset$, and $d > 1$, all parts of $\sigma$ and $\tau$ must have size at least 2.
    Moreover, since $|V^{\square}(G)| > 2$, we must have either $|\sigma| > 1$ or $|\tau| > 1$.
    Let us suppose, without loss of generality, that $|\sigma| > 1$ (otherwise we may reverse the roles of $\sigma$ and $\tau$, which amounts to transposing the ribbon diagram $G$ and the associated CGM).

    Choose any part $S \in \sigma$, and let $v \in V^{\square}(G)$ be the associated $\square$ vertex.
    We apply the factorization of Proposition~\ref{prop:factorization} with $A = \sL \cup (\partial^{\square} \sL \setminus \{v\})$, $B = \partial^{\square} \sR \cup \{v\}$, and $C = \sR$.
    Consider the diagram $G[B]$.
    We have $\sR(G[B]) = \partial^{\square} \sR$ and $V^{\square}(G[B]) = \emptyset$.
    Moreover, by the assumption of this case, every vertex of $\partial^{\square} \sR$ has a neighbor in $\sL \setminus \{v\}$.
    Therefore, $\sL(G[B]) = \partial^{\square} \sR \cup \{v\}$.
    In particular, $\sR(G[B]) \subseteq \sL(G[B])$, so $\sR(G[B]) = \sL(G[B]) \cap \sR(G[B]) = \partial^{\square} \sR$.

    Following the pinning transformation of Proposition~\ref{prop:intersection-direct-sum}, after a suitable permutation, the CGM of $G[B]$ will then be the direct sum of column vectors $\bv_{\bs} \in \RR^N$, indexed by $\bs \in [N]^{\partial^{\square} \sR}$, where
    \begin{equation}
        (\bv_{\bs})_{i} = \prod_{x \in \partial^{\square} v} M_{i,s(x)},
    \end{equation}
    where the product is understood to repeat vertices $x$ when $v$ is connected to $x$ with multiple edges.
    In particular, we have
    \begin{equation}
        \|\bv_{\bs}\|^2_2 = \sum_{i = 1}^N \prod_{x \in \partial v} M_{i,s(x)}^2,
    \end{equation}
    and since, again by the assumption of this case, $v$ has at least two different neighbors in $\partial^{\square} \sR$, we have $\|\bv_{\bs}\|_2 \leq \epsilon_{\corr}(\bM)$ whenever $\bs$ is not constant on $\partial v$.
    Thus letting $G^{\prime}$ be the diagram formed from $G$ by identifying all neighbors of $v$, and applying Proposition~\ref{prop:cgm-generic-norm-bound}, we find that
    \begin{equation}
        \|\bZ^G - \bZ^{G^{\prime}}\| \leq \|\bM\|^{3d}\epsilon_{\corr}(\bM),
    \end{equation}
    so by the inductive hypothesis,
    \begin{equation}
        \|\bZ^G - \bZ^{\tie(G)}\| \leq \|\bM\|^{3d}\epsilon_{\corr}(\bM) + \|\bM\|^{3d}(\epsilon_{\pow}(\bM) + |V^{\square}(G^{\prime})| \epsilon_{\corr}(\bM)),
  \end{equation}
  and since $|V^{\square}(G^{\prime})| < |V^{\square}(G)|$, the induction is complete.
  Finally, using again that $|V^{\square}(G)| \leq d$ since all parts of $\sigma$ and $\tau$ have size at least 2 in this case, we find the looser bound
    \begin{equation}
        \|\bZ^G - \bZ^{\tie(G)}\| \leq \|\bM\|^{3d}(\epsilon_{\pow}(\bM) + d\epsilon_{\corr}(\bM)).
        \label{eq:partition-transport-tying-bound-4}
    \end{equation}

    \emph{Conclusion.} In the four cases considered above, we have shown that either $\|\bZ^G - \bZ^{\tie(G)}\|$ is zero, or is bounded as in \eqref{eq:partition-transport-tying-bound-2}, \eqref{eq:partition-transport-tying-bound-3}, and \eqref{eq:partition-transport-tying-bound-4}.
    We then observe that the ``master bound'' in the statement is larger than each of these, completing the proof.
\end{proof}

\section{Tying Ribbon Diagrams: Combinatorial Reductions}

\subsection{Stretched Forest Ribbon Diagrams: Proof of Lemma \ref{lem:stretched-ribbon-combinatorial}}
\label{app:pf:lem:stretched-ribbon-combinatorial}

\begin{proof}[Proof of Lemma~\ref{lem:stretched-ribbon-combinatorial}]
    It suffices to consider connected forests, since both the left- and right-hand sides of the statement factorize over connected components.
    Thus, we want to show
    \begin{equation}
        \sum_{\substack{T \in \sT(\ell, m) \\ T \text{ stretched}}} \mu(F) = \One\{\ell = m\} (-1)^{m - 1}(m - 1)!\, m!.
    \end{equation}
    It will be slightly easier to work with a less stringent definition of ``stretched'' which removes the exceptions for skewed stars, and also allows sided stars if the other side has no $\bullet$ vertices.
    Let us call $F$ \emph{weakly stretched} if every terminal $\square$ vertex has a neighbor both in $\sL$ and in $\sR$, or if there is only one $\square$ vertex and one of $\sL$ and $\sR$ is empty.
    Then, our task is equivalent to showing
    \begin{equation}
        \sum_{\substack{T \in \sT(\ell, m) \\ T \text{ weakly stretched}}} \mu(F) = \begin{cases} (-1)^{m - 1}(m - 1)!\, m! & \text{if } \ell = m \geq 1, \\
            -(m - 2)! & \text{if } \ell = 0, m \geq 4 \text{ is even}, \\
            -(\ell - 2)! & \text{if } m = 0, \ell \geq 4 \text{ is even}, \\
            -(m - 1)! & \text{if } \ell = 1, m \geq 3 \text{ is odd}, \\
            -(\ell - 1)! & \text{if } m = 1, \ell \geq 3 \text{ is odd}, \\ 0 & \text{otherwise.}\end{cases}
    \end{equation}
    
    We use the convention that lowercase functions of combinatorial variables, like $f(a, b)$, give coefficients, and uppercase functions of analytic variables, like $F(x, y)$, give the corresponding exponential generating functions.

    Define the coefficients
    \begin{equation}
        c(k) = \left\{\begin{array}{ll} -(k - 2)! & \text{if } k \geq 4 \text{ is even,} \\ 0 & \text{otherwise}.\end{array}\right.
    \end{equation}
    Our goal is to compute the coefficients
    \begin{equation}
        f(\ell, m) = \sum_{\substack{T \in \sT(\ell, m) \\ T \text{ weakly stretched}}} \prod_{v \in V^{\square}(T)} c(\deg(v)),
    \end{equation}
    Equivalently, separating the terminal and non-terminal vertices, we may rewrite with the following intermediate quantities:
    \begin{align}
      g_T(m) &\colonequals \sum_{\phi: [m] \to V^{\square}(T)} \prod_{v \in V^{\square}(T)} c(\deg(v) + |\phi^{-1}(v)|) \text{ for a given } T \in \sT(m), \\
      g(\ell, m) &\colonequals \sum_{\substack{T \in \sT(\ell)}} g_T(m), \\
      h(\ell, m, n, p) &\colonequals \sum_{\substack{\pi \in \mathsf{Part}([\ell + m];\, \mathsf{odd}) \\ S \cap \{1, \dots, \ell\} \neq \emptyset \text{ for all } S \in \pi \\ S \cap \{\ell + 1, \dots, m\} \neq \emptyset \text{ for all } S \in \pi}} \prod_{S \in \pi}(-(|S| - 1)!)\, g(n + |\pi|, p), \\
      f(\ell, m) &= c(\ell + m) + \sum_{a = 0}^{\ell} \sum_{b = 0}^m \binom{\ell}{a} \binom{m}{b} h(a, b, 0, \ell + m - a - b).
    \end{align}
    The first term in the final expression counts the star tree on $[a + b]$ having a single $\square$ vertex.
    In any other tree, every terminal $\square$ vertex must be adjacent to an odd number of leaves, giving the remaining recursion.
    We will calculate the exponential generating functions of the sums appearing, in the same order as they are given above.
    We have introduced a needlessly general version of $h(\cdot, \cdot, \cdot, \cdot)$ in order to make it simpler to close a recursion to come.

    Before proceeding, we compute the exponential generating function of the $c(k)$:
    \begin{equation}
        C(x) = \sum_{k \geq 0} \frac{x^k}{k!}c(k) = -\sum_{k = 2}^{\infty} \frac{x^{2k}}{2k(2k - 1)} = \frac{1}{2}x^2 - \frac{1}{2}(1 + x)\log(1 + x) - \frac{1}{2}(1 - x)\log(1 - x)
    \end{equation}

    Next, to compute the exponential generating function of the $g_T(m)$, note that, grouping by the values of $|\phi^{-1}(v)|$ for each $v$, we may rewrite
    \begin{equation}
        g_T(m) = \sum_{\substack{\bz \in \mathbb{N}^{V^{\square}(T)} \\ |\bz| = m}} \binom{m}{\bz} \prod_{v \in V^{\square}(T)} c(\deg(v) + z_v).
    \end{equation}
    Thus, the generating function factorizes as
    \begin{align}
      G_T(x)
      &= \sum_{m \geq 0} \frac{x^{m}}{m!}g_T(m) \nonumber \\
      &= \prod_{v \in V^{\square}(T)}\left(\sum_{m \geq 0} \frac{x^{m}}{m!} f(\deg(v) + m)\right) \nonumber \\
      &= \prod_{v \in V^{\square}(T)}\frac{d^{\deg(v)}}{dx^{\deg(v)}}\left(\sum_{m \geq 0} \frac{x^{m}}{m!} c(m)\right) \nonumber \\
      &= \prod_{v \in V^{\square}(T)}C^{(\deg(v))}(x).
    \end{align}

    Next, for the $g(\ell, m)$, we have
    \begin{equation}
        G(x, y) = \sum_{\ell \geq 0} \frac{x^\ell}{\ell!} \sum_{\substack{T \in \sF(\ell) \\ T \text{ connected}}} G_T(y).
    \end{equation}
    Let us define
    \begin{equation}
        G_{\ell}(y) = \sum_{\substack{T \in \sF(\ell) \\ T \text{ connected}}} G_T(y).
    \end{equation}
    A tree on $\ell > 1$ leaves can either be a single edge between two leaves (if $\ell = 2$), or will have every leaf connected to an internal vertex.
    In the latter case, let us think of the tree as being rooted as the internal vertex that the leaf labelled $\ell$ is attached to.
    Then, recursively, the tree consists of several rooted trees, whose leaves form a partition of $[\ell - 1]$, attached to the single new root.
    (This is similar to our formalism of ``rooted odd trees'' from Definition~\ref{def:T-root}.)
    Writing this recursive structure in terms of generating functions,
    \begin{equation}
        G_{\ell}(y) = \One\{\ell = 2\} + \sum_{\pi \in \mathsf{Part}([\ell - 1])} C^{(|\pi| + 1)}(y) \prod_{S \in \pi} G_{|S| + 1}(y).
    \end{equation}
    Now, noting that $G(x, y)$ is the exponential generating function of the $G_k(y)$, we use the composition formula.
    This calculation is clearer if we reindex, defining $\widetilde{G}_{\ell}(y) = G_{\ell + 1}(y)$, which satisfy
    \begin{equation}
        \widetilde{G}_{\ell}(y) = \One\{\ell = 1\} + \sum_{\pi \in \mathsf{Part}([\ell])} C^{(|\pi| + 1)}(y) \prod_{S \in \pi} \widetilde{G}_{|S|}(y).
    \end{equation}
    Note that
    \begin{equation}
        \sum_{\ell \geq 0} \frac{x^{\ell}}{\ell!} C^{(\ell)}(y) = C(x + y)
    \end{equation}
    since this is just a Taylor expansion of $C$ about $y$.
    The generating function of $C^{(k + 1)}(y)$ is then the derivative in $x$, which is $C^{\prime}(x + y)$.
    However, we must subtract off the term that is constant in $x$ before using this in the composition formula, giving $C^{\prime}(x + y) - C^{\prime}(y)$.
    Thus, we find by the composition formula
    \begin{equation}
        \widetilde{G}(x, y) \colonequals \sum_{\ell = 0}^{\infty} \frac{x^{\ell}}{\ell!} \widetilde{G}_{\ell}(y) = x + C^{\prime}(\widetilde{G}(x, y) + y)  - C^{\prime}(y).
    \end{equation}
    We have
    \begin{equation}
        C^{\prime}(x) = x - \frac{1}{2}\log(1 + x) + \frac{1}{2}\log(1 - x)
    \end{equation}
    whereby the above functional equation is
    \begin{align*}
      &\widetilde{G}(x, y) = x + \widetilde{G}(x, y) + y - \frac{1}{2}\log(1 + \widetilde{G}(x, y) + y) + \frac{1}{2}\log(1 - \widetilde{G}(x, y) - y) - y \\
      &\hspace{1.95cm} + \frac{1}{2}\log(1 + y) - \frac{1}{2}\log(1 - y),
    \end{align*}
    which, after cancellations, gives
    \begin{equation}
        0 = x - \frac{1}{2}\log(1 + \widetilde{G}(x, y) + y) + \frac{1}{2}\log(1 - \widetilde{G}(x, y) - y) + \frac{1}{2}\log(1 + y) - \frac{1}{2}\log(1 - y),
    \end{equation}
    and exponentiating we have
    \begin{equation}
        e^{-2x} = \frac{(1 - \widetilde{G}(x, y) - y)(1 + y)}{(1 + \widetilde{G}(x, y) + y)(1 - y)} = \frac{1 - \frac{\widetilde{G}(x, y)}{1 - y}}{1 + \frac{\widetilde{G}(x, y)}{1 + y}}
    \end{equation}
    solving which we find
    \begin{equation}
        \widetilde{G}(x, y) = \frac{1 - e^{-2x}}{\frac{1}{1 - y} + \frac{e^{-2x}}{1 + y}} = \frac{(1 - y^2)(1 - e^{-2x})}{1 + e^{-2x} + y(1 - e^{-2x})} = \frac{1 - y^2}{y + \coth(x)}.
    \end{equation}
    Finally, $G(x, y)$ is the integral of this with respect to $x$, with the boundary condition $G(0, y) = 0$.
    This gives
    \begin{equation}
        G(x, y) = \log(\cosh(x)) - xy + \log(1 + y\tanh(x))
    \end{equation}
    Note that, when $y = 0$, we recover the result from the proof of Lemma~\ref{lem:mobius-fn} that the sum over trees of our Mobius function gives the alternating tangent numbers from Example~\ref{ex:mobius-part-even}, whose generating function is $\log(\cosh(x))$.

We next compute the exponential generating function of the $h(\ell, m, n, p)$.
Define the simpler version of these coefficients, without the condition that each subset of the partition $\pi$ intersect both $\{1, \dots, \ell\}$ and $\{\ell + 1, \dots, m\}$:
\begin{equation}
    h(\ell, m, n, p) \colonequals \sum_{\pi \in \mathsf{Part}([\ell + m];\, \mathsf{odd})} \prod_{S \in \pi}(-(|S| - 1)!)\, g(n + |\pi|, p).
\end{equation}
Note that, by decomposing an odd partition into the parts that are contained in $k$, the parts that are contained in $\ell$, and all the other parts, we have
\begin{equation}
    \widetilde{h}(\ell, m, n, p) = \sum_{q = 0}^{\ell}\sum_{r = 0}^{m} \binom{\ell}{q} \binom{m}{r}\sum_{\substack{\pi \in \mathsf{Part}([q]; \mathsf{odd}) \\ \rho \in \mathsf{Part}([r]; \mathsf{odd})}} \prod_{S \in \pi + \rho}(-(|S| - 1)!) \, h(\ell - q, m - r, n + |\pi| + |\rho|, p).
\end{equation}
Now, by the composition formula this implies
\begin{equation}
    \widetilde{H}(w, x, y, z) = H(w, x, y - \tanh^{-1}(w) - \tanh^{-1}(x), z)
\end{equation}
and therefore we can conversely recover $H$ from $\widetilde{H}$ by
\begin{equation}
    H(w, x, y, z) = \widetilde{H}(w, x, y + \tanh^{-1}(w) + \tanh^{-1}(x), z).
\end{equation}
On the other hand, again by the composition formula and the addition formula,
\begin{equation}
    \widetilde{H}(w, x, y, z) = \sum_{m \geq 0} \frac{y^m}{m!} \frac{\partial^m}{\partial t^m}[G(t, z)]_{t = -\tanh^{-1}(w + x)} = G(y - \tanh^{-1}(w + x), z)
\end{equation}
and thus
\begin{equation}
    H(w, x, y, z) = G(y + \tanh^{-1}(w) + \tanh^{-1}(x) - \tanh^{-1}(w + x), z).
\end{equation}

Lastly, by the addition formula we have
\begin{align}
F(x, y) &= C(x + y) + H(x, y, 0, x + y) \nonumber \\
&= C(x + y) + G(\tanh^{-1}(x) + \tanh^{-1}(y) - \tanh^{-1}(x + y), x + y) \nonumber \\
&= \frac{1}{2}(x + y)^2 - \frac{1}{2}(1 + x + y)\log(1 + x + y) - \frac{1}{2}(1 - x - y)\log(1 - x - y) \nonumber \\
  &\hspace{1cm} + \log(\cosh(\tanh^{-1}(x) + \tanh^{-1}(y) - \tanh^{-1}(x + y))) \nonumber \\
  &\hspace{1cm} - (x + y)(\tanh^{-1}(x) + \tanh^{-1}(y) - \tanh^{-1}(x + y)) \nonumber \\
        &\hspace{1cm} + \log(1 + (x + y)\tanh(\tanh^{-1}(x) + \tanh^{-1}(y) - \tanh^{-1}(x + y))),
          \intertext{which after some algebra is equivalent to}
&= C(x) + C(y) - x \tanh^{-1}(y) - y\tanh^{-1}(x) + xy + \log(1 + xy).
\end{align}
Expanding the exponential generating function coefficients of the final line then gives the result.
\end{proof}

\subsection{Partition Transport Ribbon Diagrams: Proof of Lemma \ref{lem:partition-transport-ribbon-combinatorial}}

\begin{proof}[Proof of Lemma~\ref{lem:partition-transport-ribbon-combinatorial}]
    Let us say that a partition transport plan $\bD \in \Plans(\sigma, \tau)$ is \emph{connected} if its diagram $G^{(\sigma, \tau, \bD)}$ is connected, and refer to the connected components of $\bD$, denoted $\conn(\bD)$, as the subsets of $\sigma + \tau$ that belong to each connected component of the diagram.
    As in the previous Lemma, both sides of the result factorize over connected components, so it suffices to show that
    \begin{equation}
        \sum_{\sigma, \tau \in \Part([d])} (-1)^{|\sigma| + |\tau|}\prod_{A \in \sigma + \tau}(|A| - 1)!(|A|)!\sum_{\substack{\bD \in \Plans(\sigma, \tau) \\ \bD \text{ connected}}} \frac{1}{\bD!} = (-1)^{d - 1}(d - 1)! \, d!,
    \end{equation}

    Let us first work with the innermost sum.
    For $\sigma, \tau$ arbitrary partitions, writing $\|\sigma\| = \sum_{S \in \sigma} |S|$ and likewise for $\|\tau\|$, by the inner product calculation from Proposition~\ref{prop:gram-cgm-expansion} we have
    \begin{equation}
        \sum_{\bD \in \mathsf{Plans}(\sigma, \tau)}\frac{1}{\bD!} = \One\{\|\sigma \| = \|\tau\|\}\, \frac{d!}{\prod_{A \in \sigma + \tau} (|A|)!}.
    \end{equation}
    We now compute the restriction to connected $\bD$ using \Mobius\ inversion (a similar calculation to the proof of Lemma~\ref{lem:mobius-fn}).
    For $\pi \in \Part(\sigma + \tau)$, let us define
    \begin{equation}
        b(\pi) \colonequals \sum_{\substack{\bD \in \Plans(\sigma, \tau) \\ \conn(\bD) = \pi}} \frac{1}{\bD!}.
    \end{equation}
    Then, the quantity we are interested in is $b(\{\sigma + \tau\})$.

    We may compute the downward sums of $b(\pi)$ as
    \begin{align}
      c(\pi)
      &\colonequals \sum_{\pi^{\prime}\leq \pi} b(\pi^{\prime}) \nonumber \\
      &= \sum_{\substack{\bD \in \mathsf{Plans}(\sigma, \tau) \\ \conn(D) \leq \pi}} \frac{1}{\bD!} \nonumber \\
      &= \prod_{\rho \in \pi}\left(\sum_{\bD \in \mathsf{Plans}(\sigma \cap \rho, \tau \cap \rho)} \frac{1}{\bD!}\right) \nonumber \\
      &= \prod_{\rho \in \pi} \One\{\|\sigma \cap \rho\| = \|\tau \cap \rho \|\}\,  \frac{(\|\rho\| / 2)!}{ \prod_{R \in \rho} (|R|)!}  \nonumber \\
      &= \One\{\|\sigma \cap \rho \| = \|\tau \cap \rho\| \text{ for all } \rho \in \pi\}\, \frac{\prod_{\rho \in \pi} (\|\rho \| / 2)!}{\prod_{A \in \sigma + \tau} (|A|)!}.
    \end{align}
    Therefore, by \Mobius\ inversion (in the poset $\Part(\sigma + \tau)$),
    \begin{align}
      \sum_{\substack{\bD \in \mathsf{Plans}(\sigma, \tau) \\ \bD \text{ connected}}}\frac{1}{\bD!}
      &= b(\{\sigma + \tau\}) \nonumber \\
      &= \sum_{\pi \in \mathsf{Part}(\sigma + \tau)}\mu_{\Part}(\pi, \{\sigma + \tau\})\, c(\pi) \nonumber \\
      &= \frac{1}{\prod_{A \in \sigma + \tau}(|A|)!} \sum_{\substack{\pi \in \mathsf{Part}(\sigma + \tau) \\ \|\sigma \cap \rho\| = \|\tau \cap \rho\| \text{ for all } \rho \in \pi}}(-1)^{|\pi| - 1} (|\pi| - 1)!\prod_{\rho \in \pi} (\|\rho\| / 2)!
    \end{align}

    Now, we substitute this into the left-hand side in the initial statement:
\begin{align}
  &\sum_{\substack{\sigma, \tau \in \Part([d])}} (-1)^{|\sigma| + |\tau|}\prod_{A \in \sigma + \tau}(|A| - 1)!(|A|)!\sum_{\substack{\bD \in \Plans(\sigma, \tau) \\ \bD \text{ connected}}} \frac{1}{\bD!} \nonumber \\
  &\hspace{0.5cm} = \sum_{\sigma, \tau \in \mathsf{Part}([d])} (-1)^{|\sigma| + |\tau|} \prod_{A \in \sigma + \tau} (|A| - 1)! \sum_{\substack{\pi \in \mathsf{Part}(\sigma + \tau) \\ \|\sigma \cap \rho\| = \|\tau \cap \rho\| \text{ for all } \rho \in \pi}}(-1)^{|\pi| - 1} (|\pi| - 1)!\prod_{\rho \in \pi} (\|\rho\| / 2)! \nonumber
  \intertext{and exchanging the order of summation,}
  &\hspace{0.5cm} = \sum_{\substack{\pi \in \mathsf{Part}([2d]) \\ |R \cap \{1, \dots, d\}| = |R \cap \{d + 1, \dots, 2d\}| \\ \text{for all } R \in \pi}} (-1)^{|\pi| - 1}(|\pi| - 1)! \prod_{R \in \pi}(|R| / 2)! \nonumber \\
  &\hspace{3.5cm} \sum_{\substack{\sigma \in \Part(\{1, \dots, d\}) \\ \tau \in \Part(\{d + 1, \dots, 2d\}) \\ \sigma + \tau \leq \pi}} (-1)^{|\sigma| + |\tau|} \prod_{A \in \sigma + \tau}(|A| - 1)!.
\end{align}

We again think in terms of \Mobius\ functions, but now on a different poset: on the product poset $\Part(\{1, \dots, d\}) \times \Part(\{d + 1, \dots, 2d\})$, the \Mobius\ function is (see \cite{Rota-1964-Foundations} for this fact)
\begin{equation}
    \mu_{\Part \times \Part}((\{\{1\}, \dots, \{\ell\}\}, \{\{\ell + 1\}, \dots, \{2\ell\}\}), (\sigma, \tau)) = (-1)^{|\sigma| + |\tau|} \prod_{A \in \sigma + \tau} (|A| - 1)!,
\end{equation}
and $\{(\sigma, \tau): \sigma + \tau \leq \pi\}$ is an interval.
Therefore, the inner sum above is zero unless every part of $\pi$ has size two, so that $\pi$ is a matching of $\{1, \dots, d\}$ with $\{d + 1, \dots, 2d\}$, in which case the inner sum is 1.
There are $d!$ such matchings $\pi$, so we find that the above expression equals $(-1)^{d - 1} (d - 1)! \, d!$, giving the result.
\end{proof}

\end{document}